\documentclass[a4paper,12pt]{amsart}
\usepackage{kantlipsum} 
\calclayout
\usepackage{setspace}
\usepackage[papersize={8.5in,11in}]{geometry}
\usepackage{enumerate}
\usepackage{enumitem}
\usepackage{amsmath}    
\usepackage{amssymb}    
\usepackage{bm}         
\usepackage{graphicx}   
\usepackage{color}      
\usepackage{hyperref}   
\usepackage{natbib}    
\usepackage{epsf}
\usepackage{lscape}
\bibpunct{(}{)}{;}{a}{,}{,}
\usepackage{caption}
\usepackage{colortbl}

\numberwithin{equation}{section}
\theoremstyle{plain}

\theoremstyle{plain}

\newtheorem{theorem}{Theorem}[section]

\usepackage[norelsize]{algorithm2e}
\usepackage{bigints}
\usepackage{amsmath} 
\theoremstyle{remark}

\usepackage{enumitem}

\begin{document}



\title[A Bayesian approach for the completeness of death registration]{\textbf{A Bayesian approach  to estimate the completeness of death registration}}

%


\author[Jairo F\'uquene-Pati\~no  and  Tim Adair]{Jairo F\'uquene-Pati\~no$^{\dagger}$ and Tim Adair$^{\ddagger}$. \\
\scriptsize Assistant Professor. Department of Statistics. UC Davis, USA. (jafuquenep@ucdavis.edu.co, corresponding author). $^{\dagger}$.\\
\scriptsize The Nossal Institute for Global Health, Melbourne School of Population and Global Health, The University of Melbourne, Australia, timothy.adair@unimelb.edu.au. $^{\ddagger}.$}


\maketitle







\maketitle

\begin{center}
\textbf{Abstract}

\end{center}
\setstretch{2}

Civil registration and vital statistics (CRVS) systems should be the primary source of mortality data for governments. Accurate and timely measurement of the completeness of death registration helps inform interventions to improve CRVS systems and to generate reliable mortality indicators. In this work
we propose the use of hierarchical Bayesian linear mixed models with Global-Local (GL) priors to estimate
the completeness of death registration at global, national and subnational levels. 
The use of GL priors  in this paper is motivated for situations where demographic covariates can explain much of  the observed completeness but where unexplained within-country (i.e. by year) and between-country
variability also play an important role. The use of
our approach can allow institutions improve model parameter estimates
and more accurately predict completeness of death registration. 
Our models are based on a  dataset which uses Global Burden of Disease (GBD) death estimates based on the GBD 2019 and comprises 120 countries and 2,748 country-years from 1970-2019  \citep{collaborators2020global}. To illustrate the effectiveness of our
proposal  we consider the completeness of
death registration in the departments of Colombia in 2017.

\vspace*{.3in}

\noindent \textsc{Keywords}: Completeness of death registration, Unit-Level models, Global-Local priors, Mortality, Vital statistics, Data quality.

\section{Introduction}

The primary source of mortality data for national governments should be a complete civil registration and vital statistics (CRVS) \citep{mikkelsen2015global}. However, only  59\% of all deaths that occur globally are estimated to be registered, with almost all unregistered deaths occurring in low- and middle-income countries \citep{dicker2018global}, hence this severely limits the utility of this data source to provide evidence on mortality trends to inform government policymaking. Reliable measurement of the completeness of death registration – that is, the percentage of deaths that are registered – is important firstly to measure the extent of under-registration of deaths and to inform interventions that aim to attain complete registration. The completeness of registration can also be used to adjust data from a death registration system and, together with model life tables, produce estimates of key mortality indicators.

Various methods to estimate completeness of death registration have been developed over several decades. A group of demographic methods named death distribution methods (DDMs) estimate completeness at ages five years and above by assessing the internal consistency of the age pattern of population and registered deaths, and make assumptions of population dynamics  \citep{brass1975methods, preston1980estimating, bennett1984mortality, hill1987estimating, adair2018estimating, dorrington2013brass, dorrington2013generalized, dorrington2013synthetic}. DDMs either measure completeness of death registration assuming a stable population (i.e., a constant population growth rate and no migration) and using population data from a specific point in time, or assuming a closed population (i.e., no migration) and using population from two censuses to enable measurement of completeness of inter-censal death registration. The most reliable DDM has been identified as a hybrid approach using the Generalised Growth Balance and Bennett-Horiuchi methods to estimate inter-censal death registration, together with the use of age trims to minimize the impact of the assumption of no migration \citep{murray2010can}  on the accuracy of completeness.
However, DDMs have several limitations, including inaccuracy of estimates and limited applicability at the subnational level because of the incompatibility of the restrictive assumptions to contemporary population dynamics, and the lack of timeliness of estimates where they are made for the period between the most recent two censuses \citep{adair2018estimating}. Another set of methods to estimate completeness are capture-recapture methods, which involved linking of death registration data with another independent data source \citep{rao2017overview}. However, limitations of these methods are the availability of another independent data source, and the time required to link data sources where data quality is suboptimal. A more straightforward means of calculating completeness is to divide the number of registered deaths by the number of deaths in the population estimated by either the Global Burden of Disease (GBD) Study or United Nations (UN) World Population Prospects
\citep{wang2020global}. However, with the exception of very few countries where the GBD estimates deaths for subnational populations, these estimates are only available nationally.

More recently, the empirical completeness method has been proposed to estimate completeness of death registration \citep{adair2018estimating}. The method, which was developed from an extensive database of 110 countries from 1970 to 2015 reported in the Global Burden of Disease 2015 database, models completeness based on the key drivers of a population’s crude death rate (the number of deaths per 1,000 population). The empirical completeness method, compared with DDMs, uses only relatively limited data (and which are readily available at subnational levels), provides completeness estimates for the most recent year of death registration data, and is not reliant on assumptions about population dynamics such as closed migration that adversely affect subnational estimates. A limitation of the method, however, is that it is less reliable where the level of adult mortality is relatively  high given the level of   under-five mortality, such as where HIV/AIDS mortality is high or there has been a mortality shock such as a natural disaster. The method has been applied in several settings, including to estimate excess mortality from COVID-19 in Peru, to calculate subnational completeness for Indian states and 2,844 Chinese counties, to monitor completeness of community death reporting systems in Bangladesh, and to estimate cause-specific mortality rates to track attainment of Sustainable Development Goals in Myanmar \citep{zeng2020measuring, sempe2021estimation, adair2021monitoring, basu2021have, shawon2021routine}.

The probabilistic empirical method proposed by \citep{adair2018estimating} considers  two linear mixed models  with random effects at the country level. The covariates
in the model of \citep{adair2018estimating} include information typically available from multiple sources, e.g., surveys, censuses and administrative records. Specifically, the first model includes predictors with information about the registered crude death rate, fraction of the population aged 65 years and over, under-five mortality rate and calendar year of death, while
the second model also includes also covariate of the completeness of registration for children ages less than five years. \citep{adair2018estimating}  also consider models for males and females
to allow predictions of  completeness in each case.  

The modeling of demographic components using Bayesian methods is a growing and active research area with important recent  advances, for example in population projections  \citep{raftery2014bayesian, wheldon2016bayesian}, international migration flows \citep{azose2015bayesian},
projections of life expectancy \citep{godwin2017bayesian} and estimation of maternal mortality \citep{alkema2017bayesian}. Although the use of new Bayesian methods for demographic indicators is extensive, there have not been any attempts to use these models to estimate the completeness of death registration. Therefore, this paper adapts the empirical completeness method \citep{adair2018estimating} to a Bayesian framework, using first a common prior for the scale of the errors and random effects respectively and extending this to a more flexible Bayesian modeling  that considers Global-Local priors (GL) for the scales of the errors and random effects respectively.  The use of GL priors  in this paper is motivated for situations where demographic covariates can explain much of  the observed completeness but where unexplained within-country (i.e. by year) and between-country
variability also play an important role.

The new models are based on a dataset updated to 2019, which uses GBD death estimates based on the GBD 2019 and now comprises 120 countries and 2,748 country-years from 1970-2019  \citep{collaborators2020global}.  The dataset calculates observed completeness for a country-year as the number of registered deaths divided by the GBD’s estimate of total deaths. The GBD estimates total deaths by estimating under-five mortality and adult mortality rates and then using model life tables to develop complete life tables; further detail of their methods can be found in GBD 2019 Demographics Collaborators (2020). Consistent with \cite{adair2018estimating}, some country-years were excluded from the database because they either experienced mortality shocks (e.g. natural disaster, conflict) or high adult mortality compared with child mortality due to, for example, HIV/AIDS or alcohol (e.g. sub-Saharan Africa, Russia in the 1990s). In such countries, the under-five mortality rate is not a reliable proxy for mortality level because of the high level of adult mortality.

To assess the performance of the proposed models to national and subnational levels we consider the completeness of death registration for the  country of
Colombia and its 33 departments in 2017 for both sexes, males and females based on a question in the Colombian Population Census  2018 \citep{DANE2018} that asks if households  experienced a death during the
calendar year 2017 and whether that death was registered with the civil registration system. This comparator dataset is not included in GBD 2019 data used to develop the models.

We found the Census data can be regarded as a reliable independent estimate of completeness in departments because its estimate of
 national completeness of 90\% (i.e. the \% of reported household deaths that were reported as being registered with the civil registration system) is close to the completeness of 85\% calculated by using vital registration data as the numerator and UN estimated deaths for Colombia as the denominator \citep{desa2019united}. Also, an estimated 90\% of deaths in 2017 in Colombia were reported in the Census \citep{desa2019united}.

This paper is organized as follows. Section  \ref{Methodology} introduces the methodology implemented in this work. Next, we introduce
GL priors for the scales of the random effects and errors respectively in hierarchical linear mixed models and study their theoretical properties to allow
small and large random effects when it is required. Section \ref{Computational} contains the  Markov chain Monte Carlo (MCMC) schemes for posterior
inference of the parameters and the prior elicitation. In Section \ref{Results} we implement our proposal to estimate the completeness of the death registration at the global level and consider a prior sensitivity study. The application of the models  to national and subnational levels in Colombia are given in Section \ref{Results2}. Finally, the
concluding remarks and discussion are presented in Section \ref{Conclusions}.



\section{Methodology}
\label{Methodology}



\subsection{Linear mixed models with Global-Local priors for  random effects and errors}


Consider the proposed model

\begin{align}
y_{ij}&= \theta_{ij}    + \epsilon_{ij},  & j&=1,...,n_{i}, \notag \\
\theta_{ij}&= \boldsymbol{x}_{ij}^\textsf{T}\boldsymbol{\beta} + u_{i}, & i&=1,...,m,
\label{model}
\end{align}

where $n_{i}>p=\dim( \boldsymbol{x}_{ij})$ and the errors $\epsilon_{ij}$ and random effects  $u_{i}$ are independent and normally distributed   
and $\epsilon_{ij} \overset{\mathrm{ind}}{\sim} \text{Normal}(\epsilon_{ij};0,1/(\lambda_{i}\tau))$  and $u_{i} \overset{\mathrm{iid}}{\sim} \text{Normal}(u_{i};0,1/(\omega_{i}\phi))$ and where $\lambda_{i}$
and $\omega_{i}$ denote the Local scales and $\tau$ and $\phi$ denote the Global scales for the errors and random effects, respectively. In this work $y_{ij}=\log(c_{ij}/(1-c_{ij}))$ denote the  logit of the completeness of registration at all ages for country $i$ and  year $j$.  GL priors where originally proposed by \cite{carvalho2010horseshoe} and \cite{polson2012half}. These priors have been implemented in an important number of proposals 
for  variable selection,  estimation  and model selection (e.g., \cite{armagan2013generalized, ghosh2016asymptotic, tang2018bayesian, bhattacharya2015dirichlet, bhadra2017horseshoe+, bai2018high, tang2018modeling}. 
  The global parameter  in model (\ref{model}) controls the total shrinkage of the random effects (posterior mean) towards zero (country-mean) and the local parameter allows  random effects (posterior mean) to avoid this shrinkage.   In the context of model (\ref{model}) the local priors produces a conditional posterior distribution for the random effects (posterior mean) with a peak at zero to shrink small random effects (posterior mean) and heavy tails  to avoid this shrinkage when required.  In this work we consider, for the first time to the best of our knowledge, the use of GL priors for the random effects and errors in the linear mixed model given by (\ref{model}). As we will discuss, the use of GL priors is motivated for the applied site where demographic covariates can explain the completeness of death registration  well but variability of the data at the country and country-year levels play an important role.  The list of shrinkage priors in the literature is extensive. To evaluate the performance of  different GL priors in model (\ref{model}) to estimate the completeness of death registration we use
 local shrinkage behaviors with different  shrinkage behaviors around the origin and tails. Specifically, we consider the Horseshoe \citep{carvalho2010horseshoe},  Laplace  \citep{bernardo2011bayesian, brown2010inference}  and local Student-t  \citep{griffin2005alternative} priors for $\omega_{i}$. The priors are illustrated in Table \ref{tab:priors} where HS and Half-Cauchy priors belong to the family of polynomial heavy tailed priors while the LA priors correspond to the family of exponential heavy tailed priors. We also consider the local Student-t prior with a weakly informative prior for the degrees of freedom we present in Section \ref{Computational}. Although both the HS and local Student-t priors produce a conditional distributions for the random effects with polynomial tails,  the conditional distribution under the HS is more peaked around zero than the conditional distribution obtained with the local Student-t prior. As a result, 
more shrinkage towards zero for small random effects  is produced under the HS than under the local Student-t. For the local scale of the errors $\lambda_{i}$ we consider the Half-Cauchy prior proposed by  \cite{polson2012half}  for top-levels scale parameters in Bayesian hierarchical models.

{\renewcommand{\arraystretch}{1.2}
\begin{table}[h]
\scriptsize
\begin{center}
\begin{tabular}{lccc}
\hline \hline
Name & $\pi_{\omega}(x)$ & Reference & Tail\\
\hline
Laplace & $x^{-2}\exp(-1/x)$ & \cite{bernardo2011bayesian, brown2010inference}    & Exponential\\
Horseshoe & $x^{-1/2}(1+x)^{-1}$ &  \cite{carvalho2010horseshoe} & Polynomial\\
Local Student-t & $x^{\nu/2-1}\exp(-x\nu/2)$ & \cite{griffin2005alternative}  & Polynomial\\  \hline \hline
 & $\pi_{\lambda}(x)$ & Reference & \\ \hline
Half-Cauchy  & $(1+x)^{-2}$ & \cite{polson2012half} & Polynomial\\ \hline \hline
\end{tabular}
\caption{Priors distributions  $\pi_{\omega}(x)$ and  $\pi_{\lambda}(x)$ for the local scales of the random effects and errors respectively. Tails refer to the
tail behavior of the marginal distribution for $u_{i}$ and $\boldsymbol{\theta}_{i}$.}
\label{tab:priors}
\end{center}
\end{table}

We assume
non-informative Gamma priors for each of the global scales, $\tau$ and $\phi$.  To compare the performance of our proposal we also consider non-informative Gamma priors under the assumption of a common scale $\zeta_{\epsilon}$ for the errors,  $\epsilon_{ij} \overset{\mathrm{ind}}{\sim} \text{Normal}(\epsilon_{ij};0,1/\zeta_{\epsilon})$. In addition, in order to compare the frequentist version of model  (\ref{model}) we also consider a common scale  $\zeta_{u}$ for the random effects, $u_{i} \overset{\mathrm{iid}}{\sim} \text{Normal}(u_{i};0,1/\zeta_{u})$,  and a common scale $\zeta_{\epsilon}$ for the errors, respectively. We consider the same covariates given in \cite{adair2018estimating} in our proposed Bayesian models given in (\ref{model}) to estimate the completeness of death registration for both sexes (males and females) as follow

\begin{align}
\label{models}
\text{Model 1:}& \;\;\; x_{ij}=(1,\text{RegCDR}_{ij},  \text{RegCDR}_{ij}^{2}, \%65_{ij}^{2}, \text{ln(5q0)}_{ij},  \text{C5q0}_{ij}, \text{year}_{ij}), \notag \\
\text{Model 2:}& \;\;\; x_{ij}=(1,\text{RegCDR}_{ij},  \text{RegCDR}_{ij}^{2}, \%65_{ij}^{2}, \text{ln(5q0)}_{ij},  \text{year}_{ij}),
\end{align}

where RegCDR is the registered crude death rate (registered deaths divided by population multiplied by 1000), RegCDR$^{2}$ is the registered crude death rate squared, \%65 is the fraction of the population aged 65 years and over, ln(5q0) is the natural log of the estimate of the true under-five mortality rate, C5q0 is the is the completeness of registered under-five deaths (estimated as the under-five mortality rate from registration data divided by the estimate of the true under-five mortality rate), and year is calendar year of death.  As is pointed out by  \cite{adair2018estimating},  the  under-five mortality rate and the population aged 65 years and over are important drivers of the  registered crude death rates in populations. We also consider in  (\ref{models}) models for males and females independently with the same predictors but sharing the same information for C5q0 which are likely to be similar for males and females as is pointed out by  \cite{wang2014global}  and successfully implemented in \cite{adair2018estimating}.

The joint prior distribution for the parameters $\boldsymbol{\beta}$, $\boldsymbol{\omega}=(\omega_{1},...,\omega_{m})^\textsf{T}$,
$\boldsymbol{\lambda}=(\lambda_{1},...,\lambda_{m})^\textsf{T}$, $\tau$ and $\phi$ is given as follows:

\begin{align}
\pi(\boldsymbol{\beta},\boldsymbol{\lambda},\tau,\boldsymbol{\omega},\phi)  \propto \pi(\phi) \pi(\tau)  \prod_{i=1}^{m} \pi(\omega_{i}) \pi(\lambda_{i}),
\label{eq:LG2}
\end{align}
where the conditional posterior mean of  the random effects $u_{i}$  is given by
\begin{align}
E(u_{i}|\boldsymbol{\beta},\boldsymbol{\lambda},\tau,\boldsymbol{\omega},\phi,\boldsymbol{y}_{i})=\gamma_{i}(\bar{\boldsymbol{y}}_{i}-\bar{\boldsymbol{x}}_{i}^\textsf{T}\boldsymbol{\beta}_{i}),
\label{eq:shrink}
\end{align}
and $\gamma_i = (\lambda_{i}\tau)/(\lambda_{i}\tau  + \omega_{i}\phi/n_{i})$ is the shrinkage factor associated with the conditional posterior mean of the random effect $u_{i}$. According
to (\ref{eq:shrink}) the shrinkage factors $\gamma_i \in (0,1)$  control the values of the random effects and measure the unexplained between-country
variability $(\omega_{i}\phi)^{-1}$ and the within-country-year $(\lambda_{i}\tau)^{-1}/n_{i}$ variability. If the unexplained between-country variability is relative small  (i.e., $\phi$ large) then  the shrinkage factors $\gamma_i$ will be closer to zero and therefore the random effects will be smaller. On the other hand, given the scale of the random effects, if the within-country-year variability is relative small (i.e., $\tau$ large) then  the shrinkage factor will concentrate near one, leading to random effects obtained by the difference of the logit of the observed completeness and  the regression fit, i.e., $\bar{\boldsymbol{y}}_{i}-\bar{\boldsymbol{x}}_{i}^\textsf{T}\boldsymbol{\beta}_{i}$.
Then to produce smaller or larger random effects according  to the quality of the regression fit and also to the unexplained between country and within country-year variabilities we should consider GL priors for the scales of the random effects and errors, respectively.  This is   theoretically supported in the following Theorem. 


\begin{theorem}\label{th2}
Under  model (\ref{model}), the shrinkage factor $\gamma_i$ has the following properties
\begin{enumerate}[label=(\roman*)]         
\item Suppose $\pi(\omega_{i})$ is a proper probability density function with support $(0,\infty)$,  for any $\epsilon \in (0, 1)$ if $\phi \rightarrow \infty$
         $$P(\gamma_i  > \epsilon \mid \boldsymbol{\beta},\boldsymbol{\lambda},\tau,\boldsymbol{y}_{i}) \rightarrow 0.$$
         Specifically,  
\begin{itemize}    
\item  For the Beta-Prime distribution, $\pi(\omega_{i})  = \dfrac{\Gamma(a+b)}{\Gamma(b)\Gamma(a)} \dfrac{\omega_{i}^{b-1}}{(1+\omega_{i})^{a+b}}$ with hyperparameters $0<a\leq 1$ and $0<b\leq 1$
  $$P(\gamma_i  > \epsilon \mid \boldsymbol{\beta},\boldsymbol{\lambda},\tau,\boldsymbol{y}_{i}) \asymp (1/\phi)^{b}$$ 
as $\phi \rightarrow \infty,$ 
\item  For the LA prior,  $\pi(\omega_{i}) = (\omega_{i})^{-2}\exp(-1/\omega_{i})$
$$P(\gamma_i  > \epsilon \mid \boldsymbol{\beta},\boldsymbol{\lambda},\tau,\boldsymbol{y}_{i})\asymp \exp(-\phi/c_{1})$$
as $\phi \rightarrow \infty.$
\end{itemize}
    \item Suppose $\pi(\lambda_{i})$ is the either the LA or  Beta-Prime distribution with $a=1$ and $b \in (0,  \infty)$,  for any $\epsilon \in (0, 1)$
    $$P(\gamma_i < \epsilon \mid \boldsymbol{\beta},\boldsymbol{\omega},\phi,\boldsymbol{y}_{i}) \rightarrow 0$$ as $\tau \rightarrow \infty$.
       \item Suppose $\pi(\omega_{i})$ is a proper probability density function with support $(0,\infty)$,  for any $\epsilon > 0$,  $$P(\gamma_i  < \epsilon \mid \boldsymbol{\beta},\boldsymbol{\lambda},\tau,\boldsymbol{y}_{i}) \rightarrow 0,$$  as $|\bar{\boldsymbol{y}}_{i}-\bar{\boldsymbol{x}}_{i}^\textsf{T}\boldsymbol{\beta}_{i}| \rightarrow \infty$.
\end{enumerate}
\end{theorem}
\begin{proof}  The proof is in the supplementary material, Section \ref{proof2}.
\end{proof}

In Theorem \ref{th2}  for any real-valued functions $f$ and $g$ we write $f(\xi) \asymp g(\xi)$ for $\xi\rightarrow \infty$
when $K_{1}\leq \liminf_{\xi\rightarrow\infty} f(\xi)/g(\xi) \leq \limsup_{\xi\rightarrow\infty} f(\xi)/g(\xi)  \leq K_{2}$.  Theorem \ref{th2} part (i) shows that
given the scale  $1/(n_{i}\lambda_{i}\tau)$ if the global scale (variance) for the random effects is large  (small) then the posterior distribution of the shrinkage factor will concentrate near of zero. The rate of concentration of $\gamma_i$ near to one is polynomial for the HS prior and exponential for the LA prior. On the other hand, according to  (ii) in Theorem \ref{th2}, given the scale $1/(\omega_{i}\phi)$ if the global scale (variance) for the errors is large  (small) then the posterior distribution of the shrinkage factor will concentrate near of one. In addition, part (iii) of Theorem \ref{th2} shows
that the shrinkage can be offset for large differences of logit completeness and regression fit at the country-level, i.e., $|\bar{\boldsymbol{y}}_{i}-\bar{\boldsymbol{x}}_{i}^\textsf{T}\boldsymbol{\beta}_{i}| \rightarrow \infty$.  

The methodological proposal in this paper is fully general and can be applied not only to the important problem of estimating the completeness of death registration but also  to other relevant problems involving linear mixed models with GL priors as in (\ref{model}).


 \section{Computational Algorithms and prior formulation}

\label{Computational}

Consider $\boldsymbol{\xi}=(\boldsymbol{\beta} , \phi, \boldsymbol{\omega},\tau, \boldsymbol{\lambda},
\boldsymbol{u})$ and the full joint posterior  for  model (\ref{model}) using the prior formulation in (\ref{eq:LG2}) as follow

\begin{align}
\label{poster}
p(\boldsymbol{\xi}\mid \boldsymbol{y})&\propto \prod_{i=1}^{m}\prod_{j=1}^{n_{i}}  (\lambda_{i}\tau)^{1/2}
\exp\left\{-\frac{1}{2}(\lambda_{i}\tau)(y_{ij}-\boldsymbol{x_{ij}}^{T}\boldsymbol{\beta} - u_{i})^{2}\right\}  \notag  \\   &  \hspace{+3.5cm}  \times
(\omega_{i}\phi)^{1/2}\exp\left\{-\frac{1}{2}u_{i}^{2}\omega_{i}\phi\right\} \pi(\phi)\pi(\tau)\pi(\omega_{i})\pi(\lambda_{i}),
\end{align}

The full conditionals distributions for the parameters $\boldsymbol{\xi} = (\boldsymbol{\beta} , \phi, \boldsymbol{\omega},\tau, \boldsymbol{\lambda}, \boldsymbol{u})$  
are obtained in closed form for  and Gibbs Sample algorithms \citep{gelfand1990sampling} can be implemented.  Algorithms \ref{MCMC1} to \ref{MCMC3} display the MCMC schemes  to obtained samples from the full conditional distributions of the parameters in (\ref{poster}). Algorithm \ref{MCMC1}  (See supplementary material, Section \ref{MCMCsec}) contains the full conditionals when  a common Gamma prior is considered for the scales of the errors and GL priors are considered for the scales of the random effects.
Algorithm~\ref{MCMC2} (See supplementary material, Section \ref{MCMCsec}) shows the scheme when a Half-Cauchy prior  is considered for the local scale of the errors.  Algorithm \ref{MCMC3} (See supplementary material, Section \ref{MCMCsec}) contains the steps to sampling from the  full conditionals when HS, LA and Student-t local scale priors for the random effects are implemented.
We consider Gamma distributions for the global-scale parameters $\phi$ and $\tau$ where  
 $\phi \sim \text{Gamma}(\phi;  a_{\phi}, b_{\phi})$ and  $\tau \sim \text{Gamma}(\tau; a_{\tau}, b_{\tau})$.    For $\phi$ and $\tau$ we assign small values  for the hyperparameters following the results in  \cite{tang2018modeling}, then
$ a_{\phi} = b_{\phi} = a_{\tau} = b_{\tau} = 10^{-10}$. In the common Gamma priors for the scales of the random effects and errors denoted by  $\zeta_{\epsilon}$ and
$\zeta_{u}$ respectively we also consider $ a_{\zeta_{\epsilon}} = b_{\zeta_{\epsilon}} = a_{\zeta_{u}}= b_{\zeta_{u}}= 10^{-10}$.

The conditional distributions of  the local-scales of the random effects  are obtained in closed form given that the full conditional can be 
obtained in closed form, i.e., $p(\omega_{i} \mid  \boldsymbol{\beta} , \phi, \boldsymbol{\omega},\tau, \boldsymbol{\lambda}) \propto  (\omega_{i}\phi)^{1/2}\exp\{-\frac{1}{2}u_{i}^{2}\omega_{i}\phi\} \pi(\omega_{i})$. Specifically, the conditional distribution under the LA is  a Generalized Inverse Gaussian  (GIG) distribution and  posterior samples under the HS are obtained using Gamma distributions in two steps.  The local Student-t prior
leads to a Student-t conditional distribution for the random effects with $\nu_{i}$ degrees of freedom and scale $\phi$.  There is a large number of proposals
for the prior distribution of the degrees of freedom of the Student-t distribution in the literature (see for instance, \cite{villa2014objective, rubio2015bayesian}). However, we consider in this work the proposed gamma-gamma prior in \cite{juarez2010model} also implemented by \cite{rossell2019continuous}  where the prior for $\nu_{i}$ is  weakly informative and given by $\pi(\nu_{i}) \propto  \nu_{i}/(\nu_{i} + k_{\nu_{i}})$. We use  $k_{\nu_{i}}=2.84$ as in \cite{rossell2019continuous} in order $\pi(\nu_{i})$ assigns a prior probability of $P(\nu_{i} \leq 2) =0.25$ and a prior median equal to 5. We also consider
a discrete support with prior probabilities $\nu_{i} \in \{1,...,30\}$ as in \cite{rossell2019continuous}  where samples from $\nu_{i} $ are obtained using probabilities proportional to
$(\nu_{i} /(k_{\nu_{i}}+\nu_{i} ))^{3}\text{Student-t}_{\nu_{i}^{(t)}=k}(u^{(t)}_{i},0,\sqrt{1/\phi^{(t)}})$  as is illustrated in algorithm \ref{MCMC3} (See supplementary material, Section \ref{MCMCsec}).

\begin{theorem}\label{th1}
The posterior distribution resulting from the GL priors for the scales of the random effects and errors respectively in model (\ref{model}) is proper if
the priors $\pi(\phi)$, $\pi(\tau)$, $\pi(\omega_{i})$ and $\pi(\lambda_{i})$ are proper.
\end{theorem}
\begin{proof}
The proof is in the supplementary material, Section \ref{proof1}.
\end{proof}

Because we consider an improper uniform prior for $\boldsymbol{\beta}$ we need to verify if the resulting posterior in (\ref{poster}) is proper. Theorem \ref{th1}
shows that if $\pi(\phi)$, $\pi(\tau)$, $\pi(\omega_{i})$ and $\pi(\lambda_{i})$ are proper then the resulting posterior is proper. Then our proposal is not limited
to the suitable priors considered in Table \ref{tab:priors}
 and can implemented to other shrinkage  GL priors with the only requirement of having  local and  global proper prior distributions. 

The posterior estimates of the completeness for all country-years are obtained considering the inverse of the logit proportions,  $\delta_{ij}=e^{\theta_{ij}}/(1+e^{\theta_{ij}})$.  Then, for each fit we estimate $\delta_{ij}$ using its posterior mean   $\hat{\delta}_{ij}$  obtained from the MCMC
samples in Algorithms \ref{MCMC1}-\ref{MCMC3} (See supplementary material, Section \ref{MCMCsec}).   To evaluate the performance of the models and GL priors we consider two deviance measures, the Mean Absolute Error (MAE) and the Root Mean Square Error (RMSE) given by,

\begin{align}
\text{MAE} &= \dfrac{1}{n}\sum_{i=1}^{m} \sum_{j=1}^{n_{i}} \mid \hat{\delta}_{ij} - c_{ij} \mid, & \text{RMSE} &= \sqrt{\dfrac{1}{n}\sum_{i=1}^{m} \sum_{j=1}^{n_{i}}  (\hat{\delta}_{ij} - c_{ij} )^{2}}.
\end{align}

We also evaluate the goodness of fit of the models under GL priors using an estimate of the R-square  as follow,

\begin{equation}
\text{R-square}=1-\dfrac{\sum_{i=1}^{m}\sum_{j=1}^{n_{i}} (c_{ij}-  \bar{c}_{ij})^{2}}{\sum_{i=1}^{m}\sum_{j=1}^{n_{i}} (c_{ij}-  \hat{\delta}_{ij,-u} )},
\end{equation}

where the posterior predictive mean $\hat{\delta}_{ij,-u}$ is computed by excluding the country level-random effects and $\bar{c}_{ij}$ is the mean of the observed completeness, $c_{ij}$. To asses
the posterior estimates by different levels of observed completeness  we stratified the MAE and RMSE  results as follow

\begin{align}
\text{MAE}_{k} &= \dfrac{1}{n_{j}} \sum_{\{(i,j):\varrho_{i}=k\}} \mid \hat{\delta}_{ij} - c_{ij} \mid, & \text{RMSE}_{k} &= \sqrt{\dfrac{1}{n_{j}} \sum_{\{(i,j):\varrho_{i}=k\}} (\hat{\delta}_{ij} - c_{ij} )^{2}},
\end{align}

where $n_{j}$ is the cardinality of  $\{(i,j):\varrho_{i}=k\}$ and  $k$ indicates the observed level of completeness, where $k \in \{(0,30\%)$, $[30\%, 60\%)$,   $[60\%,80\%)$,  $[80\%,90\%)$ and $[90\%,100\%]\}$.

\section{Results: implementing the models at the Global level}
\label{Results}

Tables  \ref{tab:tabglobal1} to \ref{tab:tabglobal6} in the supplementary material in Section \ref{sup_post} show the posterior mean estimates for the regression parameters in the models and the corresponding credible intervals. According to the credible intervals, the covariates are good predictors of the logit of completeness of the death registration. Positive and negative values of the parameters are as expected and consistent with those of \cite{adair2018estimating}. For instance, the under-five mortality rates and percentage of the population over the age of 65 have a negative relationship with the observed logit of the completeness (i.e., as these variables are higher, expected deaths increase and therefore completeness declines).
\begin{figure}[ht]
\small
\begin{center}
\begin{tabular}{ccc}
Both sexes (Half Cauchy) & Females (Half Cauchy) & Males (Half Cauchy)\\
\includegraphics[width=0.3\textwidth]{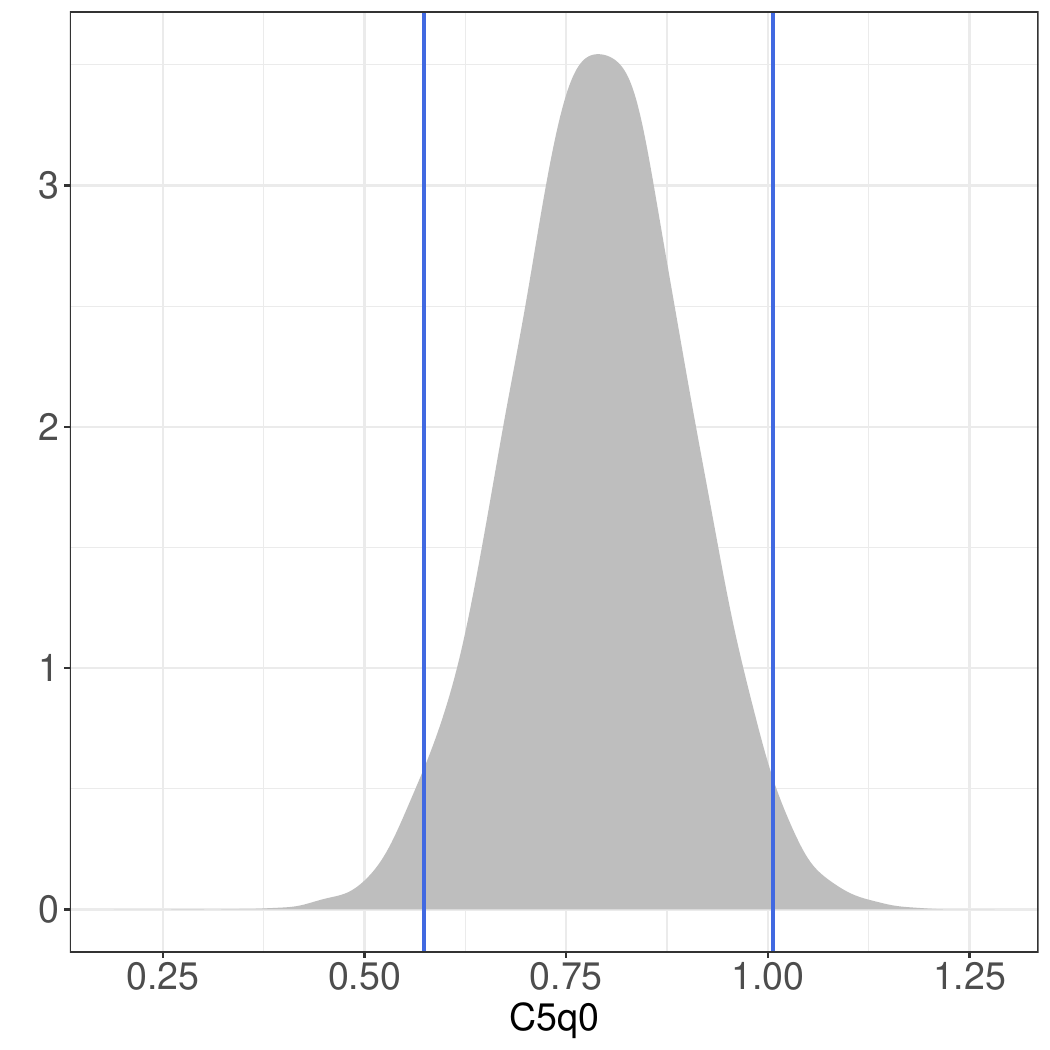} &
\includegraphics[width=0.3\textwidth]{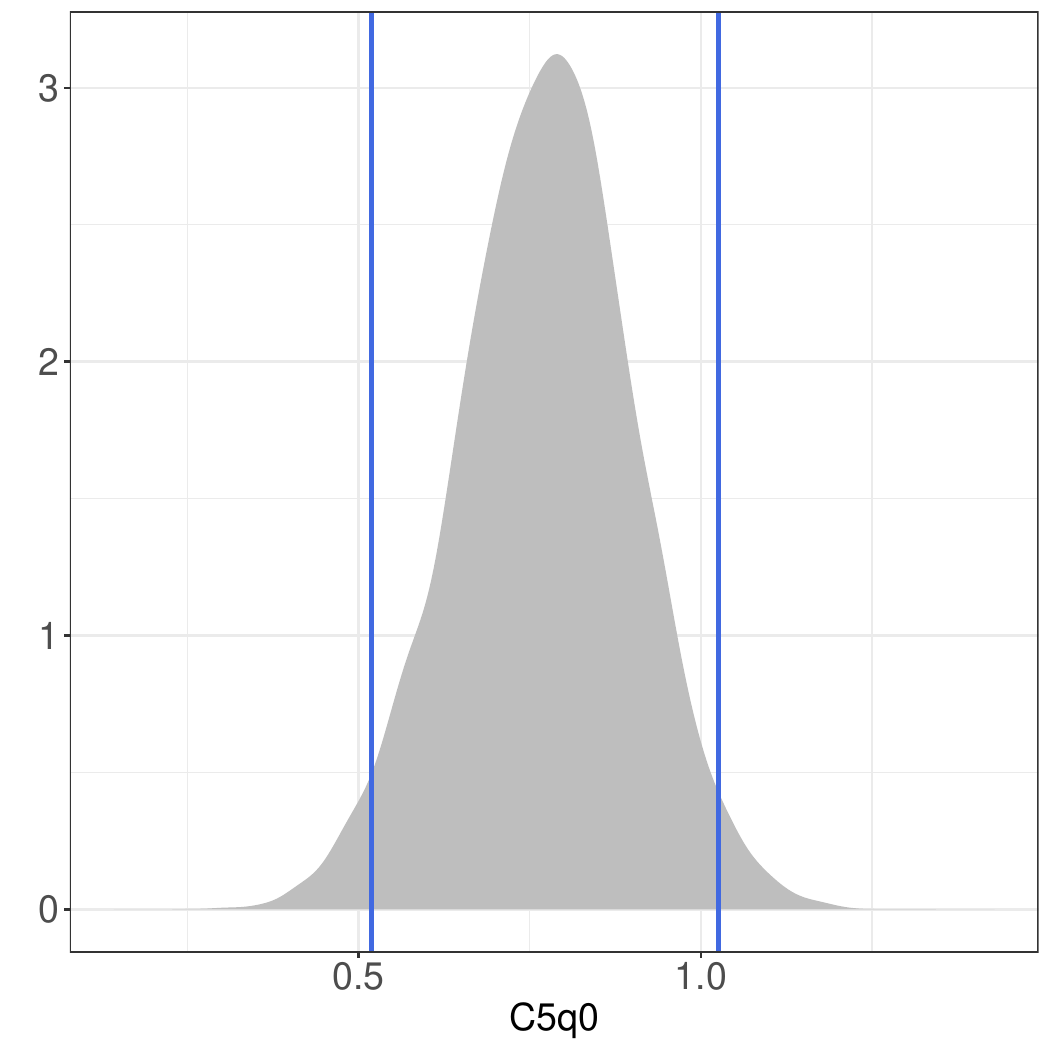} &
\includegraphics[width=0.3\textwidth]{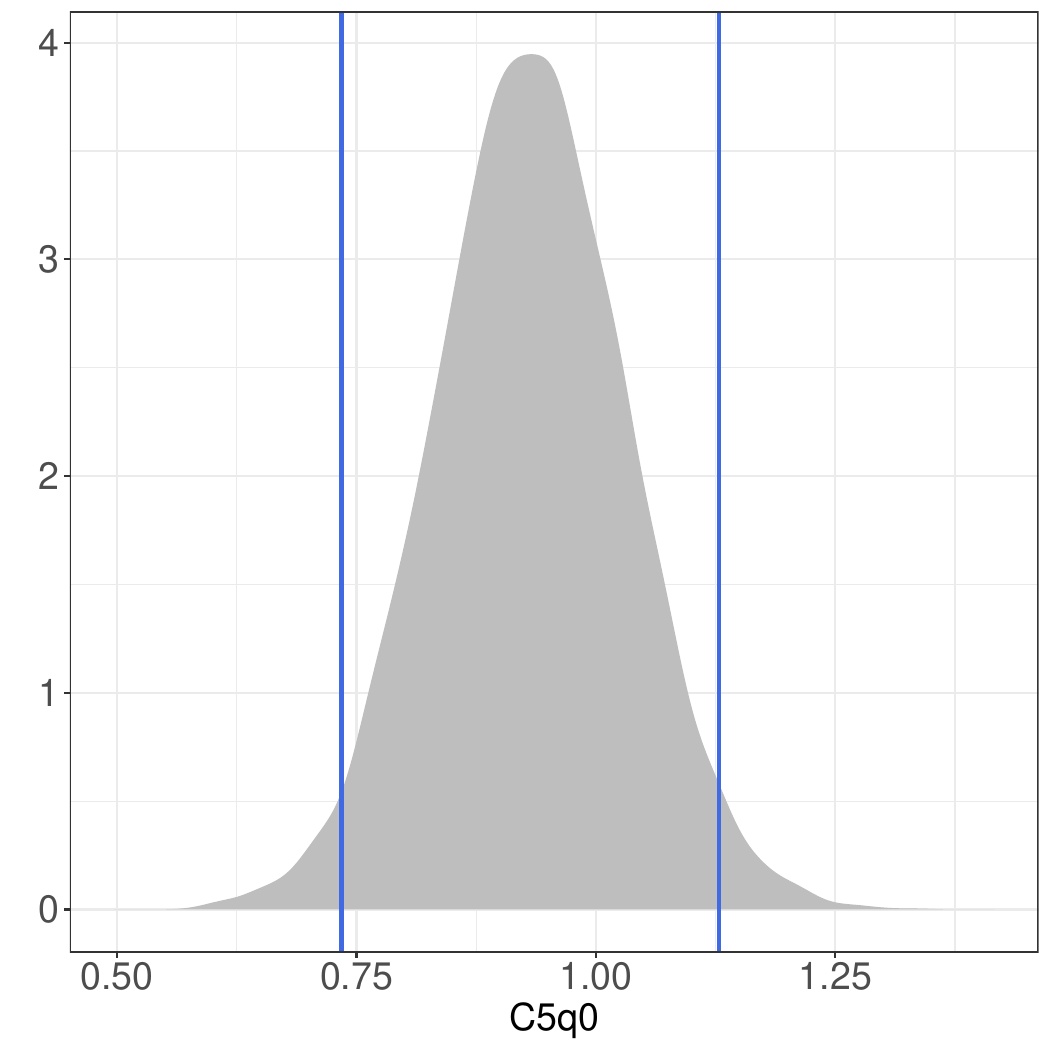}\\
Both sexes (Gamma) & Females (Gamma) & Males (Gamma)\\ \vspace{1cm}
\includegraphics[width=0.3\textwidth]{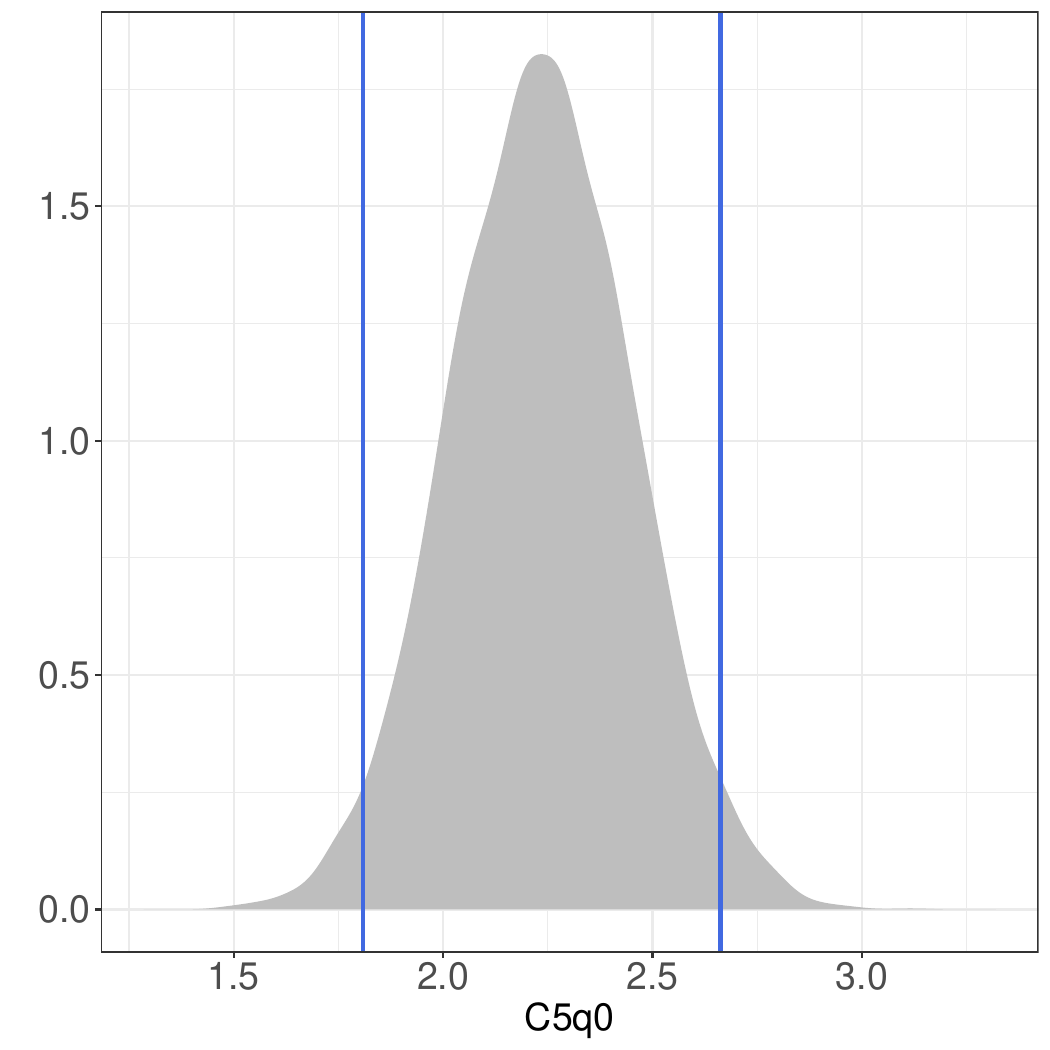} &
\includegraphics[width=0.3\textwidth]{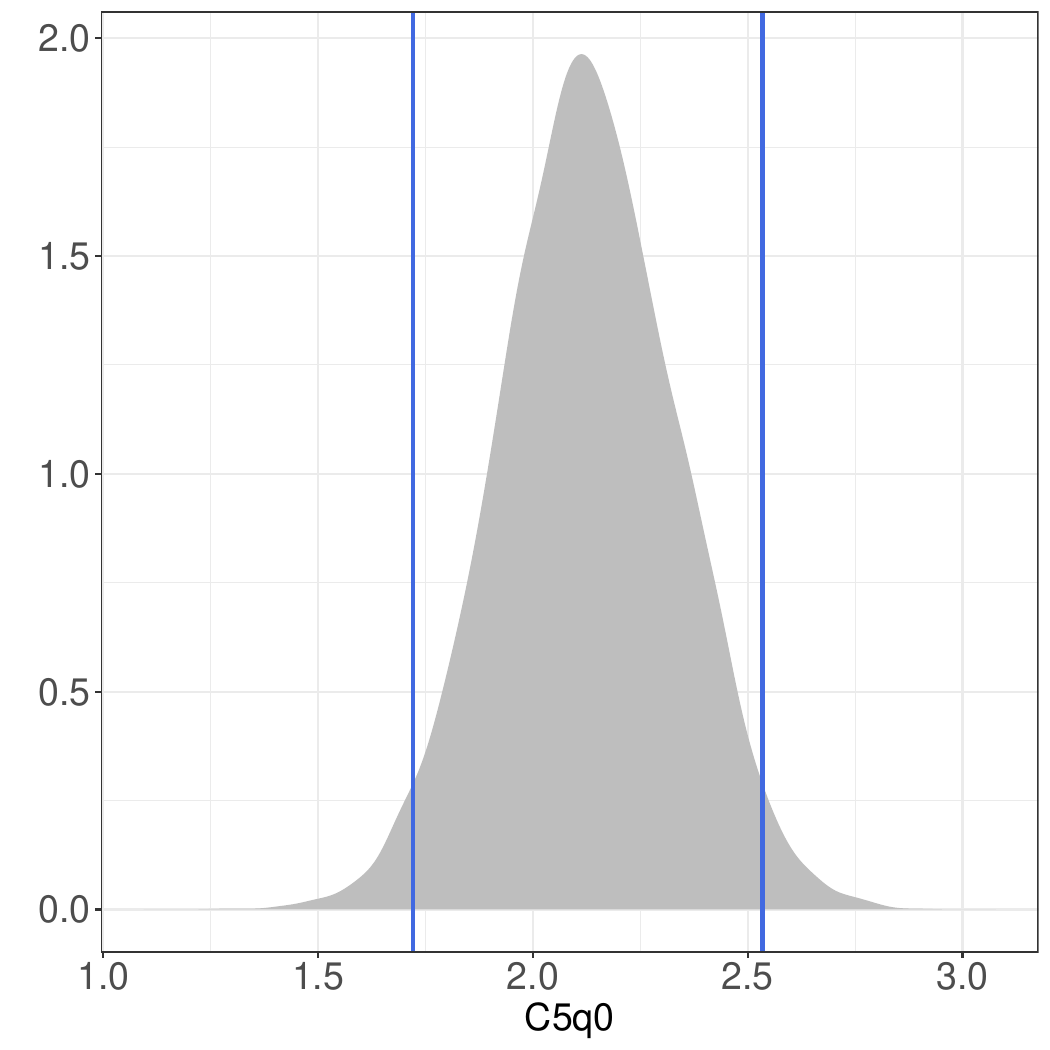} &
\includegraphics[width=0.3\textwidth]{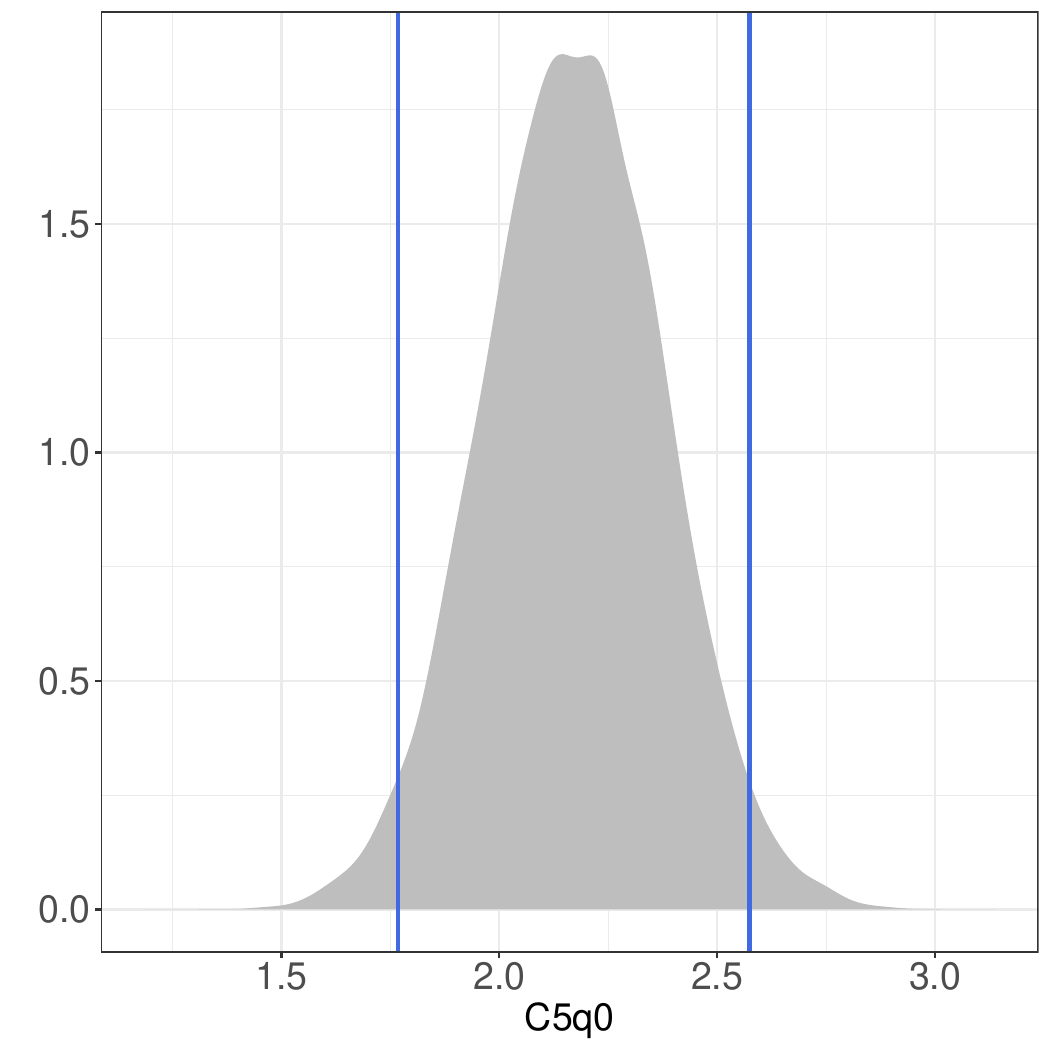}\\
\end{tabular}
\end{center}
\vspace{-0.5cm}
\caption{\small
Marginal posterior  of C5q0 with quantile-based credible intervals in blue lines for model 1 considering both sexes, females and males with
LA prior for the local scales of the random effects and local Half-Cauchy and a common Gamma prior for the scale of the errors respectively.}
\label{C5qO}
\end{figure}

\begin{table*}
\scriptsize
\caption{\small
R-squared, RMSE and MAE for the considered models using different HS, LA and Student-t priors for the local scales of the random effects and
a common Gamma prior for  the scale of the random effects and errors  respectively (without Global-Local structure).  Frequentist refers to the original model proposed in \cite{adair2018estimating} fitted using restricted maximum likelihood estimation. }
\label{tab_measures1}
\footnotesize
\renewcommand{\arraystretch}{1.5}
\scalebox{0.75}{
\begin{tabular}{rr|rrrrr|rrrrrr} \\
   \hline    \hline
 \multicolumn{12}{c}{Prior for the errors: Gamma} \\
 \hline
   &  &  &     & Model 1 &  &   &    &    & Model 2 &&\\ \hline \hline
 & & Frequentist  & Gamma & Student-t & HS & LA  &  Frequentist & Gamma & Student-t & HS & LA \\ \hline \hline
Both sexes & R-square & 0.8360 & 0.8362 & 0.8348 & 0.8370 & 0.8386 & 0.7894 & 0.7895 & 0.7930 & 0.7947 & 0.7982 \\
   & RMSE & 2.6172 & 2.6171 & 2.6276 & 2.6319 & 2.6612 & 2.6153 & 2.6154 & 2.6210 & 2.6225 & 2.6449 \\
   & MAE & 4.2804 & 4.2795 & 4.2971 & 4.3016 & 4.3384 & 4.4552 & 4.4558 & 4.4701 & 4.4738 & 4.5268 \\
  Females & R-square & 0.8477 & 0.8474 & 0.8454 & 0.8461 & 0.8432 & 0.8108 & 0.8104 & 0.8145 & 0.8151 & 0.8206 \\
   & RMSE & 2.6380 & 2.6389 & 2.6505 & 2.6564 & 2.6974 & 2.6459 & 2.6466 & 2.6546 & 2.6589 & 2.7088 \\
   & MAE & 4.2391 & 4.2398 & 4.2564 & 4.2636 & 4.3160 & 4.3755 & 4.3778 & 4.4007 & 4.4100 & 4.5087 \\
  Males  & R-square & 0.8227 & 0.8229 & 0.8222 & 0.8243 & 0.8261 & 0.7795 & 0.7795 & 0.7851 & 0.7883 & 0.7971 \\
   & RMSE & 2.7053 & 2.7041 & 2.7132 & 2.7180 & 2.7483 & 2.7841 & 2.7831 & 2.7896 & 2.7938 & 2.8163 \\
   & MAE & 4.3582 & 4.3566 & 4.3691 & 4.3759 & 4.4143 & 4.6852 & 4.6842 & 4.6895 & 4.6999 & 4.7435 \\
      \hline    \hline
\end{tabular}}
\vspace{1.5cm}
\caption{\small
R-squared, RMSE and MAE for the considered models using different Gamma, HS, LA and Student-t  priors for the local scales of the random effects and
a Half-Cauchy  for the local scale of the errors and a common Gamma for the scale of the random effects, respectively. Frequentist refers to the original model proposed in \cite{adair2018estimating} fitted using  restricted maximum likelihood estimation.}
\label{tab_measures2}
\renewcommand{\arraystretch}{1.5}
\scalebox{0.75}{
\begin{tabular}{rr|rrrr>{\columncolor[gray]{0.7}}r|rrrr>{\columncolor[gray]{0.7}}rr} \\
   \hline    \hline
 \multicolumn{12}{c}{Prior for the errors: Local prior - Half-Cauchy and Global prior - Gamma } \\
 \hline
   &  &  &     & Model 1 &  &   &      &   & Model 2  &&\\ \hline \hline
 &   & Frequentist & Gamma & Student-t & HS & LA  & Frequentist & Gamma & Student-t & HS & LA \\ \hline \hline
Both sexes & R-square & 0.836 & 0.795 & 0.799 & 0.805 & 0.827 & 0.789 & 0.802 & 0.807 & 0.808 & 0.815 \\
   & RMSE & 2.617 & 2.329 & 2.332 & 2.334 & 2.349 & 2.615 & 2.477 & 2.480 & 2.482 & 2.494 \\
   & MAE & 4.280 & 3.863 & 3.866 & 3.872 & 3.900 & 4.455 & 4.199 & 4.206 & 4.209 & 4.230 \\
 Females  & R-square & 0.848 & 0.826 & 0.828 & 0.829 & 0.830 & 0.811 & 0.825 & 0.830 & 0.829 & 0.832 \\
   & RMSE & 2.638 & 2.399 & 2.404 & 2.409 & 2.450 & 2.646 & 2.507 & 2.510 & 2.515 & 2.542 \\
   & MAE & 4.239 & 3.873 & 3.877 & 3.889 & 3.964 & 4.376 & 4.144 & 4.157 & 4.163 & 4.215 \\
  Males & R-square & 0.823 & 0.783 & 0.793 & 0.803 & 0.822 & 0.780 & 0.790 & 0.793 & 0.797 & 0.798 \\
   & RMSE & 2.705 & 2.385 & 2.389 & 2.392 & 2.412 & 2.784 & 2.585 & 2.586 & 2.587 & 2.593 \\
   & MAE & 4.358 & 3.901 & 3.904 & 3.911 & 3.934 & 4.685 & 4.382 & 4.383 & 4.387 & 4.386 \\
   \hline    \hline
\end{tabular}}
\end{table*}

 However, the posterior mean of the constant parameter and C5q0 are slightly different when Gamma and Half-Cauchy priors are considered for the local scale of the errors. Figure  \ref{C5qO} illustrates that the posterior mean of C5q0 is around 0.75 and 2.2 when a local Half-Cauchy and a common Gamma prior is considered for scales of the errors, respectively. This suggests that models using the Gamma prior are less sensitive to the completeness of under-five death registration when compared to those using the Half-Cauchy prior. Considering  GL priors for
the errors and random effects in both models 1 and 2 improves the estimation of 
 the completeness of death registration. This is illustrated in Tables \ref{tab_measures1}-\ref{tab_measures2}  where there are important reductions
 in terms of RMSE and MAE when a Half-Cauchy prior is considered for the scale of the errors in model (\ref{model})
and implemented in models 1 and 2   and for both sexes, males and females.  The differences  of RMSE and MAE  are practically negligible when
 Gamma, Student-t, HS or  LA priors are considered for the local scales of the random effects.  Interesting, similar values of RMSE, MAE and R-square
 are obtained under the HS and the local Student-t priors likely due to the similar tail behavior  in the resulting marginal posterior distribution of the random effects.

\begin{figure}[ht]
\small
\begin{center}
\begin{tabular}{ccc}
\hspace{-0.7cm} Model 1  (Both sexes) & Model 2  (Both sexes)\\
\hspace{-1.2cm} \includegraphics[width=0.55\textwidth]{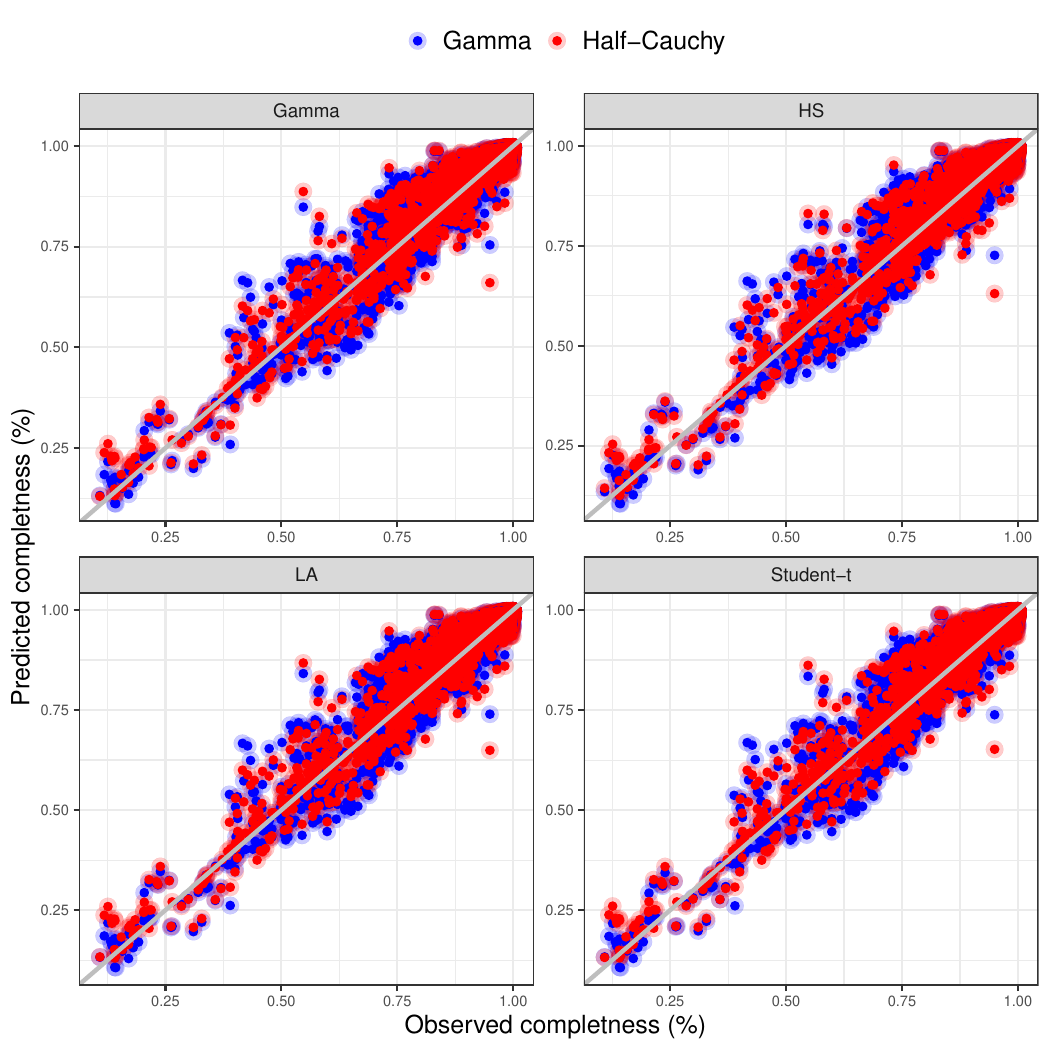}&
\includegraphics[width=0.55\textwidth]{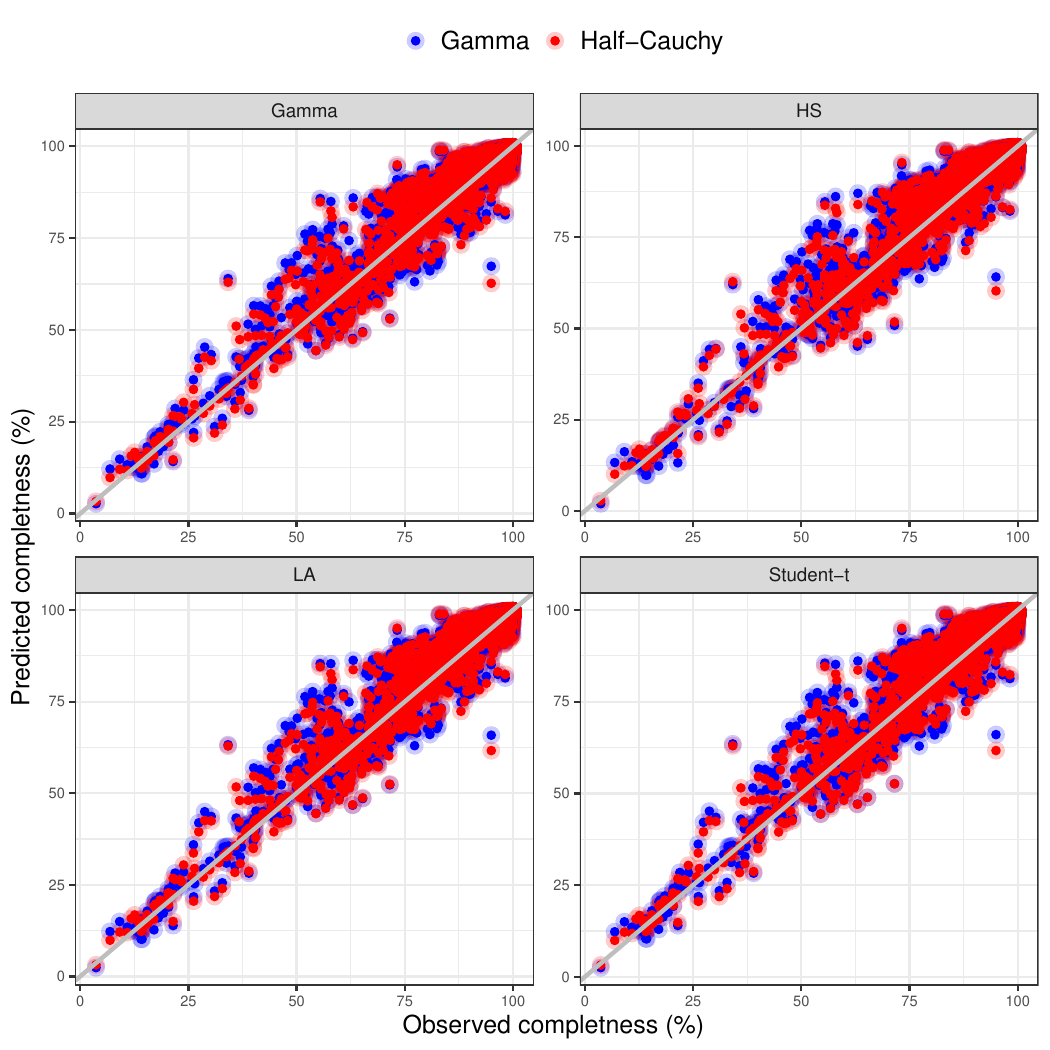}\\
\end{tabular}
\end{center}
\vspace{-0.5cm}
\caption{ \small Posterior  predicted mean versus observed death registration completeness, and using
HS LA and Student-t priors for the local scales of the random effects and a Half Cauchy prior
for the local scale of the errors and common Gamma distributions for the random effects and errors respectively.  The  models are based on an  dataset updated to 2019, which uses GBD death estimates based on the GBD 2019 and  comprises 120 countries and 2,748 country-years from 1970-2019  \citep{collaborators2020global}. Both sexes with model 1 and model 2.}
\label{fig:fig1}
\end{figure}

As we expected, similar but larger values of the RMSE and MAE  are obtained by considering the Frequentist model and the Bayesian model under a common Gamma prior for the scales of the random effects and errors, respectively. However,  when  LA priors are considered for the random effects the R-square values are higher and around 80\%  of those obtained with Gamma, Student-t and HS local priors and similar of those obtained with the frequentist (and the Bayesian approach) model.  This result shows that the use of GL priors rather than a common Gamma priors is crucial because observed completeness and covariates may exhibit considerable inter-country variability. This result also supports that the use of Half-Cauchy priors for the local scales of the errors rather than a common Gamma prior is particularly important because observed completeness may exhibit considerable within-country variability, particularly
for the 30\%-90\% completeness range where the data is concentrated.  This is illustrated in Figure ~\ref{fig:fig1}
where a large amount of the observed completeness is between around 30\% to 90\%.

\begin{figure}[h!]
\begin{tabular}{ccc}
Model 1 - Half-Cauchy (Both Sexes) & Model 1 - Gamma (Both Sexes) \\
\hspace{-1.0cm}\includegraphics[width=0.55\textwidth]{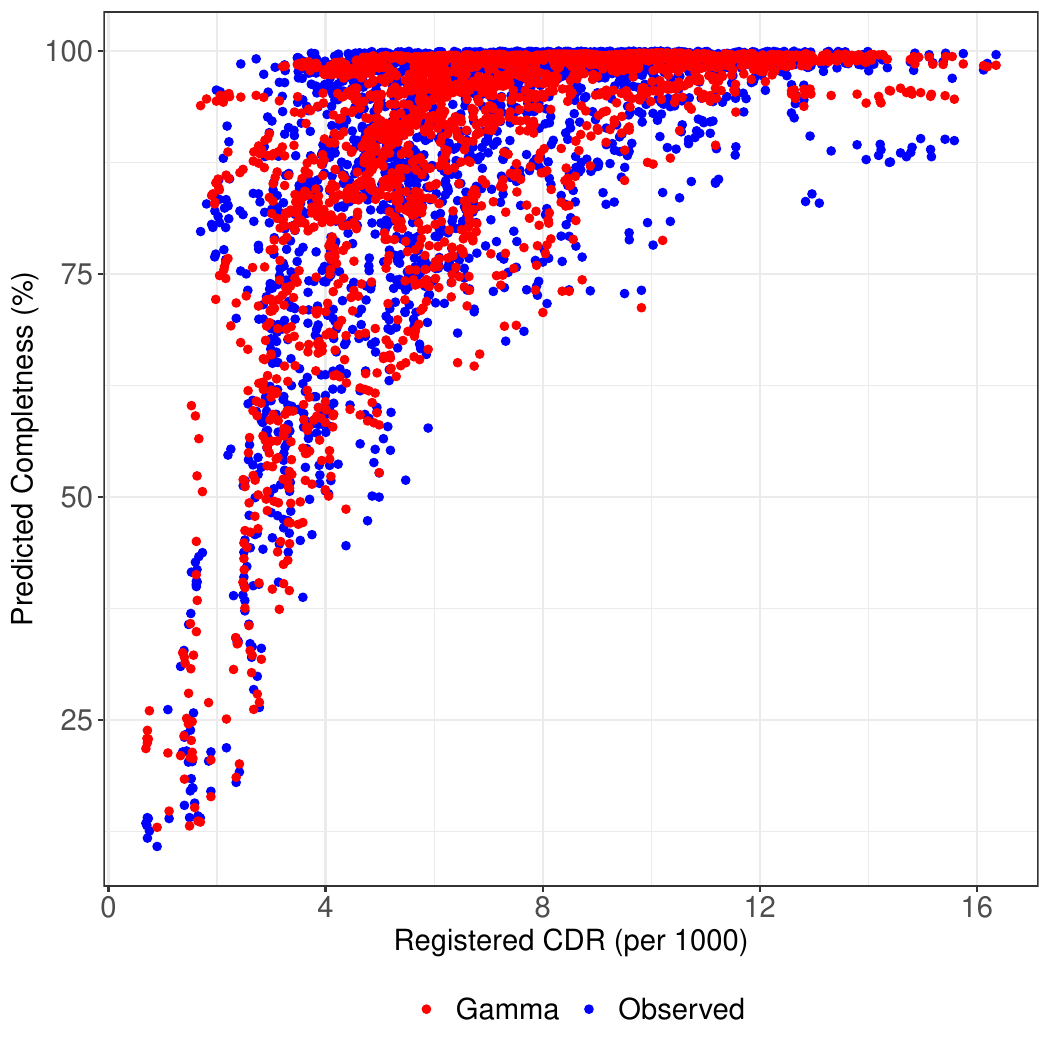} &
\includegraphics[width=0.55\textwidth]{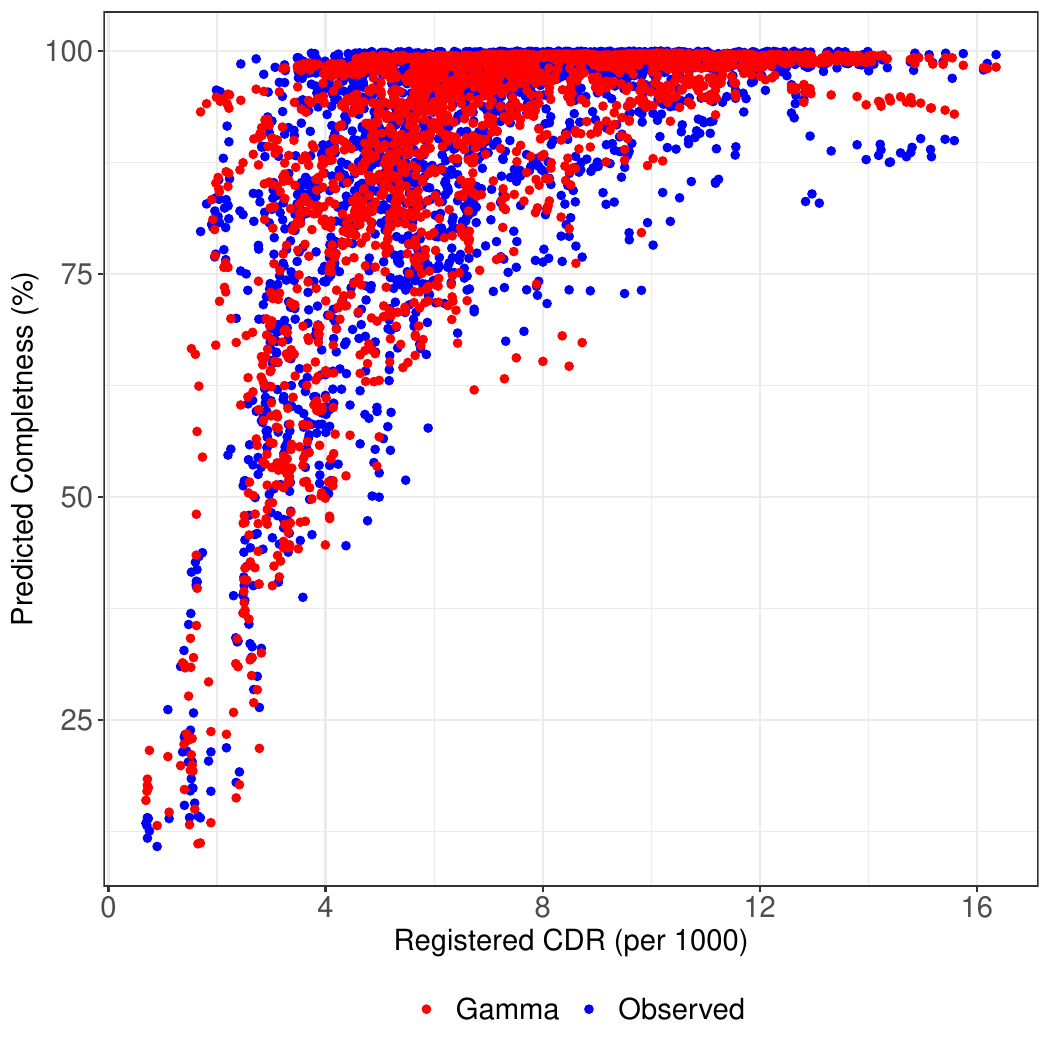}\\
\end{tabular}
\vspace{-0.5cm}
\caption{ \small Predicted versus observed death registration completeness by registered CDR
Student-t priors for the local scales of the random effects.The  models are based on an  dataset updated to 2019, which uses GBD death estimates based on the GBD 2019 and  comprises 120 countries and 2,748 country-years from 1970-2019  \citep{collaborators2020global}. Model 1, both sexes with
Half-Cauchy and Gamma priors for the scales, respectively.}
\label{fig:CDRa}
\end{figure}

Figure ~\ref{fig:fig1} also displays how GL priors and the completeness of registered under-five deaths as a predictor  dramatically improves
the estimation over this range of completeness. This is also confirmed in Tables ~\ref{tab:levels1} and ~\ref{tab:levels2} where the suitable posterior predictive estimates of completeness
are obtained according to the RMSE and MAE. Taking a closer look at the results for ranges of observed completeness between 0\% to 30\% for model
1 and for both sexes and males, the MAE values are lower when a Gamma prior is considered for the scale of the errors of the random effects. However, for other models and also for higher completeness levels the reduction in terms of RMSE and MAE is large when a Half-Cauchy prior is considered for the scale of the errors for each country $i=1,...,m$. The predictor RegCDR shows a curvilinear relationship with observed completeness, as was first pointed out by \citep{adair2018estimating}, and  in this paper the posterior predictive mean estimates also show this relationship. 

Importantly, Figure ~\ref{fig:CDRa} (left) for Half-Cauchy priors illustrates how the observed variability of the registered CDR in the completeness range 30\% to 90\% is captured for the proposed models, with the posterior predictive mean estimates and observed completeness close to each other. This range of completeness is particularly important because for these populations the method has greatest utility (i.e. where completeness is above 90\% there is more certainty about the true level of completeness, while where completeness is less than 30\% it is not recommended to use estimated completeness to adjust mortality rates). This makes Half-Cauchy priors suitable choices for the scale of the errors. 

\clearpage

\begin{table}[ht]
\renewcommand{\arraystretch}{1.4}
\scalebox{0.75}{
\begin{tabular}{r|rrrr|>{\columncolor[gray]{0.8}}r>{\columncolor[gray]{0.8}}r>{\columncolor[gray]{0.8}}r>{\columncolor[gray]{0.8}}rrrrrrr}
\hline    \hline
Global-Prior &  &     Gamma & &   &  &     Half-Cauchy & &&  &\\ \hline
Local-Prior  & Gamma & Student-t & HS & LA  & Gamma & Student-t & HS & LA \\
  \hline \hline
 \multicolumn{11}{c}{Model 1 (both sexes)} \\ \hline \hline
    90$<$100 & 1.87 & 1.88 & 1.88 & 1.90 & 1.96 & 1.97 & 1.97 & 2.00 \\
   80$<$90 & 5.36 & 5.37 & 5.37 & 5.36 & 4.83 & 4.83 & 4.82 & 4.76 \\
   60$<$80 & 7.92 & 7.94 & 7.93 & 7.95 & 6.85 & 6.88 & 6.89 & 6.96 \\
   30$<$60 & 9.14 & 9.21 & 9.26 & 9.48 & 7.71 & 7.66 & 7.71 & 7.91 \\
   $<$30 & 4.67 & 4.77 & 4.86 & 5.34 & 6.12 & 6.12 & 6.08 & 5.96 \\
   \hline \hline
 \multicolumn{11}{c}{Model 2 (both sexes)} \\ \hline \hline
   90$<$100 & 2.08 & 2.09 & 2.09 & 2.10 & 2.11 & 2.13 & 2.13 & 2.15 \\
   80$<$90 & 5.10 & 5.09 & 5.08 & 5.00 & 4.98 & 4.97 & 4.96 & 4.89 \\
   60$<$80 & 7.33 & 7.35 & 7.37 & 7.49 & 6.88 & 6.88 & 6.89 & 6.92 \\
   30$<$60 & 10.10 & 10.15 & 10.19 & 10.48 & 9.00 & 9.04 & 9.06 & 9.22 \\
   $<$30 & 4.84 & 4.84 & 4.80 & 4.71 & 4.13 & 4.14 & 4.12 & 4.15 \\
   \hline \hline
 \multicolumn{11}{c}{Model 1 (females)} \\ \hline  \hline
   90$<$100 & 1.89 & 1.91 & 1.91 & 1.95 & 1.95 & 1.97 & 1.97 & 1.99 \\
   80$<$90 & 5.42 & 5.43 & 5.43 & 5.47 & 4.96 & 4.97 & 4.98 & 5.04 \\
   60$<$80 & 7.80 & 7.82 & 7.82 & 7.81 & 6.96 & 6.97 & 6.99 & 7.18 \\
   30$<$60 & 8.53 & 8.57 & 8.62 & 8.83 & 7.20 & 7.13 & 7.19 & 7.40 \\
   $<$30 & 5.15 & 5.33 & 5.38 & 5.99 & 5.91 & 5.94 & 5.92 & 5.66 \\
  \hline \hline
 \multicolumn{11}{c}{Model 2 (females)} \\
  \hline \hline
   90$<$100 & 2.06 & 2.08 & 2.07 & 2.10 & 2.07 & 2.09 & 2.09 & 2.11 \\
   80$<$90 & 5.26 & 5.27 & 5.28 & 5.30 & 5.13 & 5.14 & 5.14 & 5.15 \\
   60$<$80 & 7.26 & 7.30 & 7.31 & 7.44 & 6.92 & 6.92 & 6.93 & 7.01 \\
   30$<$60 & 9.70 & 9.78 & 9.83 & 10.26 & 8.65 & 8.68 & 8.71 & 8.89 \\
   $<$30 & 4.34 & 4.29 & 4.32 & 4.44 & 3.80 & 3.83 & 3.84 & 4.04 \\
 \hline \hline
 \multicolumn{11}{c}{Model 1 (males)} \\ \hline \hline
   90$<$100 & 1.90 & 1.92 & 1.92 & 1.96 & 1.97 & 1.98 & 1.99 & 2.04 \\
   80$<$90 & 5.24 & 5.24 & 5.24 & 5.24 & 4.69 & 4.70 & 4.70 & 4.67 \\
   60$<$80 & 7.94 & 7.96 & 7.95 & 7.93 & 6.91 & 6.93 & 6.94 & 6.99 \\
   30$<$60 & 8.69 & 8.73 & 8.80 & 9.06 & 7.29 & 7.22 & 7.25 & 7.25 \\
   $<$30 & 5.39 & 5.41 & 5.46 & 5.71 & 5.62 & 5.61 & 5.56 & 5.66 \\
 \hline \hline
 \multicolumn{11}{c}{Model 2 (males)} \\ \hline \hline
   90$<$100 & 2.12 & 2.13 & 2.14 & 2.18 & 2.15 & 2.16 & 2.17 & 2.23 \\
   80$<$90 & 5.05 & 5.05 & 5.04 & 4.98 & 4.82 & 4.82 & 4.80 & 4.73 \\
   60$<$80 & 7.72 & 7.73 & 7.73 & 7.77 & 7.39 & 7.39 & 7.37 & 7.33 \\
   30$<$60 & 10.27 & 10.27 & 10.36 & 10.64 & 8.73 & 8.71 & 8.76 & 8.82 \\
   $<$30 & 6.24 & 6.22 & 6.19 & 5.94 & 5.91 & 5.88 & 5.88 & 5.82 \\
   \hline
\end{tabular}}
\caption{\small RMSE measures and
sensitivity analysis of the Half-Cauchy local scale prior and a common Gamma prior for the errors in model 1 and 2 considering different levels of completeness.}
\label{tab:levels1}
\end{table}

\begin{table}[ht]
\renewcommand{\arraystretch}{1.4}
\scalebox{0.75}{
\begin{tabular}{r|rrrr|>{\columncolor[gray]{0.8}}r>{\columncolor[gray]{0.8}}r>{\columncolor[gray]{0.8}}r>{\columncolor[gray]{0.8}}rrrrrrr}
\hline    \hline
Global-Prior &  &     Gamma & &   &  &     Half-Cauchy & &&  &\\ \hline
Local-Prior  & Gamma & Student-t & HS & LA  & Gamma & Student-t & HS & LA \\
  \hline \hline
 \multicolumn{11}{c}{Model 1 (both sexes)} \\ \hline \hline
   90$<$100 & 1.18 & 1.18 & 1.18 & 1.18 & 1.18 & 1.19 & 1.19 & 1.20 \\
   80$<$90 & 4.49 & 4.48 & 4.49 & 4.47 & 3.79 & 3.79 & 3.78 & 3.72 \\
   60$<$80 & 6.63 & 6.67 & 6.67 & 6.72 & 5.47 & 5.49 & 5.51 & 5.59 \\
  30$<$60 & 6.90 & 7.01 & 7.07 & 7.45 & 5.53 & 5.53 & 5.56 & 5.70 \\
  $<$30 & 3.78 & 3.87 & 3.91 & 4.11 & 4.74 & 4.76 & 4.74 & 4.68 \\
  \hline \hline
 \multicolumn{11}{c}{Model 2 (both sexes)} \\ \hline \hline
  90$<$100 & 1.25 & 1.25 & 1.25 & 1.25 & 1.24 & 1.24 & 1.24 & 1.25 \\
  80$<$90 & 3.99 & 3.98 & 3.98 & 3.94 & 3.85 & 3.85 & 3.84 & 3.79 \\
  60$<$80 & 5.66 & 5.68 & 5.69 & 5.80 & 5.22 & 5.22 & 5.23 & 5.27 \\
  30$<$60 & 7.47 & 7.52 & 7.53 & 7.73 & 6.64 & 6.67 & 6.69 & 6.83 \\
  $<$30 & 3.33 & 3.40 & 3.41 & 3.35 & 2.95 & 2.98 & 2.98 & 3.03 \\
  \hline \hline
 \multicolumn{11}{c}{Model 1 (females)} \\ \hline  \hline
  90$<$100 & 1.24 & 1.24 & 1.25 & 1.26 & 1.26 & 1.26 & 1.26 & 1.27 \\
 80$<$90 & 4.50 & 4.51 & 4.51 & 4.55 & 3.94 & 3.94 & 3.95 & 4.01 \\
  60$<$80 & 6.39 & 6.41 & 6.41 & 6.44 & 5.39 & 5.40 & 5.42 & 5.53 \\
  30$<$60 & 6.59 & 6.68 & 6.75 & 7.07 & 5.50 & 5.51 & 5.55 & 5.84 \\
  $<$30 & 4.34 & 4.48 & 4.51 & 4.89 & 4.73 & 4.77 & 4.76 & 4.66 \\
  \hline \hline
 \multicolumn{11}{c}{Model 2 (females)} \\
  \hline \hline
 90$<$100 & 1.32 & 1.32 & 1.32 & 1.33 & 1.30 & 1.30 & 1.30 & 1.31 \\
 80$<$90 & 4.15 & 4.16 & 4.18 & 4.27 & 3.98 & 3.99 & 3.99 & 4.03 \\
 60$<$80 & 5.60 & 5.63 & 5.63 & 5.74 & 5.20 & 5.19 & 5.20 & 5.24 \\
 30$<$60 & 7.36 & 7.39 & 7.41 & 7.63 & 6.53 & 6.55 & 6.59 & 6.75 \\
 $<$30 & 3.58 & 3.53 & 3.54 & 3.54 & 3.20 & 3.21 & 3.22 & 3.33 \\    \hline \hline
 \multicolumn{11}{c}{Model 1 (males)} \\ \hline \hline
 90$<$100 & 1.20 & 1.20 & 1.20 & 1.21 & 1.19 & 1.20 & 1.20 & 1.22 \\
 80$<$90 & 4.38 & 4.38 & 4.39 & 4.40 & 3.68 & 3.69 & 3.69 & 3.68 \\
 60$<$80 & 6.61 & 6.65 & 6.64 & 6.66 & 5.51 & 5.53 & 5.54 & 5.58 \\
 30$<$60 & 6.29 & 6.35 & 6.43 & 6.80 & 5.16 & 5.16 & 5.17 & 5.21 \\
 $<$30 & 4.59 & 4.64 & 4.66 & 4.73 & 4.24 & 4.23 & 4.21 & 4.27 \\ \\ \hline \hline
 \multicolumn{11}{c}{Model 2 (males)} \\ \hline \hline
 90$<$100 & 1.29 & 1.29 & 1.30 & 1.31 & 1.26 & 1.26 & 1.27 & 1.29 \\
 80$<$90 & 3.89 & 3.89 & 3.88 & 3.87 & 3.69 & 3.69 & 3.68 & 3.64 \\
 60$<$80 & 5.96 & 5.98 & 5.98 & 6.03 & 5.66 & 5.66 & 5.66 & 5.61 \\
 30$<$60 & 7.85 & 7.87 & 7.91 & 8.06 & 6.25 & 6.26 & 6.28 & 6.38 \\
 $<$30 & 3.95 & 4.03 & 4.06 & 3.94 & 3.02 & 3.02 & 3.02 & 3.05 \\
   \hline
\end{tabular}}
\caption{\small MAE measures and
sensitivity analysis of the Half-Cauchy local scale prior and a common Gamma prior for the errors in model 1 and 2 considering different levels of completeness..}
\label{tab:levels2}
\end{table}

\clearpage

To illustrate that the inclusion of GL priors  affect the predictions
of the logit of completeness we consider two deviances at the country and at the country-year levels  with

\begin{align*}
\text{Deviance}^{\text{countr- level}} &= (\bar{y}_{i}-\boldsymbol{\bar{x}_{i}^\textsf{T}} \hat{\boldsymbol{\beta}})^2, &
\text{Deviance}^{\text{country-year level}} &= \sum \limits_{j=1}^{n_{i}}(y_{ij}-\boldsymbol{x_{ij}}^{T}\hat{\boldsymbol{\beta}})^{2}/n_{i}
\end{align*}

where $\hat{\boldsymbol{\beta}}$ denotes the ordinary least squares estimates from a multiple regression model. Notice that these two deviances  in some sense measure the within country and between country-year variabilities. According to Theorem  \ref{th2}  when the difference $|\bar{\boldsymbol{y}}_{i}-\bar{\boldsymbol{x}}_{i}^\textsf{T}\boldsymbol{\beta}_{i}|$ increases then the shrinkage under the posterior mean of the random effects towards zero is offset and the random effects depend on the observed completeness and  the regression fit. This is illustrated in Figure \ref{3D} where the random effects increases under GL priors with the deviance at the country-level. However as is pointed out  in Theorem  \ref{th2}  the posterior concentration of the shrinkage factors are also affected
for the GL priors of the errors. Notice that the shape of the conditional Gamma  posterior distribution of the global parameter $\tau$
 in algorithm \ref{MCMC2} depends on the country-year level variability. This behavior is also illustrated in Figure \ref{3D}  where the values of the random effects under GL priors are very sensible even to  smaller values of within country-year variability.

\begin{figure}[ht]
\small
\begin{center}
\begin{tabular}{ccc}
\hspace{-2.5cm} Both sexes  & Females & Males \\
\hspace{-2.5cm} \includegraphics[width=0.45\textwidth]{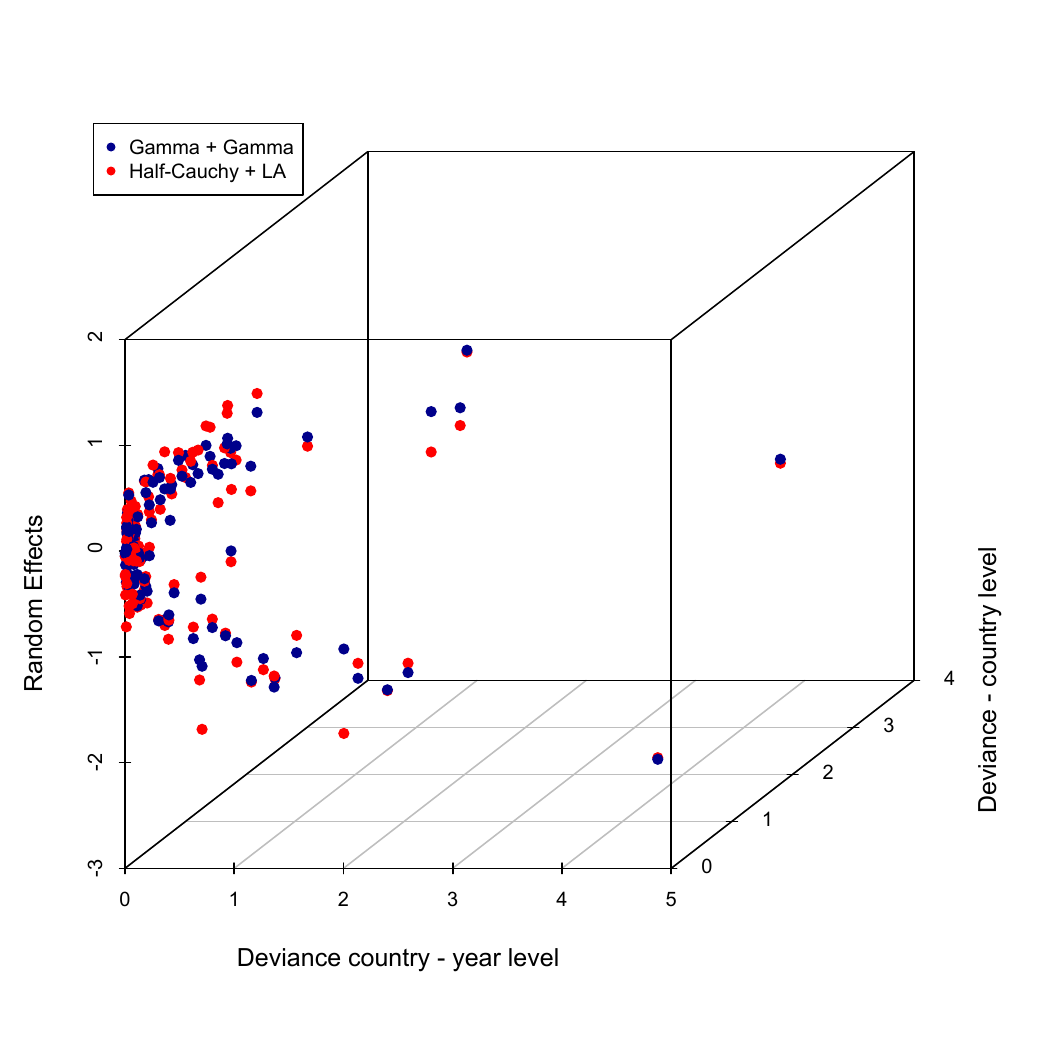} & \includegraphics[width=0.45\textwidth]{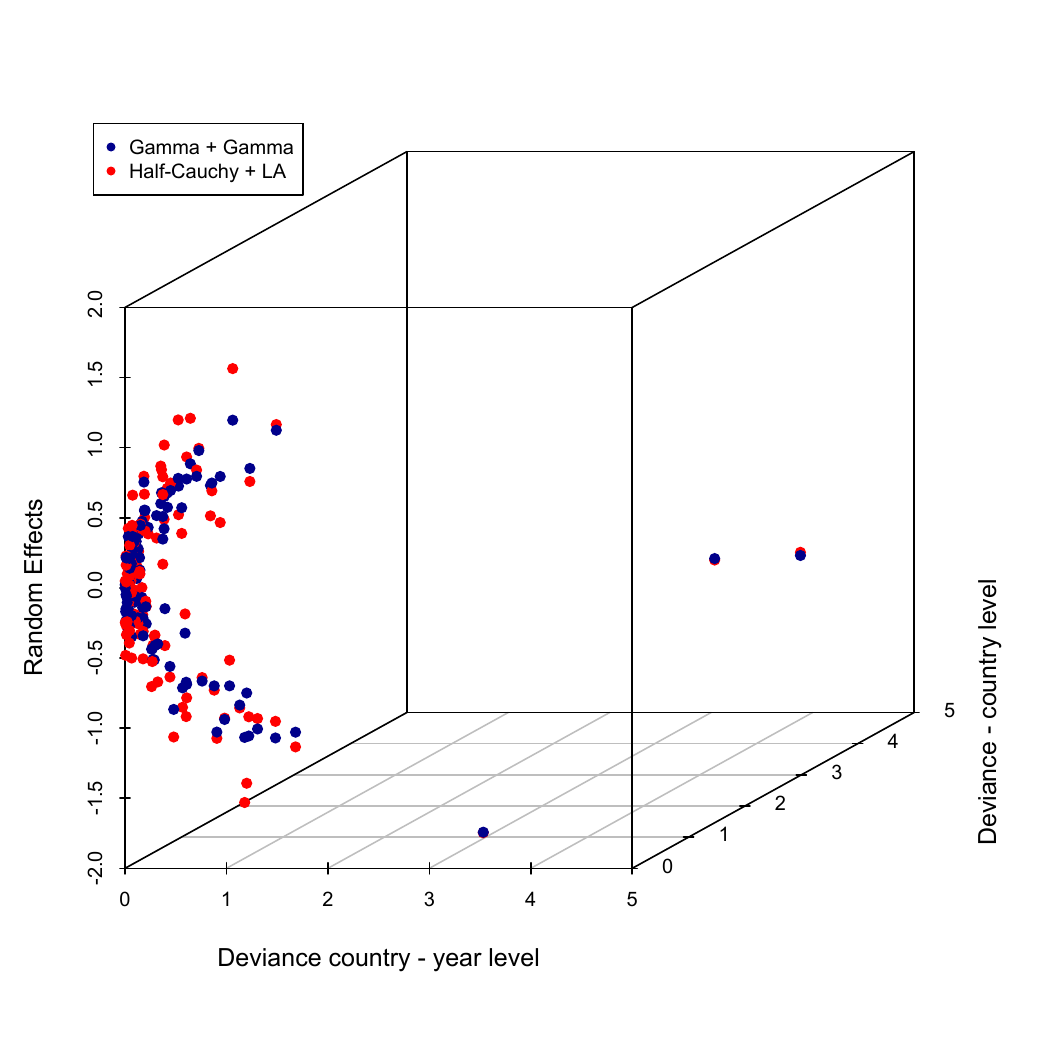} &
\includegraphics[width=0.45\textwidth]{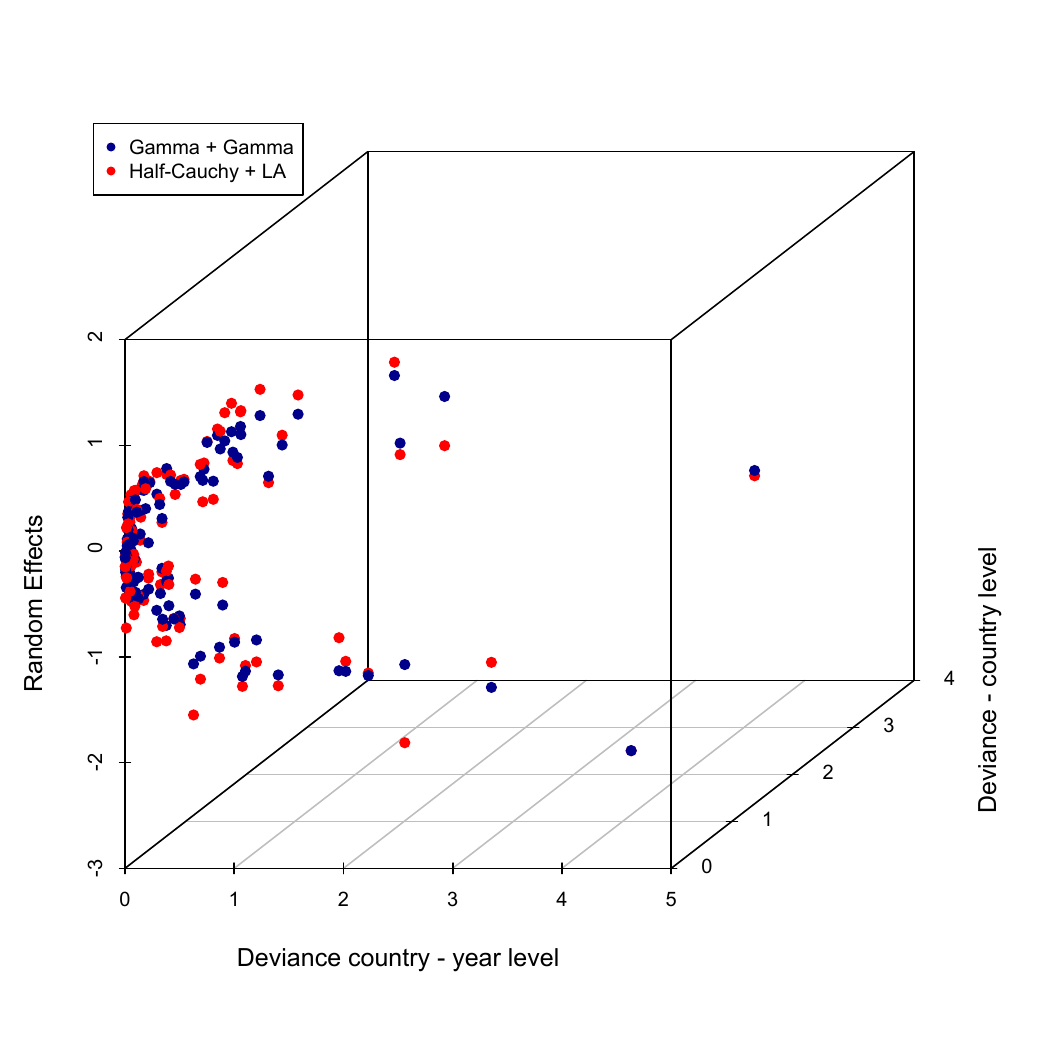}\\
\end{tabular}
\end{center}
\vspace{-0.5cm}
\caption{\small
 Posterior means of random effects from model 1 using the  LA local prior for the scales of the random effects 
 and a common Gamma or a Half-Cauchy prior for the scales of the errors versus
the deviances at country and country - year levels.}
\label{3D}
\end{figure}

\section{National and subnational implementation: completeness of death registration in Colombia and its  Departments in 2017}

\label{Results2}

To evaluate the performance and demonstrate the utility of the models in Section  \ref{Results},  in this section we use them to estimate the completeness of death registration $c_{k}$ for both sexes, males and females for the Country of Colombia and the corresponding regional division by 33 departments in 2017, $k = 1,...,33$. The results are compared to data from  the Colombian population census described earlier. We call the estimates from the Census the observed estimates $c_{k}$ during for the purpose of analysis of the results. The models estimated completeness of deaths in 2017 from vital registration, the same data source for which Census respondents report whether a household death was registered \citep{DANE2017}. The 5q0 was estimated considering the same procedure implemented in \cite{adair2018estimating}; we considered the average from the 2010 and 2015 Demographic and Health Surveys \citep{Profamilia, Profamilia2}  and scaled those estimates to the GBD of 5q0 for Colombia in 2017. We consider the projected values of population aged 65 years and over in 2017 using the official population projections in Colombia  \citep{DANE20172}. To estimate C5q0 we computed the ratio between the 5q0 obtained from standard life tables using the registration data \citep{DANE2017}  divided by the estimate of 5q0 as in \cite{adair2018estimating}.

\begin{table}[ht]

\caption{\small
MAE, MSE, and the number of departments including national level with absolute deviations of the predictive 
completeness $\hat{\delta}_{k}$ and the comparator  $c_{k}$ less than10 percentage points.
The table contains the results for  HS, LA and Student-t local scale priors for the random effects and
the Half-Cauchy for the local scale of the errors and common Gamma  prior for the scale of the random effects and errors, respectively. }
 \label{tab:tab_measuressub}
\renewcommand{\arraystretch}{1.4}
\scalebox{0.85}{
\begin{tabular}{rr|rrrr|rrrrrr} \\
   \hline    \hline
 \multicolumn{12}{c}{A common prior for the scale of the errors: Gamma} \\
 \hline
   &  &  &    Model 1 & &  &   &     Model 2 & & &&\\ \hline \hline
 &  & Gamma & Student-t & HS & LA  & Gamma & Student-t & HS & LA \\ \hline \hline
   &  &  &    Model 1 & &  &   &     Model 2 & & &&\\ \hline \hline
   & MAE & 0.060 & 0.054 & 0.061 & 0.060 & 0.054 & 0.054 & 0.053 & 0.054 \\
Both sexes & MSE & 0.009 & 0.006 & 0.010 & 0.009 & 0.006 & 0.006 & 0.006 & 0.006 \\
  &   \# diff<10\% & 29 & 28 & 28 & 29 & 29 & 28 & 29 & 28 \\
  & MAE & 0.062 & 0.062 & 0.062 & 0.062 & 0.054 & 0.055 & 0.052 & 0.054 \\
Males  & MSE & 0.009 & 0.009 & 0.009 & 0.009 & 0.006 & 0.006 & 0.005 & 0.006 \\
  & \# diff<10\% & 29 & 29 & 28 & 29.000 & 30 & 30 & 30 & 30 \\
 &  MAE & 0.059 & 0.059 & 0.061 & 0.060 & 0.052 & 0.052 & 0.054 & 0.052 \\
Females    & MSE & 0.010 & 0.010 & 0.010 & 0.010 & 0.005 & 0.005 & 0.006 & 0.005 \\
  & \# diff<10\% & 29 & 29 & 29 & 30 & 28 & 28 & 29 & 28 \\
        \hline    \hline
\end{tabular}}
\renewcommand{\arraystretch}{1.4}
\scalebox{0.85}{
\begin{tabular}{rr|>{\columncolor[gray]{0.8}}r>{\columncolor[gray]{0.8}}r>{\columncolor[gray]{0.8}}r>{\columncolor[gray]{0.8}}r|
|>{\columncolor[gray]{0.8}}r>{\columncolor[gray]{0.8}}r>{\columncolor[gray]{0.8}}r>{\columncolor[gray]{0.8}}r|rrrrrr} \\
   \hline    \hline
 \multicolumn{12}{c}{Local scale prior for the errors: Half-Cauchy} \\
 \hline
   &  &  &    Model 1 & &  &   &     Model 2 & & &&\\ \hline \hline
 &  & Gamma & Student-t & HS & LA  & Gamma & Student-t & HS & LA \\ \hline \hline
 & MAE & 0.055 & 0.045 & 0.053 & 0.054 & 0.045 & 0.045 & 0.046 & 0.045 \\
Both sexes  & MSE & 0.008 & 0.004 & 0.008 & 0.008 & 0.004 & 0.004 & 0.004 & 0.004 \\
  & \# diff<10\% & 29 & 30 & 30 & 29 & 30 & 30 & 30 & 30 \\
 &  MAE & 0.056 & 0.056 & 0.055 & 0.055 & 0.047 & 0.047 & 0.048 & 0.047 \\
Males &  MSE & 0.008 & 0.007 & 0.007 & 0.007 & 0.005 & 0.005 & 0.005 & 0.005 \\
 &  \# diff<10\% & 30 & 30 & 30 & 30 & 29 & 29 & 30 & 30 \\
 &  MAE & 0.053 & 0.052 & 0.051 & 0.052 & 0.048 & 0.048 & 0.048 & 0.048 \\
Females   & MSE & 0.008 & 0.008 & 0.007 & 0.008 & 0.005 & 0.005 & 0.005 & 0.005 \\
  & \# diff<10\% & 30 & 30 & 30 & 30 & 30 & 29 & 29 & 29 \\
   \hline    \hline
\end{tabular}}
\end{table}

To evaluate the fit of the models at the national and subnational level (the national level and the 33 departments) we consider two deviance measures, the MAE and the Mean Square Error (MSE) given by,

\begin{align}
\text{MAE} &= \dfrac{1}{34}\sum_{k=1}^{34} \mid \hat{\delta}_{k} - c_{k} \mid, & \text{MSE} &= \dfrac{1}{34}\sum_{k=1}^{34}  ( \hat{\delta}_{k} - c_{k} )^2,
\end{align}

where  $\hat{\delta}_{k}$  is the posterior mean of the  posterior predictive distribution.  We also compute the number of departments
with absolute deviations  $\mid \hat{\delta}_{k} - c_{k} \mid$ smaller than 10 percentage points. Table ~\ref{tab:tab_measuressub} shows that in 29 or 30 departments and at the national level the 
absolute deviances between the observed completeness $c_{k}$ and the posterior predictive mean estimate is less than 10 percentage points for the models considering
a Half-Cauchy prior for the local scale of the errors (depending on the local prior).

\begin{figure}[h!]
\begin{tabular}{ccc}
\hspace{1.2cm} Model 1 (Both Sexes) & \hspace{2.5cm}  Model 2 (Both Sexes) \\
\hspace{-2.5cm}\includegraphics[width=0.65\textwidth]{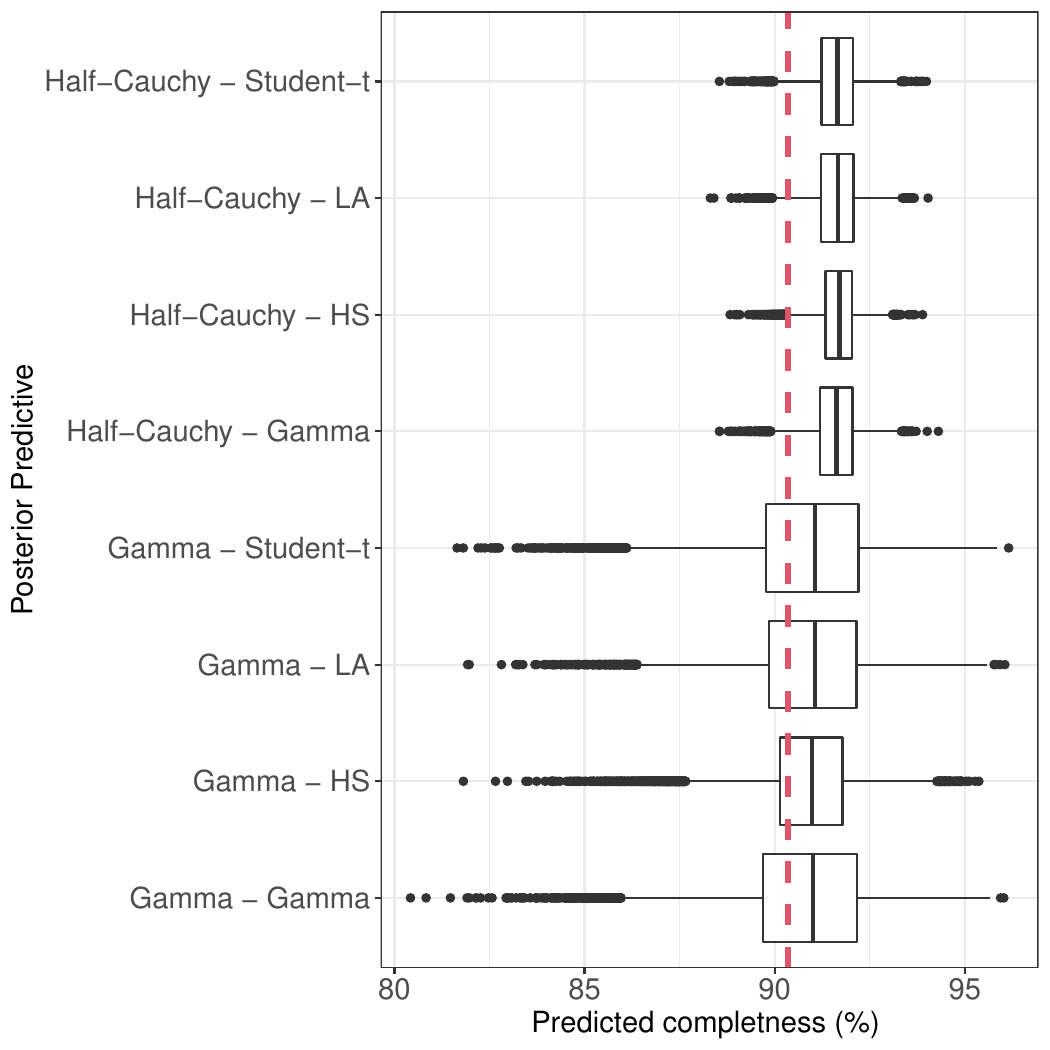} &
\includegraphics[width=0.65\textwidth]{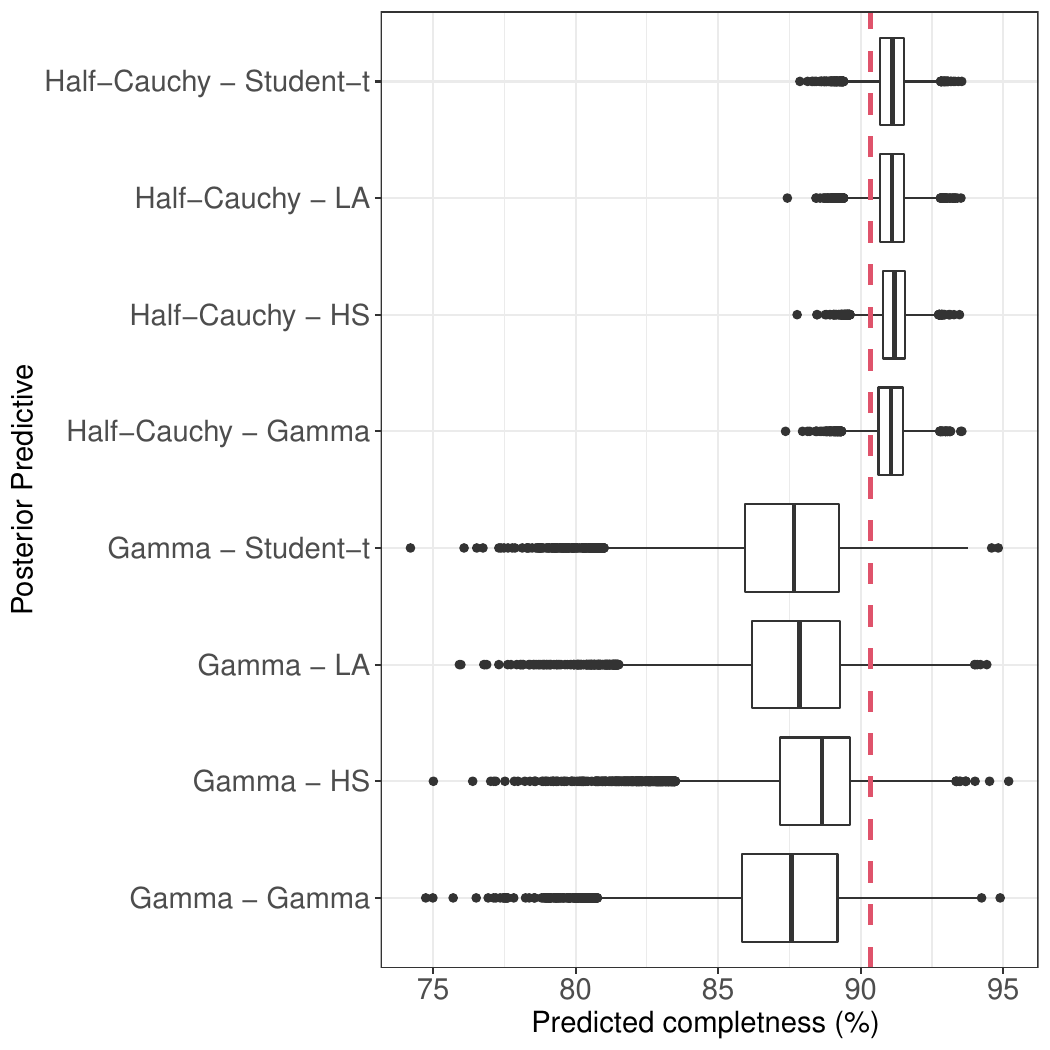}\\
\end{tabular}
\vspace{-0.5cm}
\caption{ \small Marginal posterior predictive distributions of the completeness at national level under Gamma, Student-t, HS and LA local scales for the random effects and
Half-Cauchy for the local scales of the errors, Gamma-Gamma refers to the model with a common Gamma prior distribution for the random effects and errors
respectively,  Model 1 and 2, both sexes. The red  lines illustrates the observed values obtained from Census 2018.}
\label{fig:CDR}
\end{figure}

\begin{figure}[h!]
\begin{tabular}{ccc}
\hspace{1.2cm} Model 1 (Males) & \hspace{2.5cm}  Model 2 (Males) \\
\hspace{-2.5cm}\includegraphics[width=0.65\textwidth]{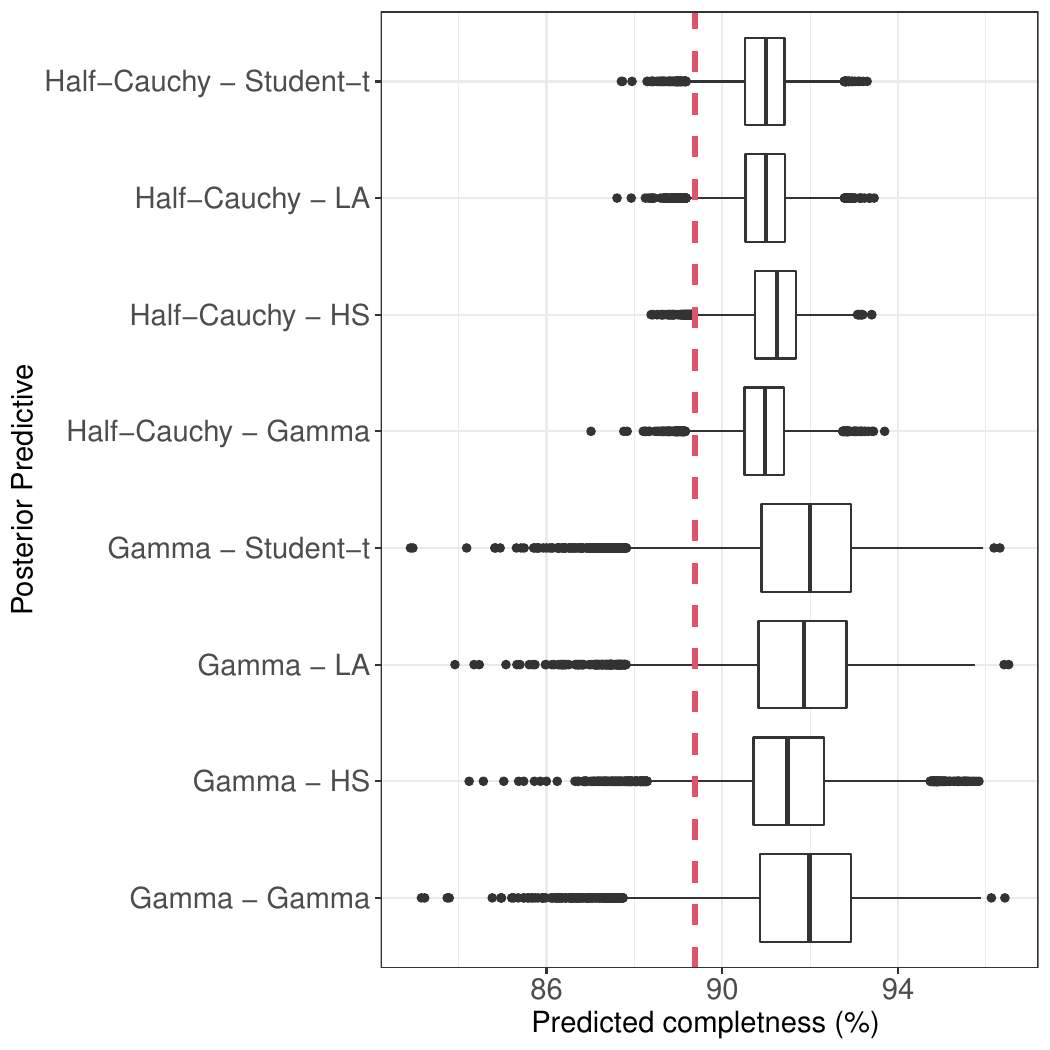} &
\includegraphics[width=0.65\textwidth]{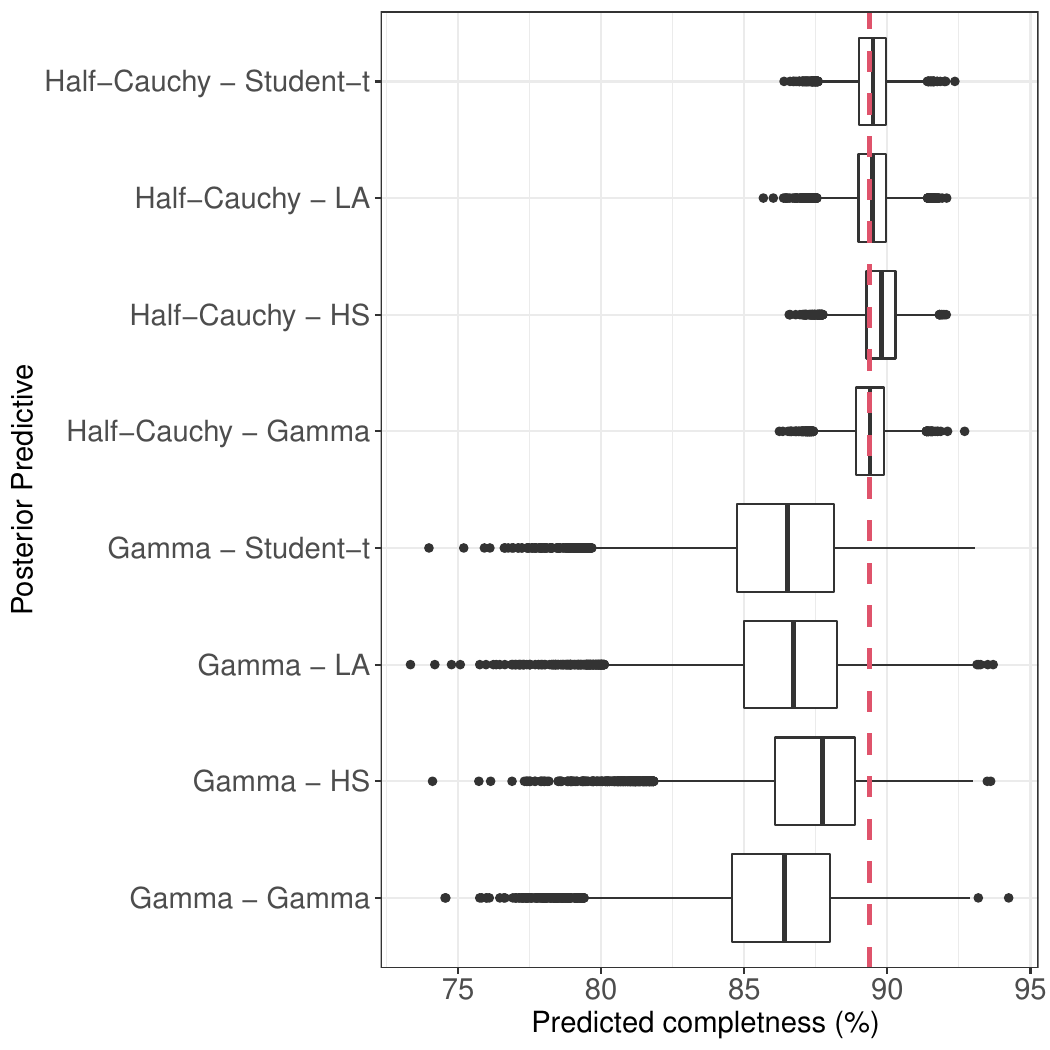}\\
\end{tabular}
\vspace{-0.5cm}
\caption{ \small Marginal posterior predictive distributions of the completeness at national level under Gamma, Student-t, HS and LA local scales for the random effects and
Half-Cauchy for the local scales of the errors, Gamma-Gamma refers to the model with a common Gamma prior distribution for the random effects and errors
respectively, Model 1 and 2, males. The red lines illustrates the observed values obtained from Census 2018.}
\label{fig:CDR2}
\end{figure}

In addition, there are  important reductions
in terms of MAE and MSE and slightly better posterior predictive estimates are obtained  when the HS is considered as a local scale for the random effects
compared with the Gamma prior distribution. We also illustrate
the posterior predictive distributions in Figures ~\ref{fig:CDR}-~\ref{fig:CDR3} at the national level. The inclusion of Half-Cauchy as a local-scale prior under
the errors that control the shrinkage under the posterior mean of the random effects makes a substantial difference. The posterior predictive distributions
obtained at the national level are more precise in these Figures. In addition the posterior predictive mean is closer
to the observed completeness at the National level in Colombia for 2017. These results are also observed for both sexes, males and females and models 1 and 2.
Further, the HS prior for the local scales of random effects produces posterior predictive estimates closer to
the observed completeness  with higher precision than the other local scale priors, as is displayed in Figures  ~\ref{fig:CDR}-~\ref{fig:CDR3}.

\begin{figure}[h!]
\begin{tabular}{ccc}
\hspace{1.2cm} Model 1 (Females) & \hspace{2.5cm}  Model 2 (Females) \\
\hspace{-2.5cm}\includegraphics[width=0.65\textwidth]{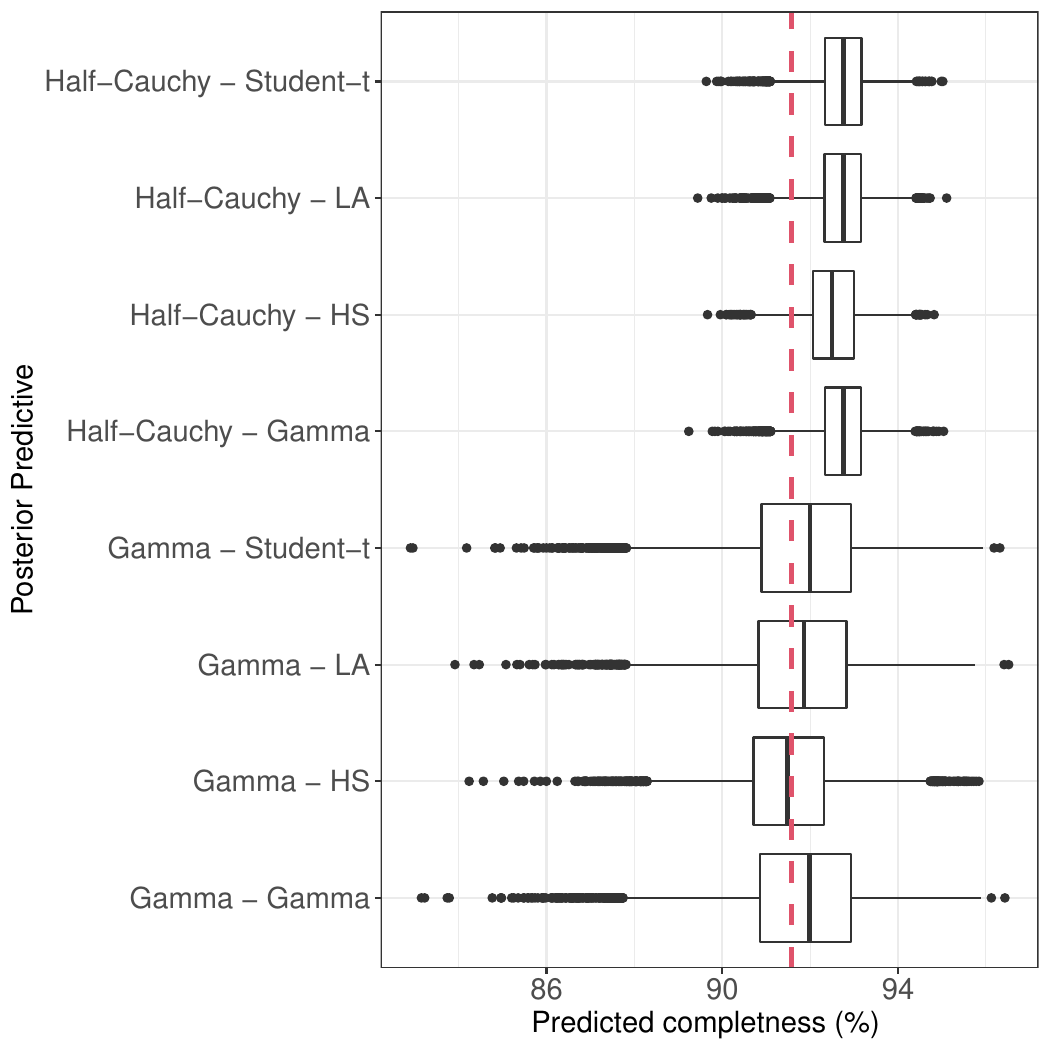} &
\includegraphics[width=0.65\textwidth]{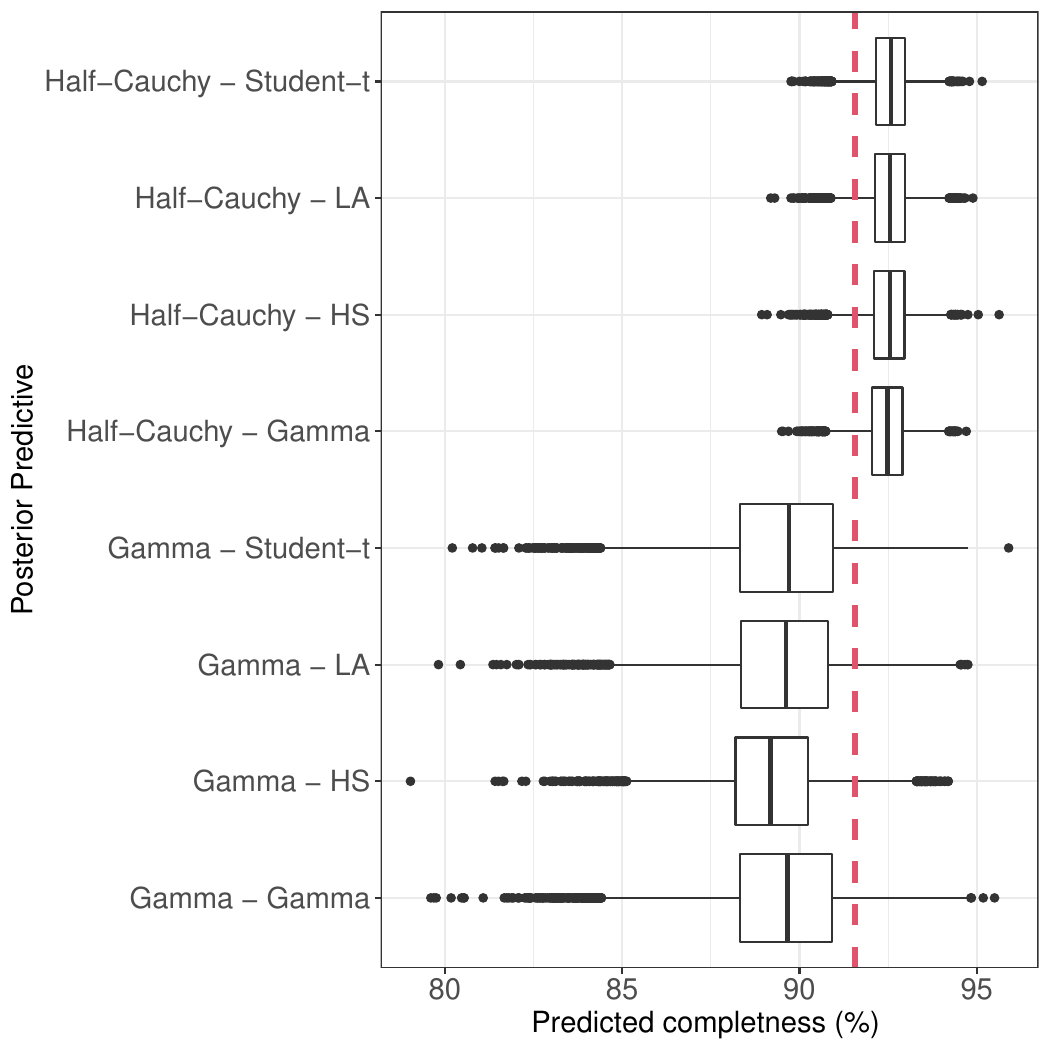}\\
\end{tabular}
\vspace{-0.5cm}
\caption{ \small Marginal posterior predictive distributions of the completeness at national level under Gamma, Student-t, HS and LA local scales for the random effects and
Half-Cauchy for the local scales of the errors, Gamma-Gamma refers to the model with a common Gamma prior distribution for the random effects and errors
respectively,  Model 1 and 2, females. The red lines illustrates the truth values obtained from Census 2018.}
\label{fig:CDR3}
\end{figure}


Table \ref{tab:tab_measures2} and  Tables \ref{tab:tab_measures3}-\ref{tab:tab_measures6}  in  the supplementary material in Section \ref{sup_post2} show the posterior predictive mean estimates and the corresponding credible intervals for models
1 and 2 both sexes, females and males, respectively.  As is presented in Tables \ref{tab:tab_measures2} and \ref{tab:tab_measures3}-\ref{tab:tab_measures6} most of the posterior mean estimates
are closer to the observed value of completeness in the departments of Colombia in 2017 and for almost all departments the credible intervals under GL priors contain  the observed completeness. These important results are also obtained when models 1 and 2 are considered for females and males as shown in Tables
\ref{tab:tab_measures3}-\ref{tab:tab_measures6}. Departments such as  Choc\'o,  Amazonas,  Vaup\'es and  La Guajira  with some of the higher  mortality rates and historically affected by problems in the 
civil registration system present important lower values of completeness according to Census.  Considering
GL priors for the random effects and errors in models 1 and 2 allow us to obtain posterior predictive estimates closer to the observed completeness for those departments.
Figure \ref{fig:post1} displays the spatial patterns of the observed completeness according to the Census 2018 in Colombia and posterior predictive estimates obtained
by considering a HS prior for the scales of the random effects and a Half-Cauchy prior for scale of the errors.  As is illustrated in the Figures,
Vaup\'es,  Guain\'ia,  Vichada,   Amazonas,  La Guajira  and  Choc\'o have  lower completeness values of death registration observed
in the Census. Those completeness patterns in the departments of Colombia  are also effectively estimated by considering the proposed models.
Model 2 with respect to Model 1 for both sexes, females and males produce posterior predictive completeness  estimates slightly closer  (for those departments) to the obtained considering the comparator from Census.


\clearpage

\begin{table}[ht]
\tiny
\renewcommand{\arraystretch}{1.2}
\scalebox{0.60}{
\begin{tabular}{r>{\columncolor[gray]{0.8}}rrrrrrrrrrrrrr}
\hline \hline
  &&&&& Model 1 (Both Sexes) &&&&& \\  \hline
\hline
  Department & Census  & Gamma &  $2.5\%$ &  $97.5\%$ & Student-t &  $2.5\%$ &  $97.5\%$ & LA &  $2.5\%$ &  $97.5\%$ & HS &  $2.5\%$ &  $97.5\%$  \\
 \hline
 Antioquia & 92 & 95 & 94 & 96 & 95 & 94 & 96 & 95 & 94 & 96 & 95 & 94 & 96 \\
Atl\'antico  & 96 & 94 & 92 & 95 & 94 & 92 & 94 & 94 & 92 & 94 & 94 & 92 & 94 \\
Bogot\'a, D.C. & 96 & 92 & 91 & 93 & 92 & 91 & 93 & 92 & 91 & 93 & 92 & 91 & 93 \\
 Bol\'ivar & 88 & 90 & 89 & 92 & 90 & 89 & 92 & 90 & 89 & 92 & 90 & 89 & 91 \\
Boyac\'a & 93 & 93 & 91 & 94 & 93 & 91 & 94 & 93 & 91 & 94 & 93 & 92 & 94 \\
 Caldas & 93 & 97 & 97 & 98 & 98 & 97 & 98 & 98 & 97 & 98 & 98 & 97 & 98 \\
 Caquet\'a & 87 & 93 & 92 & 94 & 93 & 92 & 94 & 93 & 92 & 94 & 93 & 92 & 94 \\
 Cauca & 83 & 86 & 84 & 88 & 86 & 84 & 88 & 86 & 84 & 88 & 86 & 84 & 87 \\
 Cesar & 88 & 84 & 81 & 86 & 83 & 81 & 86 & 83 & 81 & 86 & 83 & 81 & 85 \\
C\'ordoba & 88 & 86 & 84 & 88 & 86 & 84 & 88 & 86 & 84 & 88 & 86 & 84 & 88 \\
 Cundinamarca & 93 & 96 & 95 & 96 & 96 & 95 & 96 & 96 & 95 & 96 & 96 & 95 & 96 \\
  \rowcolor[gray]{0.8}
 Choc\'o & 65 & 61 & 56 & 66 & 61 & 56 & 66 & 61 & 56 & 65 & 58 & 54 & 63 \\
 Huila & 92 & 91 & 89 & 92 & 91 & 89 & 92 & 91 & 89 & 92 & 91 & 90 & 92 \\
  \rowcolor[gray]{0.8}
 La Guajira & 39 & 52 & 46 & 57 & 51 & 46 & 56 & 51 & 45 & 56 & 48 & 43 & 52 \\
 Magdalena & 88 & 90 & 88 & 91 & 90 & 88 & 91 & 90 & 88 & 91 & 89 & 88 & 91 \\
 Meta & 91 & 96 & 95 & 97 & 96 & 95 & 97 & 96 & 95 & 97 & 96 & 96 & 97 \\
 Nari\~no & 82 & 85 & 82 & 87 & 85 & 83 & 87 & 85 & 83 & 87 & 85 & 83 & 87 \\
 Norte de Santander & 90 & 93 & 91 & 94 & 93 & 91 & 94 & 93 & 91 & 94 & 93 & 92 & 93 \\
 Quindio & 94 & 96 & 95 & 96 & 96 & 95 & 96 & 96 & 95 & 96 & 96 & 95 & 96 \\
 Risaralda & 93 & 96 & 96 & 97 & 96 & 96 & 97 & 96 & 96 & 97 & 97 & 96 & 97 \\
 Santander & 94 & 92 & 91 & 93 & 92 & 91 & 93 & 92 & 91 & 93 & 92 & 91 & 93 \\
 Sucre & 86 & 93 & 92 & 94 & 93 & 92 & 94 & 93 & 92 & 94 & 93 & 92 & 94 \\
 Tolima & 91 & 91 & 89 & 92 & 91 & 89 & 92 & 91 & 89 & 92 & 91 & 89 & 92 \\
 Valle del Cauca & 95 & 96 & 95 & 96 & 96 & 95 & 96 & 96 & 95 & 96 & 96 & 95 & 97 \\
 Arauca & 85 & 96 & 95 & 97 & 96 & 95 & 97 & 96 & 95 & 97 & 96 & 96 & 97 \\
 Casanare & 87 & 94 & 93 & 95 & 94 & 93 & 95 & 94 & 93 & 95 & 94 & 93 & 95 \\
 Putumayo & 79 & 84 & 81 & 86 & 83 & 81 & 86 & 83 & 81 & 85 & 83 & 81 & 85 \\
San Andr\'es & 93 & 95 & 94 & 96 & 95 & 94 & 96 & 95 & 94 & 96 & 95 & 94 & 96 \\
  \rowcolor[gray]{0.75}
 Amazonas & 67 & 70 & 66 & 74 & 70 & 66 & 74 & 70 & 66 & 73 & 68 & 65 & 71 \\
 Guain\'ia & 51 & 87 & 84 & 89 & 86 & 83 & 89 & 86 & 83 & 88 & 86 & 83 & 88 \\
 Guaviare & 81 & 94 & 93 & 95 & 94 & 93 & 95 & 94 & 93 & 95 & 94 & 93 & 95 \\
  \rowcolor[gray]{0.8}
 Vaup\'es & 57 & 60 & 56 & 65 & 60 & 55 & 64 & 59 & 55 & 64 & 57 & 54 & 61 \\
 Vichada & 60 & 86 & 83 & 88 & 86 & 82 & 88 & 85 & 82 & 88 & 85 & 82 & 87 \\
 \rowcolor[gray]{0.8}
 National & 90 & 92 & 90 & 93 & 92 & 90 & 93 & 92 & 90 & 93 & 92 & 90 & 93 \\   \hline
\end{tabular}}
\tiny
\renewcommand{\arraystretch}{1.2}
\scalebox{0.60}{
\begin{tabular}{r>{\columncolor[gray]{0.8}}rrrrrrrrrrrrrr}
\hline \hline
  &&&&& Model 2 (Both Sexes) &&&&& \\  \hline
\hline
  Department & Census & Gamma &  $2.5\%$ &  $97.5\%$ & Student-t &  $2.5\%$ &  $97.5\%$ & LA &  $2.5\%$ &  $97.5\%$ & HS &  $2.5\%$ &  $97.5\%$  \\
 \hline
Antioquia & 92 & 95 & 94 & 95 & 95 & 94 & 95 & 95 & 94 & 95 & 95 & 94 & 95 \\
Atl\'antico & 96 & 92 & 91 & 93 & 92 & 91 & 93 & 92 & 91 & 93 & 92 & 91 & 93 \\
Bogot\'a, D.C. & 96 & 91 & 90 & 92 & 91 & 90 & 92 & 91 & 90 & 92 & 91 & 90 & 92 \\
 Bol\'ivar & 88 & 88 & 86 & 89 & 88 & 86 & 89 & 88 & 86 & 89 & 88 & 86 & 89 \\
 Boyac\'a & 93 & 93 & 92 & 94 & 93 & 92 & 94 & 93 & 92 & 94 & 93 & 92 & 94 \\
Caldas & 93 & 97 & 97 & 98 & 97 & 97 & 98 & 97 & 97 & 98 & 98 & 97 & 98 \\
 Caquet\'a & 87 & 92 & 91 & 93 & 92 & 91 & 93 & 92 & 91 & 93 & 92 & 91 & 93 \\
 Cauca & 83 & 84 & 82 & 86 & 84 & 82 & 86 & 84 & 82 & 86 & 84 & 82 & 86 \\
 Cesar & 88 & 80 & 78 & 82 & 80 & 77 & 82 & 80 & 77 & 82 & 80 & 77 & 82 \\
C\'ordoba & 88 & 83 & 80 & 85 & 83 & 80 & 85 & 83 & 80 & 85 & 83 & 81 & 85 \\
 Cundinamarca & 93 & 94 & 94 & 95 & 95 & 94 & 95 & 95 & 94 & 95 & 95 & 94 & 95 \\
  \rowcolor[gray]{0.8}
  Choc\'o & 65 & 53 & 48 & 57 & 53 & 48 & 57 & 52 & 48 & 57 & 52 & 47 & 56 \\
 Huila & 92 & 92 & 90 & 93 & 92 & 90 & 93 & 92 & 90 & 93 & 92 & 90 & 93 \\
  \rowcolor[gray]{0.8}
 La Guajira & 39 & 46 & 41 & 51 & 45 & 41 & 50 & 45 & 40 & 50 & 44 & 40 & 48 \\
 Magdalena & 88 & 86 & 85 & 88 & 86 & 85 & 88 & 86 & 85 & 88 & 86 & 85 & 88 \\
 Meta & 91 & 95 & 94 & 95 & 95 & 94 & 95 & 95 & 94 & 95 & 95 & 94 & 95 \\
 Nari\~no & 82 & 86 & 83 & 88 & 86 & 83 & 88 & 86 & 83 & 88 & 86 & 84 & 88 \\
 Norte de Santander & 90 & 93 & 92 & 94 & 93 & 92 & 94 & 93 & 92 & 94 & 93 & 92 & 94 \\
 Quindio & 94 & 96 & 95 & 97 & 96 & 95 & 97 & 96 & 95 & 97 & 96 & 96 & 97 \\
 Risaralda & 93 & 96 & 96 & 97 & 96 & 96 & 97 & 96 & 96 & 97 & 96 & 96 & 97 \\
 Santander & 94 & 93 & 92 & 94 & 93 & 92 & 94 & 93 & 92 & 94 & 93 & 92 & 94 \\
 Sucre & 86 & 91 & 90 & 92 & 91 & 90 & 92 & 91 & 90 & 92 & 91 & 90 & 92 \\
 Tolima & 91 & 93 & 91 & 94 & 93 & 91 & 94 & 93 & 92 & 94 & 93 & 92 & 94 \\
 Valle del Cauca & 95 & 96 & 95 & 97 & 96 & 95 & 97 & 96 & 95 & 97 & 96 & 96 & 97 \\
 Arauca & 85 & 95 & 94 & 95 & 95 & 94 & 95 & 95 & 94 & 95 & 95 & 94 & 95 \\
 Casanare & 87 & 91 & 90 & 92 & 91 & 90 & 92 & 91 & 90 & 92 & 91 & 90 & 92 \\
 Putumayo & 79 & 82 & 80 & 85 & 82 & 80 & 84 & 82 & 80 & 84 & 82 & 80 & 84 \\
San Andr\'es & 93 & 94 & 93 & 94 & 94 & 93 & 94 & 94 & 93 & 94 & 94 & 93 & 94 \\
  \rowcolor[gray]{0.8}
 Amazonas & 67 & 63 & 59 & 67 & 63 & 59 & 67 & 63 & 59 & 67 & 62 & 59 & 66 \\
 Guain\'ia & 51 & 75 & 72 & 78 & 75 & 72 & 78 & 75 & 72 & 78 & 75 & 72 & 78 \\
 Guaviare & 81 & 92 & 90 & 93 & 92 & 90 & 93 & 92 & 90 & 93 & 92 & 91 & 93 \\
  \rowcolor[gray]{0.8}
 Vaup\'es & 57 & 55 & 51 & 59 & 55 & 50 & 59 & 54 & 50 & 59 & 54 & 50 & 58 \\
 Vichada & 60 & 73 & 69 & 76 & 73 & 69 & 76 & 73 & 69 & 76 & 72 & 69 & 76 \\
 \rowcolor[gray]{0.8}
 National & 90 & 91 & 90 & 92 & 91 & 90 & 92 & 91 & 90 & 92 & 91 & 90 & 92 \\
   \hline
\end{tabular}}
\caption{\small
Posterior predictive estimates and 95\% credible intervals under Gamma, HS, LA and Student-t local priors for the random effects and  a Half-Cauchy local prior for the scale of the errors and  the observed estimates from
the census. Inf
and Sup represent the 2.5\%  and 97.5\% quantiles from the  marginal posterior distribution
of the completeness in each Department and at the national level. Both Sexes and models 1 and 2. }
\label{tab:tab_measures2}
\end{table}

\clearpage

\begin{figure}[ht]
\begin{center}
\scriptsize
\begin{tabular}{ccc}
Census 2018 (Both-sexes) & Model 1 (Both Sexes) \hspace{-2cm}  & Model 2 (Both Sexes) \hspace{-2cm}  \\
\hspace{-3cm} & HS + Half-Cauchy  \hspace{-2cm}  & HS + Half-Cauchy \hspace{-2cm}   \vspace{-0.5cm} \\
\hspace{-3cm}  \includegraphics[width=0.43\textwidth]{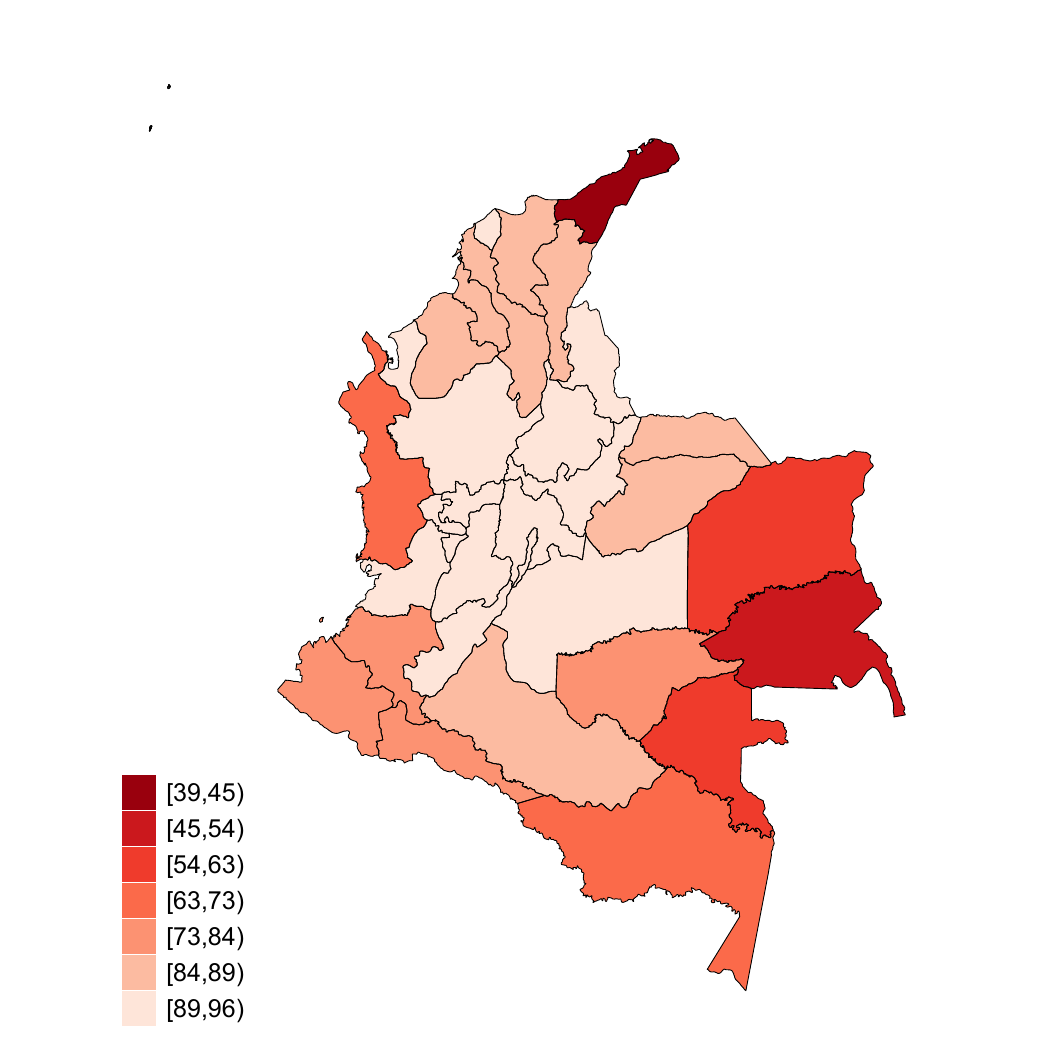} \hspace{-2cm} & \includegraphics[width=0.43\textwidth]{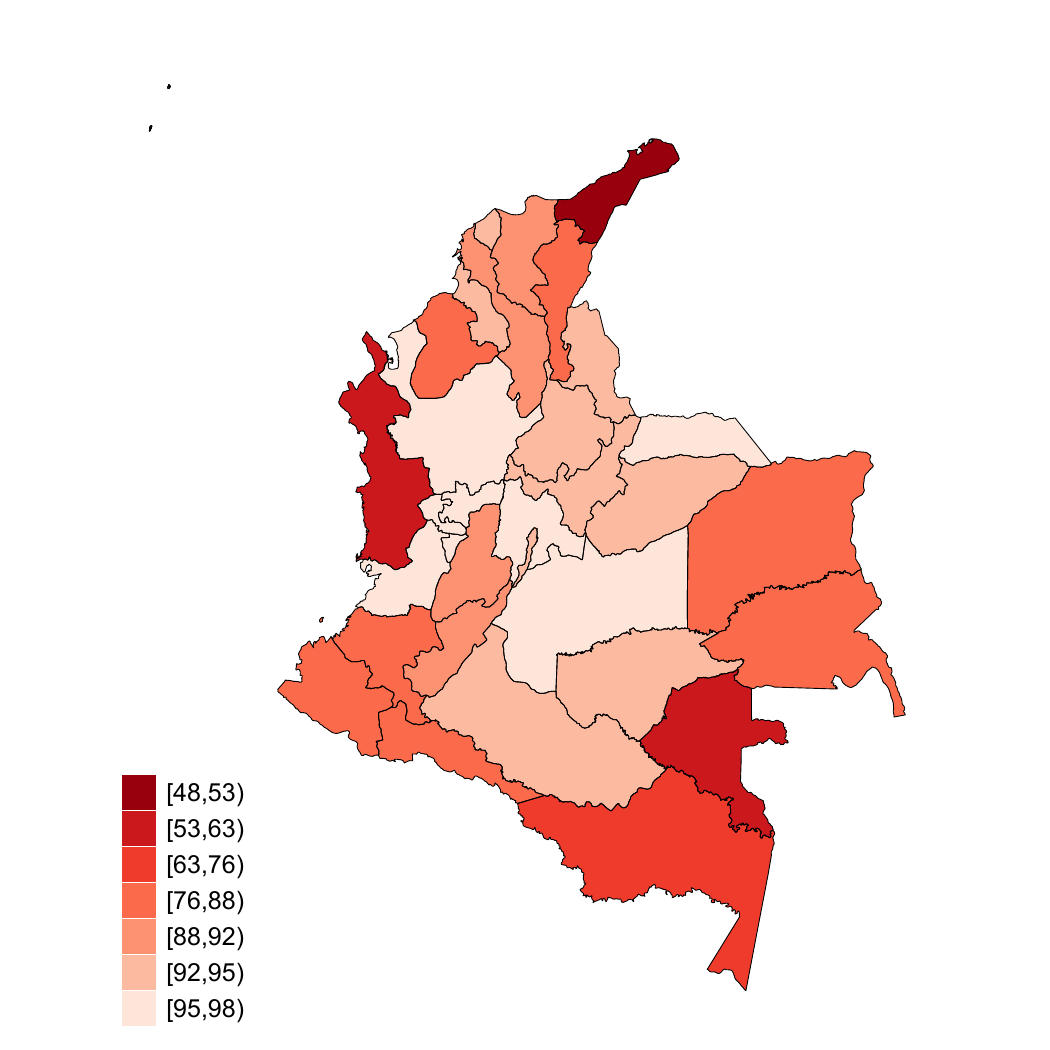}  \hspace{-2cm}  & \includegraphics[width=0.43\textwidth]{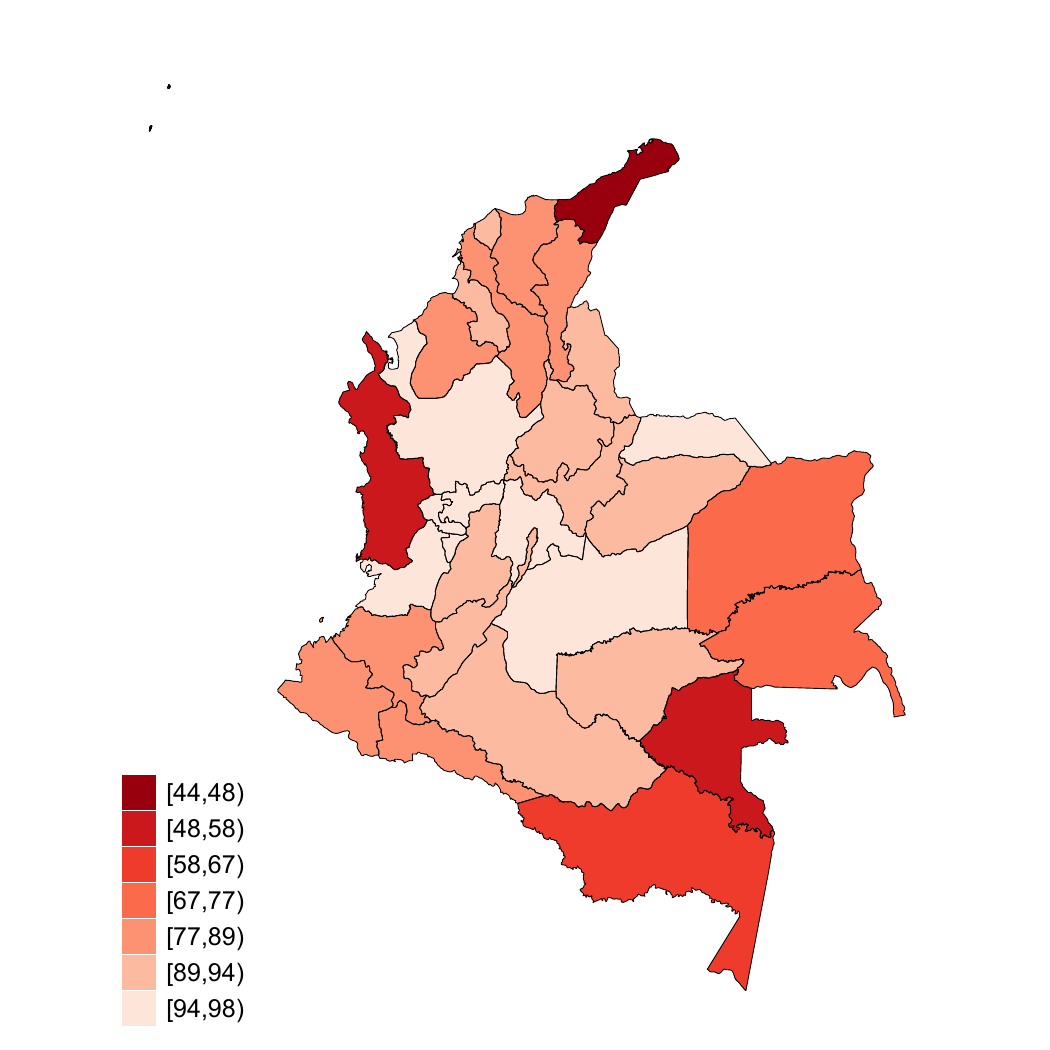}  \\
Census 2018 (Female) & Model 1 (Female) \hspace{-2cm}  & Model 2 (Female) \hspace{-2cm}  \\
\hspace{-3cm} & HS + Half-Cauchy  \hspace{-2cm}  & HS + Half-Cauchy \hspace{-2cm}   \vspace{-0.5cm} \\
\hspace{-3cm}  \includegraphics[width=0.43\textwidth]{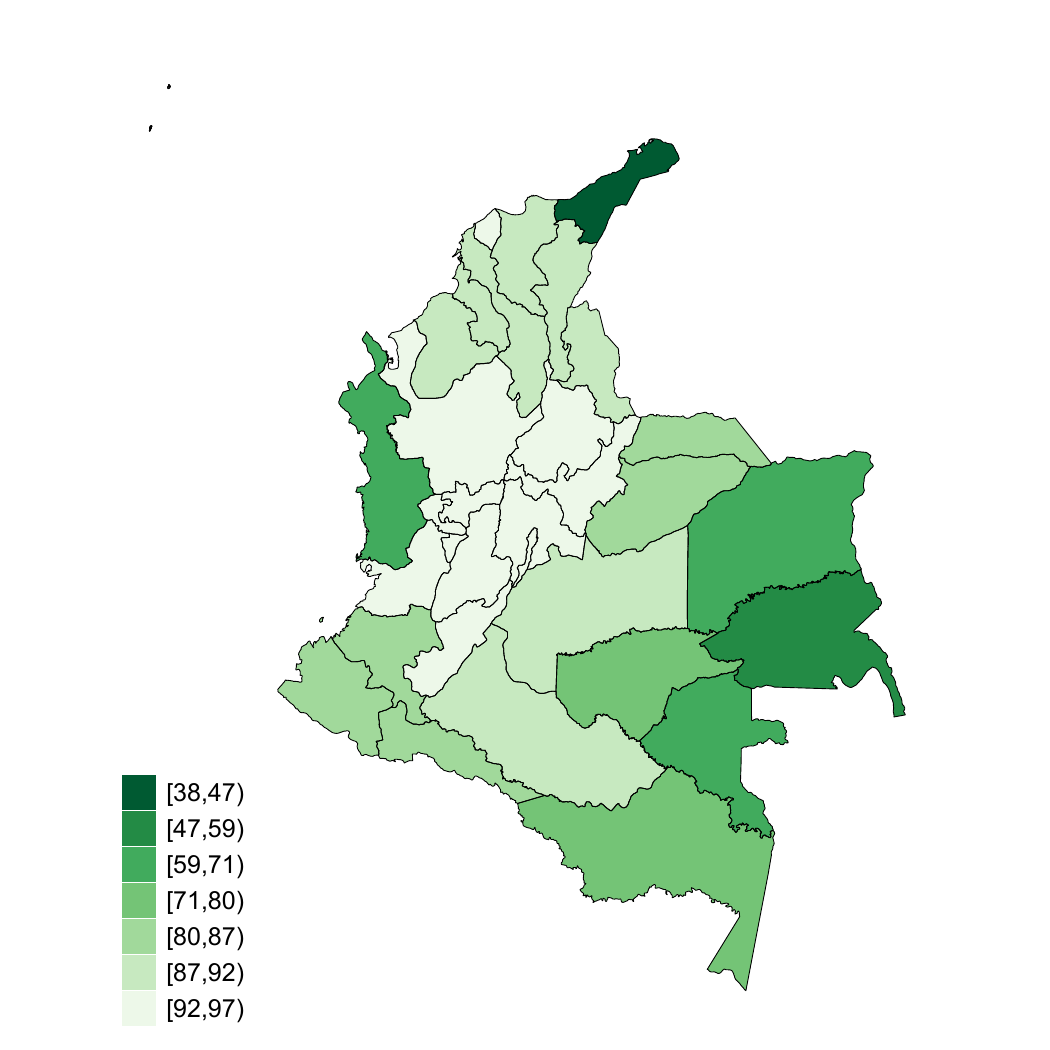} \hspace{-2cm} & \includegraphics[width=0.43\textwidth]{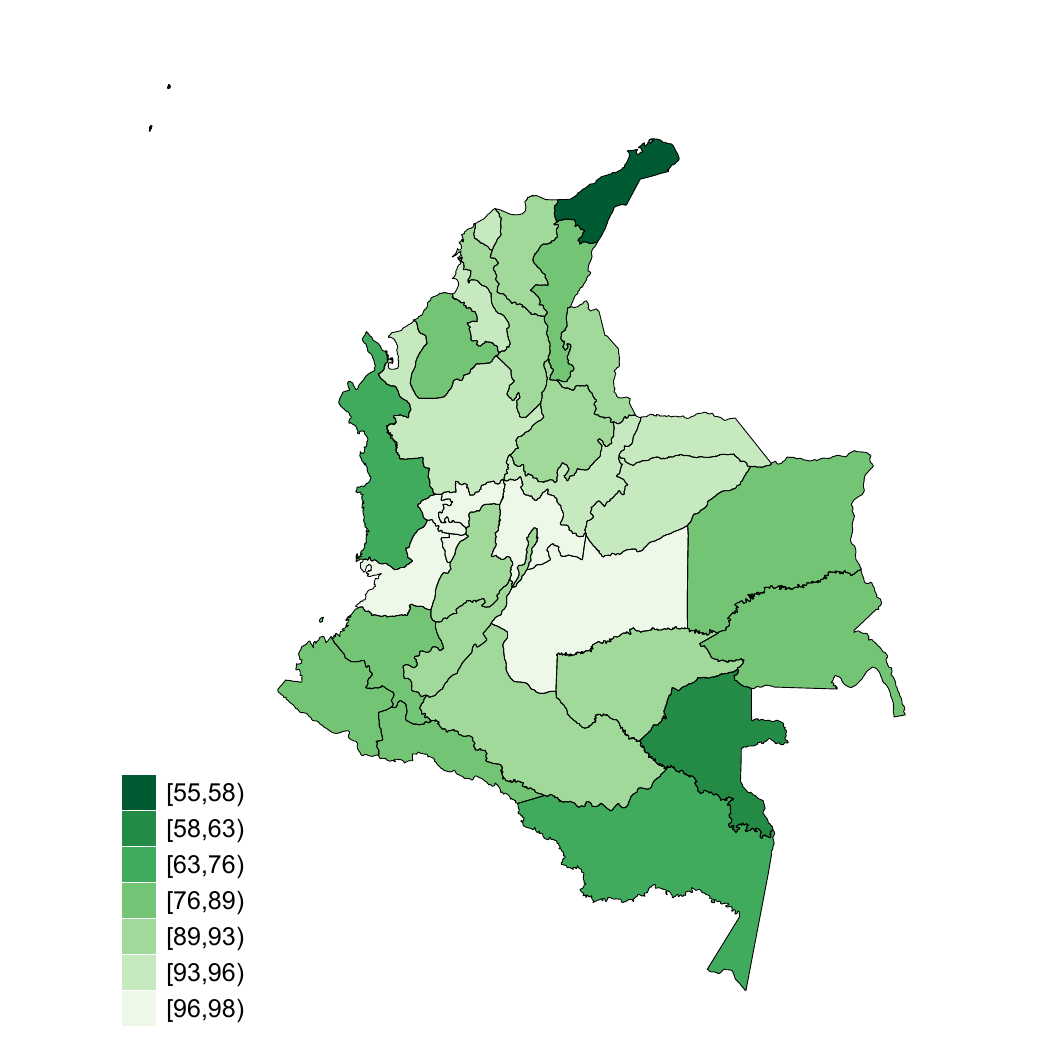}  \hspace{-2cm}  & \includegraphics[width=0.43\textwidth]{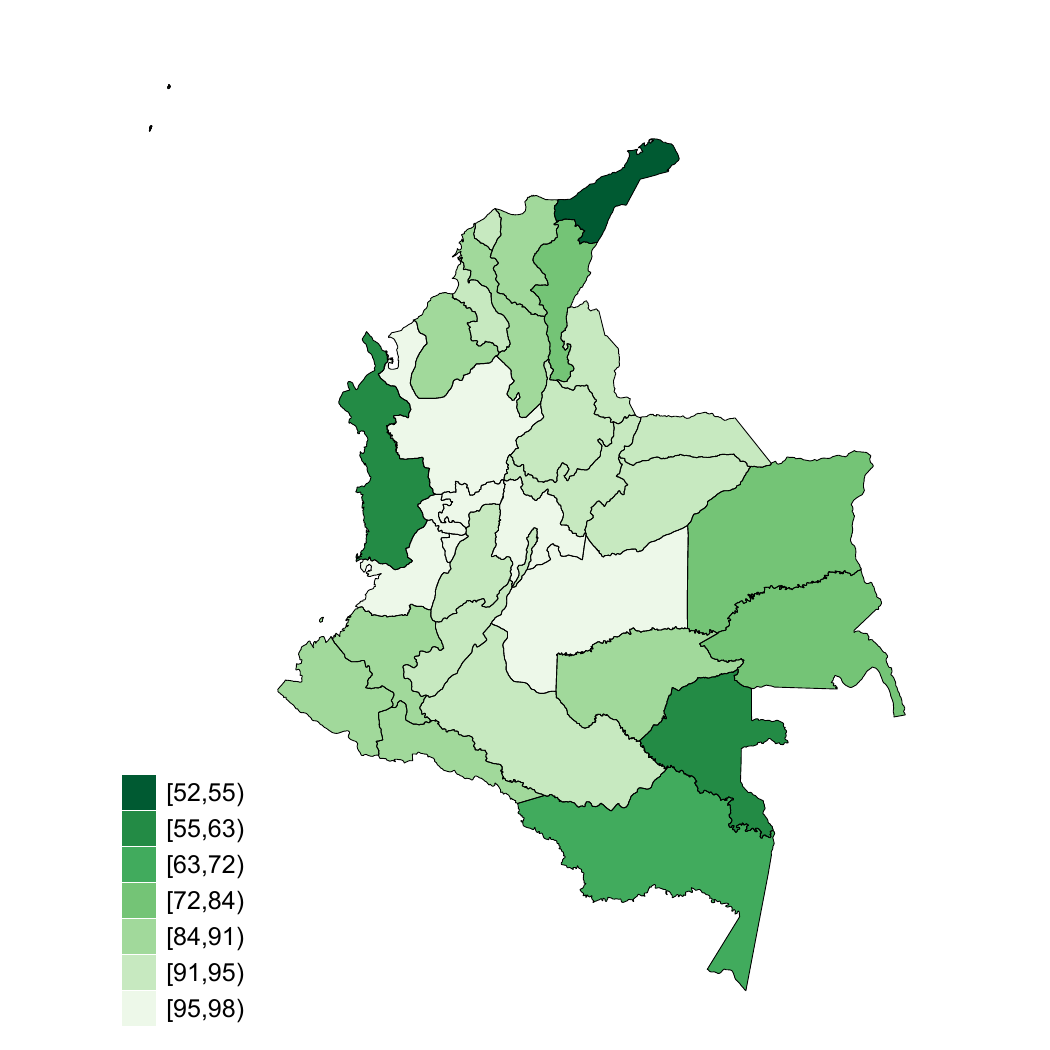}  \\
Census 2018 (Male) & Model 1 (Male) \hspace{-2cm}  & Model 2 (Male) \hspace{-2cm}  \\
\hspace{-3cm}  & HS + Half-Cauchy  \hspace{-2cm}  & HS + Half-Cauchy \hspace{-2cm}   \vspace{-0.5cm} \\
\hspace{-3cm}  \includegraphics[width=0.43\textwidth]{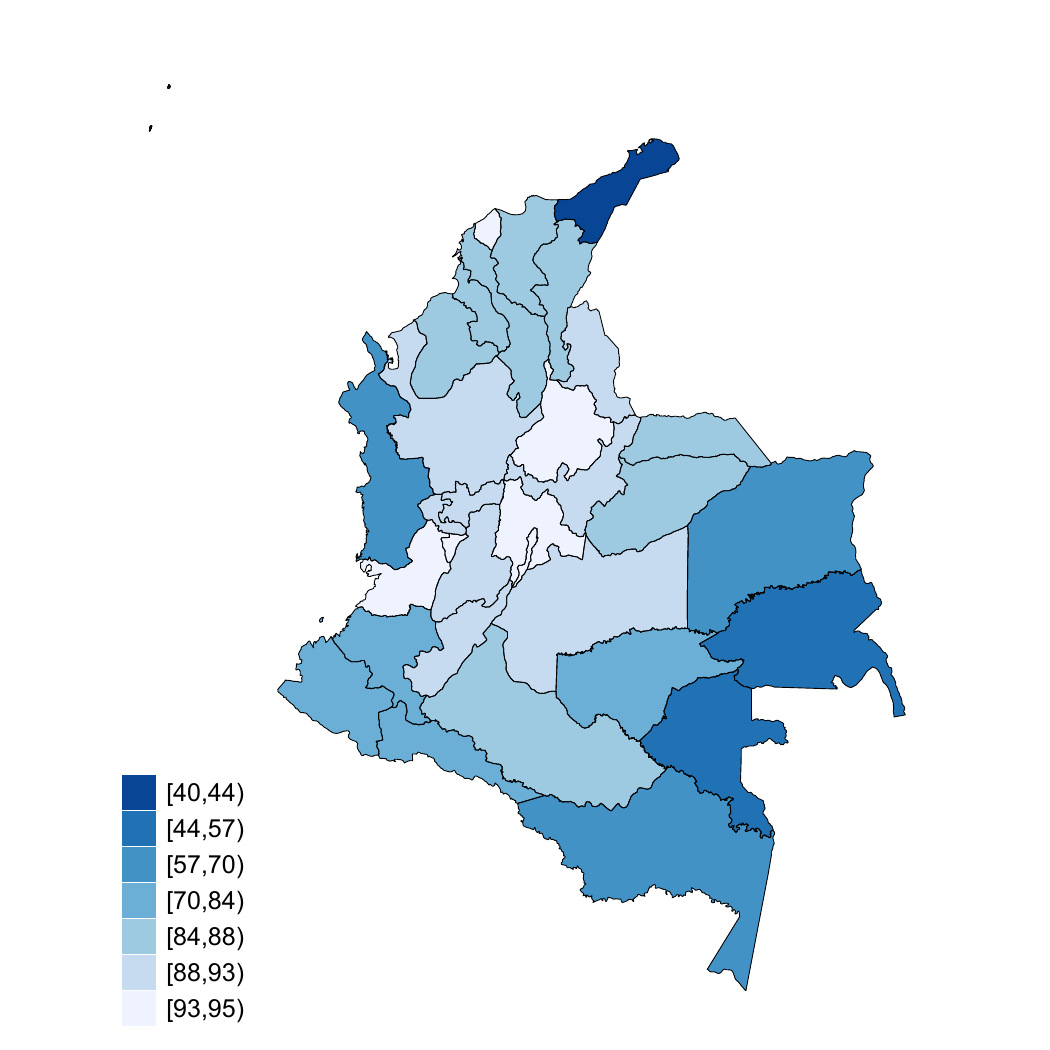} \hspace{-2cm} & \includegraphics[width=0.43\textwidth]{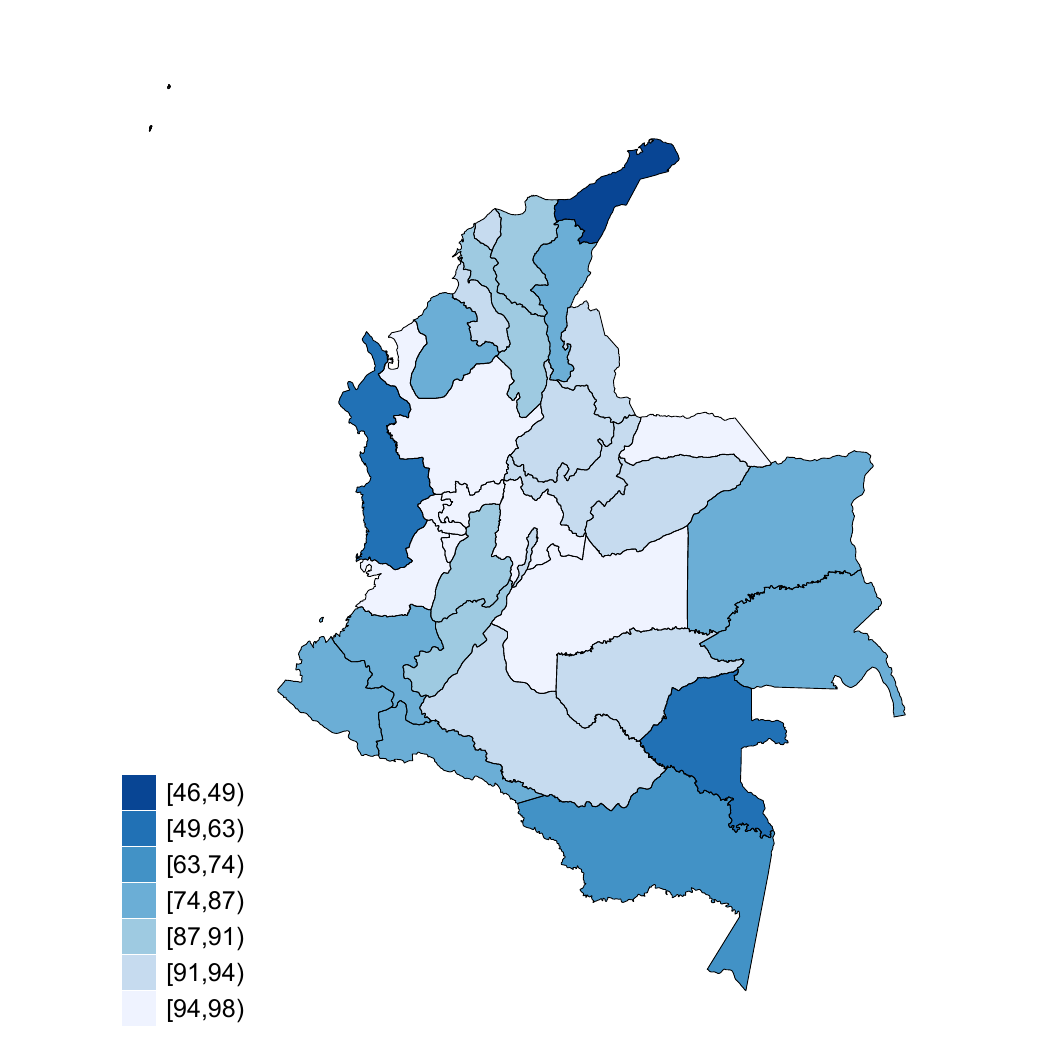}  \hspace{-2cm}  & \includegraphics[width=0.43\textwidth]{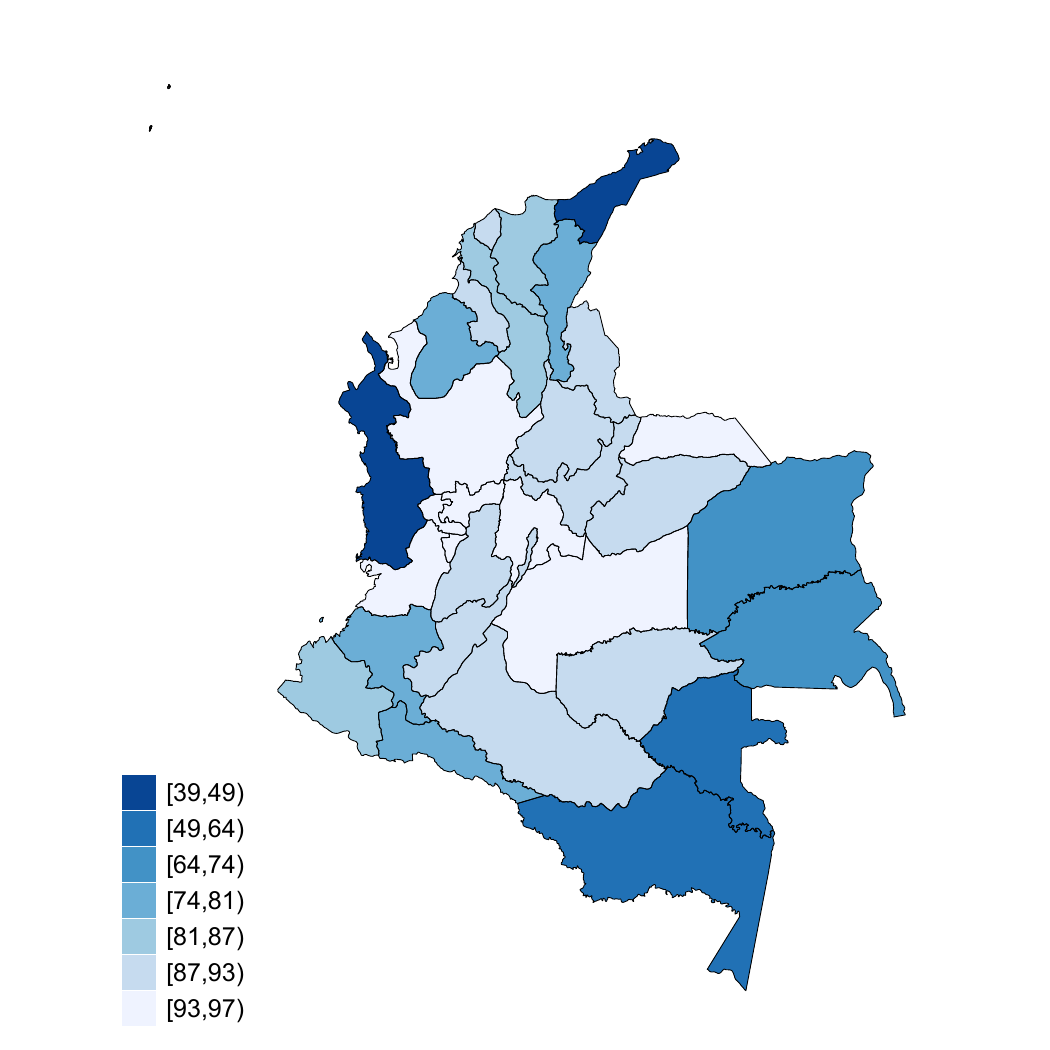}  \\
\end{tabular}
\end{center}
\vspace{-0.5cm}
\caption{\small
Completeness of death registration at  departmental level in Colombia according to  Census 2018 (left), Bayesian Model 1 (middle) and Bayesian model 2 (right) with HS priors for the local scales of the random effects and a
Half-Cauchy for the local scales of the errors by both sexes, females and males.
\label{fig:post1}}
\end{figure}

\clearpage

\section{Concluding remarks and discussion}
\label{Conclusions}

In this work we propose the use of GL priors  in this paper for situations where demographic covariates can explain much of  the observed completeness but where unexplained within-country (i.e. by year) and between-country variability also plays an important role. Therefore, we introduce GL priors for the scales of the random effects and errors respectively in hierarchical linear
mixed models to estimate the completeness of death registration and study their theoretical properties to allow small and large random effects
when it is required. Our proposal is inspired by the original model proposed by  \cite{adair2018estimating} which  includes information typically available from multiple sources, e.g., surveys, censuses and administrative records. 

The use of a Bayesian framework with GL priors to extend the existing \cite{adair2018estimating}  models further strengthens the models effectiveness at estimating completeness of death registration. These models are immensely useful to enable estimation of completeness of death registration in a timely manner using available data, as demonstrated by their wide use in several different settings \citep{zeng2020measuring, sempe2021estimation, adair2021monitoring}. In particular, they overcome the limitations of existing methods which rely on inaccurate assumptions of population dynamics and often provide estimates that lack timeliness. Furthermore, this paper uses an updated GBD dataset using data to 2019 from 120 countries and 2,748 country-years and also improves the model's ability to predict completeness in contemporary populations. 
 
A limitation of the use of the Colombia Census as a comparator is that the data are based on households responding to a question about whether reported deaths, which comprise about 90\% of total deaths in the country, were registered. These figures may differ from actual death registration because of incorrect recollection by the household, whether intentional or otherwise, and also if death registration status of the 10\% of unreported deaths differ substantially from the 90\% of deaths that are reported. However the similarity in results for the predicted registration completeness and the completeness reported in the Census suggests that any such issues are minimal.

We also proposed new Markov chain Monte Carlo (MCMC) algorithms for hierarchical linear mixed models under GL priors to estimate the uncertainty of
the estimated completeness of the death registration  at the global, national and subnational levels. Finally, the methodological results in this paper
are complementary to existing alternatives under GL priors  and can be implemented in a general framework to estimate other demographic indicators at subnational levels.


\section*{Acknowledgments}
Jairo F\'uquene-Pati\~no was supported for the CAMPOS scholar initiative, UC Davis.
\bibliographystyle{plainnat}
\bibliography{bibliosmallb}

\begin{thebibliography}{48}
\providecommand{\natexlab}[1]{#1}
\providecommand{\url}[1]{\texttt{#1}}
\expandafter\ifx\csname urlstyle\endcsname\relax
  \providecommand{\doi}[1]{doi: #1}\else
  \providecommand{\doi}{doi: \begingroup \urlstyle{rm}\Url}\fi

\bibitem[Adair and Lopez(2018)]{adair2018estimating}
Tim Adair and Alan~D Lopez.
\newblock Estimating the completeness of death registration: An empirical
  method.
\newblock \emph{PloS one}, 13\penalty0 (5):\penalty0 e0197047, 2018.

\bibitem[Adair et~al.(2021)Adair, Firth, Phyo, Bo, and
  Lopez]{adair2021monitoring}
Tim Adair, Sonja Firth, Tint Pa~Pa Phyo, Khin~Sandar Bo, and Alan~D Lopez.
\newblock Monitoring progress with national and subnational health goals by
  integrating verbal autopsy and medically certified cause of death data.
\newblock \emph{BMJ global health}, 6\penalty0 (5):\penalty0 e005387, 2021.

\bibitem[Alkema et~al.(2017)Alkema, Zhang, Chou, Gemmill, Moller, Fat, Say,
  Mathers, and Hogan]{alkema2017bayesian}
Leontine Alkema, Sanqian Zhang, Doris Chou, Alison Gemmill, Ann-Beth Moller,
  Doris~Ma Fat, Lale Say, Colin Mathers, and Daniel Hogan.
\newblock A bayesian approach to the global estimation of maternal mortality.
\newblock \emph{The Annals of Applied Statistics}, 11\penalty0 (3):\penalty0
  1245--1274, 2017.

\bibitem[Armagan et~al.(2013)Armagan, Dunson, and Lee]{armagan2013generalized}
Artin Armagan, David~B Dunson, and Jaeyong Lee.
\newblock Generalized double pareto shrinkage.
\newblock \emph{Statistica Sinica}, 23\penalty0 (1):\penalty0 119, 2013.

\bibitem[Azose and Raftery(2015)]{azose2015bayesian}
Jonathan~J Azose and Adrian~E Raftery.
\newblock Bayesian probabilistic projection of international migration.
\newblock \emph{Demography}, 52\penalty0 (5):\penalty0 1627--1650, 2015.

\bibitem[Bai and Ghosh(2018)]{bai2018high}
Ray Bai and Malay Ghosh.
\newblock High-dimensional multivariate posterior consistency under
  global--local shrinkage priors.
\newblock \emph{Journal of Multivariate Analysis}, 167:\penalty0 157--170,
  2018.

\bibitem[Basu and Adair(2021)]{basu2021have}
Jayanta~Kumar Basu and Tim Adair.
\newblock Have inequalities in completeness of death registration between
  states in india narrowed during two decades of civil registration system
  strengthening?
\newblock \emph{International journal for equity in health}, 20\penalty0
  (1):\penalty0 1--9, 2021.

\bibitem[Bennett and Horiuchi(1984)]{bennett1984mortality}
Neil~G Bennett and Shiro Horiuchi.
\newblock Mortality estimation from registered deaths in less developed
  countries.
\newblock \emph{Demography}, 21\penalty0 (2):\penalty0 217--233, 1984.

\bibitem[Bhadra et~al.(2017)Bhadra, Datta, Polson, and
  Willard]{bhadra2017horseshoe+}
Anindya Bhadra, Jyotishka Datta, Nicholas~G Polson, and Brandon Willard.
\newblock The horseshoe+ estimator of ultra-sparse signals.
\newblock \emph{Bayesian Analysis}, 12\penalty0 (4):\penalty0 1105--1131, 2017.

\bibitem[Bhattacharya et~al.(2015)Bhattacharya, Pati, Pillai, and
  Dunson]{bhattacharya2015dirichlet}
Anirban Bhattacharya, Debdeep Pati, Natesh~S Pillai, and David~B Dunson.
\newblock Dirichlet--laplace priors for optimal shrinkage.
\newblock \emph{Journal of the American Statistical Association}, 110\penalty0
  (512):\penalty0 1479--1490, 2015.

\bibitem[Brass et~al.(1975)]{brass1975methods}
William Brass et~al.
\newblock Methods for estimating fertility and mortality from limited and
  defective data.
\newblock \emph{Methods for estimating fertility and mortality from limited and
  defective data.}, 1975.

\bibitem[Brown and Griffin(2010)]{brown2010inference}
Philip~J Brown and Jim~E Griffin.
\newblock Inference with normal-gamma prior distributions in regression
  problems.
\newblock \emph{Bayesian analysis}, 5\penalty0 (1):\penalty0 171--188, 2010.

\bibitem[Carvalho et~al.(2010)Carvalho, Polson, and
  Scott]{carvalho2010horseshoe}
Carlos~M Carvalho, Nicholas~G Polson, and James~G Scott.
\newblock The horseshoe estimator for sparse signals.
\newblock \emph{Biometrika}, 97\penalty0 (2):\penalty0 465--480, 2010.

\bibitem[DANE(2017{\natexlab{a}})]{DANE2017}
DANE.
\newblock Defunciones no fetales 2017.
\newblock Technical report, Departamento Administrativo Nacional de
  Estad{\'{i}}stica, 2017{\natexlab{a}}.

\bibitem[DANE(2017{\natexlab{b}})]{DANE20172}
DANE.
\newblock Proyecciones de poblaci{\'{o}}n.
\newblock Technical report, Departamento Administrativo Nacional de
  Estad{\'{i}}stica, 2017{\natexlab{b}}.

\bibitem[DANE(2018)]{DANE2018}
DANE.
\newblock Censo nacional de población y vivienda - cnpv - 2018.
\newblock Technical report, Departamento Administrativo Nacional de
  Estad{\'{i}}stica, 2018.

\bibitem[Demographics~Collaborators(2020)]{collaborators2020global}
GBD~2019 Demographics~Collaborators.
\newblock Global age-sex-specific fertility, mortality, healthy life expectancy
  (hale), and population estimates in 204 countries and territories,
  1950--2019: a comprehensive demographic analysis for the global burden of
  disease study 2019.
\newblock \emph{The Lancet}, 396\penalty0 (10258):\penalty0 1160--1203, 2020.

\bibitem[DHS(2010)]{Profamilia}
DHS.
\newblock Encuesta nacional de demograf{\'{i}}a y salud – ends 2011.
  profamilia.
\newblock Technical report, 2010.

\bibitem[DHS(2015)]{Profamilia2}
DHS.
\newblock Encuesta nacional de demograf{\'{i}}a y salud – ends 2015.
  profamilia.
\newblock Technical report, 2015.

\bibitem[Dicker et~al.(2018)Dicker, Nguyen, Abate, Abate, Abay, Abbafati,
  Abbasi, Abbastabar, Abd-Allah, Abdela, et~al.]{dicker2018global}
Daniel Dicker, Grant Nguyen, Degu Abate, Kalkidan~Hassen Abate, Solomon~M Abay,
  Cristiana Abbafati, Nooshin Abbasi, Hedayat Abbastabar, Foad Abd-Allah, Jemal
  Abdela, et~al.
\newblock Global, regional, and national age-sex-specific mortality and life
  expectancy, 1950--2017: a systematic analysis for the global burden of
  disease study 2017.
\newblock \emph{The lancet}, 392\penalty0 (10159):\penalty0 1684--1735, 2018.

\bibitem[Dorrington(2013{\natexlab{a}})]{dorrington2013brass}
R~Dorrington.
\newblock The brass growth balance method.
\newblock \emph{Tools for demographic estimation. Paris: IUSSP}, pages
  196--208, 2013{\natexlab{a}}.

\bibitem[Dorrington(2013{\natexlab{b}})]{dorrington2013generalized}
R~Dorrington.
\newblock The generalized growth balance method.
\newblock \emph{Tools for demographic estimation}, pages 258--274,
  2013{\natexlab{b}}.

\bibitem[Dorrington(2013{\natexlab{c}})]{dorrington2013synthetic}
Rob Dorrington.
\newblock Synthetic extinct generations methods.
\newblock \emph{Tools for demographic estimation}, pages 275--292,
  2013{\natexlab{c}}.

\bibitem[Fr{\"u}hwirth-Schnatter and Wagner(2011)]{bernardo2011bayesian}
S~Fr{\"u}hwirth-Schnatter and H~Wagner.
\newblock Bayesian variable selection for random intercept modeling of gaussian
  and non-gaussian data.
\newblock \emph{Bayesian statistics}, 9:\penalty0 165--185, 2011.

\bibitem[Gelfand and Smith(1990)]{gelfand1990sampling}
Alan~E Gelfand and Adrian~FM Smith.
\newblock Sampling-based approaches to calculating marginal densities.
\newblock \emph{Journal of the American statistical association}, 85\penalty0
  (410):\penalty0 398--409, 1990.

\bibitem[Ghosh et~al.(2016)Ghosh, Tang, Ghosh, and
  Chakrabarti]{ghosh2016asymptotic}
Prasenjit Ghosh, Xueying Tang, Malay Ghosh, and Arijit Chakrabarti.
\newblock Asymptotic properties of bayes risk of a general class of shrinkage
  priors in multiple hypothesis testing under sparsity.
\newblock \emph{Bayesian Analysis}, 11\penalty0 (3):\penalty0 753--796, 2016.

\bibitem[Godwin and Raftery(2017)]{godwin2017bayesian}
Jessica Godwin and Adrian~E Raftery.
\newblock Bayesian projection of life expectancy accounting for the hiv/aids
  epidemic.
\newblock \emph{Demographic research}, 37:\penalty0 1549, 2017.

\bibitem[Griffin and Brown(2005)]{griffin2005alternative}
JE~Griffin and PJ~Brown.
\newblock Alternative prior distributions for variable selection with very many
  more variables than observations.
\newblock Technical report, Technical report, University of Warwick, 2005.

\bibitem[Hill(1987)]{hill1987estimating}
Kenneth Hill.
\newblock Estimating census and death registration completeness.
\newblock In \emph{Asian and Pacific population forum/East-West Population
  Institute, East-West Center}, volume~1, pages 8--13. The Asian \& Pacific
  Population Forum, 1987.

\bibitem[Ju{\'a}rez and Steel(2010)]{juarez2010model}
Miguel~A Ju{\'a}rez and Mark~FJ Steel.
\newblock Model-based clustering of non-gaussian panel data based on skew-t
  distributions.
\newblock \emph{Journal of Business \& Economic Statistics}, 28\penalty0
  (1):\penalty0 52--66, 2010.

\bibitem[Mikkelsen et~al.(2015)Mikkelsen, Phillips, AbouZahr, Setel,
  De~Savigny, Lozano, and L~opez]{mikkelsen2015global}
Lene Mikkelsen, David~E Phillips, Carla AbouZahr, Philip~W Setel, Don
  De~Savigny, Rafael Lozano, and Alan~D L~opez.
\newblock A global assessment of civil registration and vital statistics
  systems: monitoring data quality and progress.
\newblock \emph{The Lancet}, 386\penalty0 (10001):\penalty0 1395--1406, 2015.

\bibitem[Murray et~al.(2010)Murray, Rajaratnam, Marcus, Laakso, and
  Lopez]{murray2010can}
Christopher~JL Murray, Julie~Knoll Rajaratnam, Jacob Marcus, Thomas Laakso, and
  Alan~D Lopez.
\newblock What can we conclude from death registration? improved methods for
  evaluating completeness.
\newblock \emph{PLoS medicine}, 7\penalty0 (4):\penalty0 e1000262, 2010.

\bibitem[Polson and Scott(2012)]{polson2012half}
Nicholas~G Polson and James~G Scott.
\newblock On the half-cauchy prior for a global scale parameter.
\newblock \emph{Bayesian Analysis}, 7\penalty0 (4):\penalty0 887--902, 2012.

\bibitem[Preston et~al.(1980)Preston, Coale, Trussell, and
  Weinstein]{preston1980estimating}
Samuel Preston, Ansley~J Coale, James Trussell, and Maxine Weinstein.
\newblock Estimating the completeness of reporting of adult deaths in
  populations that are approximately stable.
\newblock \emph{Population index}, pages 179--202, 1980.

\bibitem[Raftery et~al.(2014)Raftery, Alkema, and Gerland]{raftery2014bayesian}
Adrian~E Raftery, Leontine Alkema, and Patrick Gerland.
\newblock Bayesian population projections for the united nations.
\newblock \emph{Statistical science: a review journal of the Institute of
  Mathematical Statistics}, 29\penalty0 (1):\penalty0 58, 2014.

\bibitem[Rao and Kelly(2017)]{rao2017overview}
C~Rao and M~Kelly.
\newblock Overview of the principles and international experiences in
  implementing record linkage mechanisms to assess completeness of death
  registration.
\newblock \emph{Population Division, Department of Economic and Social Affairs,
  United Nations, New York}, \penalty0 (2017/5), 2017.

\bibitem[Rossell and Steel(2019)]{rossell2019continuous}
David Rossell and Mark~FJ Steel.
\newblock Continuous mixtures with skewness and heavy tails.
\newblock In \emph{Handbook of Mixture Analysis}, pages 219--237. Chapman and
  Hall/CRC, 2019.

\bibitem[Rubio and Steel(2015)]{rubio2015bayesian}
FJ~Rubio and MFJ Steel.
\newblock Bayesian modelling of skewness and kurtosis with two-piece scale and
  shape distributions.
\newblock \emph{Electronic Journal of Statistics}, 9\penalty0 (2):\penalty0
  1884--1912, 2015.

\bibitem[Semp{\'e} et~al.(2021)Semp{\'e}, Lloyd-Sherlock, Mart{\'\i}nez,
  Ebrahim, McKee, and Acosta]{sempe2021estimation}
Lucas Semp{\'e}, Peter Lloyd-Sherlock, Ram{\'o}n Mart{\'\i}nez, Shah Ebrahim,
  Martin McKee, and Enrique Acosta.
\newblock Estimation of all-cause excess mortality by age-specific mortality
  patterns for countries with incomplete vital statistics: a population-based
  study of the case of peru during the first wave of the covid-19 pandemic.
\newblock \emph{The Lancet Regional Health-Americas}, 2:\penalty0 100039, 2021.

\bibitem[Shawon et~al.(2021)Shawon, Ashrafi, Azad, Firth, Chowdhury, Mswia,
  Adair, Riley, Abouzahr, and Lopez]{shawon2021routine}
Md~Toufiq~Hassan Shawon, Shah Ali~Akbar Ashrafi, Abul~Kalam Azad, Sonja~M
  Firth, Hafizur Chowdhury, Robert~G Mswia, Tim Adair, Ian Riley, Carla
  Abouzahr, and Alan~D Lopez.
\newblock Routine mortality surveillance to identify the cause of death pattern
  for out-of-hospital adult (aged 12+ years) deaths in bangladesh: introduction
  of automated verbal autopsy.
\newblock \emph{BMC public health}, 21\penalty0 (1):\penalty0 1--11, 2021.

\bibitem[Tang et~al.(2018{\natexlab{a}})Tang, Ghosh, Ha, and
  Sedransk]{tang2018modeling}
Xueying Tang, Malay Ghosh, Neung~Soo Ha, and Joseph Sedransk.
\newblock Modeling random effects using global--local shrinkage priors in small
  area estimation.
\newblock \emph{Journal of the American Statistical Association}, 113\penalty0
  (524):\penalty0 1476--1489, 2018{\natexlab{a}}.

\bibitem[Tang et~al.(2018{\natexlab{b}})Tang, Xu, Ghosh, and
  Ghosh]{tang2018bayesian}
Xueying Tang, Xiaofan Xu, Malay Ghosh, and Prasenjit Ghosh.
\newblock Bayesian variable selection and estimation based on global-local
  shrinkage priors.
\newblock \emph{Sankhya A}, 80\penalty0 (2):\penalty0 215--246,
  2018{\natexlab{b}}.

\bibitem[UN(2019)]{desa2019united}
UN.
\newblock New york: United nations department of economic and social affairs.
\newblock \emph{Population Division. World Population Prospects: The 2019
  Revision}, 2019.

\bibitem[Villa and Walker(2014)]{villa2014objective}
Cristiano Villa and Stephen~G Walker.
\newblock Objective prior for the number of degrees of freedom of at
  distribution.
\newblock \emph{Bayesian Analysis}, 9\penalty0 (1):\penalty0 197--220, 2014.

\bibitem[Wang et~al.(2014)Wang, Liddell, Coates, Mooney, Levitz, Schumacher,
  Apfel, Iannarone, Phillips, Lofgren, et~al.]{wang2014global}
Haidong Wang, Chelsea~A Liddell, Matthew~M Coates, Meghan~D Mooney, Carly~E
  Levitz, Austin~E Schumacher, Henry Apfel, Marissa Iannarone, Bryan Phillips,
  Katherine~T Lofgren, et~al.
\newblock Global, regional, and national levels of neonatal, infant, and
  under-5 mortality during 1990--2013: a systematic analysis for the global
  burden of disease study 2013.
\newblock \emph{The Lancet}, 384\penalty0 (9947):\penalty0 957--979, 2014.

\bibitem[Wang et~al.(2020)Wang, Abbas, Abbasifard, Abbasi-Kangevari,
  Abbastabar, Abd-Allah, Abdelalim, Abolhassani, Abreu, Abrigo,
  et~al.]{wang2020global}
Haidong Wang, Kaja~M Abbas, Mitra Abbasifard, Mohsen Abbasi-Kangevari, Hedayat
  Abbastabar, Foad Abd-Allah, Ahmed Abdelalim, Hassan Abolhassani,
  Lucas~Guimar{\~a}es Abreu, Michael~RM Abrigo, et~al.
\newblock Global age-sex-specific fertility, mortality, healthy life expectancy
  (hale), and population estimates in 204 countries and territories,
  1950--2019: a comprehensive demographic analysis for the global burden of
  disease study 2019.
\newblock \emph{The Lancet}, 396\penalty0 (10258):\penalty0 1160--1203, 2020.

\bibitem[Wheldon et~al.(2016)Wheldon, Raftery, Clark, and
  Gerland]{wheldon2016bayesian}
Mark~C Wheldon, Adrian~E Raftery, Samuel~J Clark, and Patrick Gerland.
\newblock Bayesian population reconstruction of female populations for less
  developed and more developed countries.
\newblock \emph{Population studies}, 70\penalty0 (1):\penalty0 21--37, 2016.

\bibitem[Zeng et~al.(2020)Zeng, Adair, Wang, Yin, Qi, Liu, Liu, Lopez, and
  Zhou]{zeng2020measuring}
Xinying Zeng, Tim Adair, Lijun Wang, Peng Yin, Jinlei Qi, Yunning Liu, Jiangmei
  Liu, Alan~D Lopez, and Maigeng Zhou.
\newblock Measuring the completeness of death registration in 2844 chinese
  counties in 2018.
\newblock \emph{BMC medicine}, 18\penalty0 (1):\penalty0 1--11, 2020.

\end{thebibliography}




\section{Proof of Theorems}
We consider in this section  the proofs of Theorems  \ref{th1} and \ref{th2}.

\subsection{Proof of Theorem \ref{th1}}\label{proof1}

\begin{proof}

Consider model (\ref{model})  in matrix form as follow

\begin{align}
\boldsymbol{y_{i}}&=\boldsymbol{X_{i}}\boldsymbol{\beta} + u_{i}\boldsymbol{1}_{n_{i}} + \boldsymbol{\epsilon_{i}}, & i=1,...,m,
\label{model2}
\end{align}

where $\boldsymbol{y_{i}}=(y_{i1},...,y_{in_{i}})^\textsf{T}$,  $\boldsymbol{X_{i}}=(\boldsymbol{x}_{i1},...,\boldsymbol{x}_{in_{i}})^\textsf{T}$, $\boldsymbol{\epsilon_{i}}=(\boldsymbol{\epsilon}_{i1},...,\boldsymbol{\epsilon}_{in_{i}})^\textsf{T}$, $\boldsymbol{u}=(u_{1},...,u_{m})^\textsf{T}$. The joint posterior density
can be written as follow

\begin{align}
\label{eq:beta}
p(\boldsymbol{u},\phi,\boldsymbol{\omega}, \boldsymbol{\beta} ,\tau,\boldsymbol{\lambda} \mid \boldsymbol{y})&\propto\prod_{i=1}^{m}|(\lambda_{i}\tau)^{-1}I_{m})|^{-1/2}
\exp\left\{-\frac{1}{2}(\boldsymbol{z_{i}} - \boldsymbol{X_{i}}\boldsymbol{\beta})^\textsf{T}(\lambda_{i}\tau)I_{m}(\boldsymbol{z_{i}} - \boldsymbol{X_{i}}\boldsymbol{\beta})\right\}  \notag  \\  & \hspace{1cm}  \times
(\omega_{i}\phi)^{1/2}\exp\left\{-\frac{1}{2}u_{i}^{2}\omega_{i}\phi\right\}  \pi(\tau) \pi(\lambda_{i}) \pi(\phi) \pi(\omega_{i}).
\end{align}

Consider now the block-diagonal structure for model (\ref{model2}) given by

\begin{align*}
\boldsymbol{z}&=\text{col}_{1\leq i \leq m}(\boldsymbol{z_{i}}) = (\boldsymbol{z_{1}}^\textsf{T},...,\boldsymbol{z_{m}})^\textsf{T}, & \boldsymbol{X}&=\text{col}_{1\leq i \leq m}(\boldsymbol{X_{i}^\textsf{T}}),\\
\boldsymbol{V}&=\text{blkdiag}(\{\boldsymbol{V_{i}}\}_{i=1}^{m}),  &  \boldsymbol{\epsilon}&=\text{col}_{1\leq i \leq m}(\boldsymbol{\epsilon_{i}}),
\end{align*}

where $\boldsymbol{z_{i}}=\boldsymbol{y_{i}}- u_{i}\boldsymbol{1}_{n_{i}}$ and $\boldsymbol{V_{i}}=(\lambda_{i}\tau)I_{m}$. Since $\boldsymbol{X}$ is of full range then  $\boldsymbol{X}^\textsf{T}\boldsymbol{V}\boldsymbol{X}$ is nonsingular. Notice that

\begin{align*}
&\exp\left\{-\frac{1}{2}(\boldsymbol{z} - \boldsymbol{X}\boldsymbol{\beta})^\textsf{T}\boldsymbol{V}(\boldsymbol{z} - \boldsymbol{X}\boldsymbol{\beta})\right\} \\
 &=\exp\left\{-\frac{1}{2}(\boldsymbol{\beta}-\hat{\boldsymbol{\beta}})^\textsf{T} \boldsymbol{X}
^\textsf{T}\boldsymbol{V} \boldsymbol{X}(\boldsymbol{\beta}-\hat{\boldsymbol{\beta}}) -\frac{1}{2} (\boldsymbol{z} - \boldsymbol{X}\hat{\boldsymbol{\beta}})^\textsf{T}\boldsymbol{V}
(\boldsymbol{z} - \boldsymbol{X}\hat{\boldsymbol{\beta}})
\right\}\notag
\end{align*}

where $\hat{\boldsymbol{\beta}}=(\boldsymbol{X}^\textsf{T}\boldsymbol{V}\boldsymbol{X})^{-1}\boldsymbol{X}
^\textsf{T}\boldsymbol{V}\boldsymbol{z_{i}}$. Computing the integral  of the right-hand site in (\ref{eq:beta}) with respect to $\boldsymbol{\beta}$ and having that
$\exp\left\{-\frac{1}{2} (\boldsymbol{z} -\hat{\boldsymbol{\beta}})^\textsf{T}\boldsymbol{V}
(\boldsymbol{z} - \hat{\boldsymbol{\beta}})
 \right\} \leq 1$ then
\begin{align}
p(\boldsymbol{u},\phi,\boldsymbol{\omega}, \tau,\boldsymbol{\lambda} \mid \boldsymbol{y})&\leq K \prod_{i=1}^{m}(\omega_{i}\phi)^{1/2}\exp\left\{-\frac{1}{2}u_{i}^{2}\omega_{i}\phi\right\}\pi(\tau) \pi(\lambda_{i}) \pi(\phi) \pi(\omega_{i}),
 \label{eq:posterior2}
\end{align}
where $K$ is a generic positive constant. Integrating (\ref{eq:posterior2}) with respect to $\boldsymbol{u}$, we have
\begin{align}
p(\phi,\boldsymbol{\omega}, \boldsymbol{\beta} ,\tau,\boldsymbol{\lambda} \mid \boldsymbol{y})&\leq K \prod_{i=1}^{m}\pi(\tau) \pi(\lambda_{i}) \pi(\phi) \pi(\omega_{i}),
 \label{eq:posterior3}
\end{align}
moreover if the priors distributions of $\pi(\tau)$, $\pi(\lambda_{i})$, $\pi(\phi)$  and $\pi(\omega_{i})$ are proper then the posterior distribution is proper.
\end{proof}

\subsection{Proof of Theorem \ref{th2}}\label{proof2}

\begin{proof}  

For (i) consider

\begin{align} \small
P(\gamma_i  > \epsilon \mid \boldsymbol{\beta},\boldsymbol{\lambda},\tau,\boldsymbol{y}_{i})=\dfrac{\bigintss_{0}^{c_{1}/\phi} N (\boldsymbol{y_{i}}; \boldsymbol{X_{i}}\boldsymbol{\beta},
I_{m}(\lambda_{i}\tau)^{-1} + \boldsymbol{1}_{n_{i}} \boldsymbol{1}_{n_{i}}^\textsf{T} (\omega_{i}\phi)^{-1}) \pi(\omega_{i}) d \omega_{i}}{\bigintss_{0}^{\infty} N (\boldsymbol{y_{i}};  \boldsymbol{X_{i}}\boldsymbol{\beta},
I_{m}(\lambda_{i}\tau)^{-1} + \boldsymbol{1}_{n_{i}} \boldsymbol{1}_{n_{i}}^\textsf{T} (\omega_{i}\phi)^{-1}) \pi(\omega_{i}) d \omega_{i}} = \dfrac{N_{1}}{D_{1}},
\label{density1}
\end{align}

where $c_{1}=((1-\epsilon)/\epsilon)(n_{i}\lambda_{i}\tau)$ and

\begin{align*}
N (\boldsymbol{y_{i}};  \boldsymbol{X_{i}}\boldsymbol{\beta},
I_{m}(\lambda_{i}\tau)^{-1} + \boldsymbol{1}_{n_{i}} \boldsymbol{1}_{n_{i}}^\textsf{T} (\omega_{i}\phi)^{-1}) &= (2\pi)^{-1/2} |
I_{m}(\lambda_{i}\tau)^{-1} + \boldsymbol{1}_{n_{i}} \boldsymbol{1}_{n_{i}}^\textsf{T} (\omega_{i}\phi)^{-1}
|^{-1/2}\\ &\hspace{-2cm}\times \exp\left\{-\dfrac{1}{2}(\boldsymbol{y_{i}} - \boldsymbol{X_{i}}\boldsymbol{\beta})^\textsf{T}(I_{m}(\lambda_{i}\tau)^{-1} + \boldsymbol{1}_{n_{i}} \boldsymbol{1}_{n_{i}}^\textsf{T}(\omega_{i}\phi)^{-1})(\boldsymbol{y_{i}} - \boldsymbol{X_{i}}\boldsymbol{\beta})\right\}\\
&=(2\pi)^{-1/2} \left(1+ n_{i}(\lambda_{i}\tau)/(\omega_{i}\phi) \right)^{-1/2}(\lambda_{i}\tau)^{n_{i}/2}
\\ &\hspace{-2cm}\times 
\exp\left\{-\dfrac{1}{2}\left(
\sum_{j=1}^{n_{i}}  \dfrac{y^{*2}_{ij}}{
1/(\lambda_{i}\tau) + n_{i}/(\omega_{i}\phi)} + 
\sum_{j=1}^{n_{i}}  \dfrac{n_{i} (y^{*}_{ij} - \bar{y}^{*}_{i} )^{2}}{n_{i}/(\lambda_{i}\tau) + (\omega_{i}\phi)/(\lambda_{i}\tau)^2}
\right)\right\}\\ 
&=(2\pi)^{-1/2} \left(1+ n_{i}(\lambda_{i}\tau)/(\omega_{i}\phi) \right)^{-1/2}(\lambda_{i}\tau)^{n_{i}/2} \exp\left\{-\frac{1}{2} l_{i}^2 \right\}, 
\label{density2}
\end{align*}

with $y^{*}_{ij} = y_{ij} -  \boldsymbol{x_{ij}}^\textsf{T}\boldsymbol{\beta}$ and  $\bar{y}^{*}_{i}  = (1/n_{i})\sum_{j=1}^{n_{i}} y_{ij}^{*}$. Since $$\left(1+ n_{i}(\lambda_{i}\tau)/(\omega_{i}\phi) \right)^{-1/2} \exp\left\{-\frac{1}{2} l_{i}^2 \right\}  \pi(\omega_{i}) \leq  \pi(\omega_{i})$$ 
and $\pi(\omega_{i})$ is proper using the Lebesgue dominated convergence theorem the denominator $D_{1}$
in (\ref{density1}) converge to  $\exp\left\{-\dfrac{\lambda_{i}\tau \sum_{j=1}^{n_{i}} y^{*2}_{ij} }{2}
\right\}$ as $\phi \rightarrow \infty$. For the the numerator $N_{1}$ in (\ref{density1})  we have that

\begin{align*}
k_{1} \int_{0}^{c_{1}/\phi} 
 \left(1+ n_{i}(\lambda_{i}\tau)/(\omega_{i}\phi) \right)^{-1/2} \pi(\omega_{i}) d \omega_{i}
\leq 
N_{1} \leq  (1-\epsilon)^{1/2} \int_{0}^{c_{1}/\phi} \pi(\omega_{i}) d \omega_{i},
\end{align*}

with $k_{1}=\exp\left\{-\dfrac{\lambda_{i}\tau}{2}\left( \sum_{j=1}^{n_{i}}  (y^{*2}_{ij} +  (y^{*}_{ij} - \bar{y}^{*}_{i} )^{2}
\right)\right\}$. If we consider the Laplace prior $\pi(\omega_{i}) = (\omega_{i})^{-2}\exp(-1/\omega_{i})$ then $N_{1} \leq  (1-\epsilon)^{1/2}  \exp(-\phi/c_{1})$. For the upper bound of $N_{1}$ notice that

\begin{align*}
N_{1} &\geq k_{1} \int_{0}^{c_{1}/\phi} 
 \left(1+ n_{i}(\lambda_{i}\tau)/(\omega_{i}\phi) \right)^{-1/2}  (\omega_{i})^{-2}\exp\left\{-\dfrac{1}{\omega_{i}}\right\} d \omega_{i} \\
&\geq  k_{1} \phi \int_{c_{1}/2}^{c_{1}} 
 \left(1+ n_{i}(\lambda_{i}\tau)/x \right)^{-1/2}  (x)^{-2}\exp\left\{-\dfrac{\phi}{x}\right\} dx\\
&\geq  k_{1} \phi \left(1+ 2n_{i}(\lambda_{i}\tau)/c_{1} \right)^{-1/2} \int_{c_{1}/2}^{c_{1}} 
   (x)^{-2}\exp\{-\dfrac{\phi}{x}\} dx\\
&=   k_{1}  \left(1+ 2n_{i}(\lambda_{i}\tau)/c_{1} \right)^{-1/2}  \exp\left\{-\dfrac{\phi}{c_{1}}\right\}\left(1-\exp\left\{-\dfrac{\phi}{c_{1}}\right\}\right),
\end{align*}

hence  $P(\gamma_i  > \epsilon \mid \boldsymbol{\beta},\boldsymbol{\lambda},\tau,\boldsymbol{y}_{i})\asymp \exp(-\phi/c_{1})$
as $\phi \rightarrow \infty$. For the Beta-Prime distribution $\pi(\omega_{i})  = \dfrac{\Gamma(a+b)}{\Gamma(b)\Gamma(a)} \dfrac{\omega_{i}^{b-1}}{(1+\omega_{i})^{a+b}}$ with $0<a\leq 1$ and $0<b\leq 1$ we have

\begin{align*}
N_{1} &\leq  \dfrac{\Gamma(a+b)}{\Gamma(b)\Gamma(a)}  (1-\epsilon)^{1/2} \int_{0}^{c_{1}/\phi}   \dfrac{\omega_{i}^{b-1}}{(1+\omega_{i})^{a+b}}
d \omega_{i}\\
&= \dfrac{\Gamma(a+b)}{\Gamma(b)\Gamma(a)}  (1-\epsilon)^{1/2} \int_{0}^{c_{1}/(c_{1}+\phi)} 
x^{b-1} (1-x)^{a-1} d x \\
&\leq
\dfrac{\Gamma(a+b)}{\Gamma(b)\Gamma(a)}  (1-\epsilon)^{1/2}  \left( \dfrac{\phi}{c_{1}+\phi} \right)^{a-1} \int_{0}^{c_{1}/(c_{1}+\phi)} 
x^{b-1}  d x \\
&=
\dfrac{\Gamma(a+b)}{\Gamma(b)\Gamma(a)} \dfrac{1}{b}  (1-\epsilon)^{1/2}  \left( \dfrac{\phi}{c_{1}+\phi} \right)^{a-1}    \left( \dfrac{\phi c_{1}}{c_{1}+\phi} \right)^{b}  \left(  \dfrac{1}{\phi}  \right)^{b}
\end{align*}

and for the under bound of $N_{1}$ we have that 

\begin{align*}
N_{1} &\geq k_{1}  \dfrac{\Gamma(a+b)}{\Gamma(b)\Gamma(a)}   \int_{0}^{c_{1}/\phi}  \left(1+ n_{i}(\lambda_{i}\tau)/(\omega_{i}\phi) \right)^{-1/2}  \dfrac{\omega_{i}^{b-1}}{(1+\omega_{i})^{a+b}}
d \omega_{i}\\
&=  k_{1}  \dfrac{\Gamma(a+b)}{\Gamma(b)\Gamma(a)}   \int_{0}^{c_{1}/(c_{1}+\phi)} \left(1+ n_{i}(\lambda_{i}\tau)(1-x)/(x\phi) \right)^{-1/2}  
x^{b-1} (1-x)^{a-1} d x \\
&\geq  \left( \dfrac{\phi}{c_{1}+\phi} \right)^{b-1} k_{1}  \dfrac{\Gamma(a+b)}{\Gamma(b)\Gamma(a)}   \int_{0}^{c_{1}/(c_{1}+\phi)} \left(1+ n_{i}(\lambda_{i}\tau)(1-x)/(x\phi) \right)^{-1/2}  
d x \\
&\geq   \left( \dfrac{\phi}{c_{1}+\phi} \right)^{b-1} k_{1} \left(c_{1}/(c_{1}+\phi) + (n_{i}\lambda_{i}\tau)/\phi \right)^{-1/2}  \dfrac{\Gamma(a+b)}{\Gamma(b)\Gamma(a)}   \int_{0}^{c_{1}/(c_{1}+\phi)} x^{1/2}d x \\
&=\dfrac{1}{3/2} \left( \dfrac{\phi}{c_{1}+\phi} \right)^{b-1} k_{1} \left(1 + (n_{i}\lambda_{i}\tau)(c_{1}+\phi)/(c_{1}\phi) \right)^{-1/2}  \dfrac{\Gamma(a+b)}{\Gamma(b)\Gamma(a)}    \left( \dfrac{\phi c_{1}}{c_{1}+\phi} \right)^{b}  \left(  \dfrac{1}{\phi}  \right)^{b},
\end{align*}

therefore $P(\gamma_i  > \epsilon \mid \boldsymbol{\beta},\boldsymbol{\lambda},\tau,\boldsymbol{y}_{i}) \asymp (  1/\phi)^{b}$
as $\phi \rightarrow \infty$. In order to show (ii) consider

\begin{align}
\small
P(\gamma_i < \epsilon \mid \boldsymbol{\beta},\boldsymbol{\omega},\phi,\boldsymbol{y}_{i})=\dfrac{\bigintss_{0}^{c_{2}/\tau} N (\boldsymbol{y_{i}}; \boldsymbol{X_{i}}\boldsymbol{\beta},
I_{m}(\lambda_{i}\tau)^{-1} + \boldsymbol{1}_{n_{i}} \boldsymbol{1}_{n_{i}}^\textsf{T} (\omega_{i}\phi)^{-1}) \pi(\lambda_{i}) d \lambda_{i}}{\bigintss_{0}^{\infty} N (\boldsymbol{y_{i}};  \boldsymbol{X_{i}}\boldsymbol{\beta},
I_{m}(\lambda_{i}\tau)^{-1} + \boldsymbol{1}_{n_{i}} \boldsymbol{1}_{n_{i}}^\textsf{T} (\omega_{i}\phi)^{-1}) \pi(\lambda_{i}) d \lambda_{i}} = \dfrac{N_{2}}{D_{2}},
\end{align}

where $c_{2}=(\epsilon/(1-\epsilon))(\omega_{i}\phi/n_{i})$ and, 

\begin{align*}
\dfrac{N_{2}}{D_{2}} \leq 
\dfrac{(1-\epsilon)^{1/2} (\epsilon/(1-\epsilon))^{n_{i}/2}\bigintss_{0}^{c_{2}/\tau}  \pi(\lambda_{i}) d \lambda_{i}}{k_{i}(1-\epsilon_{1})^{1/2}(\epsilon_{1}/(1-\epsilon_{1}))^{n_{i}/2} \bigintss_{c_{3}/\tau}^{c_{4}/\tau}  \pi(\lambda_{i}) d \lambda_{i}}, 
\end{align*}

with 

$$k_{i}=\exp\left\{-\dfrac{1}{2}\left(
\sum_{j=1}^{n_{i}}  \dfrac{y^{*2}_{ij}}{n_{i}(1-\epsilon_{2}) + \epsilon_{2}} + 
\sum_{j=1}^{n_{i}}  \dfrac{ (y^{*}_{ij} - \bar{y}^{*}_{i} )^{2}}{ (n_{i} + \omega_{i}\phi)(1-\epsilon_{2})}
\right)\right\}, $$

where  $c_{3}=(\epsilon_{1}/(1-\epsilon_{1}))(\omega_{i}\phi/n_{i})$ and $c_{4}=(\epsilon_{2}/(1-\epsilon_{2}))(\omega_{i}\phi/n_{i})$ for any $\epsilon_{1}$ with $\epsilon_{2}$ with  $0<\epsilon_{1}<\epsilon_{2}<\epsilon<1$. If we consider the Laplace prior $\pi(\lambda_{i}) = (\lambda_{i})^{-2}\exp\left(-1/\lambda_{i}\right)$  then $N_{2} \leq  (1-\epsilon)^{1/2} (\epsilon/(1-\epsilon))^{n_{i}/2} \exp\left\{-\dfrac{\tau n_{i} (1-\epsilon)}{\omega_{i}\phi \epsilon}\right\}$ and 

\begin{align*}
D_{2} \geq  k_{i}(1-\epsilon_{1})^{1/2}(\epsilon_{1}/(1-\epsilon_{1}))^{n_{i}/2}   \exp\left\{-\dfrac{\tau n_{i} (1-\epsilon_{1})}{\omega_{i}\phi \epsilon_{1}}\right\}\left(1-\exp\left\{-\dfrac{\epsilon_{2}(1-\epsilon_{1})}{\epsilon_{1}(1-\epsilon_{2}) }\right\}\right), 
\end{align*}

therefore

\begin{align*}
\dfrac{N_{2}}{D_{2}} \leq \dfrac{(1-\epsilon)^{1/2} (\epsilon/(1-\epsilon))^{n_{i}/2}}{k_{i}(1-\epsilon_{1})^{1/2}(\epsilon_{1}/(1-\epsilon_{1}))^{n_{i}/2}}\left(1-\exp\left\{-\dfrac{\epsilon_{2}(1-\epsilon_{1})}{\epsilon_{1}(1-\epsilon_{2}) }\right\}\right)^{-1} \exp\left\{-\dfrac{\tau n_{i} }{\omega_{i}\phi}\left(\dfrac{\epsilon}{1-\epsilon} - \dfrac{\epsilon_{1}}{1-\epsilon_{1}}\right\}\right),
\end{align*}

which converges to zero at an exponential rate when $\tau \rightarrow \infty$. For the Beta prime prior $\pi(\lambda_{i})  = \dfrac{\Gamma(a+b)}{\Gamma(b)\Gamma(a)} \dfrac{\lambda_{i}^{b-1}}{(1+\lambda_{i})^{a+b}}$ with $a=1$ and $b \in (0,  \infty)$ we have

\begin{align*}
\dfrac{N_{2}}{D_{2}} &\leq 
\dfrac{(1-\epsilon)^{1/2} (\epsilon/(1-\epsilon))^{n_{i}/2}\bigintss_{0}^{c_{2}/\tau}  \dfrac{1}{(1+\lambda_{i})^{b+1}} d \lambda_{i}}{k_{i}(1-\epsilon_{1})^{1/2}(\epsilon_{1}/(1-\epsilon_{1}))^{n_{i}/2}\bigintss_{c_{3}/\tau}^{c_{4}/\tau}  \dfrac{1}{(1+\lambda_{i})^{b+1}} d \lambda_{i}} \\
&= \dfrac{(1-\epsilon)^{1/2} (\epsilon/(1-\epsilon))^{n_{i}/2}\bigintss_{0}^{c_{2}/(c_{2}+\tau)} x^{b-1} dx }{k_{i}(1-\epsilon_{1})^{1/2}(\epsilon_{1}/(1-\epsilon_{1}))^{n_{i}/2}\bigintss_{c_{3}/(c_{3}+\tau)}^{c_{4}/(c_{4}+\tau)} x^{b-1} dx }\\
&= \dfrac{(1-\epsilon)^{1/2} (\epsilon/(1-\epsilon))^{n_{i}/2}\dfrac{1}{b}\left(\dfrac{c_{2}}{c_{2}+\tau}\right)^{b}
}{k_{i}(1-\epsilon_{1})^{1/2}(\epsilon_{1}/(1-\epsilon_{1}))^{n_{i}/2}\bigintss_{c_{3}/(c_{3}+\tau)}^{c_{4}/(c_{4}+\tau)} x^{b-1} dx },
\end{align*}
where $$D_{2} \geq k_{i}(1-\epsilon_{1})^{1/2}(\epsilon_{1}/(1-\epsilon_{1}))^{n_{i}/2} \left(\dfrac{c_{4}}{c_{4}+\tau}\right)^{b-1},$$ 
when  $b \leq 1$ and 
$$D_{2} \geq k_{i}(1-\epsilon_{1})^{1/2}(\epsilon_{1}/(1-\epsilon_{1}))^{n_{i}/2}  \left(\dfrac{c_{3}}{c_{3}+\tau}\right)^{b-1},$$ 
when $b >1$ therefore  $$P(\gamma_{i} < \epsilon \mid \boldsymbol{\beta}, \boldsymbol{\omega},\phi,\boldsymbol{y}_{i}) \rightarrow 0$$ at a polynomial rate when $\tau \rightarrow \infty$. 

For part (iii) consider

 \begin{align*} \small
P(\gamma_i  < \epsilon \mid \boldsymbol{\beta},\boldsymbol{\lambda},\tau,\boldsymbol{y}_{i})=\dfrac{\bigintss_{c_{1}/\phi}^{\infty}  N (\boldsymbol{y_{i}}; \boldsymbol{X_{i}}\boldsymbol{\beta},
I_{m}(\lambda_{i}\tau)^{-1} + \boldsymbol{1}_{n_{i}} \boldsymbol{1}_{n_{i}}^\textsf{T} (\omega_{i}\phi)^{-1}) \pi(\omega_{i}) d \omega_{i}}{\bigintss_{0}^{\infty} N (\boldsymbol{y_{i}};  \boldsymbol{X_{i}}\boldsymbol{\beta},
I_{m}(\lambda_{i}\tau)^{-1} + \boldsymbol{1}_{n_{i}} \boldsymbol{1}_{n_{i}}^\textsf{T} (\omega_{i}\phi)^{-1}) \pi(\omega_{i}) d \omega_{i}} = \dfrac{N_{3}}{D_{3}},
\end{align*}

where where $c_{1}=((1-\epsilon)/\epsilon)(n_{i}\lambda_{i}\tau)$ and for any $0<\epsilon_{5} <1$ we have

\begin{align*} 
\dfrac{N_{3}}{D_{3}} &\leq 
\dfrac{\bigintss_{c_{1}/\phi}^{\infty}  N (\boldsymbol{y_{i}}; \boldsymbol{X_{i}}\boldsymbol{\beta},
I_{m}(\lambda_{i}\tau)^{-1} + \boldsymbol{1}_{n_{i}} \boldsymbol{1}_{n_{i}}^\textsf{T} (\omega_{i}\phi)^{-1}) \pi(\omega_{i}) d \omega_{i}}{\bigintss_{((1-\epsilon)/\epsilon)(n_{i}\lambda_{i}\tau)/\phi}^{\infty} N (\boldsymbol{y_{i}};  \boldsymbol{X_{i}}\boldsymbol{\beta},
I_{m}(\lambda_{i}\tau)^{-1} + \boldsymbol{1}_{n_{i}} \boldsymbol{1}_{n_{i}}^\textsf{T} (\omega_{i}\phi)^{-1}) \pi(\omega_{i}) d \omega_{i}} \\
&\leq \exp\left\{-\dfrac{(n_{i}\lambda_{i}\tau)^2}{2(1-\epsilon_{5})}  (\bar{\boldsymbol{y}}_{i}-\bar{\boldsymbol{x}}_{i}^\textsf{T}\boldsymbol{\beta}_{i})^2 \right\}
\dfrac{\bigintss_{c_{1}/\phi}^{\infty} \left(1+ n_{i}(\lambda_{i}\tau)/(\omega_{i}\phi) \right)^{-1/2} \pi(\omega_{i}) d \omega_{i}}{
(1- \epsilon_{5})^{1/2} \bigintss_{((1-\epsilon_{5})/\epsilon_{5})(n_{i}\lambda_{i}\tau)/\phi}^{\infty} \pi(\omega_{i}) d \omega_{i}},
\end{align*}

using that $\pi(\omega_{i})$ is proper then $P(\gamma_i  < \epsilon \mid \boldsymbol{\beta},\boldsymbol{\lambda},\tau,\boldsymbol{y}_{i}) \leq C^{*} \exp\left\{-\dfrac{(n_{i}\lambda_{i}\tau)^2}{2(1-\epsilon_{5})}  (\bar{\boldsymbol{y}}_{i}-\bar{\boldsymbol{x}}_{i}^\textsf{T}\boldsymbol{\beta}_{i})^2 \right\} $ where $C^{*}$ is a constant. Therefore $P(\gamma_i  < \epsilon \mid \boldsymbol{\beta},\boldsymbol{\lambda},\tau,\boldsymbol{y}_{i}) \rightarrow 0$  as $|\bar{\boldsymbol{y}}_{i}-\bar{\boldsymbol{x}}_{i}^\textsf{T}\boldsymbol{\beta}_{i}| \rightarrow \infty$.

\end{proof}

\clearpage

\section{Markov chain Monte Carlo algorithms}
\label{MCMCsec}

We consider in this section  the Markov chain Monte Carlo algorithms for implementing our propose models.

\RestyleAlgo{boxruled}
\vspace{1mm}
\begin{algorithm}
\small
Initialize $\boldsymbol{\xi}^{(0)}$.
\For{$t=1,...,T$}{
Draw
\begin{equation*}
u_{i}^{(t)}\sim \text{N}(u_{i};\gamma_{i}^{(t-1)}(\bar{y}_{i}-\boldsymbol{\bar{x}_{i}^\textsf{T}} \boldsymbol{\beta}^{(t-1)}),\gamma_{i}^{(t-1)}/(\zeta_{\epsilon}^{(t-1)} n_{i})).
\end{equation*}
with $\gamma_{i}^{(t-1)}=\zeta_{\epsilon}^{(t-1)}/(\zeta_{\epsilon}^{(t-1)} + (\phi^{(t-1)}\omega_{i}^{(t-1)})/n_{i})$. Draw
$$\boldsymbol{\beta}^{(t)}\sim
\text{N}_{p}(\boldsymbol{\beta}^{(t)};(\sum\limits_{i=1}^{m}\sum\limits_{j=1}^{n_{i}}\boldsymbol{x_{ij}}\boldsymbol{x_{ij}^\textsf{T}})^{-1}
\sum\limits_{i=1}^{m}\sum\limits_{j=1}^{n_{i}}\boldsymbol{x_{i}}(y_{ij}-u_{i}^{(t)}),\zeta^{(t-1)}_{\epsilon}(\sum\limits_{i=1}^{m}\sum\limits_{j=1}^{n_{i}}\boldsymbol{x_{ij}}\boldsymbol{x_{ij}^\textsf{T}})^{-1}).
$$
Draw
\begin{equation*}
\zeta^{(t)}_{\epsilon}\sim \text{Gamma}(\zeta^{(t)}; \frac{m}{2} + a_{\zeta_{\epsilon}},
\frac{1}{2}\sum\limits_{i=1}^{m}\sum\limits_{j=1}^{n_{i}}(y_{ij}-\boldsymbol{x_{ij}}^{T}\boldsymbol{\beta}^{(t)} - u_{i})^{2}+ b_{\zeta_{\epsilon}}).
\end{equation*}
Draw
\begin{align*}
\phi^{(t)}&\sim \text{Gamma}(\phi^{(t)}; \frac{n}{2} + a_{\phi},
\frac{1}{2}\sum\limits_{i=1}^{m}\omega_{i}^{(t-1)}u_{i}^{2}+b_{\phi}),  & n=&\sum_{i=1}^{m} n_{i}.
\end{align*}
Draw $\omega_{i}^{(t)}$ using Algorithm \ref{MCMC3} and update $\gamma_{i}^{(t)}$.
}
\caption{Gibbs sampling using common Gamma priors for the scales of the  errors given by $\zeta^{(t)}_{\epsilon}$ and GL priors (HS, LA and Student-t)
for the scales of the random effects in model (\ref{model}).}
\label{MCMC1}
\end{algorithm}

\RestyleAlgo{boxruled}
\vspace{3mm}
\begin{algorithm}
\small
Initialize $(\boldsymbol{\xi}^{(0)},\rho_{i}^{(0)})$.
\For{$t=1,...,T$}{
Draw
\begin{equation*}
u_{i}^{(t)}\sim \text{N}(u_{i};\gamma_{i}^{(t-1)}(\bar{y}_{i}-\boldsymbol{\bar{x}_{i}^\textsf{T}} \boldsymbol{\beta}^{(t-1)}),\gamma_{i}^{(t-1)}/(\tau^{(t-1)}\lambda_{i}^{(t-1)}) n_{i})).
\end{equation*}
with $\gamma_{i}^{(t-1)}=(\tau^{(t-1)}\lambda_{i}^{(t-1)})/(\tau^{(t-1)}\lambda_{i}^{(t-1)} + (\phi^{(t-1)}\omega_{i}^{(t-1)})/n_{i})$.  Draw 
$$\boldsymbol{\beta}^{(t)}\sim
\text{N}_{p}(\boldsymbol{\beta}^{(t)};a_{\boldsymbol{\beta}^{(t)}},b_{\boldsymbol{\beta}^{(t)}}),
$$
where
$$a_{\boldsymbol{\beta}^{(t)}}=(\sum\limits_{i=1}^{m}\sum\limits_{j=1}^{n_{i}}\boldsymbol{x_{ij}}\lambda_{i}^{(t-1)}\boldsymbol{x_{ij}^\textsf{T}})^{-1}
\sum\limits_{i=1}^{m}\sum\limits_{j=1}^{n_{i}}\lambda_{i}^{(t-1)}\boldsymbol{x_{i}}(y_{ij}-u_{i}^{(t)}),$$
$$b_{\boldsymbol{\beta}^{(t)}}=\tau^{(t-1)}(\sum\limits_{i=1}^{m}\sum\limits_{j=1}^{n_{i}}\boldsymbol{x_{ij}}\lambda_{i}^{(t-1)}\boldsymbol{x_{ij}^\textsf{T}})^{-1}.$$
Draw
\begin{equation*}
\tau^{(t)}\sim \text{Gamma}(\tau^{(t)}; \frac{m}{2} + a_{\tau},
\frac{1}{2}\sum\limits_{i=1}^{m}\sum\limits_{j=1}^{n_{i}}\lambda_{i}^{(t-1)}(y_{ij}-\boldsymbol{x_{ij}}^{T}\boldsymbol{\beta}^{(t)} - u_{i})^{2}+b_{\tau}).
\end{equation*}
Draw
\begin{equation*}
\lambda_{i}^{(t)}\sim \text{Gamma}(\lambda_{i}^{(t)}; \frac{n_{i}}{2} + 1,
\frac{1}{2}\tau^{(t)}\sum\limits_{j=1}^{n_{i}}(y_{ij}-\boldsymbol{x_{ij}}^{T}\boldsymbol{\beta}^{(t)} - u_{i})^{2}+\rho_{i}^{(t-1)}).
\end{equation*}
Draw
\begin{equation*}
\rho_{i}^{(t)}\sim \text{Gamma}(\rho_{i}^{(t)}; 2,\lambda_{i}^{(t)}+1).
\end{equation*}
Draw
\begin{align*}
\phi^{(t)}&\sim \text{Gamma}(\phi^{(t)}; \frac{n}{2} + a_{\phi},
\frac{1}{2}\sum\limits_{i=1}^{m}\omega_{i}^{(t-1)}u_{i}^{2}+b_{\phi}),  & n=&\sum_{i=1}^{m} n_{i}.
\end{align*}
Draw $\omega_{i}^{(t)}$ using Algorithm \ref{MCMC3} and update $\gamma_{i}^{(t)}$.
}
\caption{Gibbs sampling using a Half-Cauchy prior for the local  scales of the errors and GL priors (HS, LA and Student-t)
for the scales of the random effects in model (\ref{model}).}
\label{MCMC2}
\end{algorithm}

\RestyleAlgo{boxruled}
\vspace{3mm}
\begin{algorithm}
\small
\For{$t=1,...,T$}{
For the HS prior draw
\begin{equation*}
\omega_{i}^{(t)}\sim \text{Gamma}(\omega_{i}^{(t)}; 1,
\frac{1}{2}\phi^{(t)}u_{i}^{2}+\varrho_{i}^{(t-1)}).
\end{equation*}
Draw
\begin{equation*}
\varrho_{i}^{(t)}\sim \text{Gamma}(\varrho_{i}^{(t)};1,\omega_{i}^{(t)}+1).
\end{equation*}
For the LA prior draw
\begin{equation*}
\omega_{i}^{(t)}\sim \text{GIG}(\omega_{i}^{(t)}; -\frac{1}{2},2,\phi^{(t)}u_{i}^{2}).
\end{equation*}
For the Student-t draw
\begin{equation*}
\omega_{i}^{(t)}\sim \text{Gamma}(\omega_{i}^{(t)}; \frac{1}{2}\nu_{i}^{(t-1)} + \frac{1}{2},
\frac{1}{2}\nu_{i}^{(t-1)}u_{i}^{2}+\frac{\nu_{i}^{(t-1)}}{2}).
\end{equation*}
Set $\nu_{i}^{(t)}=l$ with probability proportional to
$$\frac{l}{(l+k_{\nu_{i}})^{3}}\text{Student-t}_{\nu_{i}^{(t)}=k}(u^{(t)}_{i},0,\sqrt{1/\phi^{(t)}}).$$
}
\caption{Gibbs sampling for the local scales of the random effects in model  (\ref{model}).}
\label{MCMC3}
\end{algorithm}

\clearpage

\enlargethispage{10cm}

\section{Posterior mean and credible intervals of the regression parameters  for models 1 and 2}\label{sup_post}
This section contains the posterior estimates and the corresponding credible intervals for the regression parameters of the considered models in Section \ref{Results}.  Tables \ref{tab:tabglobal1} and \ref{tab:tabglobal2} consider models 1 and 2 respectively for both sexes, Tables \ref{tab:tabglobal3} and \ref{tab:tabglobal4} contain the results for models 1 and 2 respectively for females and  Tables \ref{tab:tabglobal5} and \ref{tab:tabglobal6} show the results for models 1 and 2 respectively for males.

\begin{table}[ht]
\scriptsize
\renewcommand{\arraystretch}{1.05}
\begin{tabular}{rrrr|rrr|rrrrrrrr}\\
\hline    \hline
Prior for the errors &    & Gamma  &   &  &     Half-Cauchy (local) &  & &  &\\ \hline \hline
Parameter & Mean  & Lower CI  &  Upper CI &   Mean  & Lower
CI  &  Upper CI & Local-Prior \\ \hline    \hline
Constant & 43.01 & 29.94 & 56.07 & 31.65 & 22.67 & 40.64 & Gamma \\
RegCDR & 0.73 & 0.61 & 0.85 & 0.77 & 0.69 & 0.84 & Gamma\\
 RegCDR squared & -0.03 & -0.04 & -0.02 & -0.03 & -0.03 & -0.02 & Gamma \\
\%65+ & -12.68 & -16.19 & -9.18 & -14.33 & -17.06 & -11.60 &Gamma\\
ln(5q0) & -1.22 & -1.40 & -1.04 & -1.42 & -1.56 & -1.29 & Gamma\\
 C5q0 & 2.22 & 1.79 & 2.65 & 0.80 & 0.59 & 1.02 & Gamma \\
 year & -0.02 & -0.03 & -0.02 & -0.02 & -0.02 & -0.01 & Gamma \\
Constant & 44.96 & 32.03 & 57.89 & 33.40 & 24.64 & 42.15 & Student-t \\
 RegCDR & 0.74 & 0.61 & 0.86 & 0.77 & 0.69 & 0.85 & Student-t \\
 RegCDR squared & -0.03 & -0.04 & -0.02 & -0.03 & -0.03 & -0.02 & Student-t \\
 \%65+ & -12.51 & -16.02 & -8.99 & -14.14 & -16.94 & -11.34 & Student-t \\
 ln(5q0) & -1.25 & -1.42 & -1.07 & -1.45 & -1.58 & -1.32 & Student-t \\
  C5q0 & 2.24 & 1.81 & 2.66 & 0.79 & 0.57 & 1.01 & Student-t \\
  year & -0.03 & -0.03 & -0.02 & -0.02 & -0.02 & -0.01 & Student-t \\
  Constant & 46.14 & 33.45 & 58.84 & 34.52 & 25.75 & 43.29 & Exponential \\
  RegCDR & 0.73 & 0.61 & 0.86 & 0.77 & 0.70 & 0.85 & Exponential \\
  RegCDR squared & -0.03 & -0.04 & -0.02 & -0.03 & -0.03 & -0.02 & Exponential \\
  \%65+ & -12.36 & -15.91 & -8.81 & -14.02 & -16.72 & -11.33 & Exponential \\
  ln(5q0) & -1.26 & -1.44 & -1.08 & -1.46 & -1.59 & -1.33 & Exponential \\
  C5q0 & 2.24 & 1.81 & 2.67 & 0.78 & 0.57 & 1.00 & Exponential \\
  year & -0.03 & -0.03 & -0.02 & -0.02 & -0.02 & -0.02 & Exponential \\
  Constant & 51.81 & 39.94 & 63.67 & 41.21 & 32.68 & 49.75 & Horseshoe \\
  RegCDR & 0.72 & 0.60 & 0.84 & 0.78 & 0.71 & 0.86 & Horseshoe \\
  RegCDR squared & -0.03 & -0.03 & -0.02 & -0.03 & -0.03 & -0.02 & Horseshoe \\
  \%65+ & -11.93 & -15.49 & -8.36 & -13.79 & -16.27 & -11.32 & Horseshoe \\
   ln(5q0) & -1.32 & -1.50 & -1.14 & -1.56 & -1.68 & -1.43 & Horseshoe \\
   C5q0 & 2.24 & 1.81 & 2.67 & 0.76 & 0.54 & 0.99 & Horseshoe \\
   year & -0.03 & -0.04 & -0.02 & -0.02 & -0.03 & -0.02 & Horseshoe \\ \hline   \hline
  \end{tabular}
\caption{
Posterior mean and credible intervals of the regression parameters  for
Model 1 (both sexes).}
\label{tab:tabglobal1}
\end{table}

\clearpage

\begin{table}[ht]
\scriptsize
\renewcommand{\arraystretch}{1.4}
\begin{tabular}{rrrr|rrr|rrrrrrrr}\\
\hline    \hline
Prior for the errors &    & Gamma  &   &  &     Half-Cauchy (local) &  & &  &\\ \hline \hline
Parameter & Mean  & Lower CI  &  Upper CI &   Mean  & Lower
CI  &  Upper CI & Local-Prior \\
  \hline   \hline
 Constant & 52.65 & 39.59 & 65.71 & 33.70 & 25.62 & 41.79  & Gamma  \\
 RegCDR & 1.06 & 0.95 & 1.16 & 1.03 & 0.98 & 1.09 &  Gamma  \\
 RegCDR squared & -0.04 & -0.05 & -0.04 & -0.04 & -0.05 & -0.04 &  Gamma  \\
 \%65+ & -17.11 & -20.47 & -13.74 & -12.81 & -15.42 & -10.21 &  Gamma  \\
 ln(5q0) & -1.56 & -1.73 & -1.39 & -1.46 & -1.58 & -1.33 &  Gamma \\
 year & -0.03 & -0.04 & -0.02 & -0.02 & -0.02 & -0.02 & Gamma \\
Constant & 54.95 & 42.22 & 67.68 & 34.55 & 26.86 & 42.24 & Student-t \\
RegCDR & 1.06 & 0.96 & 1.17 & 1.03 & 0.97 & 1.09 & Student-t \\
RegCDR squared & -0.04 & -0.05 & -0.04 & -0.04 & -0.05 & -0.04 & Student-t \\
\%65+ & -16.91 & -20.29 & -13.54 & -12.55 & -15.20 & -9.91 & Student-t \\
ln(5q0) & -1.59 & -1.75 & -1.43 & -1.46 & -1.59 & -1.34 & Student-t \\
year & -0.03 & -0.04 & -0.02 & -0.02 & -0.02 & -0.02 & Student-t \\
Constant & 56.09 & 43.41 & 68.78 & 35.02 & 27.15 & 42.89 & LA \\
RegCDR & 1.07 & 0.96 & 1.17 & 1.03 & 0.97 & 1.09 & LA \\
RegCDR squared & -0.04 & -0.05 & -0.04 & -0.04 & -0.05 & -0.04 & LA \\
\%65+ & -16.84 & -20.25 & -13.43 & -12.43 & -15.04 & -9.82 & LA \\
ln(5q0) & -1.60 & -1.76 & -1.44 & -1.47 & -1.60 & -1.34 & LA \\
year & -0.03 & -0.04 & -0.02 & -0.02 & -0.02 & -0.02 & LA \\
Constant & 60.73 & 48.58 & 72.88 & 38.18 & 30.19 & 46.16 & HS \\
RegCDR & 1.08 & 0.98 & 1.19 & 1.02 & 0.97 & 1.08 & HS \\
RegCDR squared & -0.04 & -0.05 & -0.04 & -0.04 & -0.05 & -0.04 & HS \\
\%65+ & -16.71 & -20.13 & -13.30 & -12.04 & -14.76 & -9.32 & HS \\
ln(5q0) & -1.67 & -1.83 & -1.50 & -1.51 & -1.64 & -1.37 & HS \\
year & -0.03 & -0.04 & -0.03 & -0.02 & -0.03 & -0.02 & HS \\  \hline   \hline
\end{tabular}
\caption{
Posterior mean and credible intervals of the regression parameters  for
Model 2 (both sexes).}
\label{tab:tabglobal2}
\end{table}

\clearpage

\begin{table}[ht]
\scriptsize
\renewcommand{\arraystretch}{1.4}
\begin{tabular}{rrrr|rrr|rrrrrrrr}\\
\hline    \hline
Prior for the errors &    & Gamma  &   &  &     Half-Cauchy (local) &  & &  &\\ \hline \hline
Parameter & Mean  & Lower CI  &  Upper CI &   Mean  & Lower
CI  &  Upper CI & Local-Prior \\
Constant & 40.43 & 28.60 & 52.26 & 28.22 & 18.61 & 37.83 & Gamma  \\
 RegCDR & 0.79 & 0.67 & 0.91 & 0.86 & 0.75 & 0.96 & Gamma \\
RegCDR squared & -0.03 & -0.04 & -0.03 & -0.03 & -0.04 & -0.03 & Gamma  \\
 \%65+ & -12.26 & -15.36 & -9.16 & -14.20 & -16.62 & -11.79 & Gamma \\
 ln(5q0) & -1.18 & -1.35 & -1.01 & -1.37 & -1.50 & -1.24 & Gamma  \\
 C5q0 & 2.10 & 1.69 & 2.51 & 0.78 & 0.54 & 1.03 & Gamma \\
 year & -0.02 & -0.03 & -0.02 & -0.02 & -0.02 & -0.01 & Gamma  \\
Constant & 41.40 & 30.05 & 52.74 & 29.50 & 19.91 & 39.10 & Student-t \\
RegCDR & 0.79 & 0.67 & 0.91 & 0.86 & 0.76 & 0.97 & Student-t \\
RegCDR squared & -0.03 & -0.04 & -0.03 & -0.03 & -0.04 & -0.03 & Student-t \\
\%65+ & -12.17 & -15.26 & -9.08 & -14.07 & -16.48 & -11.65 & Student-t \\
ln(5q0) & -1.20 & -1.36 & -1.03 & -1.39 & -1.51 & -1.26 & Student-t \\
C5q0 & 2.13 & 1.72 & 2.53 & 0.77 & 0.51 & 1.02 & Student-t \\
year & -0.02 & -0.03 & -0.02 & -0.02 & -0.02 & -0.01 & Student-t \\
Constant & 41.88 & 30.73 & 53.04 & 30.12 & 20.72 & 39.51 & LA \\
RegCDR & 0.79 & 0.67 & 0.91 & 0.86 & 0.76 & 0.97 & LA \\
RegCDR squared & -0.03 & -0.04 & -0.03 & -0.03 & -0.04 & -0.03 & LA \\
\%65+ & -12.13 & -15.25 & -9.01 & -14.01 & -16.38 & -11.64 & LA \\
ln(5q0) & -1.20 & -1.37 & -1.04 & -1.39 & -1.52 & -1.27 & LA \\
C5q0 & 2.12 & 1.72 & 2.53 & 0.76 & 0.50 & 1.01 & LA \\
year & -0.02 & -0.03 & -0.02 & -0.02 & -0.02 & -0.01 & LA \\
Constant & 43.96 & 33.21 & 54.70 & 30.11 & 20.20 & 40.02 & HS \\
RegCDR & 0.78 & 0.65 & 0.91 & 0.87 & 0.76 & 0.99 & HS \\
RegCDR squared & -0.03 & -0.04 & -0.03 & -0.03 & -0.04 & -0.03 & HS \\
\%65+ & -12.06 & -15.20 & -8.91 & -14.01 & -16.46 & -11.56 & HS \\
ln(5q0) & -1.22 & -1.39 & -1.06 & -1.40 & -1.53 & -1.27 & HS \\
C5q0 & 2.18 & 1.76 & 2.59 & 0.71 & 0.44 & 0.99 & HS \\
year & -0.02 & -0.03 & -0.02 & -0.02 & -0.02 & -0.01 & HS \\
 \hline   \hline
\end{tabular}
\caption{
Posterior mean and credible intervals of the regression parameters  for
Model 1 (females).}
\label{tab:tabglobal3}
\end{table}

\clearpage

\begin{table}[ht]
\scriptsize
\renewcommand{\arraystretch}{1.4}
\begin{tabular}{rrrr|rrr|rrrrrrrr}\\
\hline    \hline
Prior for the errors &    & Gamma  &   &  &     Half-Cauchy (local) &  & &  &\\ \hline \hline
Parameter & Mean  & Lower CI  &  Upper CI &   Mean  & Lower
CI  &  Upper CI & Local-Prior \\
Constant & 50.83 & 39.19 & 62.47 & 26.80 & 18.69 & 34.92 &  Gamma\\
RegCDR & 1.11 & 1.01 & 1.22 & 1.13 & 1.07 & 1.19 &  Gamma\\
RegCDR squared & -0.05 & -0.05 & -0.04 & -0.05 & -0.05 & -0.04 &  Gamma\\
 \%65+ & -16.72 & -19.71 & -13.74 & -14.60 & -16.93 & -12.27 &  Gamma \\
ln(5q0) & -1.54 & -1.69 & -1.39 & -1.41 & -1.52 & -1.30 &  Gamma\\
 year & -0.03 & -0.03 & -0.02 & -0.02 & -0.02 & -0.01 &  Gamma\\
  Constant & 51.36 & 39.80 & 62.92 & 27.45 & 19.45 & 35.44 & Student-t \\
RegCDR & 1.12 & 1.02 & 1.23 & 1.13 & 1.07 & 1.19 & Student-t \\
RegCDR squared & -0.05 & -0.05 & -0.04 & -0.05 & -0.05 & -0.04 & Student-t \\
\%65+ & -16.47 & -19.48 & -13.47 & -14.19 & -16.46 & -11.92 & Student-t \\
ln(5q0) & -1.54 & -1.69 & -1.39 & -1.42 & -1.53 & -1.31 & Student-t \\
year & -0.03 & -0.03 & -0.02 & -0.02 & -0.02 & -0.01 & Student-t \\
Constant & 51.62 & 40.29 & 62.95 & 27.42 & 19.49 & 35.34 & LA \\
RegCDR & 1.12 & 1.02 & 1.23 & 1.13 & 1.07 & 1.19 & LA \\
RegCDR squared & -0.05 & -0.05 & -0.04 & -0.05 & -0.05 & -0.04 & LA \\
\%65+ & -16.43 & -19.45 & -13.42 & -14.17 & -16.47 & -11.88 & LA \\
ln(5q0) & -1.54 & -1.69 & -1.39 & -1.42 & -1.53 & -1.31 & LA \\
year & -0.03 & -0.03 & -0.02 & -0.02 & -0.02 & -0.01 & LA \\
Constant & 51.59 & 40.47 & 62.71 & 27.07 & 19.96 & 34.17 & HS \\
RegCDR & 1.14 & 1.03 & 1.25 & 1.13 & 1.07 & 1.18 & HS \\
RegCDR squared & -0.05 & -0.06 & -0.04 & -0.05 & -0.05 & -0.04 & HS \\
\%65+ & -15.88 & -18.98 & -12.77 & -13.99 & -16.28 & -11.69 & HS \\
ln(5q0) & -1.53 & -1.68 & -1.37 & -1.41 & -1.51 & -1.31 & HS \\
year & -0.03 & -0.03 & -0.02 & -0.02 & -0.02 & -0.01 & HS \\
 \hline   \hline
\end{tabular}
\caption{
Posterior mean and credible intervals of the regression parameters  for
Model 2 (females).}
\label{tab:tabglobal4}
\end{table}

\clearpage

\begin{table}[ht]
\scriptsize
\renewcommand{\arraystretch}{1.4}
\begin{tabular}{rrrr|rrr|rrrrrrrr}\\
\hline    \hline
Prior for the errors &    & Gamma  &   &  &     Half-Cauchy (local) &  & &  &\\ \hline \hline
Parameter & Mean  & Lower CI  &  Upper CI &   Mean  & Lower
CI  &  Upper CI & Local-Prior \\
 Constant & 45.14 & 31.92 & 58.36 & 29.73 & 20.16 & 39.30  & Gamma\\
 RegCDR & 0.62 & 0.51 & 0.72 & 0.72 & 0.66 & 0.77  & Gamma \\
 RegCDR squared & -0.02 & -0.03 & -0.02 & -0.03 & -0.03 & -0.02  & Gamma \\
 \%65+ & -9.63 & -13.07 & -6.19 & -13.94 & -16.77 & -11.10  & Gamma \\
 ln(5q0) & -1.12 & -1.29 & -0.95 & -1.35 & -1.48 & -1.22  & Gamma \\
 C5q0 & 2.16 & 1.75 & 2.56 & 0.94 & 0.73 & 1.14  & Gamma \\
 year & -0.03 & -0.03 & -0.02 & -0.02 & -0.02 & -0.01  & Gamma \\
Constant & 47.08 & 34.23 & 59.92 & 29.50 & 19.91 & 39.10 & Student-t \\
RegCDR & 0.62 & 0.52 & 0.72 & 0.86 & 0.76 & 0.97 & Student-t \\
RegCDR squared & -0.02 & -0.03 & -0.02 & -0.03 & -0.04 & -0.03 & Student-t \\
\%65+ & -9.32 & -12.74 & -5.90 & -14.07 & -16.48 & -11.65 & Student-t \\
ln(5q0) & -1.14 & -1.30 & -0.97 & -1.39 & -1.51 & -1.26 & Student-t \\
C5q0 & 2.17 & 1.77 & 2.57 & 0.77 & 0.51 & 1.02 & Student-t \\
year & -0.03 & -0.03 & -0.02 & -0.02 & -0.02 & -0.01 & Student-t \\
Constant & 49.05 & 36.58 & 61.51 & 30.12 & 20.72 & 39.51 & LA \\
RegCDR & 0.62 & 0.51 & 0.72 & 0.86 & 0.76 & 0.97 & LA \\
RegCDR squared & -0.02 & -0.03 & -0.02 & -0.03 & -0.04 & -0.03 & LA \\
\%65+ & -9.11 & -12.55 & -5.68 & -14.01 & -16.38 & -11.64 & LA \\
ln(5q0) & -1.16 & -1.32 & -0.99 & -1.39 & -1.52 & -1.27 & LA \\
C5q0 & 2.16 & 1.75 & 2.57 & 0.76 & 0.50 & 1.01 & LA \\
year & -0.03 & -0.03 & -0.02 & -0.02 & -0.02 & -0.01 & LA \\
Constant & 55.19 & 44.12 & 66.26 & 30.11 & 20.20 & 40.02 & HS \\
RegCDR & 0.61 & 0.51 & 0.72 & 0.87 & 0.76 & 0.99 & HS \\
RegCDR squared & -0.02 & -0.03 & -0.02 & -0.03 & -0.04 & -0.03 & HS \\
\%65+ & -8.42 & -11.81 & -5.03 & -14.01 & -16.46 & -11.56 & HS \\
ln(5q0) & -1.22 & -1.38 & -1.06 & -1.40 & -1.53 & -1.27 & HS \\
C5q0 & 2.16 & 1.75 & 2.56 & 0.71 & 0.44 & 0.99 & HS \\
year & -0.03 & -0.04 & -0.02 & -0.02 & -0.02 & -0.01 & HS \\
 \hline   \hline
\end{tabular}
\caption{
Posterior mean and credible intervals of the regression parameters  for
Model 1 (males).}
\label{tab:tabglobal5}
\end{table}

\clearpage

\begin{table}[ht]
\scriptsize
\renewcommand{\arraystretch}{1.4}
\begin{tabular}{rrrr|rrr|rrrrrrrr}\\
\hline    \hline
Prior for the errors &    & Gamma  &   &  &     Half-Cauchy (local) &  & &  &\\ \hline \hline
Parameter & Mean  & Lower CI  &  Upper CI &   Mean  & Lower
CI  &  Upper CI & Local-Prior \\
 Constant & 57.43 & 44.32 & 70.53 & 26.80 & 18.69 & 34.92 & Gamma \\
 RegCDR & 0.89 & 0.79 & 0.98 & 1.13 & 1.07 & 1.19 & Gamma\\
 RegCDR squared & -0.03 & -0.04 & -0.03 & -0.05 & -0.05 & -0.04 &  Gamma \\
 \%65+ & -13.33 & -16.74 & -9.92 & -14.60 & -16.93 & -12.27 &  Gamma\\
 ln(5q0) & -1.45 & -1.61 & -1.29 & -1.41 & -1.52 & -1.30 &  Gamma \\
 year & -0.03 & -0.04 & -0.03 & -0.02 & -0.02 & -0.01 & Gamma \\
Constant & 59.57 & 46.73 & 72.41 & 27.45 & 19.45 & 35.44 & Student-t \\
RegCDR & 0.89 & 0.80 & 0.98 & 1.13 & 1.07 & 1.19 & Student-t \\
RegCDR squared & -0.03 & -0.04 & -0.03 & -0.05 & -0.05 & -0.04 & Student-t \\
\%65+ & -12.86 & -16.27 & -9.45 & -14.19 & -16.46 & -11.92 & Student-t \\
ln(5q0) & -1.47 & -1.63 & -1.31 & -1.42 & -1.53 & -1.31 & Student-t \\
year & -0.03 & -0.04 & -0.03 & -0.02 & -0.02 & -0.01 & Student-t \\
Constant & 61.10 & 48.65 & 73.55 & 27.42 & 19.49 & 35.34 & LA \\
RegCDR & 0.89 & 0.80 & 0.99 & 1.13 & 1.07 & 1.19 & LA \\
RegCDR squared & -0.03 & -0.04 & -0.03 & -0.05 & -0.05 & -0.04 & LA \\
\%65+ & -12.69 & -16.06 & -9.33 & -14.17 & -16.47 & -11.88 & LA \\
ln(5q0) & -1.49 & -1.64 & -1.33 & -1.42 & -1.53 & -1.31 & LA \\
year & -0.03 & -0.04 & -0.03 & -0.02 & -0.02 & -0.01 & LA \\
Constant & 64.99 & 53.71 & 76.27 & 27.07 & 19.96 & 34.17 & HS \\
RegCDR & 0.91 & 0.82 & 1.00 & 1.13 & 1.07 & 1.18 & HS \\
RegCDR squared & -0.03 & -0.04 & -0.03 & -0.05 & -0.05 & -0.04 & HS \\
\%65+ & -11.73 & -14.96 & -8.51 & -13.99 & -16.28 & -11.69 & HS \\
ln(5q0) & -1.52 & -1.67 & -1.38 & -1.41 & -1.51 & -1.31 & HS \\
year & -0.04 & -0.04 & -0.03 & -0.02 & -0.02 & -0.01 & HS \\
 \hline   \hline
\end{tabular}
\caption{
Posterior mean and credible intervals of the regression parameters  for
Model 2 (males).}
\label{tab:tabglobal6}
\end{table}

\clearpage

\section{Posterior densities and additional results of the regression parameters  for the models in Section \ref{Results}}
\label{sup_post2}

\enlargethispage{10cm}
This section contains the posterior densities of the regression parameters  of the models in Section \ref{Results}.  Figures \ref{fig:fig_both1}-\ref{fig:fig_both2} and
\ref{fig:fig_both3}-\ref{fig:fig_both4} consider models 1 and 2 respectively for both sexes, Figures \ref{fig:fig_both5}-\ref{fig:fig_both6} and  \ref{fig:fig_both7}-\ref{fig:fig_both8} display the results for models 1 and 2 respectively for females and  Figures \ref{fig:fig_both9}-\ref{fig:fig_both10} and \ref{fig:fig_both11}-\ref{fig:fig_both12} illustrate the results for models 1 and 2 respectively for males.

\begin{figure}[ht]

\begin{center}
\begin{tabular}{ccc}
\includegraphics[width=0.3\textwidth]{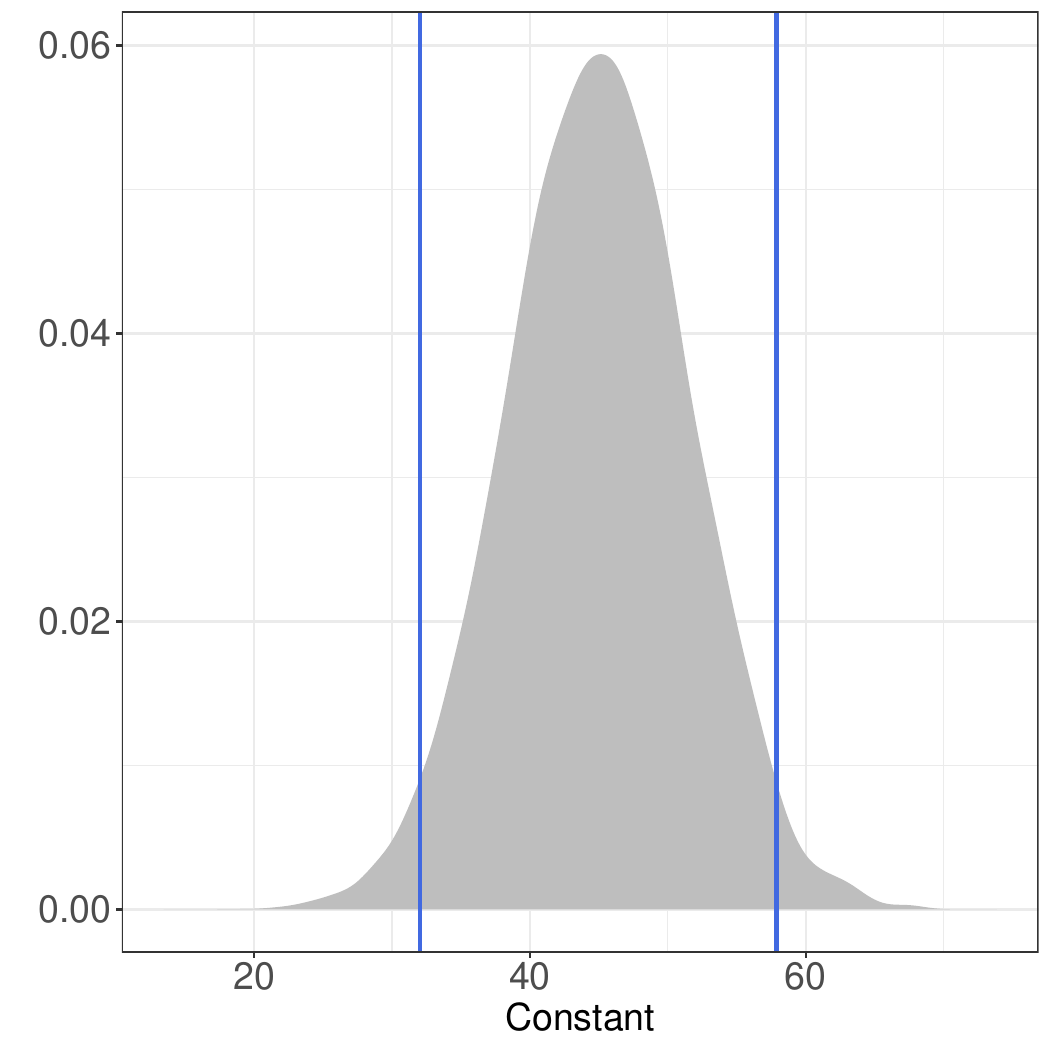} &
\includegraphics[width=0.3\textwidth]{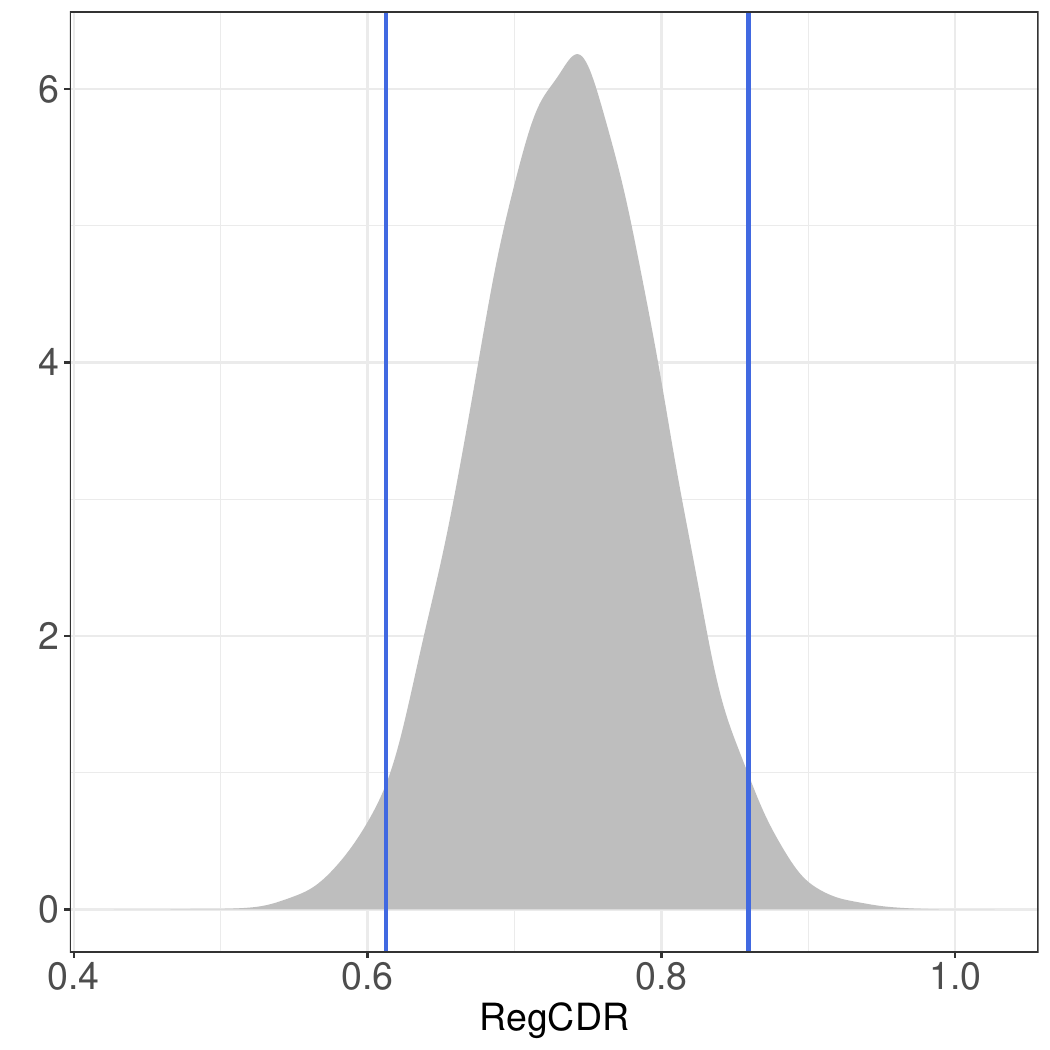}\\
\includegraphics[width=0.3\textwidth]{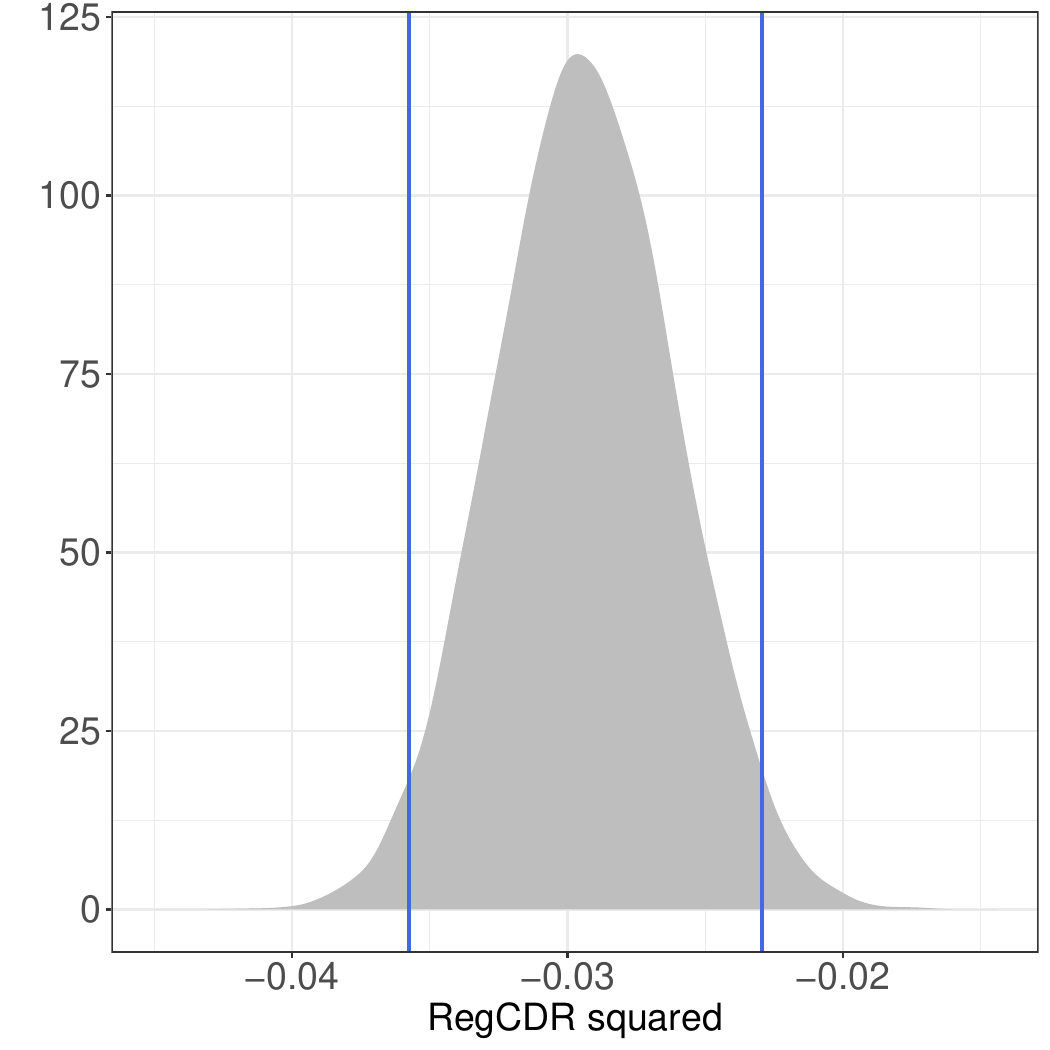} &
\includegraphics[width=0.3\textwidth]{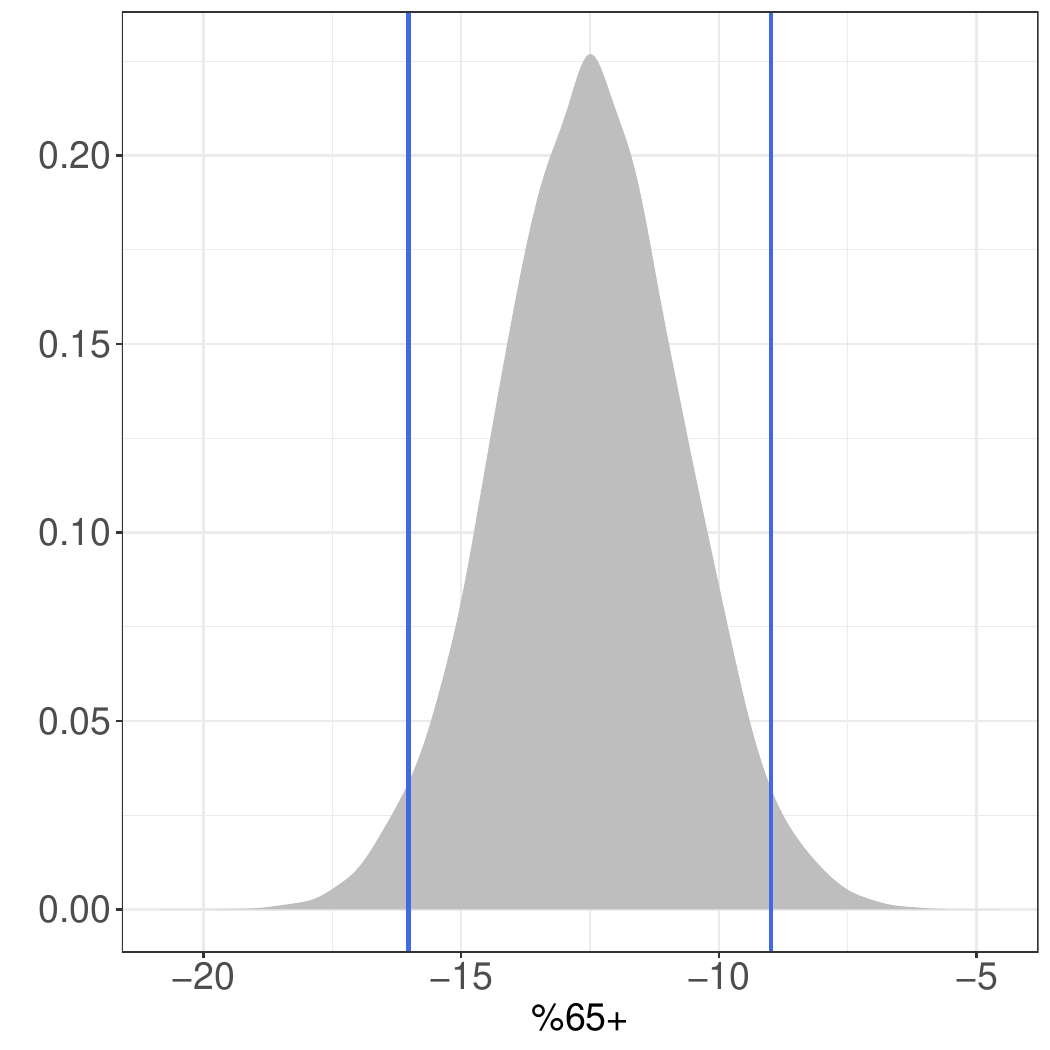}\\
\includegraphics[width=0.3\textwidth]{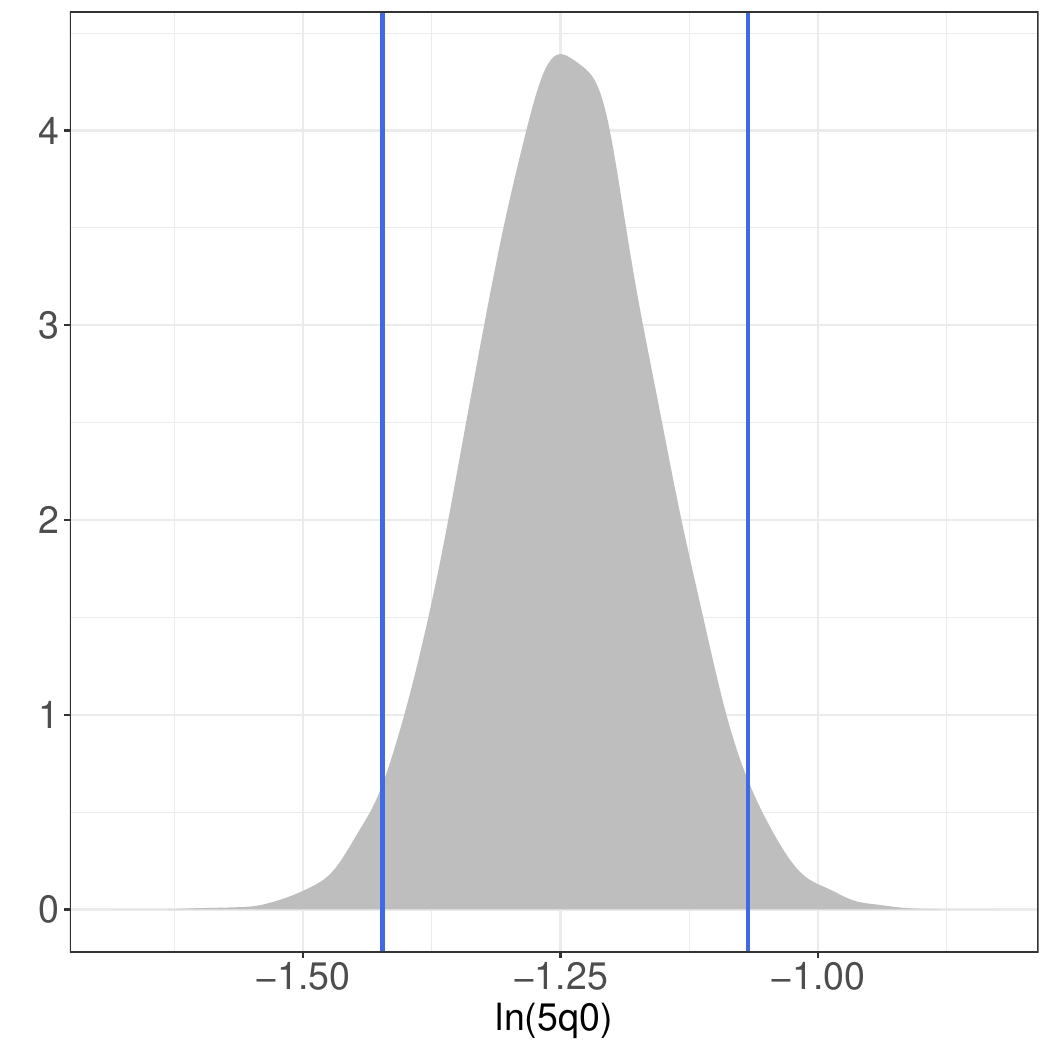} &
\includegraphics[width=0.3\textwidth]{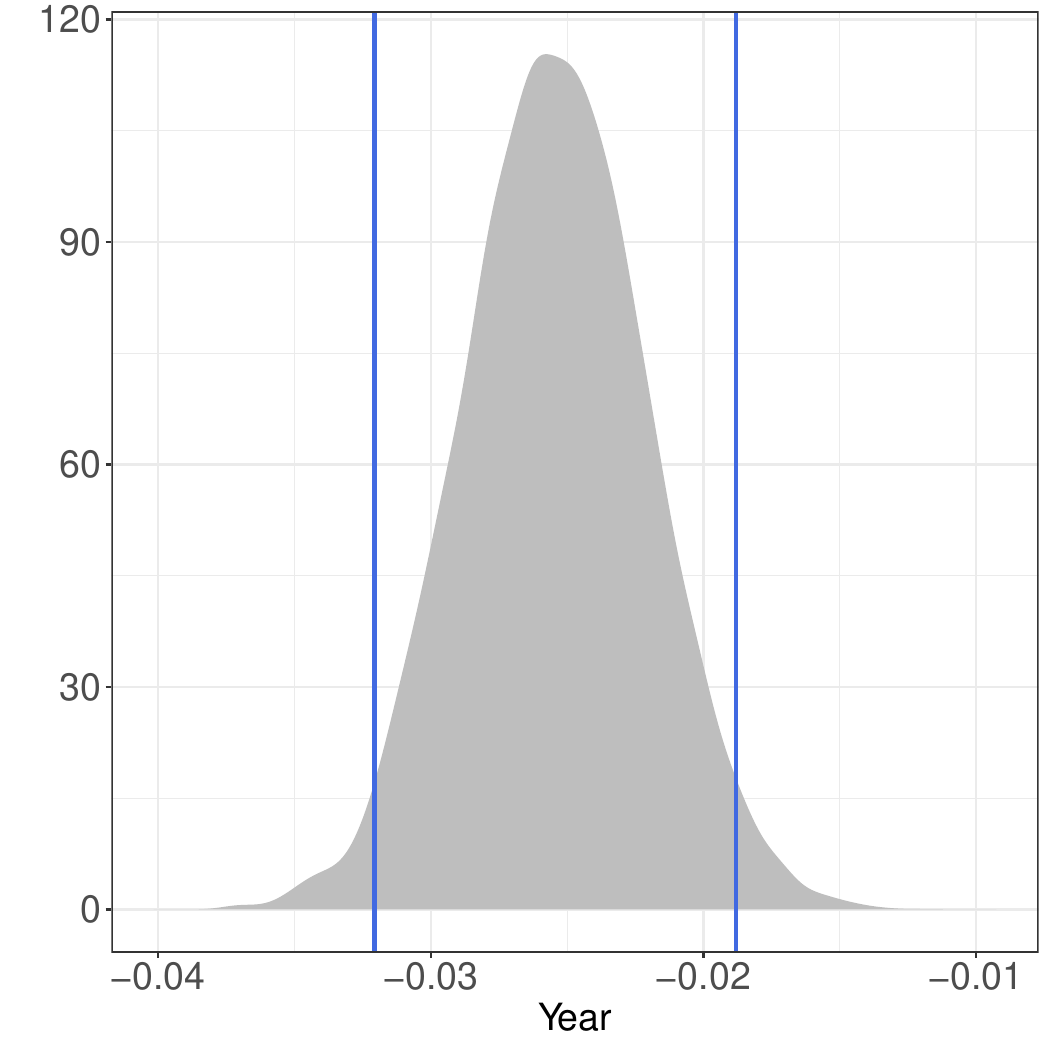}
\end{tabular}
\vspace{-0.5cm}
\caption{
Densities with credible intervals (blue lines).  The  models are based on an  dataset updated to 2019, which uses GBD death estimates based on the GBD 2019 and  comprises 120 countries and 2,748 country-years from 1970-2019  \citep{collaborators2020global}, both sexes and model 1 using a common Gamma prior
for the scale of the errors.}
\label{fig:fig_both1}
\end{center}
\end{figure}

\clearpage

\begin{figure}[ht]
\begin{center}
\begin{tabular}{ccc}
\includegraphics[width=0.4\textwidth]{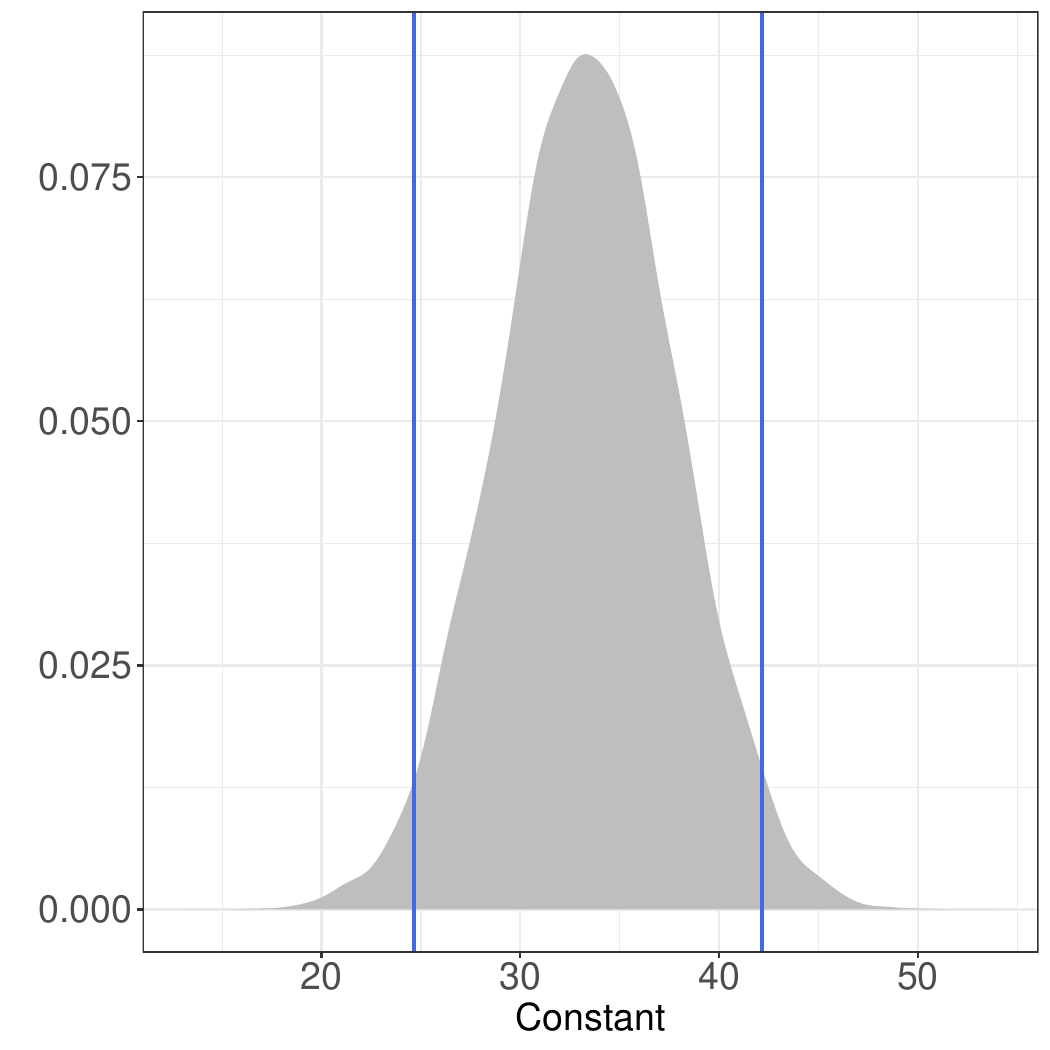} &
\includegraphics[width=0.4\textwidth]{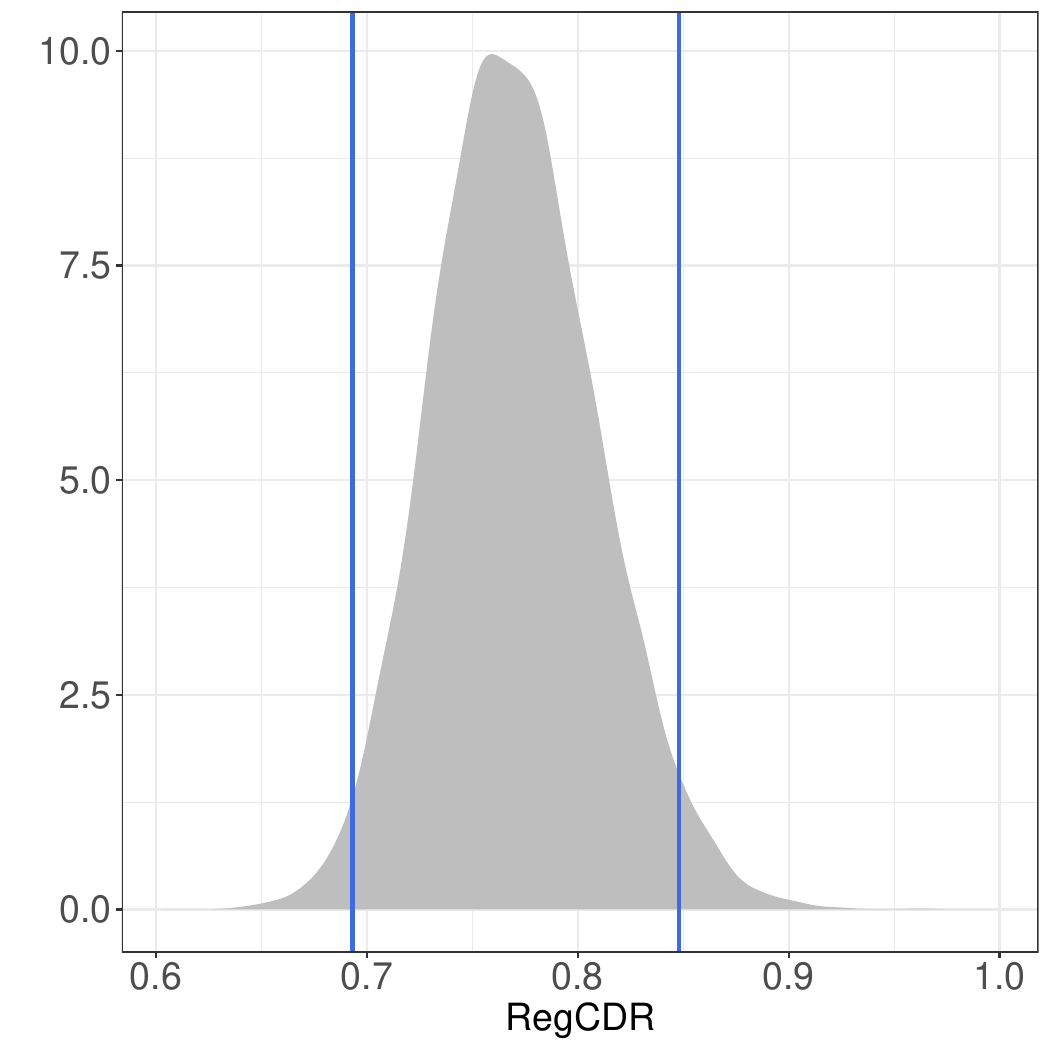}\\
\includegraphics[width=0.4\textwidth]{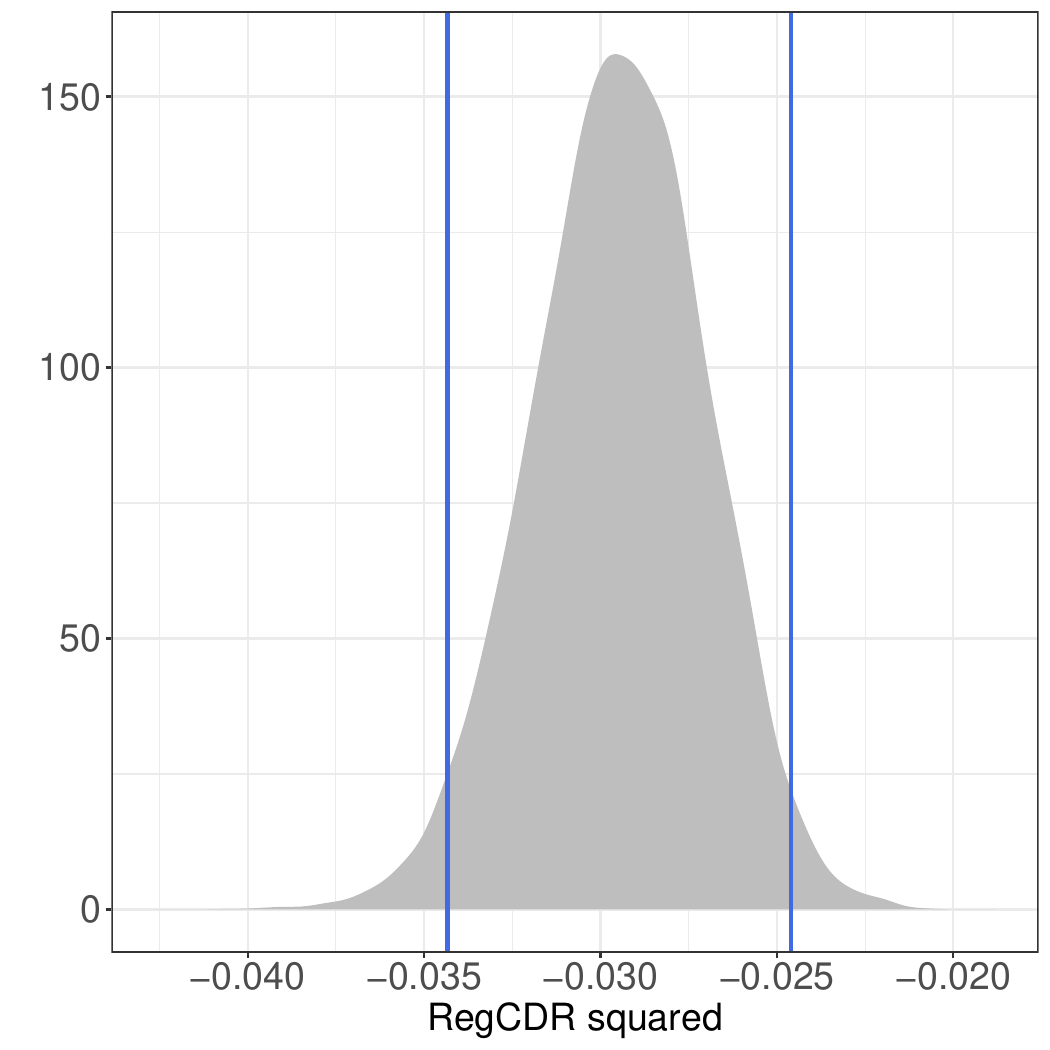} &
\includegraphics[width=0.4\textwidth]{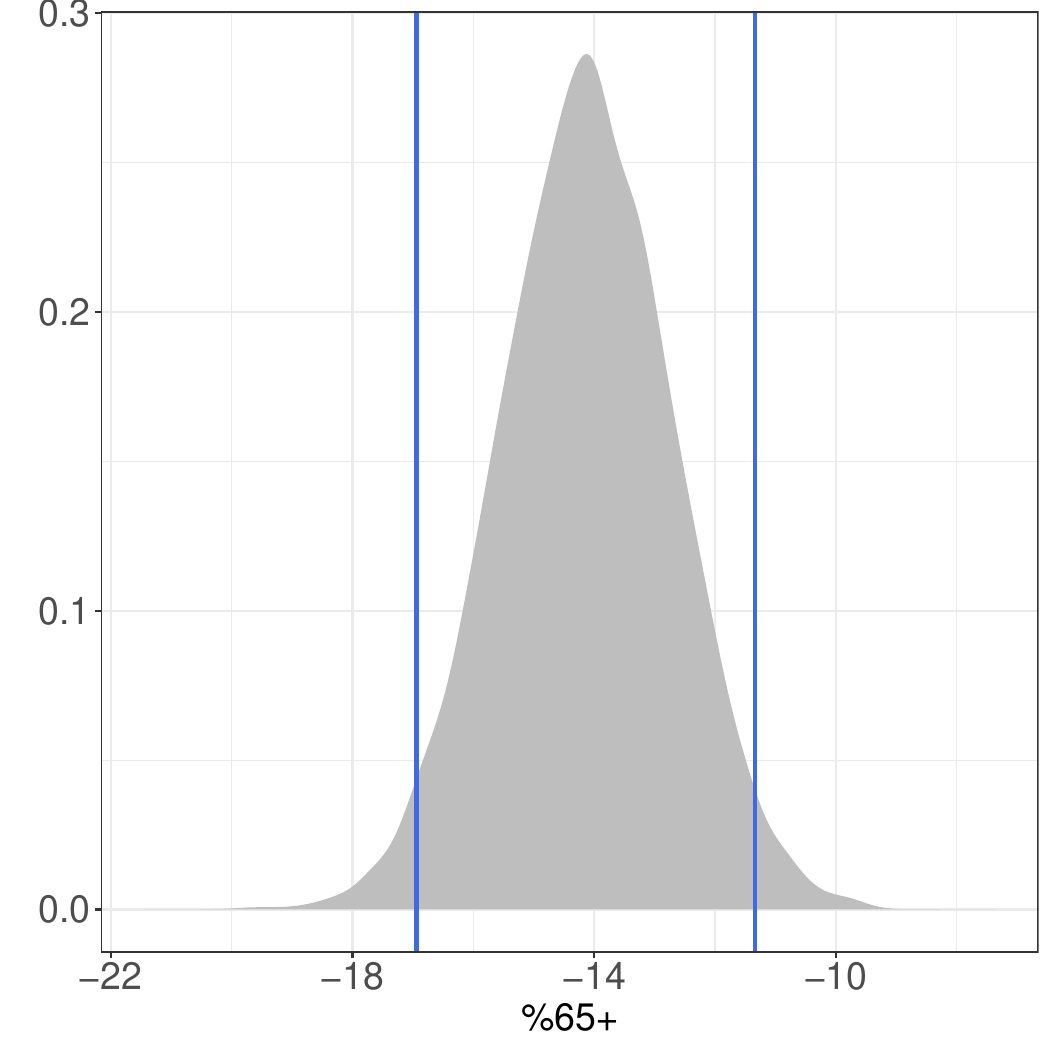}\\
\includegraphics[width=0.4\textwidth]{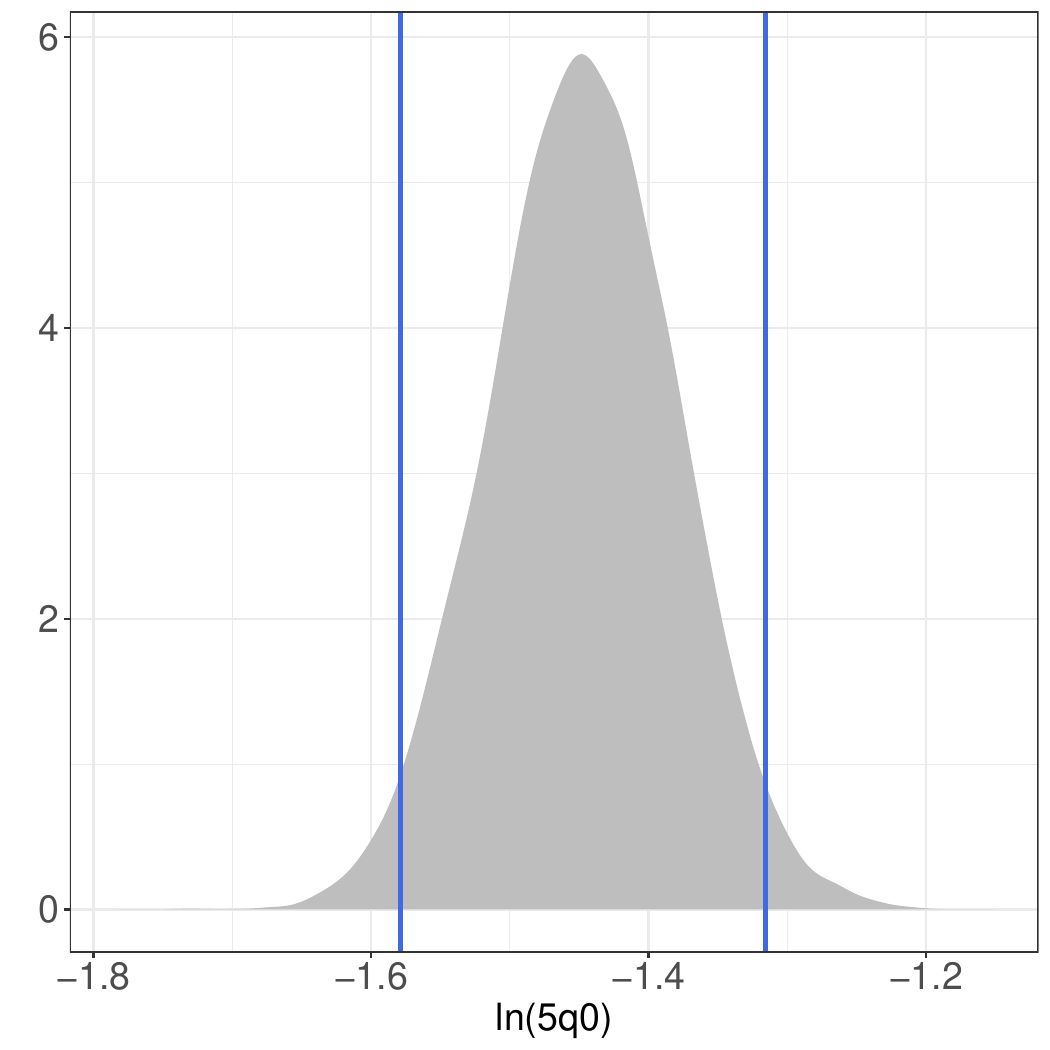} &
\includegraphics[width=0.4\textwidth]{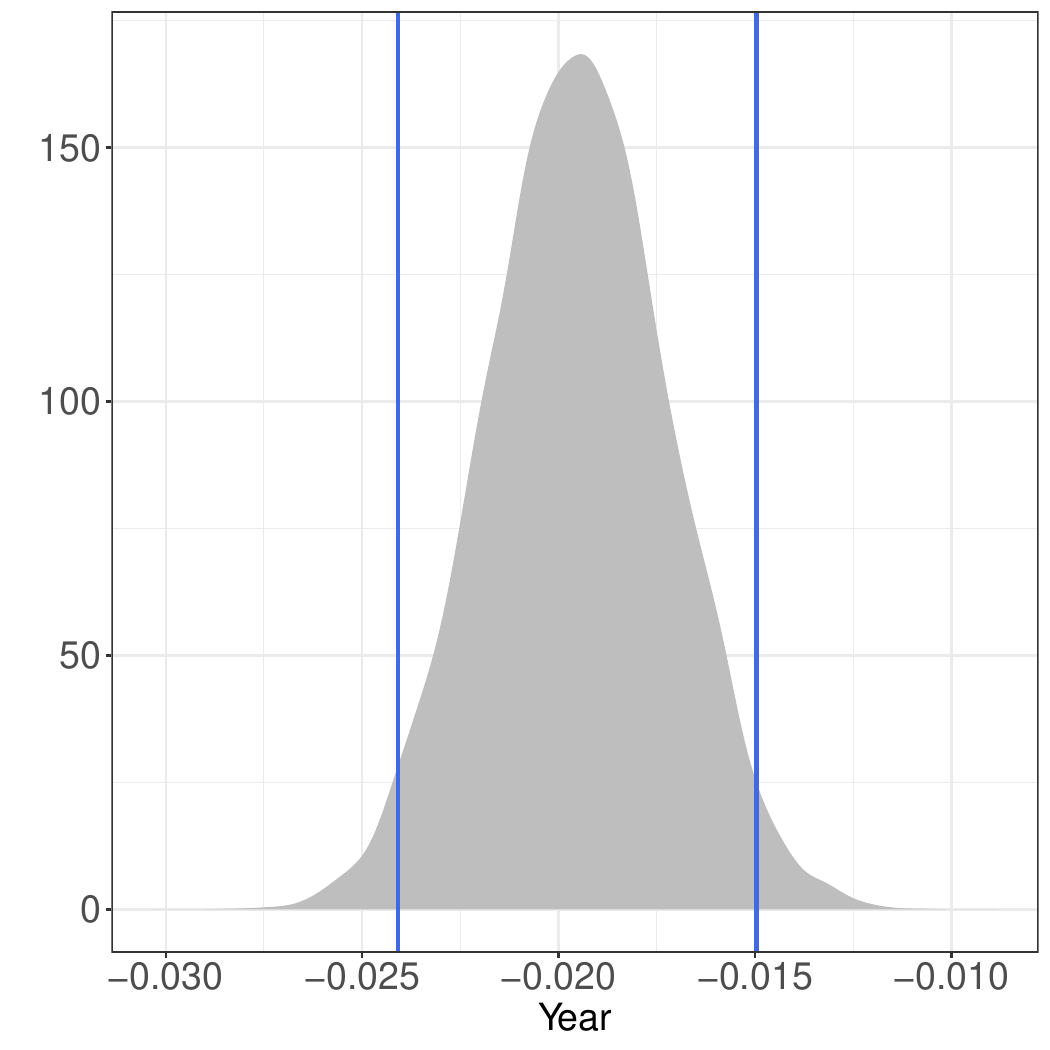}
\end{tabular}
\end{center}
\vspace{-0.5cm}
\caption{
Densities with credible intervals (blue lines). The  models are based on an  dataset updated to 2019, which uses GBD death estimates based on the GBD 2019 and  comprises 120 countries and 2,748 country-years from 1970-2019  \citep{collaborators2020global}, both sexes and model 1 using a Half-Cauchy prior
for the local scale of the errors.}
\label{fig:fig_both2}
\end{figure}

\clearpage

\begin{figure}[ht]

\begin{center}
\begin{tabular}{ccc}
\includegraphics[width=0.4\textwidth]{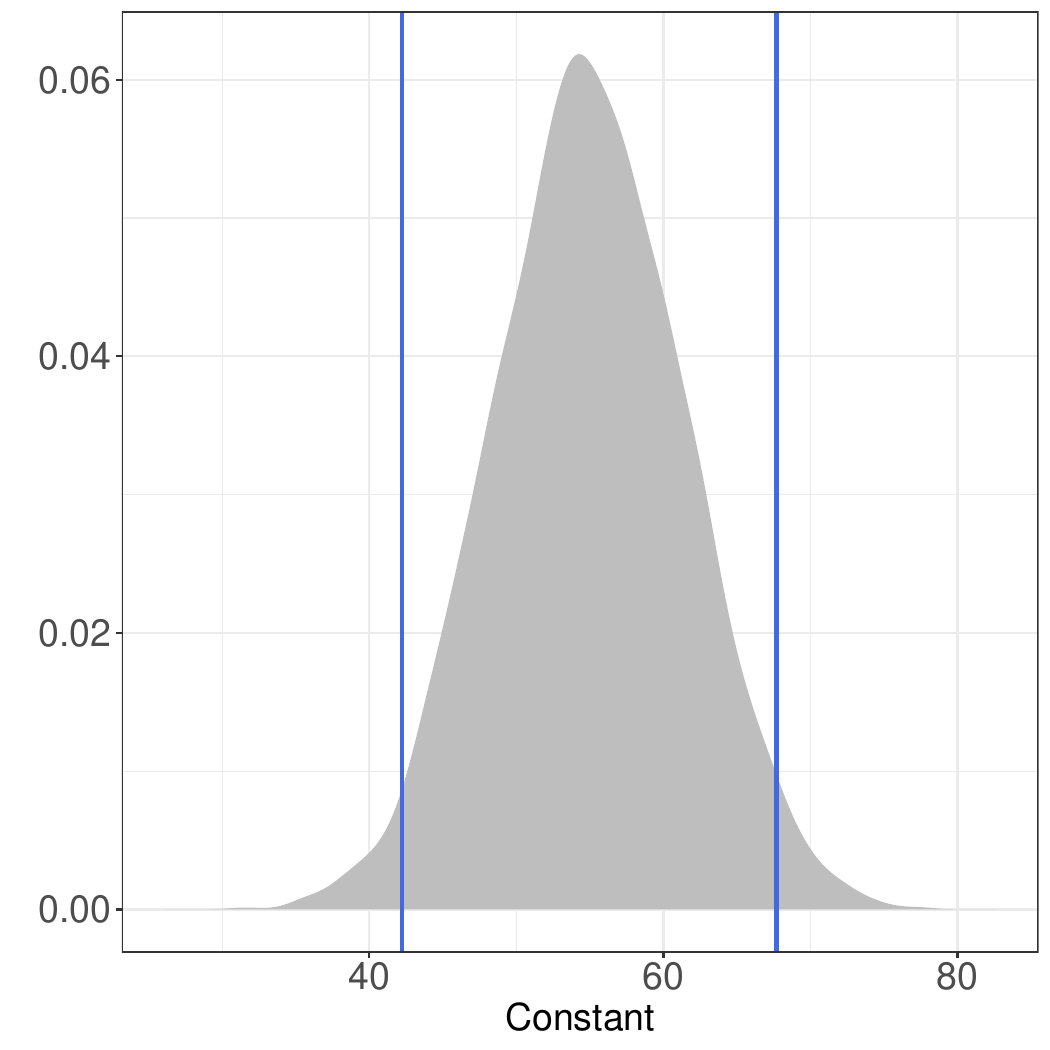} &
\includegraphics[width=0.4\textwidth]{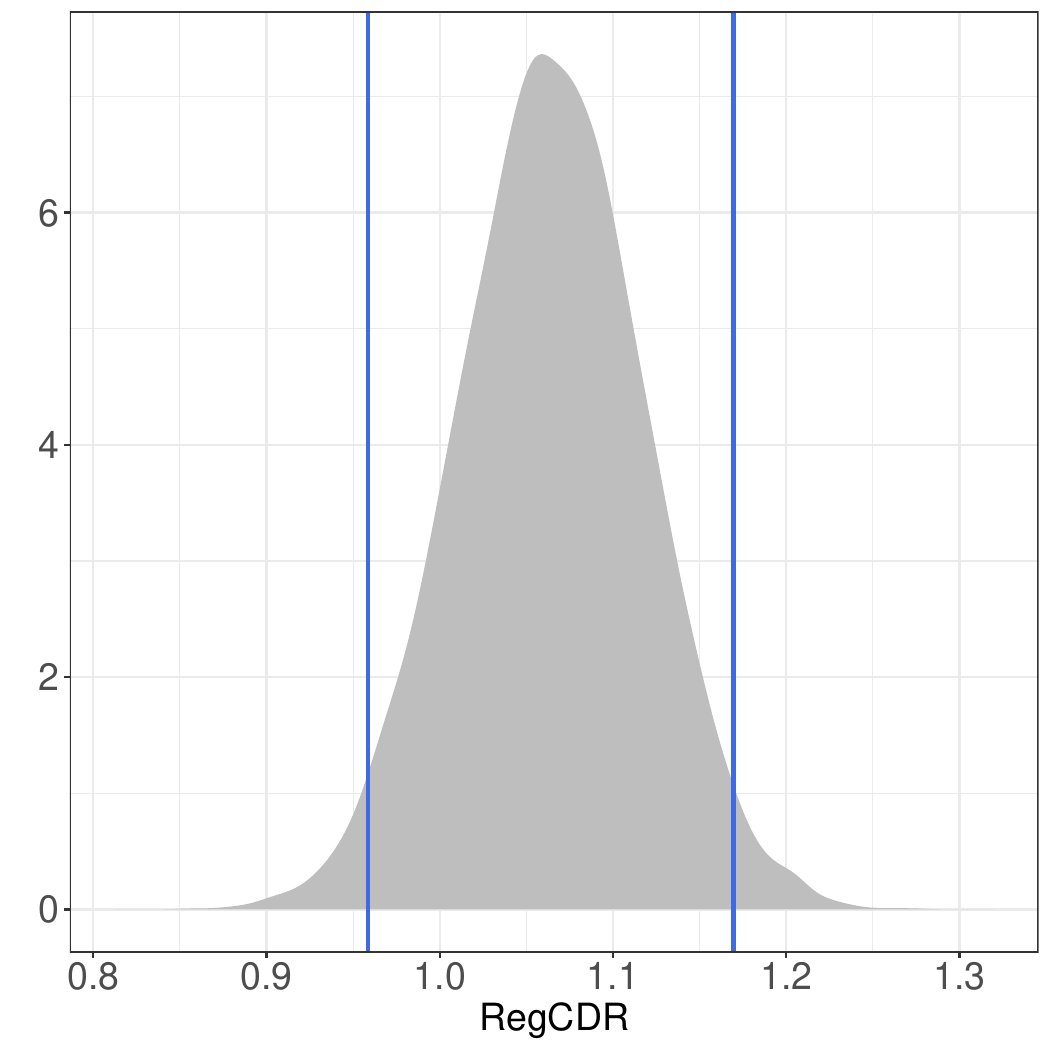}\\
\includegraphics[width=0.4\textwidth]{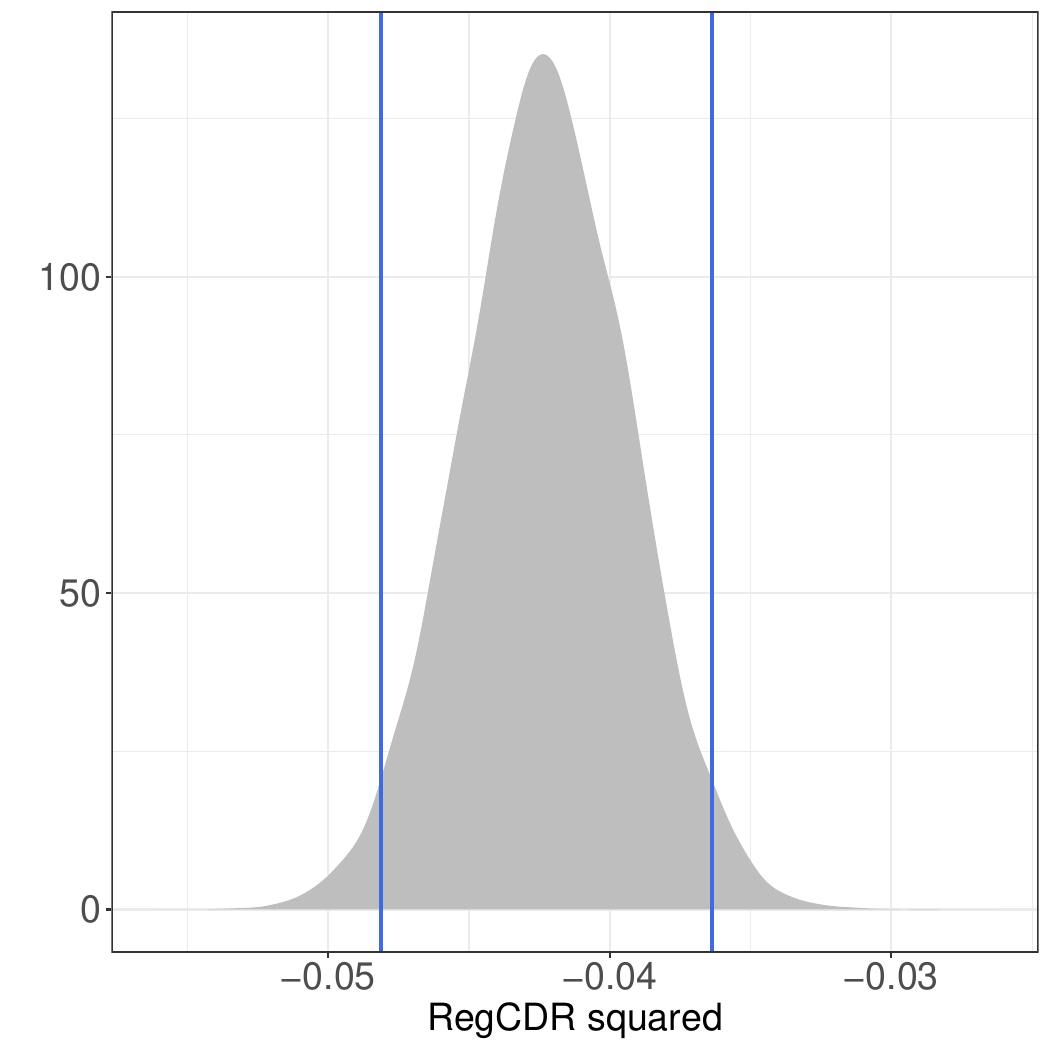} &
\includegraphics[width=0.4\textwidth]{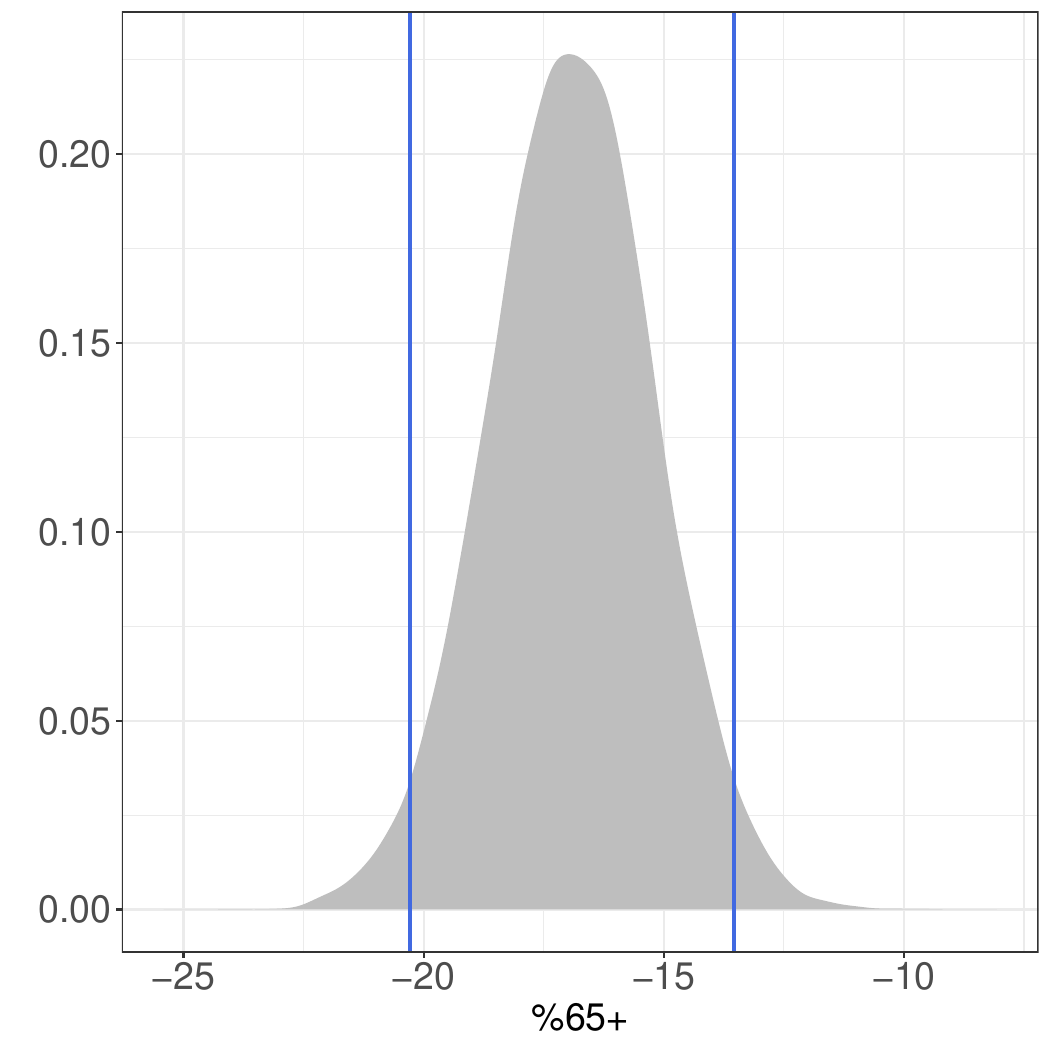}\\
\includegraphics[width=0.4\textwidth]{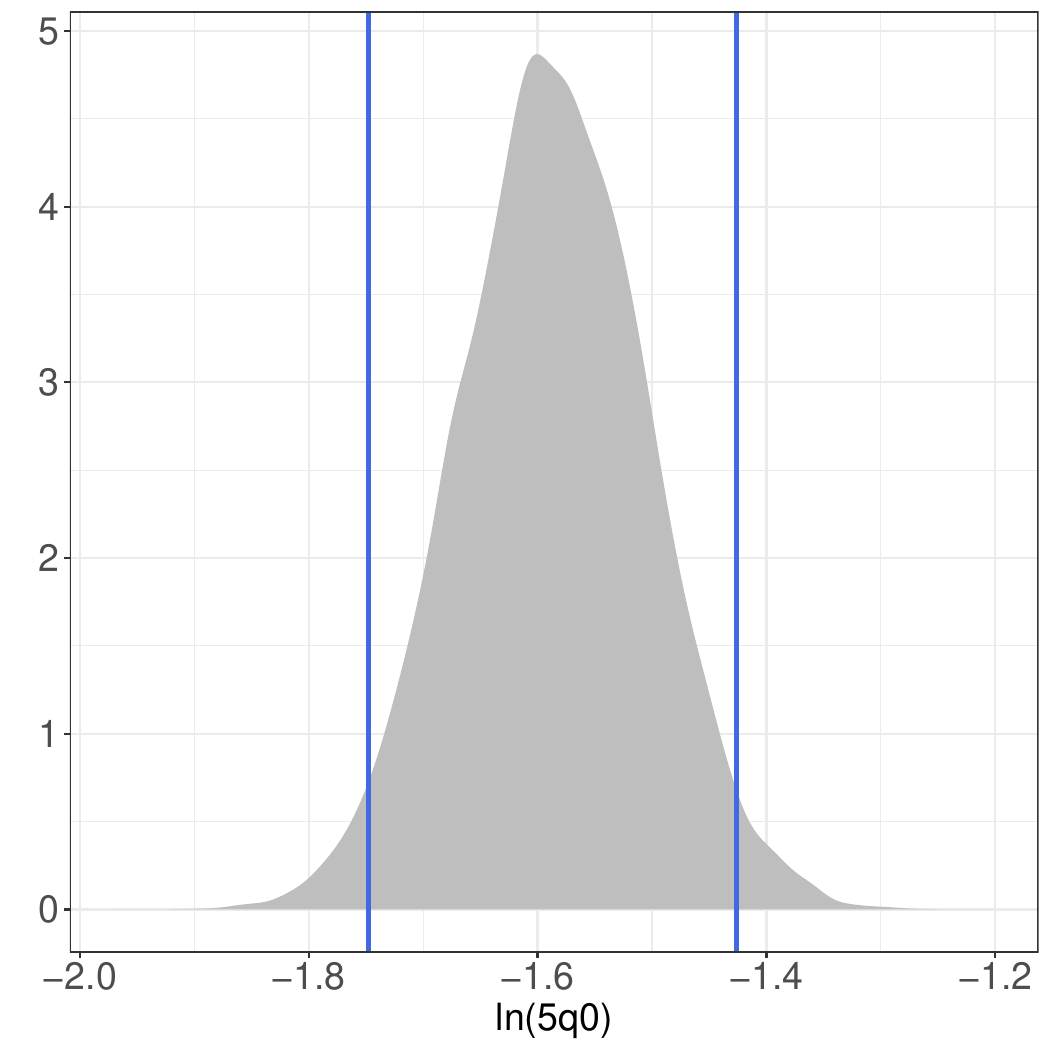} &
\includegraphics[width=0.4\textwidth]{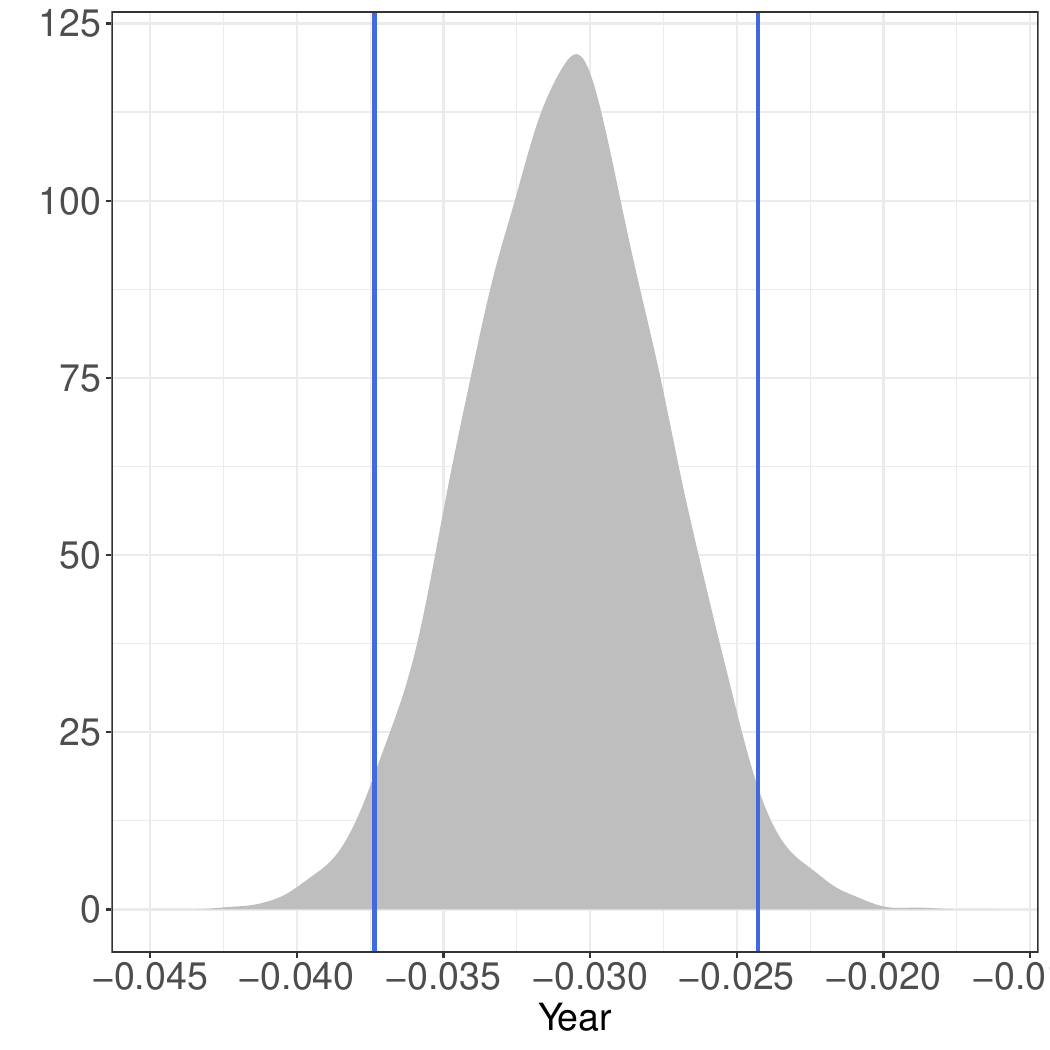}
\end{tabular}
\end{center}
\vspace{-0.5cm}
\caption{Densities with credible intervals (blue lines).  The  models are based on an  dataset updated to 2019, which uses GBD death estimates based on the GBD 2019 and  comprises 120 countries and 2,748 country-years from 1970-2019  \citep{collaborators2020global}, both sexes and model 2 using a common Gamma prior
for the scale of the errors.}
\label{fig:fig_both3}
\end{figure}

\begin{figure}[ht]
\begin{center}
\begin{tabular}{ccc}
\includegraphics[width=0.4\textwidth]{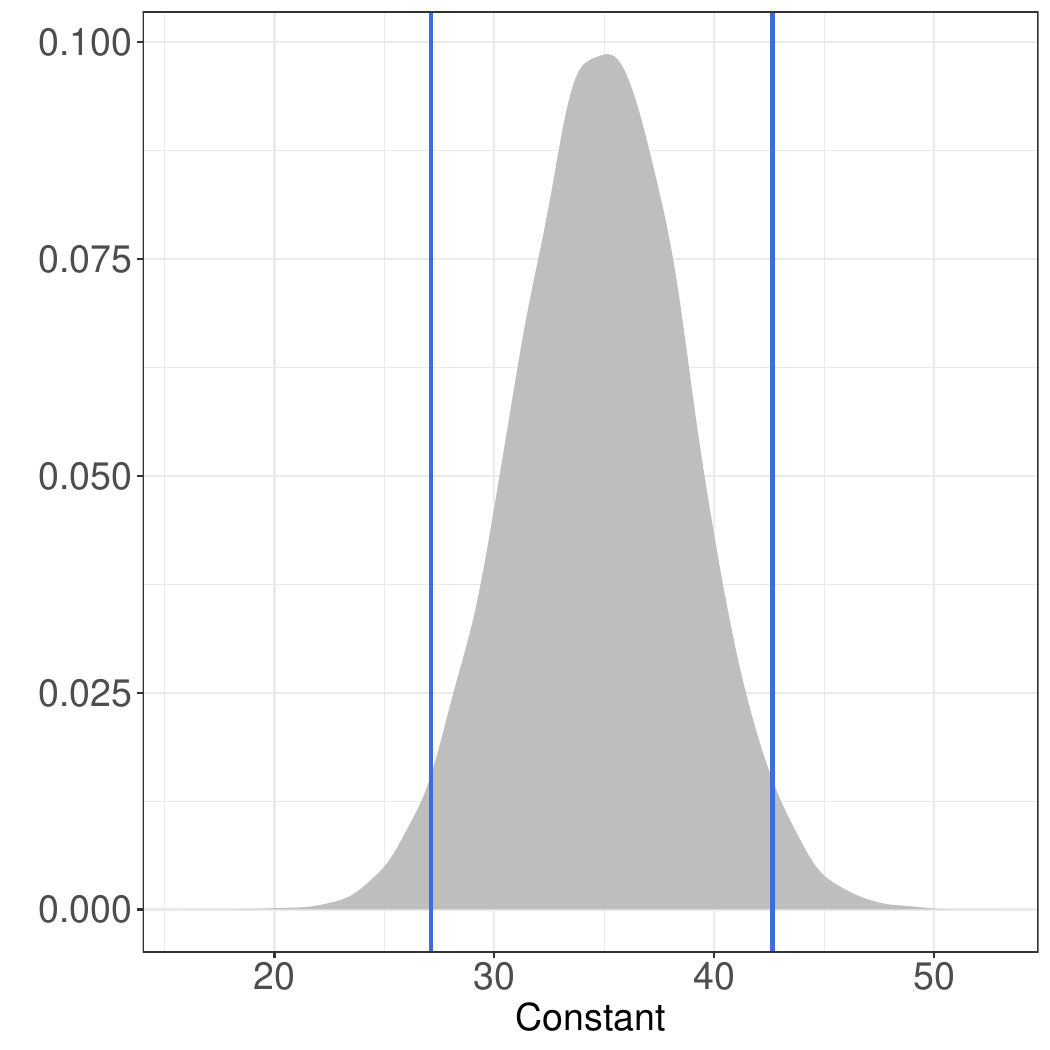} &
\includegraphics[width=0.4\textwidth]{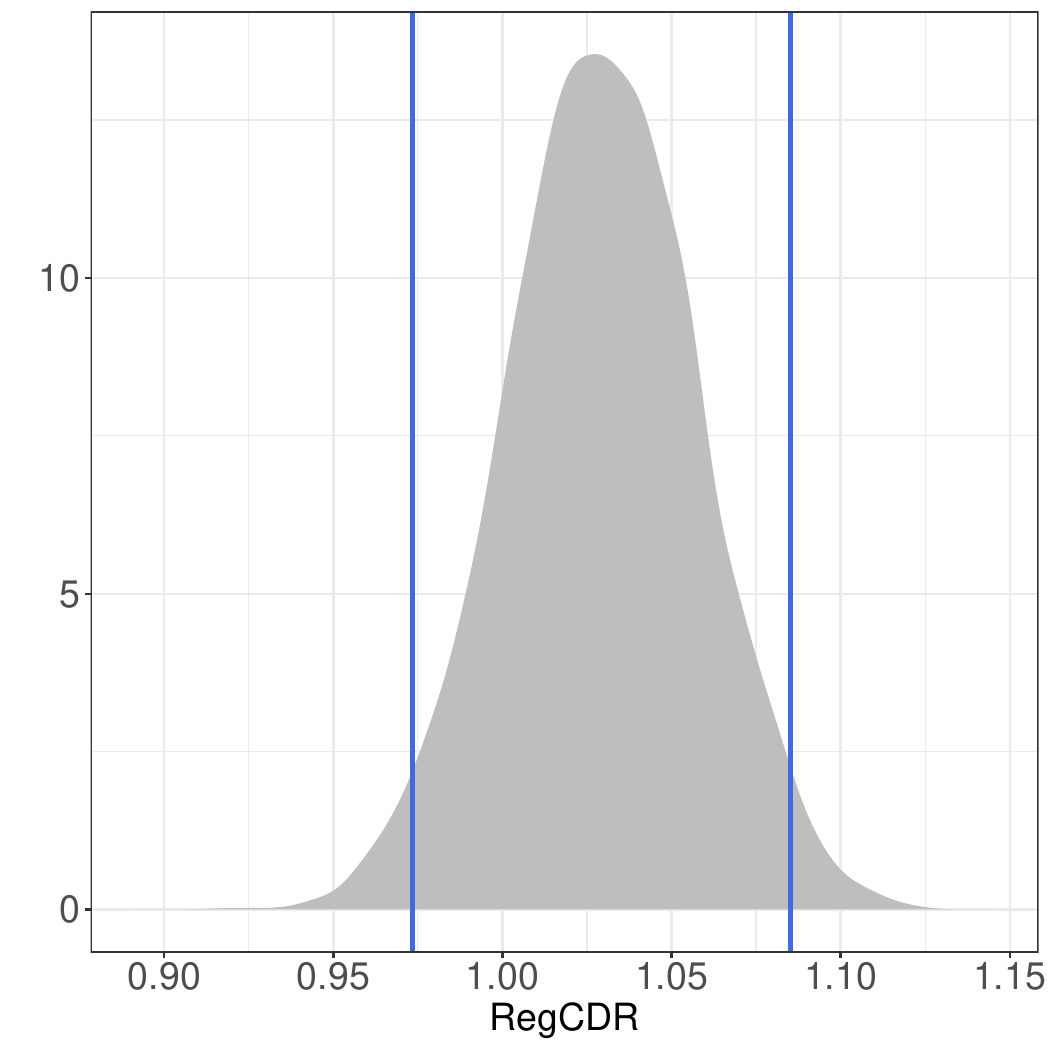}\\
\includegraphics[width=0.4\textwidth]{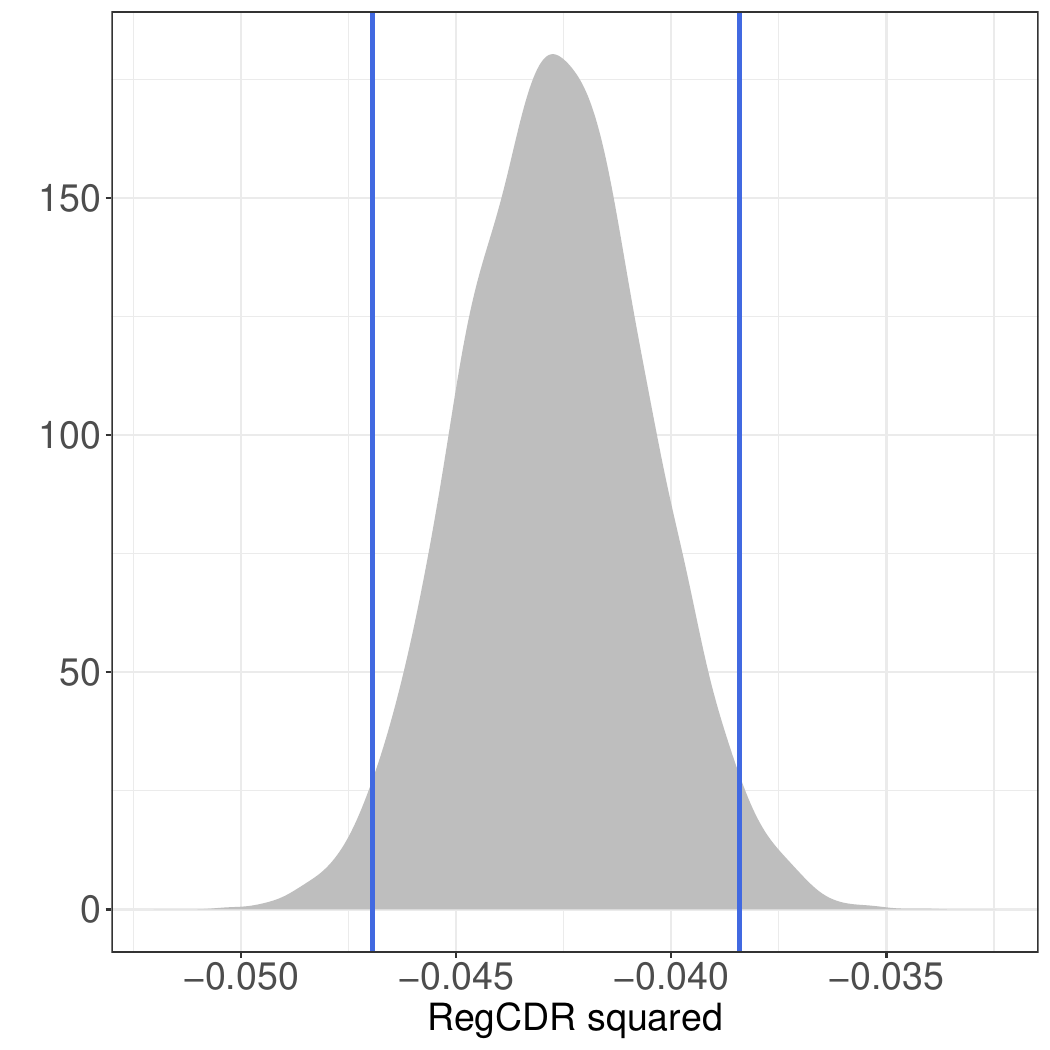} &
\includegraphics[width=0.4\textwidth]{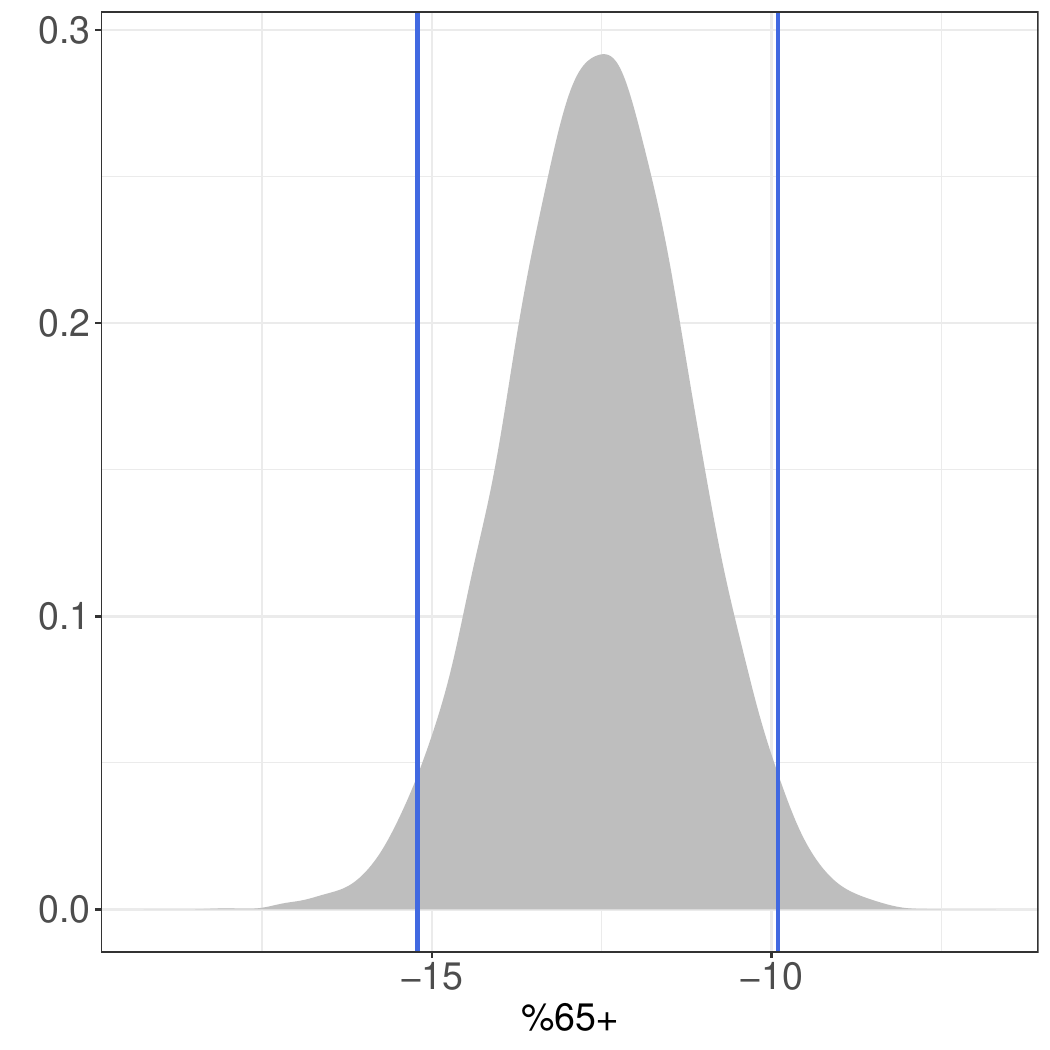}\\
\includegraphics[width=0.4\textwidth]{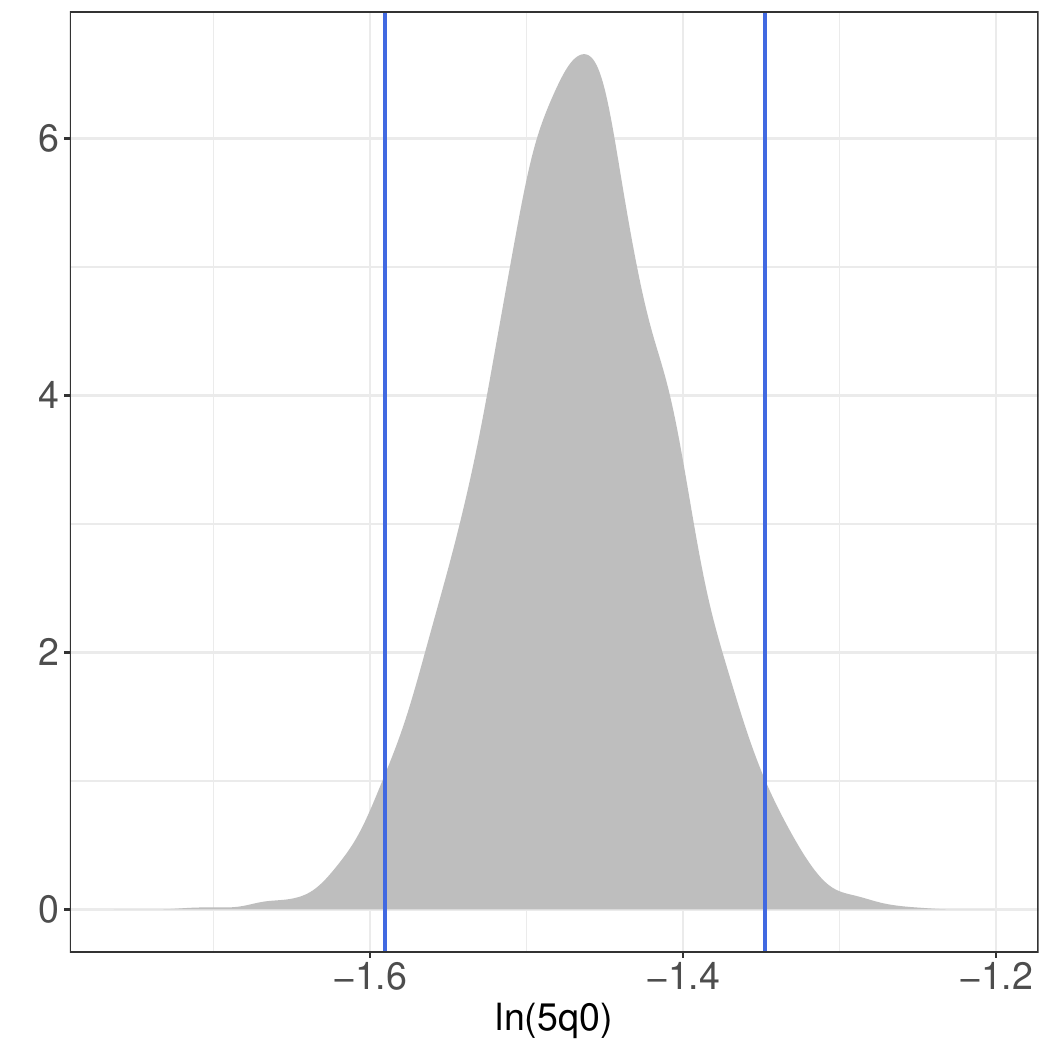} &
\includegraphics[width=0.4\textwidth]{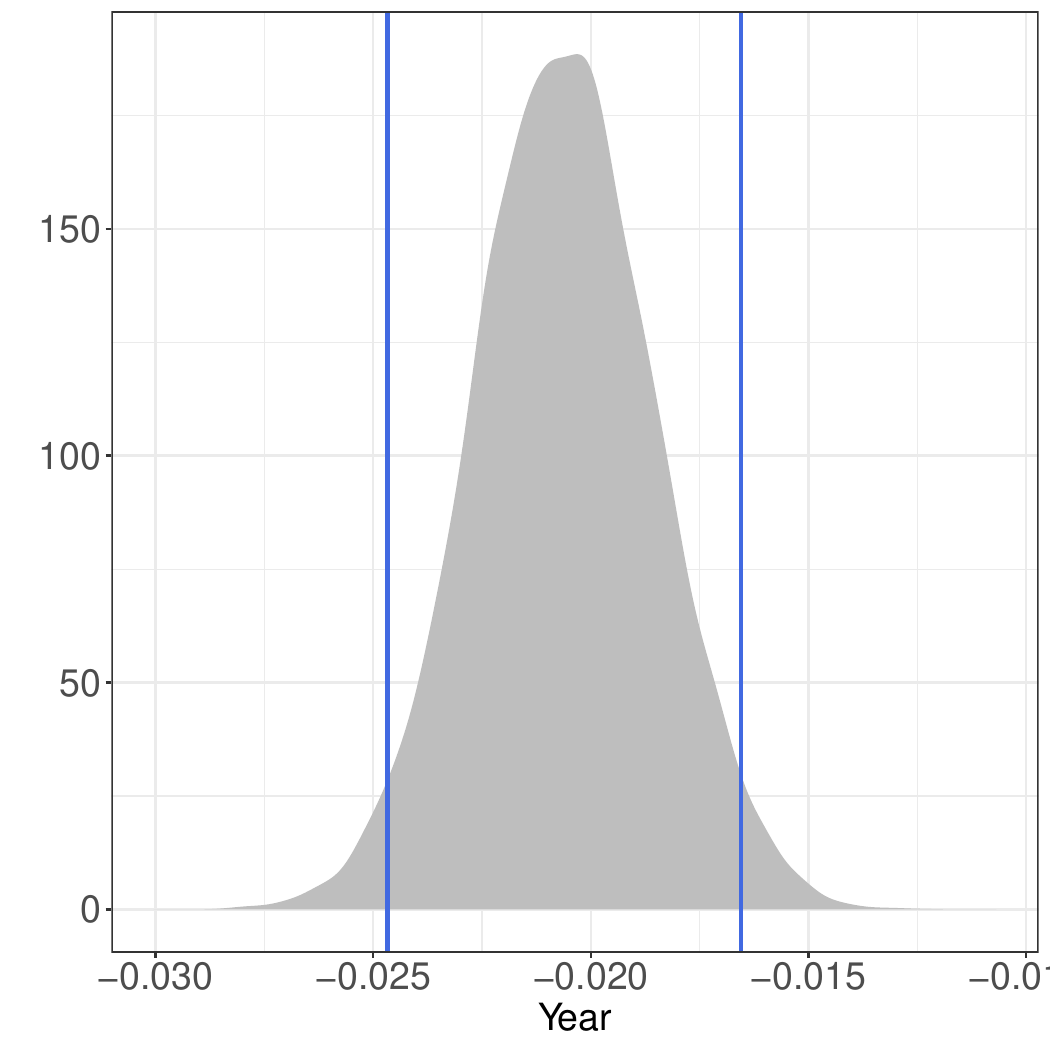}
\end{tabular}
\end{center}
\vspace{-0.5cm}
\caption{
Densities with credible intervals (blue lines). The  models are based on an  dataset updated to 2019, which uses GBD death estimates based on the GBD 2019 and  comprises 120 countries and 2,748 country-years from 1970-2019  \citep{collaborators2020global}, both sexes and model 2 using a Half-Cauchy prior
for the local scale of the errors.}
\label{fig:fig_both4}
\end{figure}

\clearpage

\begin{figure}[ht]

\begin{center}
\begin{tabular}{ccc}
\includegraphics[width=0.4\textwidth]{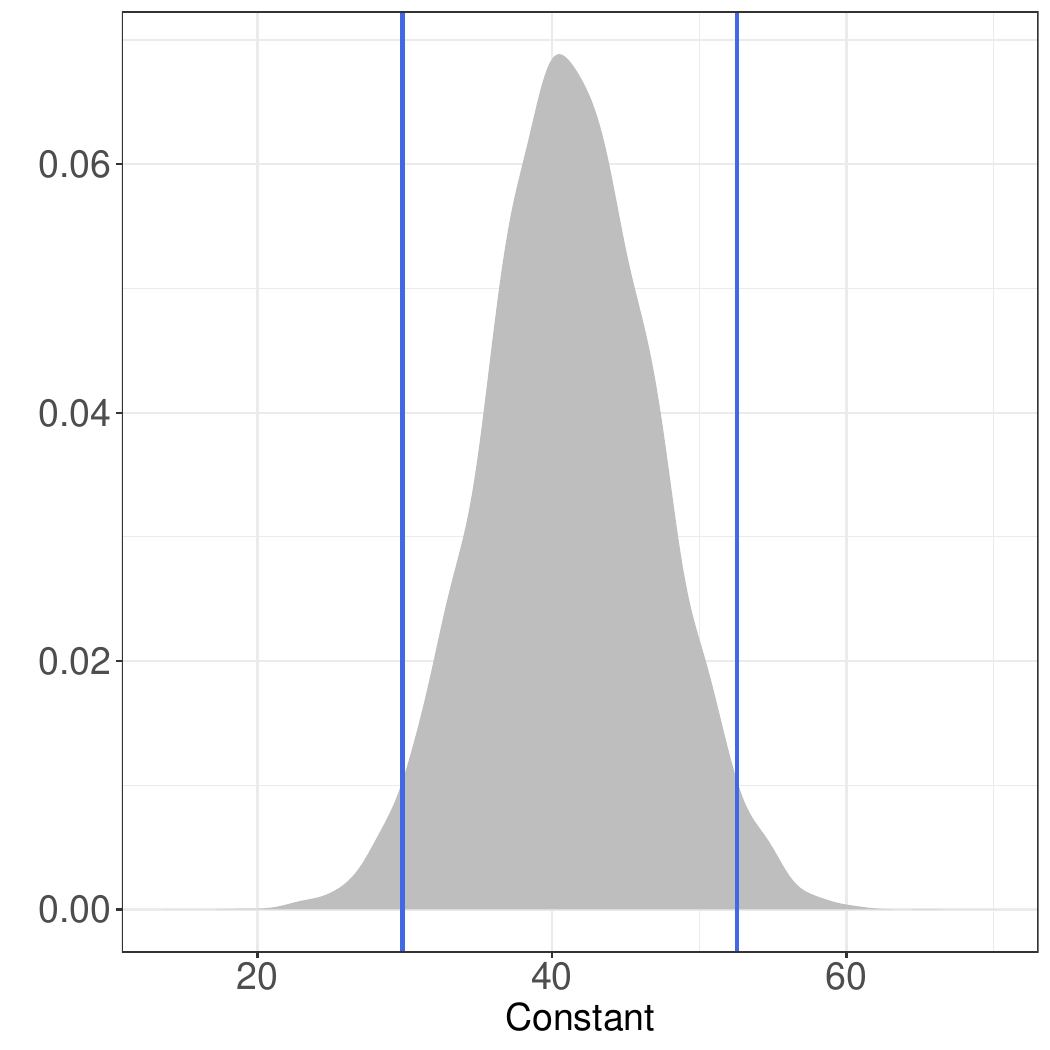} &
\includegraphics[width=0.4\textwidth]{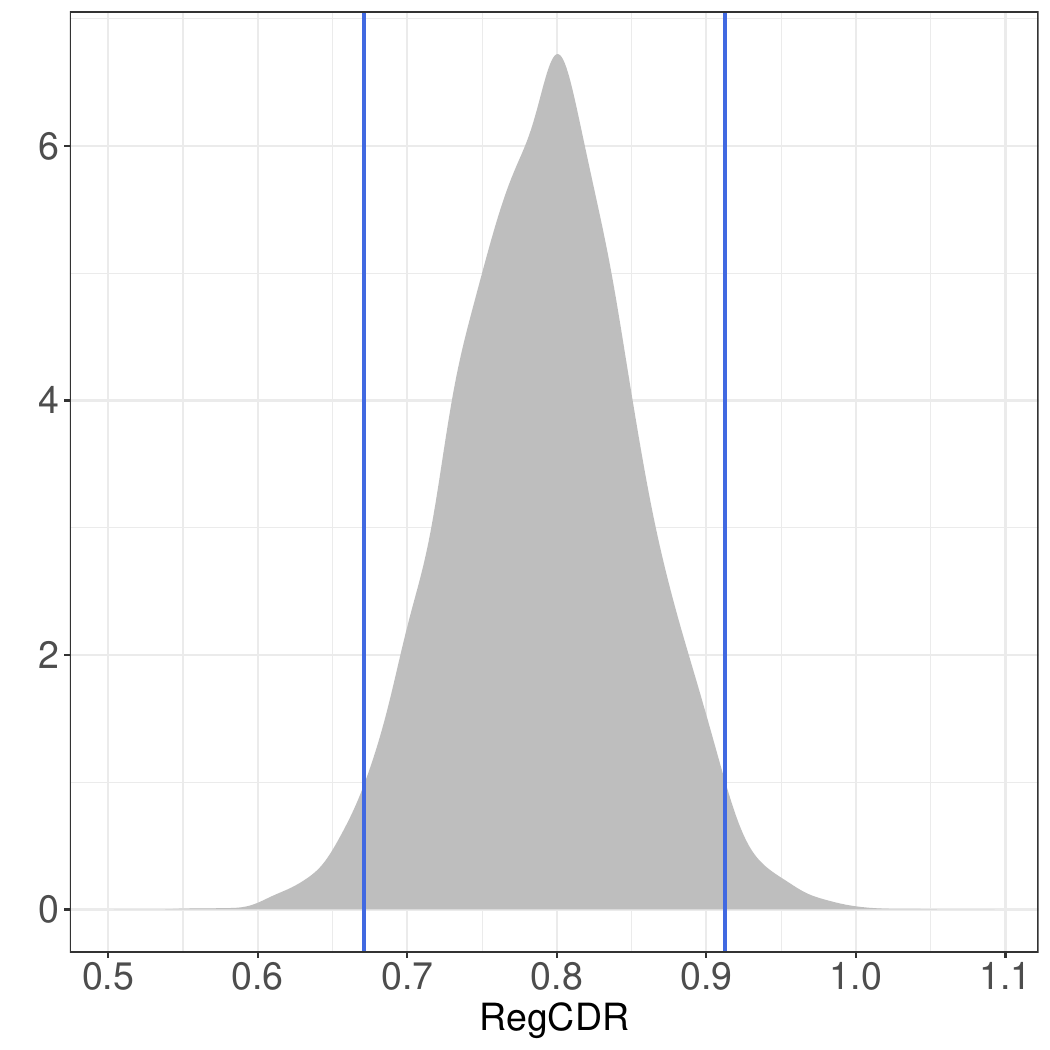}\\
\includegraphics[width=0.4\textwidth]{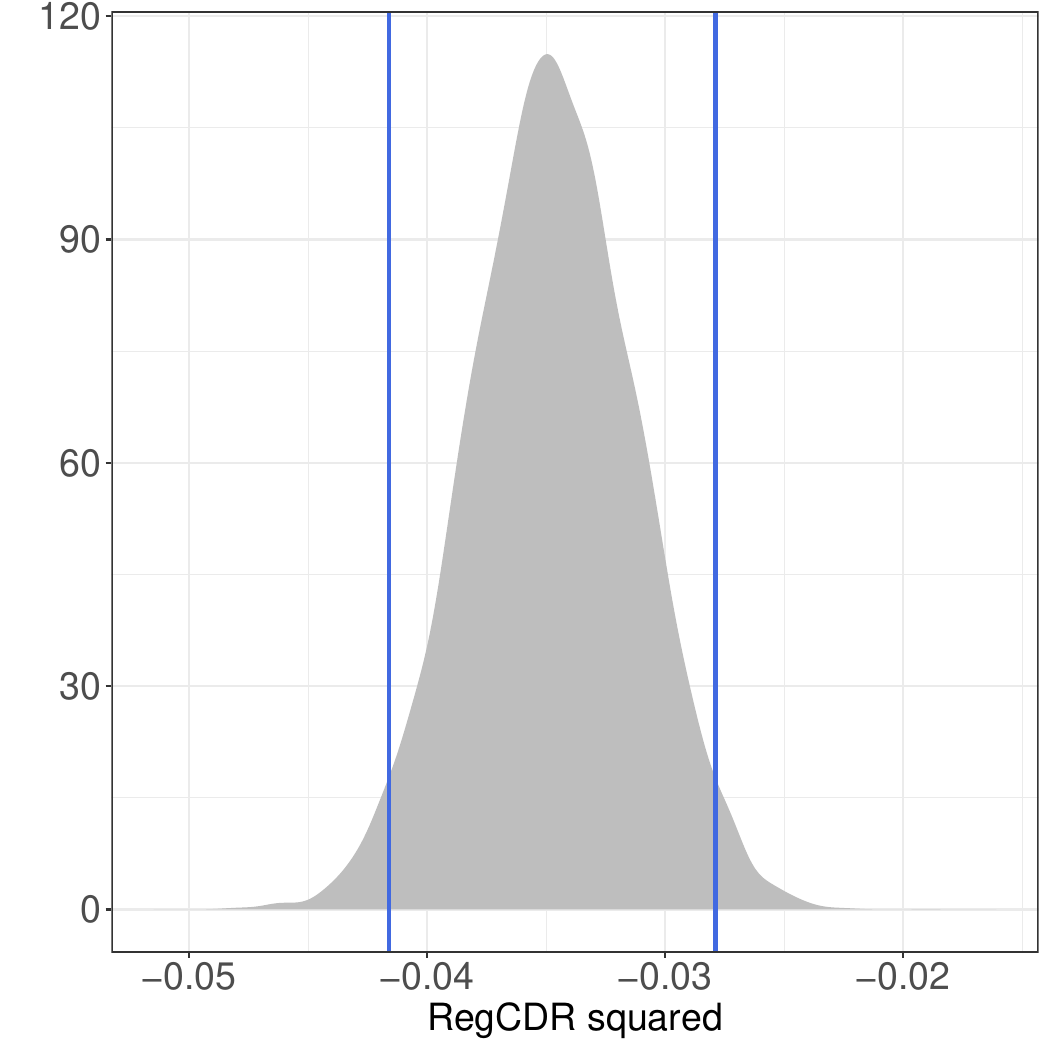} &
\includegraphics[width=0.4\textwidth]{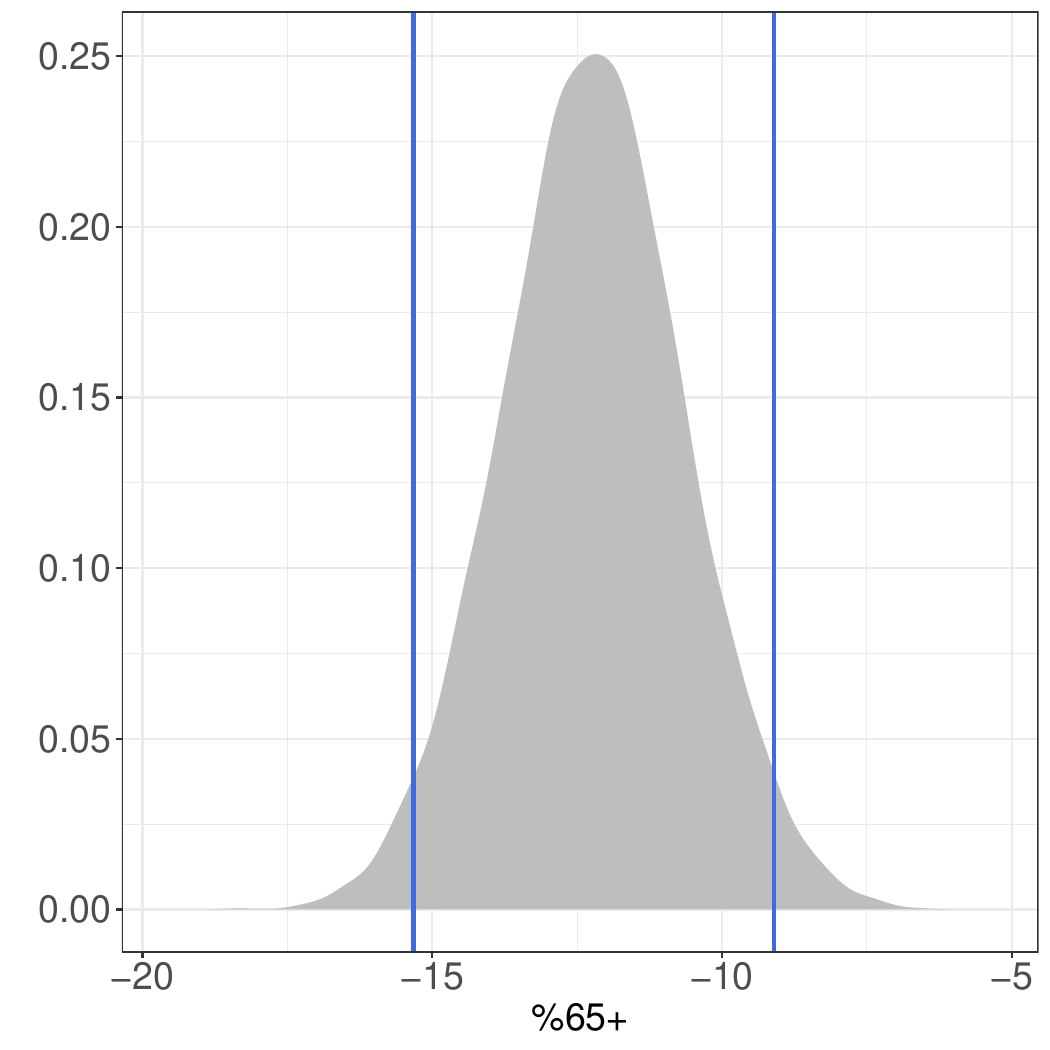}\\
\includegraphics[width=0.4\textwidth]{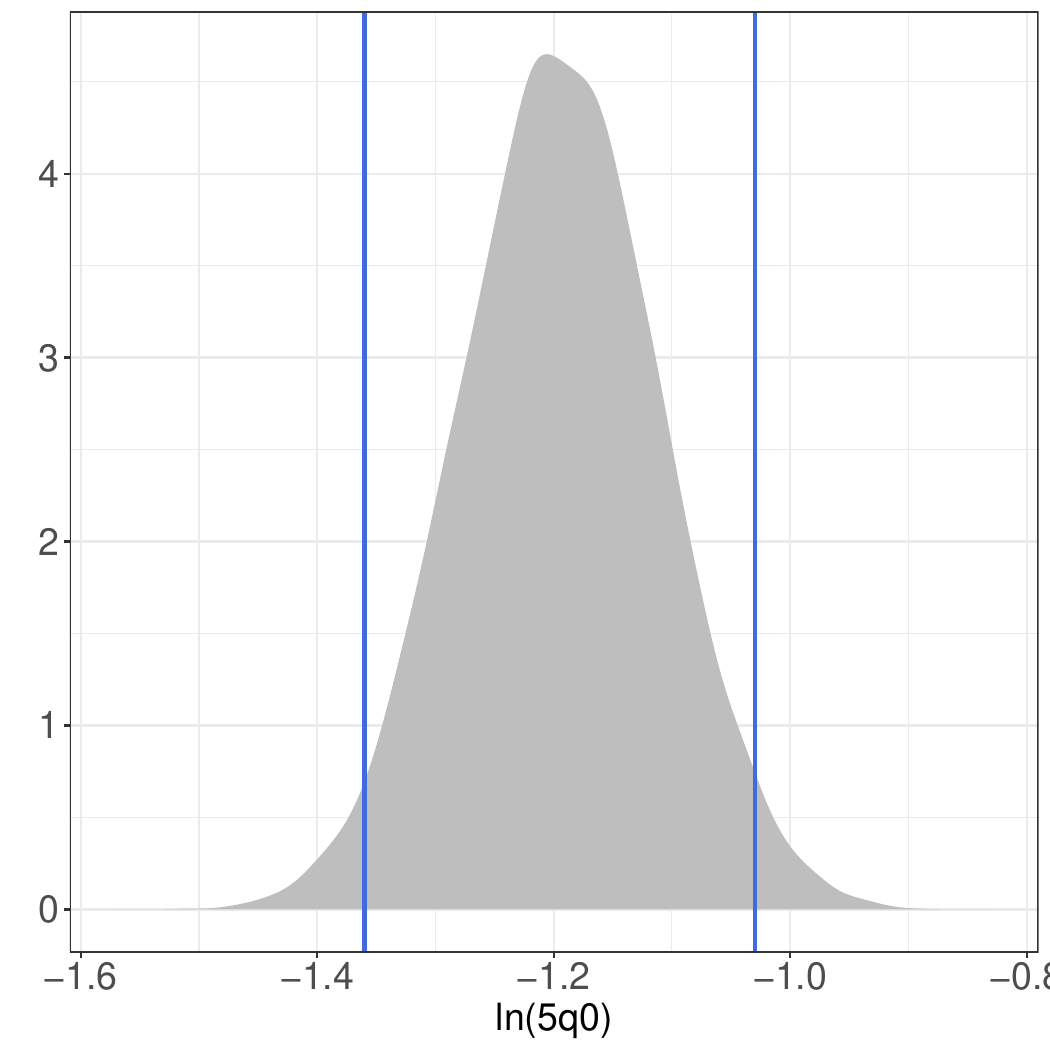} &
\includegraphics[width=0.4\textwidth]{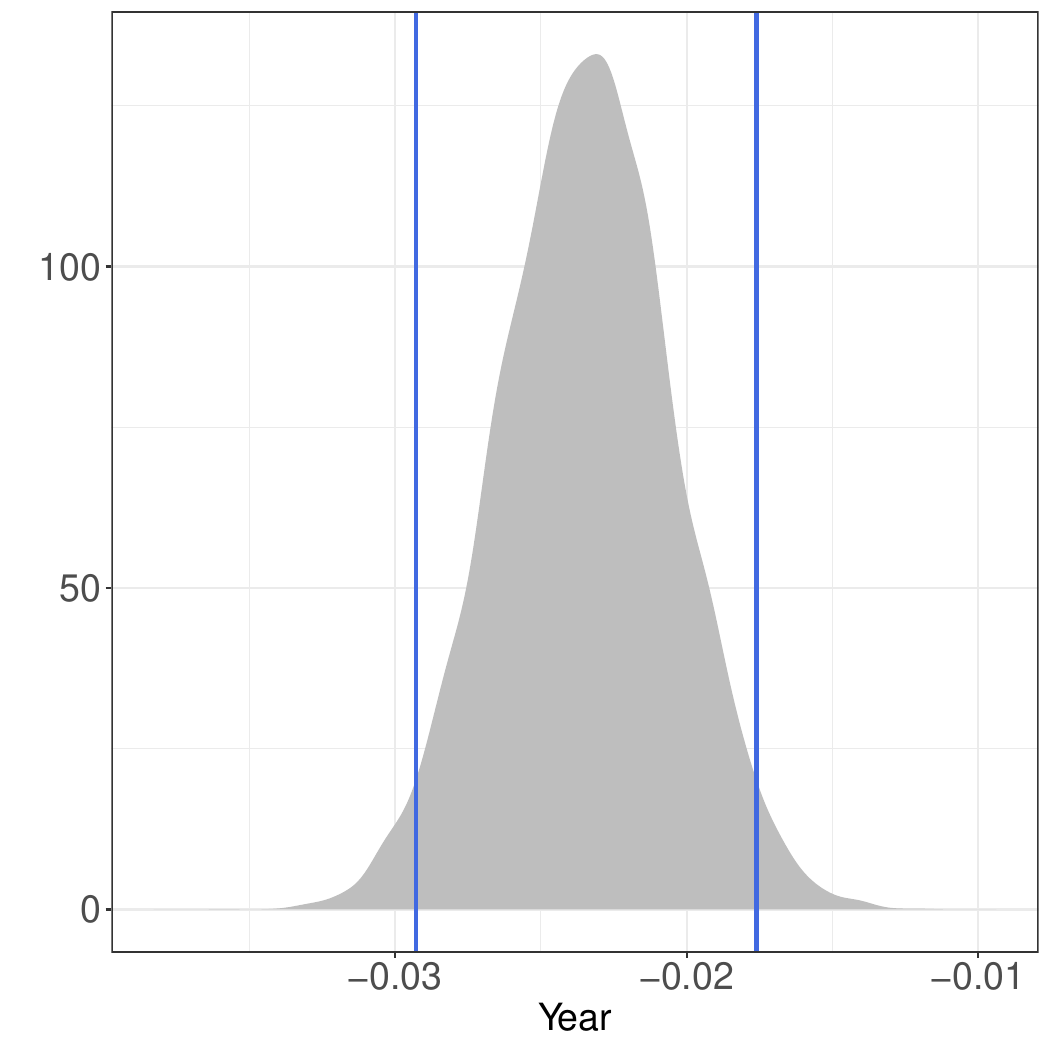}
\end{tabular}
\end{center}
\vspace{-0.5cm}
\caption{
Densities with credible intervals (blue lines).  The  models are based on an  dataset updated to 2019, which uses GBD death estimates based on the GBD 2019 and  comprises 120 countries and 2,748 country-years from 1970-2019  \citep{collaborators2020global}, females and model 1 using a common Gamma prior
for the scale of the errors.}
\label{fig:fig_both5}
\end{figure}

\begin{figure}[ht]
\begin{center}
\begin{tabular}{ccc}
\includegraphics[width=0.4\textwidth]{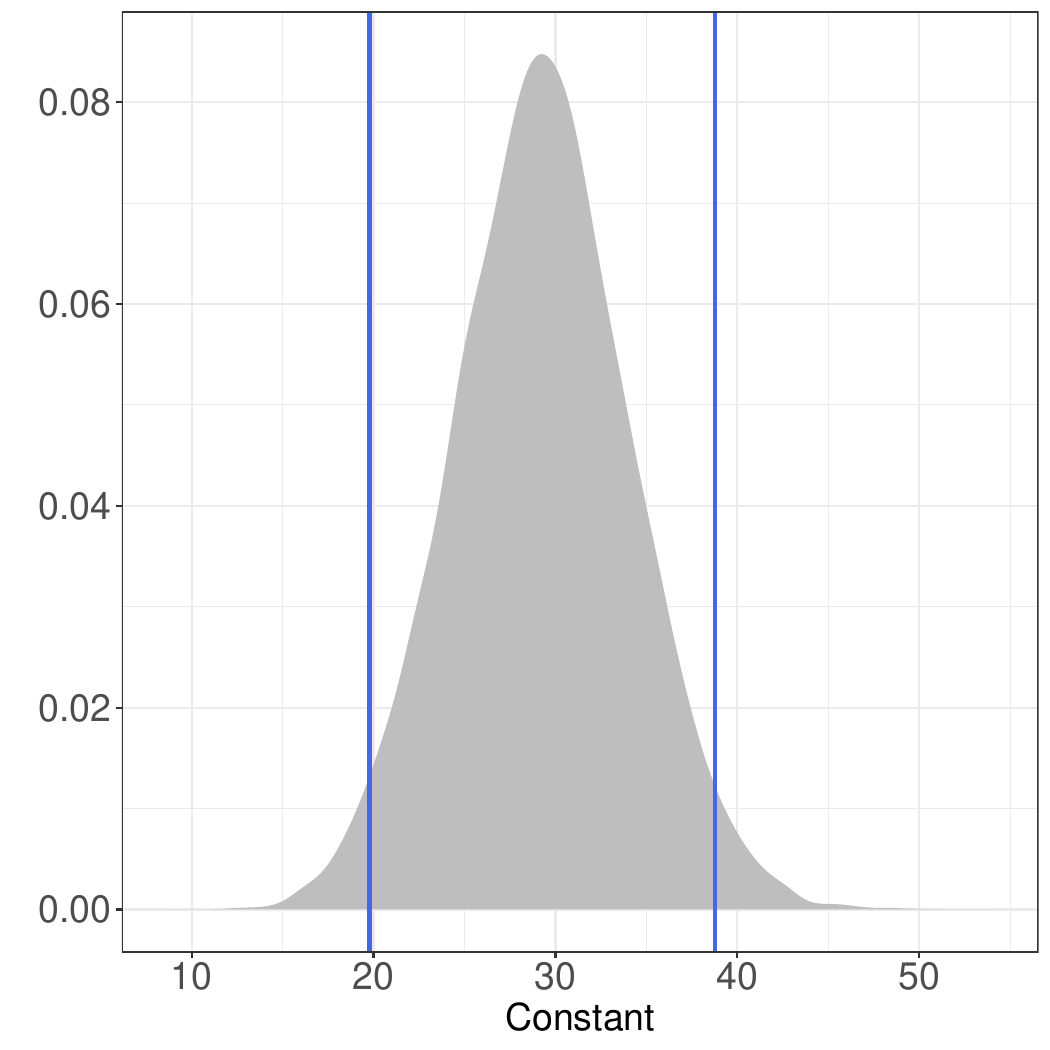} &
\includegraphics[width=0.4\textwidth]{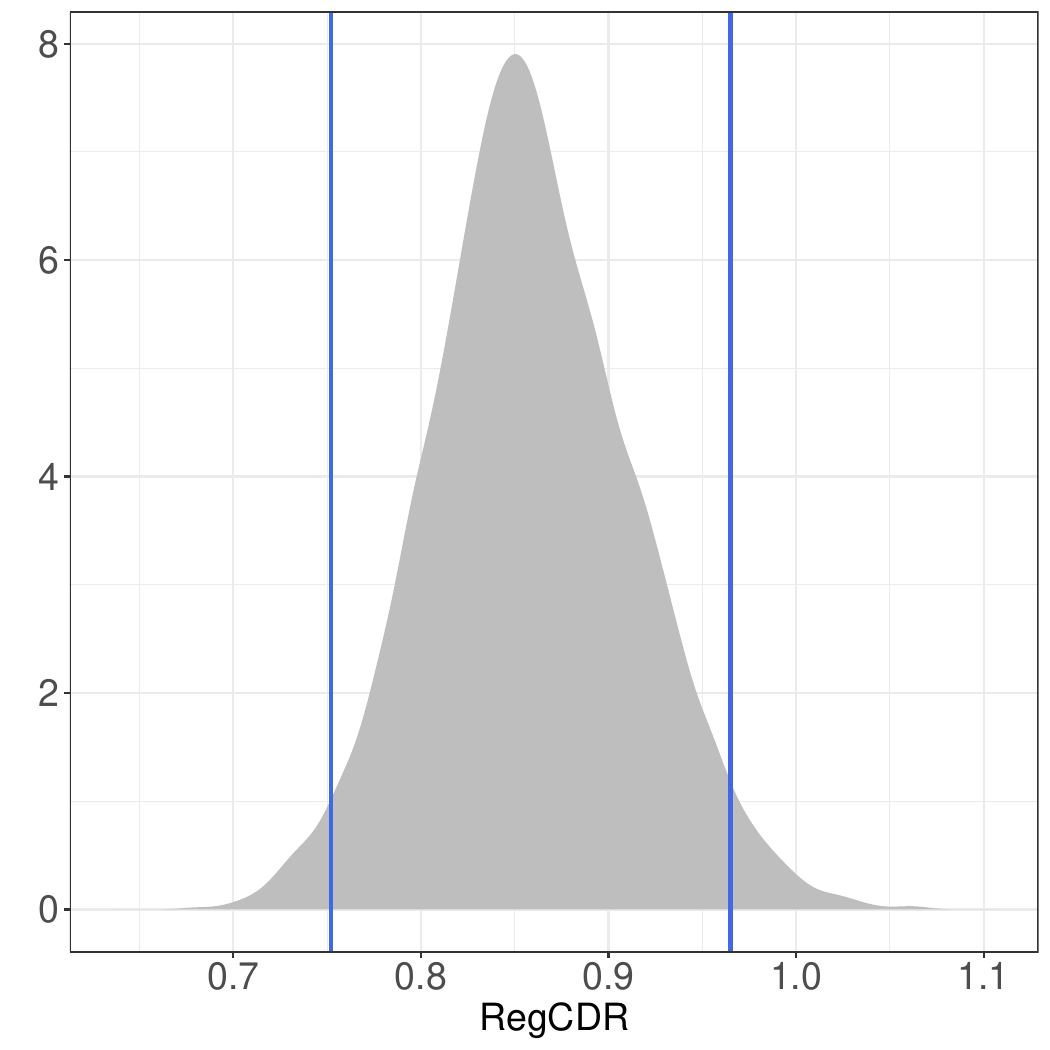}\\
\includegraphics[width=0.4\textwidth]{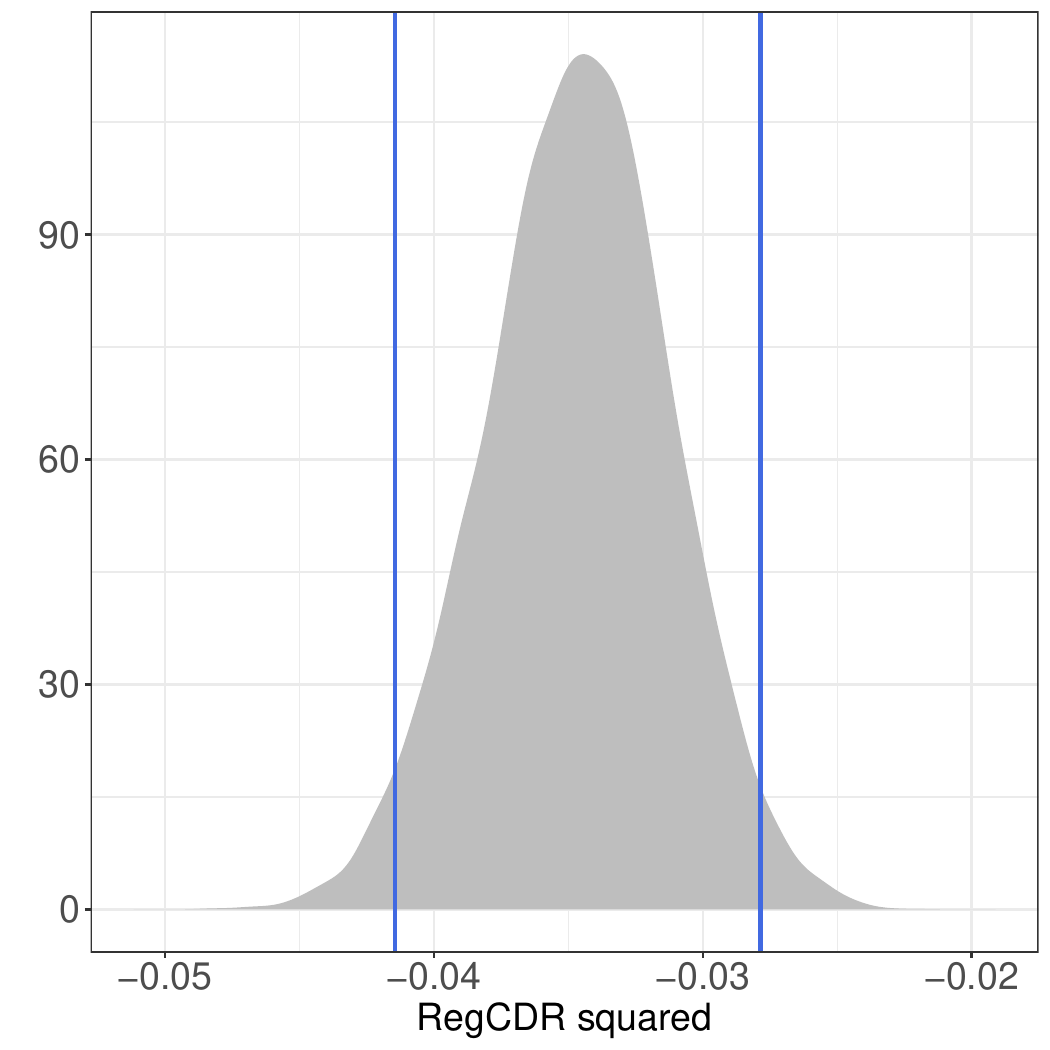} &
\includegraphics[width=0.4\textwidth]{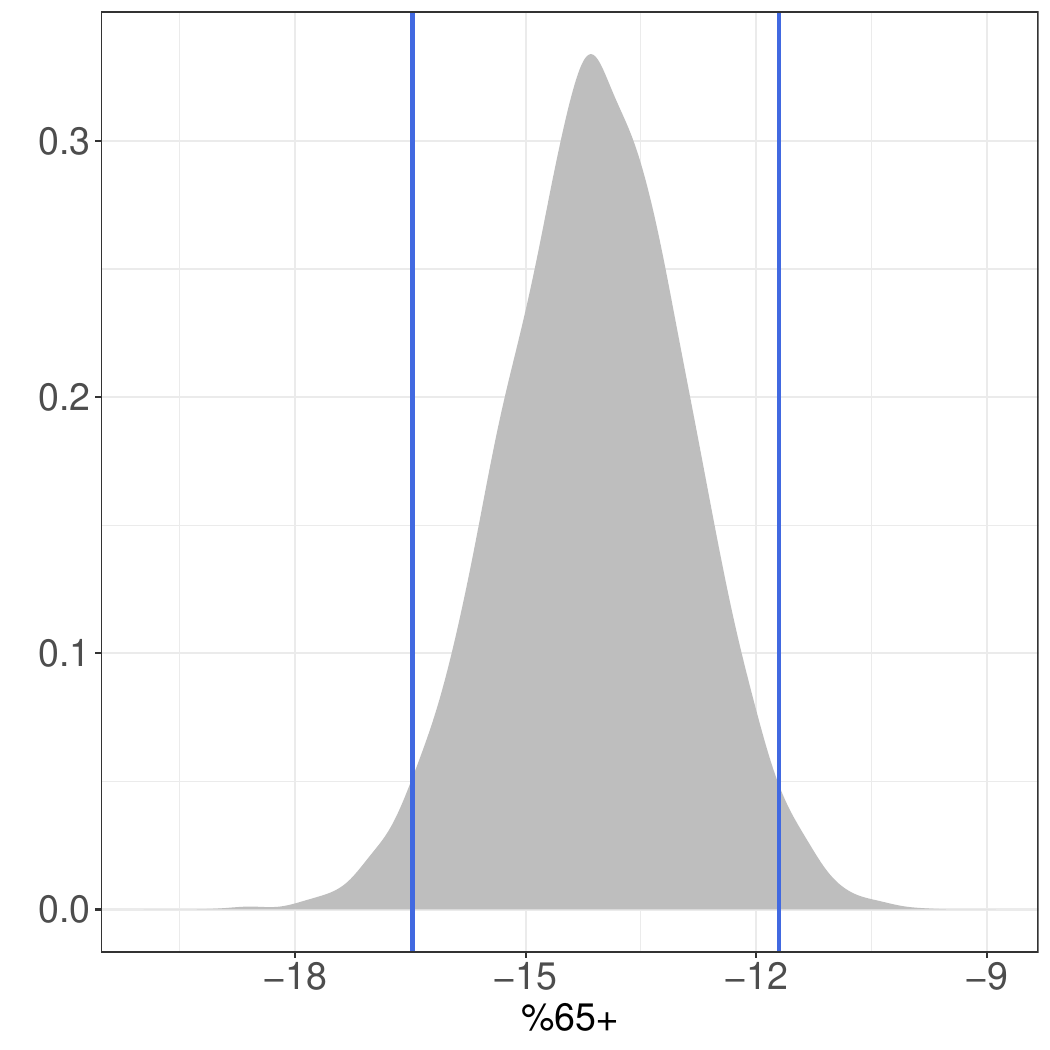}\\
\includegraphics[width=0.4\textwidth]{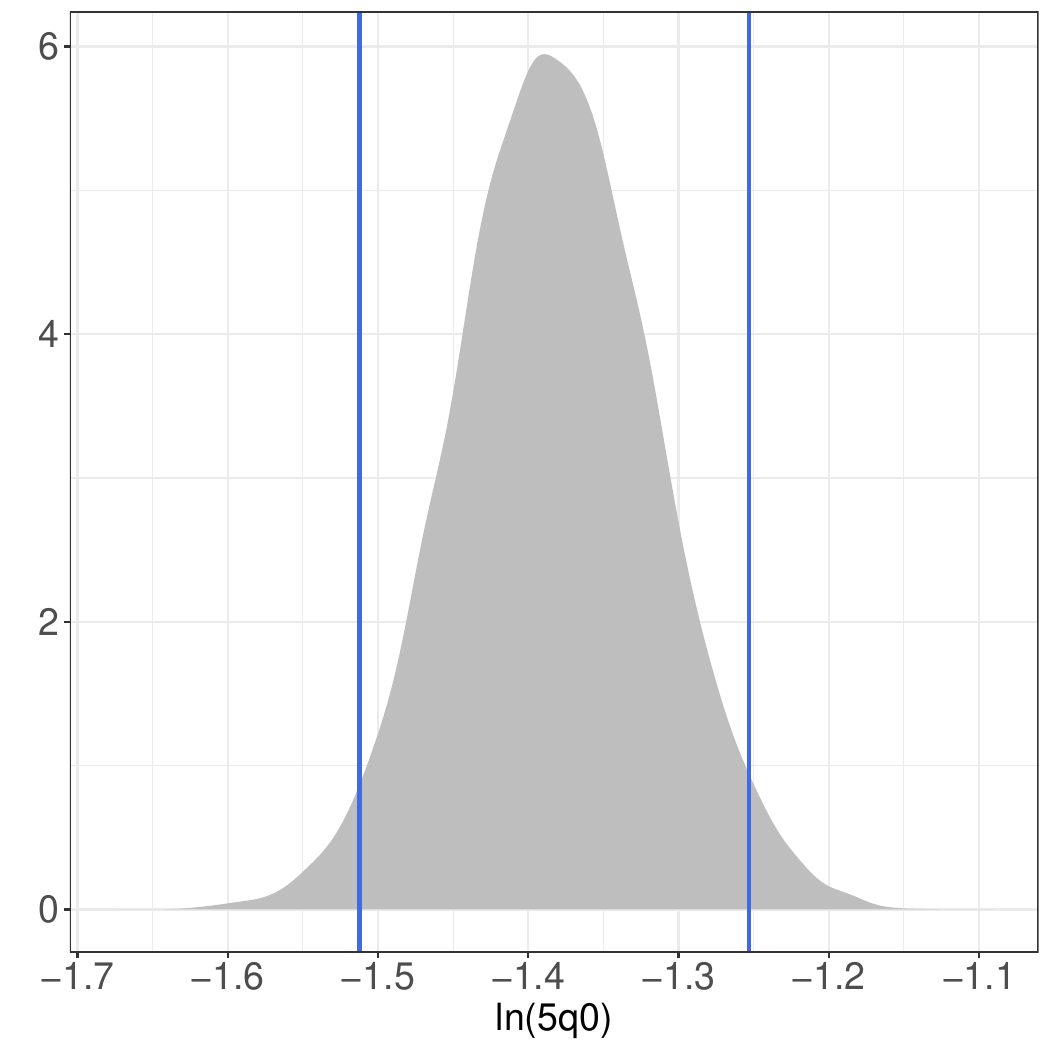} &
\includegraphics[width=0.4\textwidth]{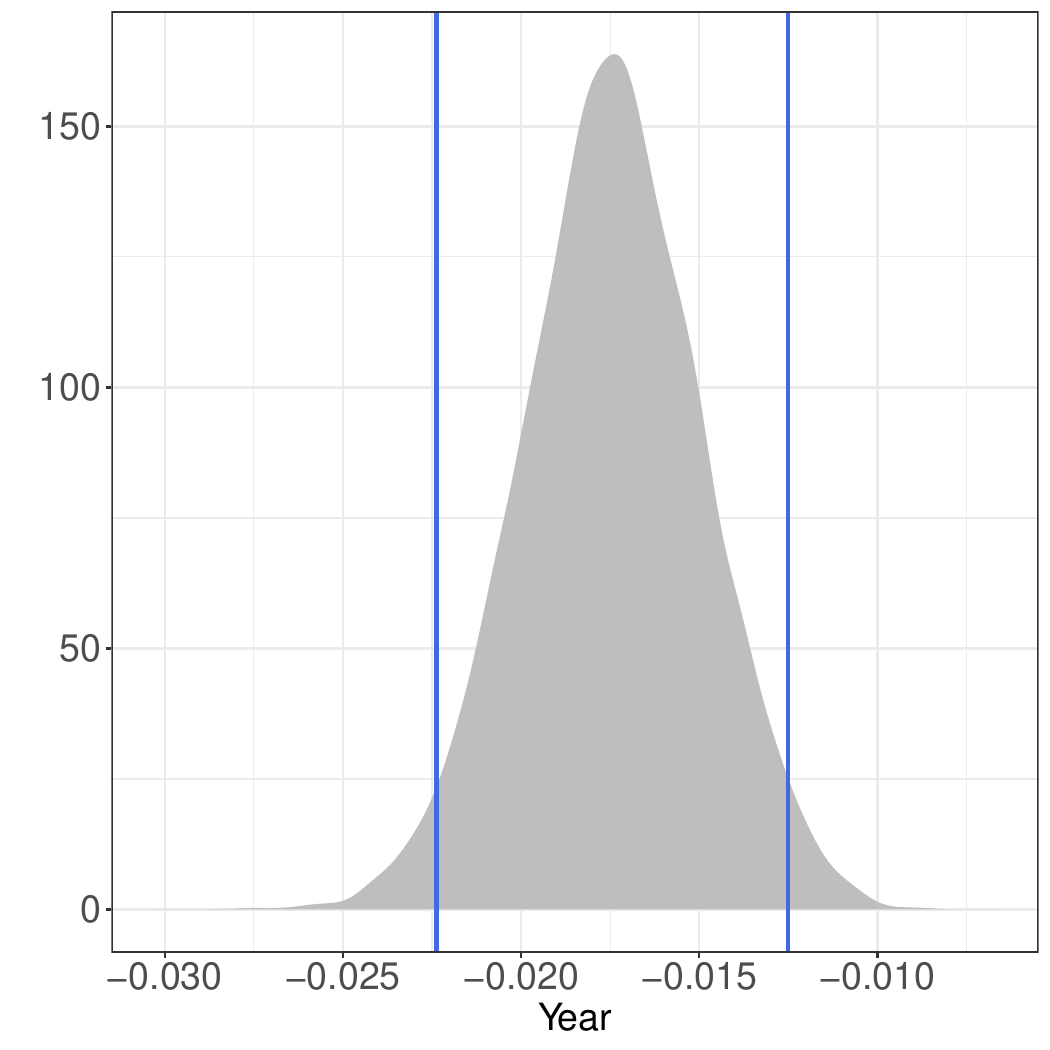}
\end{tabular}
\end{center}
\vspace{-0.5cm}
\caption{
Densities with credible intervals (blue lines). The  models are based on an  dataset updated to 2019, which uses GBD death estimates based on the GBD 2019 and  comprises 120 countries and 2,748 country-years from 1970-2019  \citep{collaborators2020global}, females and model 1 using a Half-Cauchy prior
for the local scale of the errors.}
\label{fig:fig_both6}
\end{figure}

\clearpage

\begin{figure}[ht]

\begin{center}
\begin{tabular}{ccc}
\includegraphics[width=0.4\textwidth]{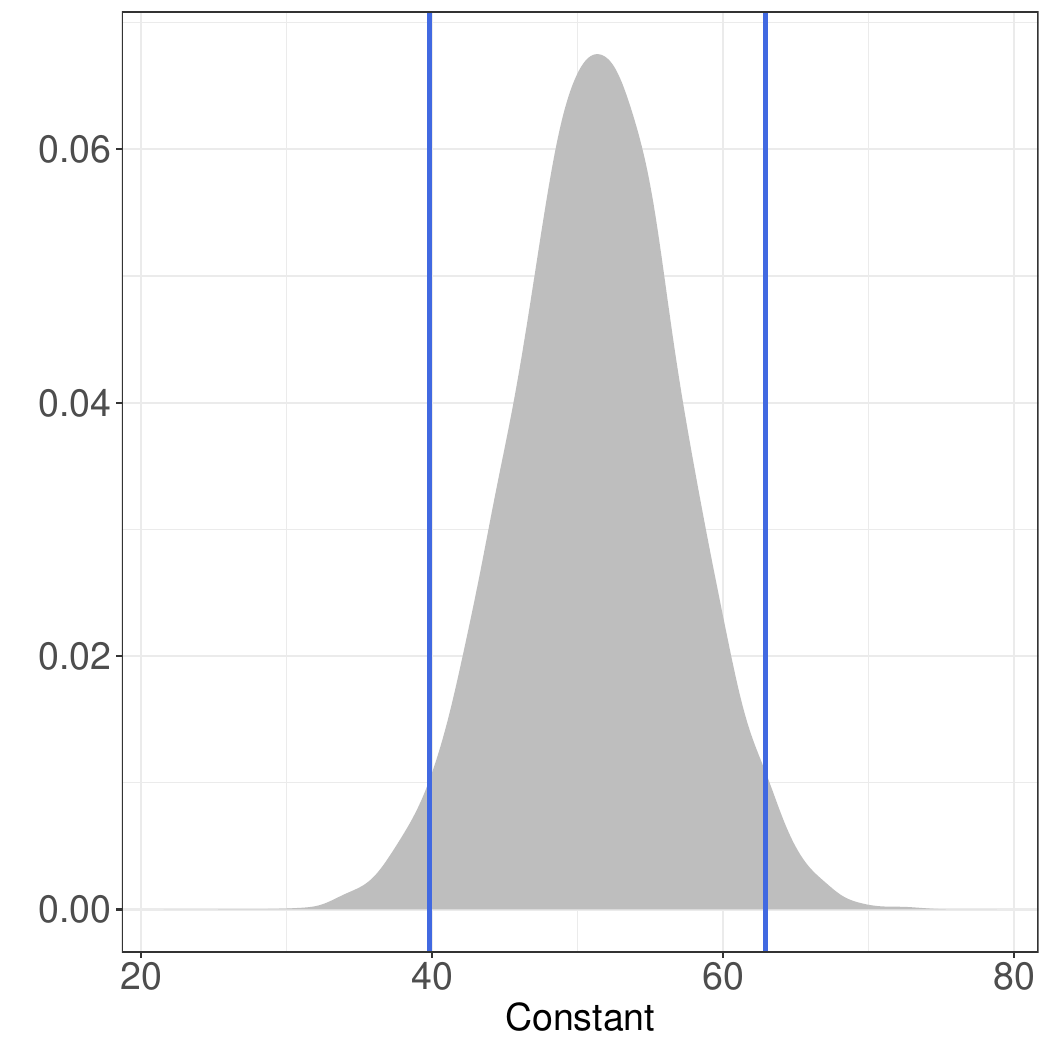} &
\includegraphics[width=0.4\textwidth]{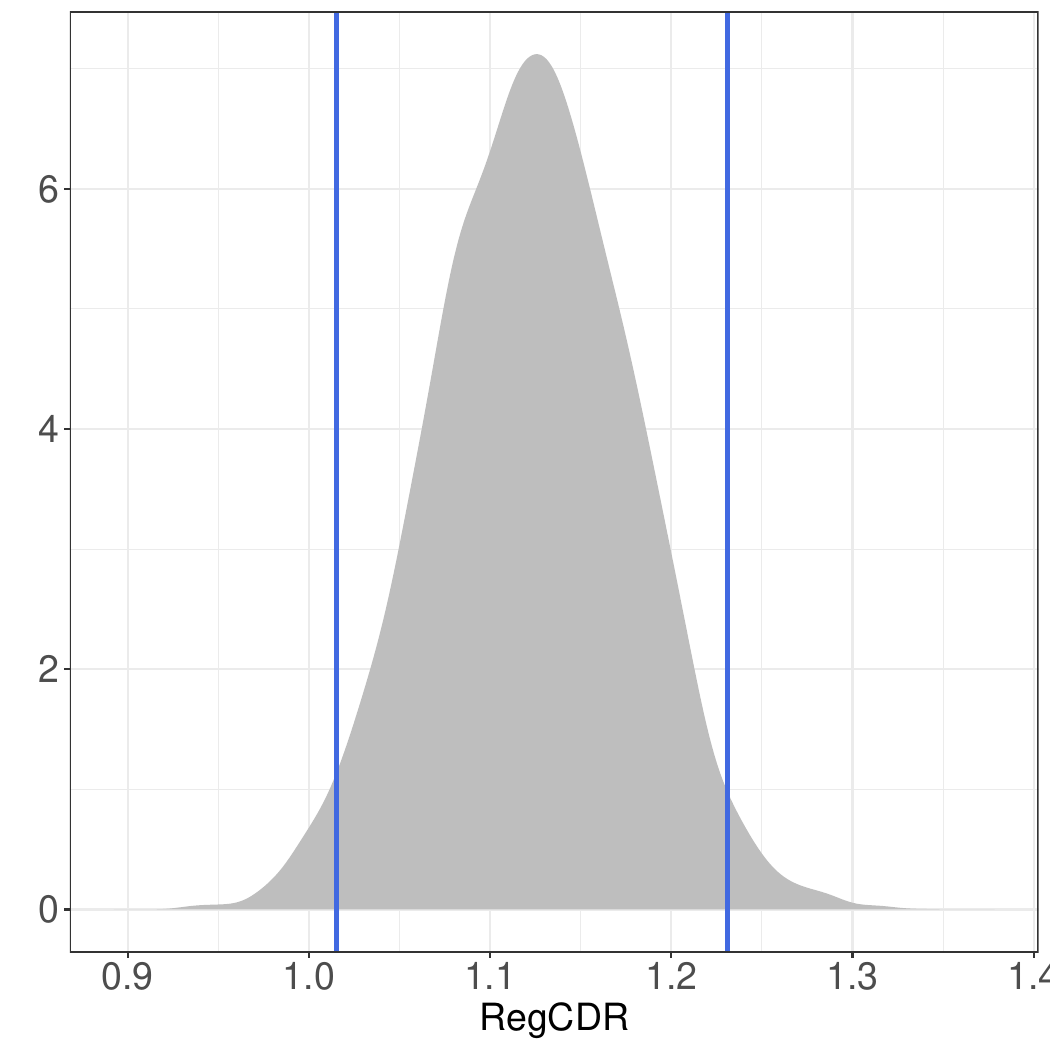}\\
\includegraphics[width=0.4\textwidth]{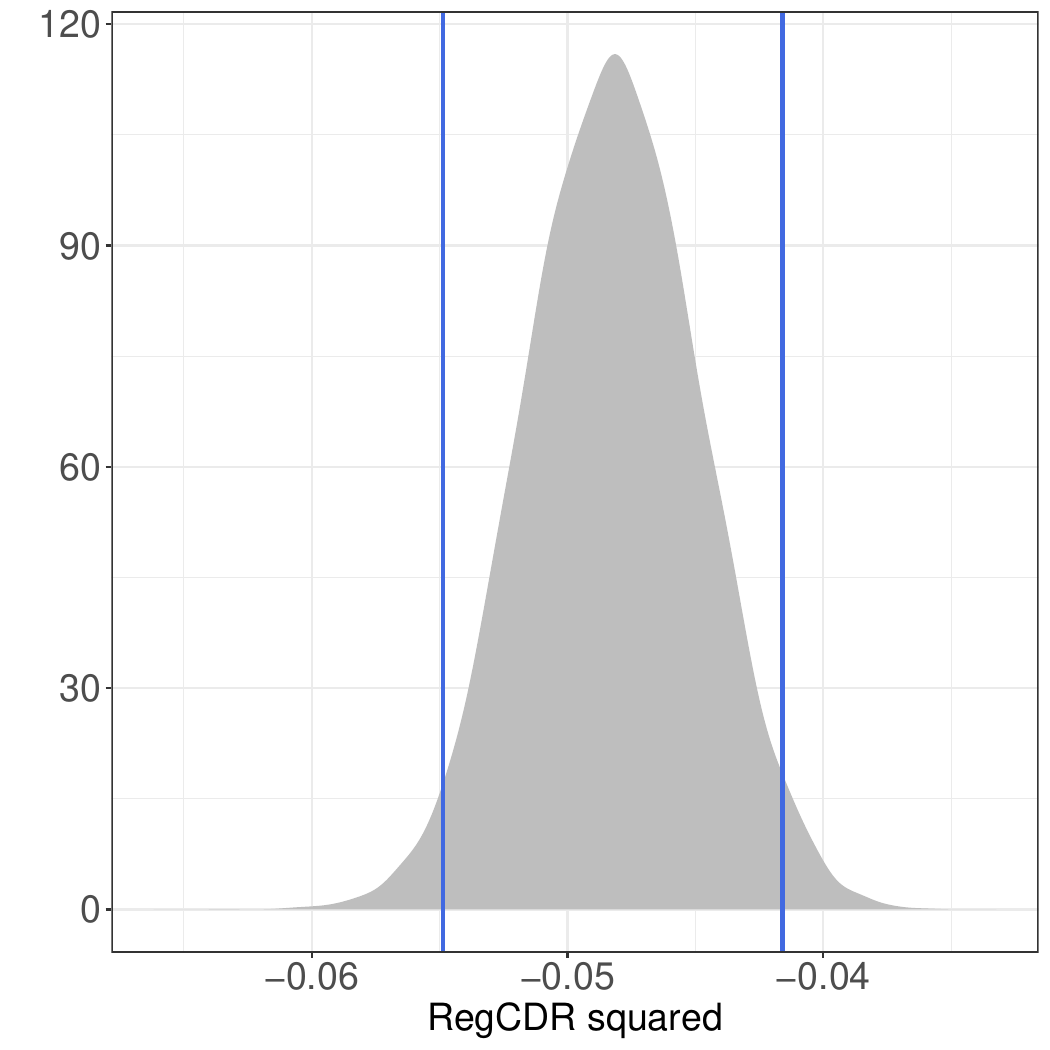} &
\includegraphics[width=0.4\textwidth]{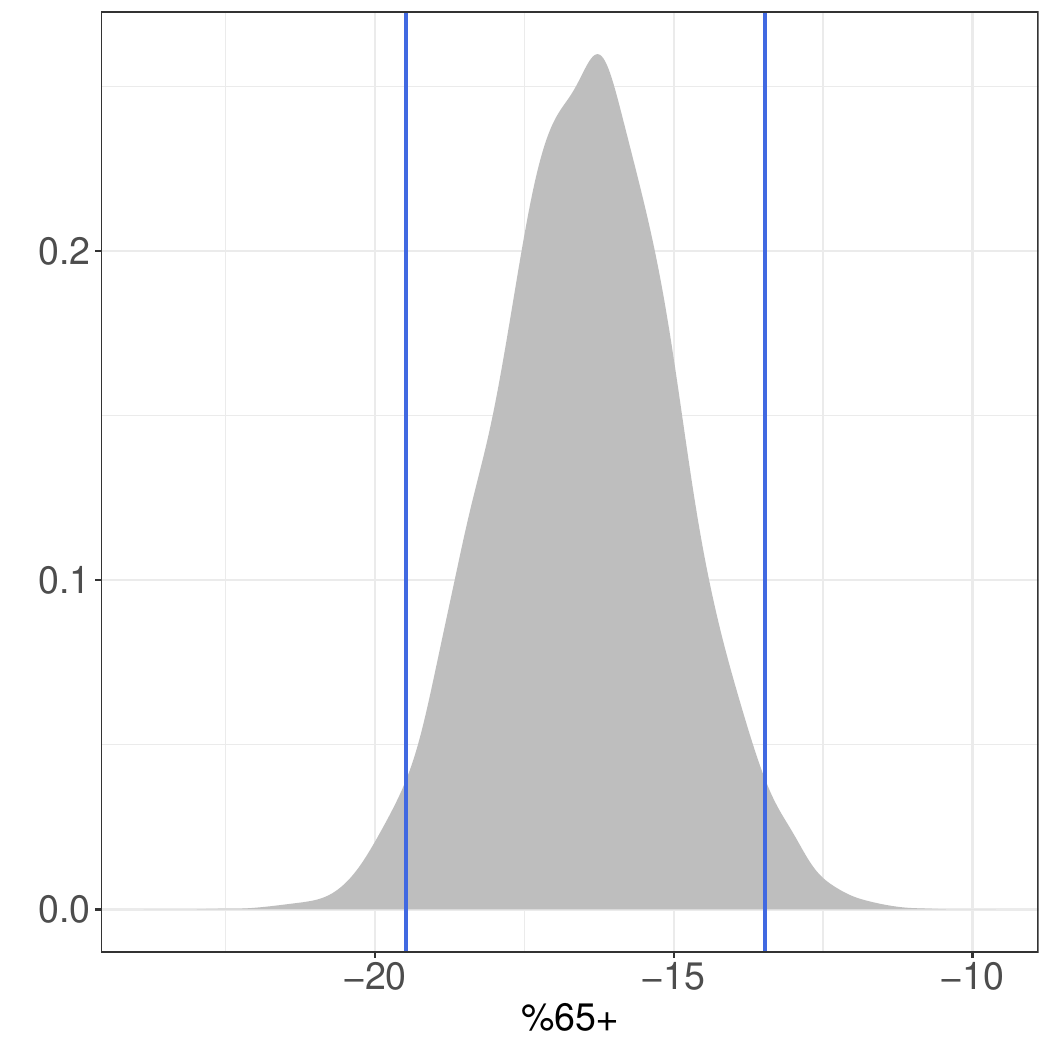}\\
\includegraphics[width=0.4\textwidth]{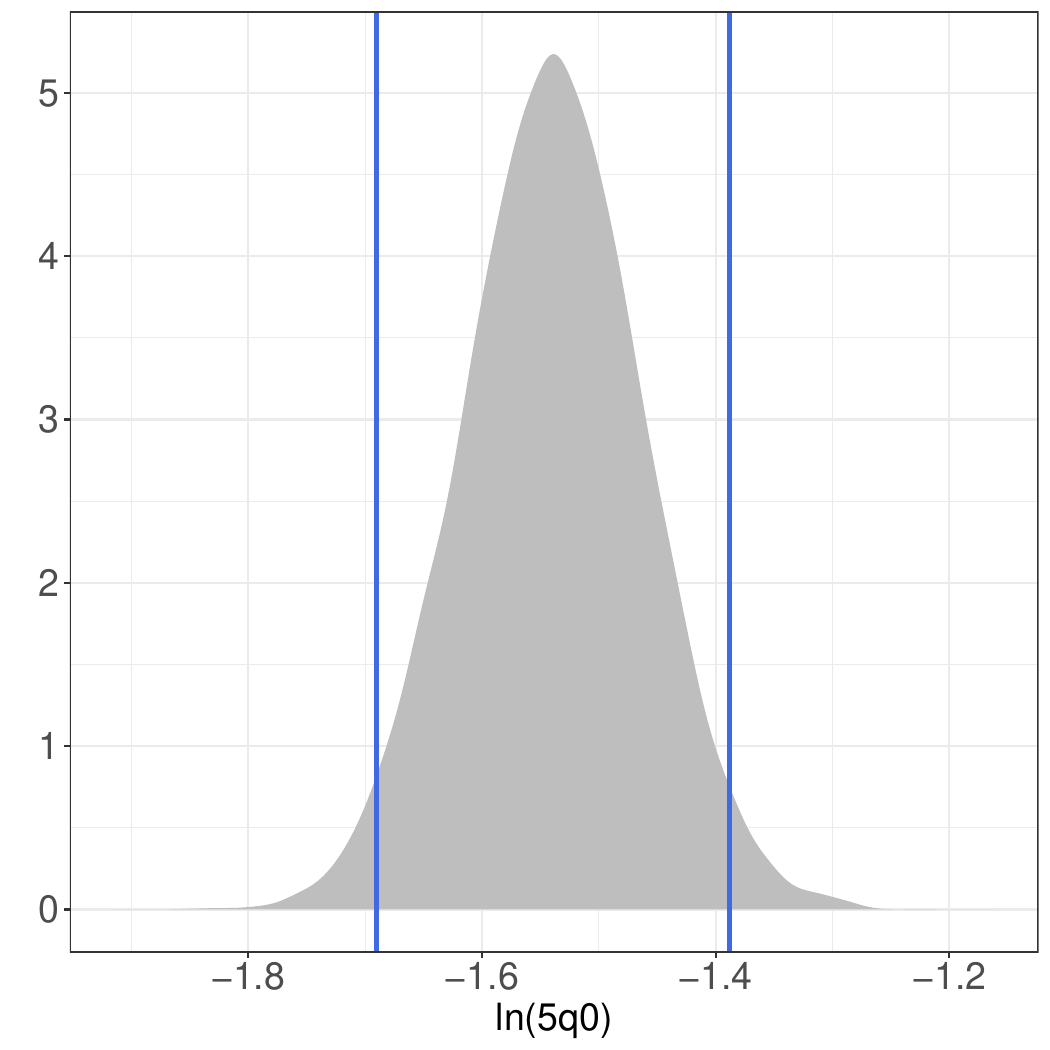} &
\includegraphics[width=0.4\textwidth]{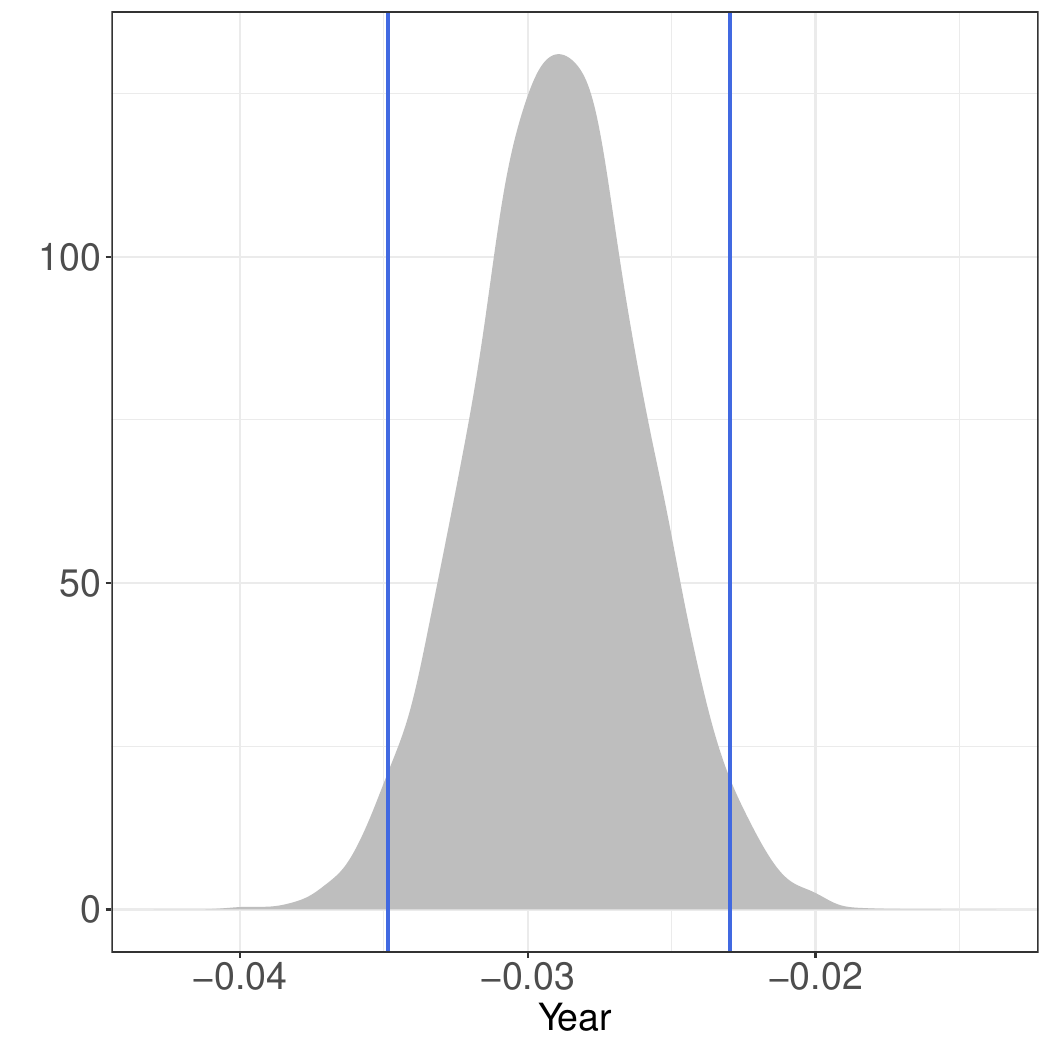}
\end{tabular}
\end{center}
\vspace{-0.5cm}
\caption{
Densities with credible intervals (blue lines).  The  models are based on an  dataset updated to 2019, which uses GBD death estimates based on the GBD 2019 and  comprises 120 countries and 2,748 country-years from 1970-2019  \citep{collaborators2020global}, females and model 2 using a common Gamma prior
for the scale of the errors.}
\label{fig:fig_both7}
\end{figure}

\begin{figure}[ht]
\begin{center}
\begin{tabular}{ccc}
\includegraphics[width=0.4\textwidth]{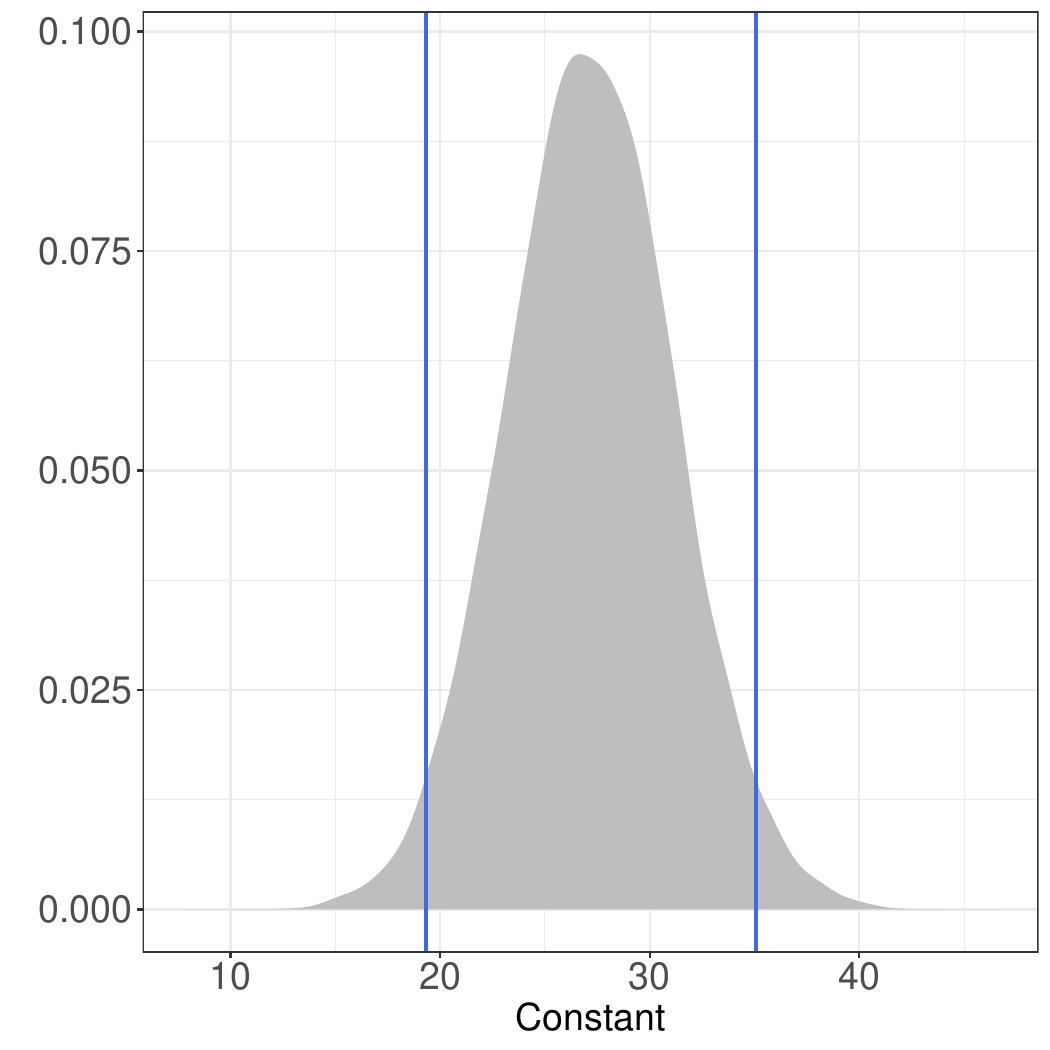} &
\includegraphics[width=0.4\textwidth]{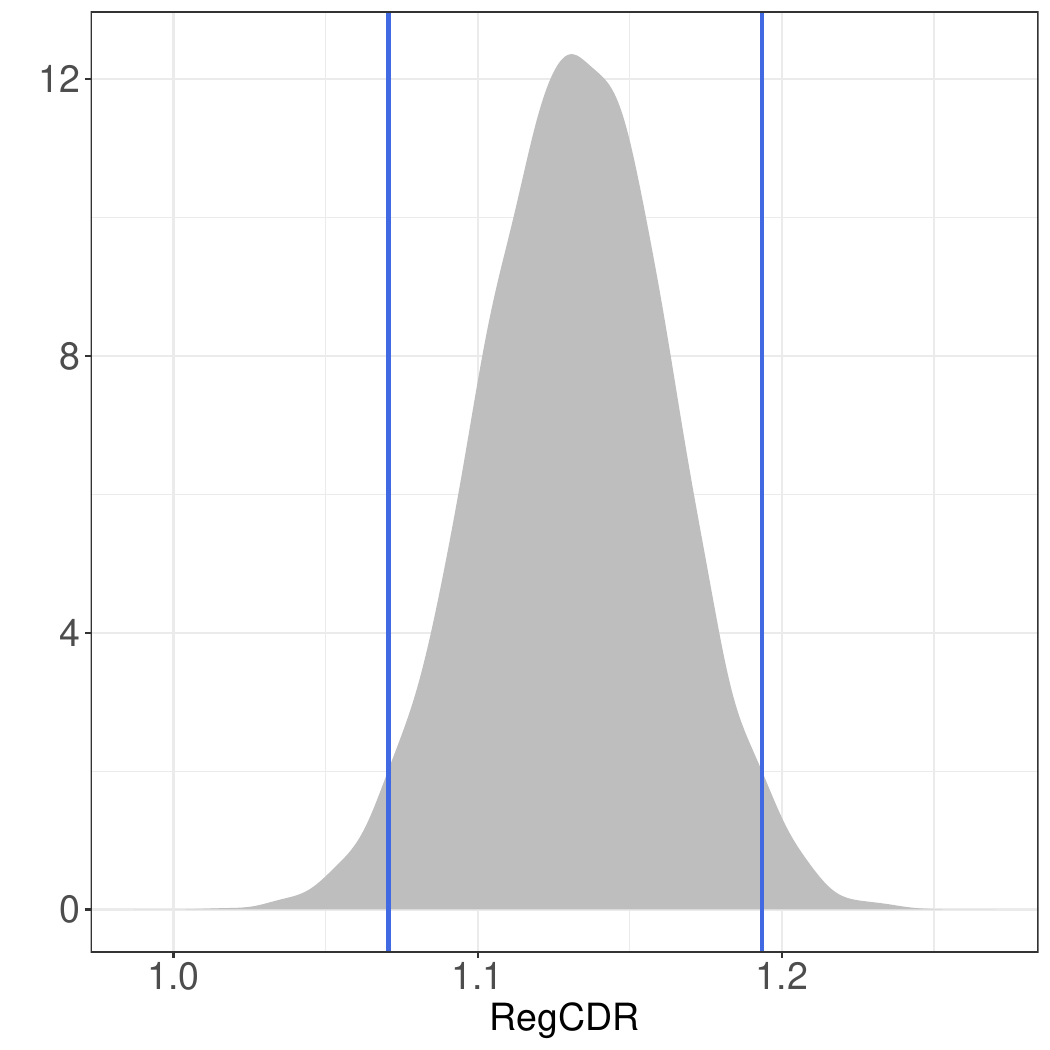}\\
\includegraphics[width=0.4\textwidth]{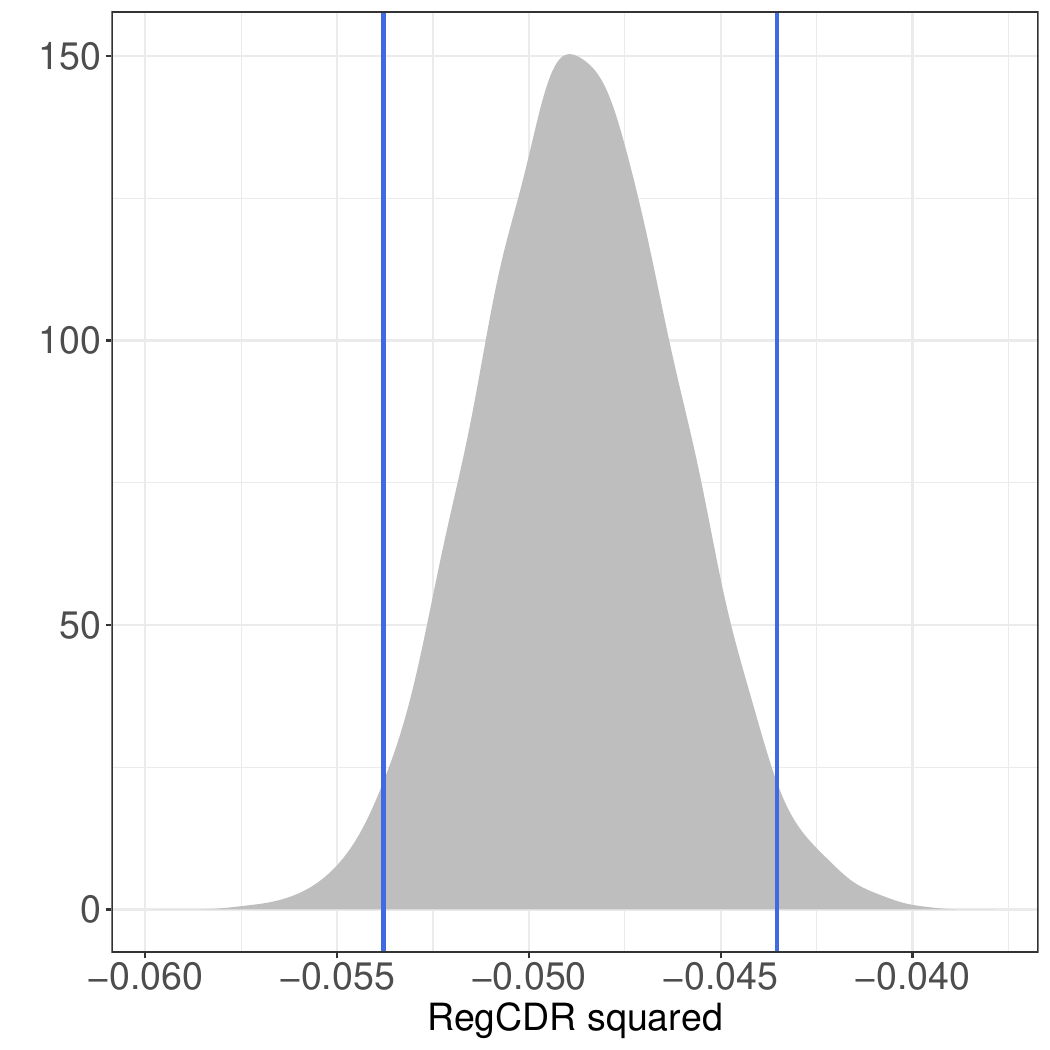} &
\includegraphics[width=0.4\textwidth]{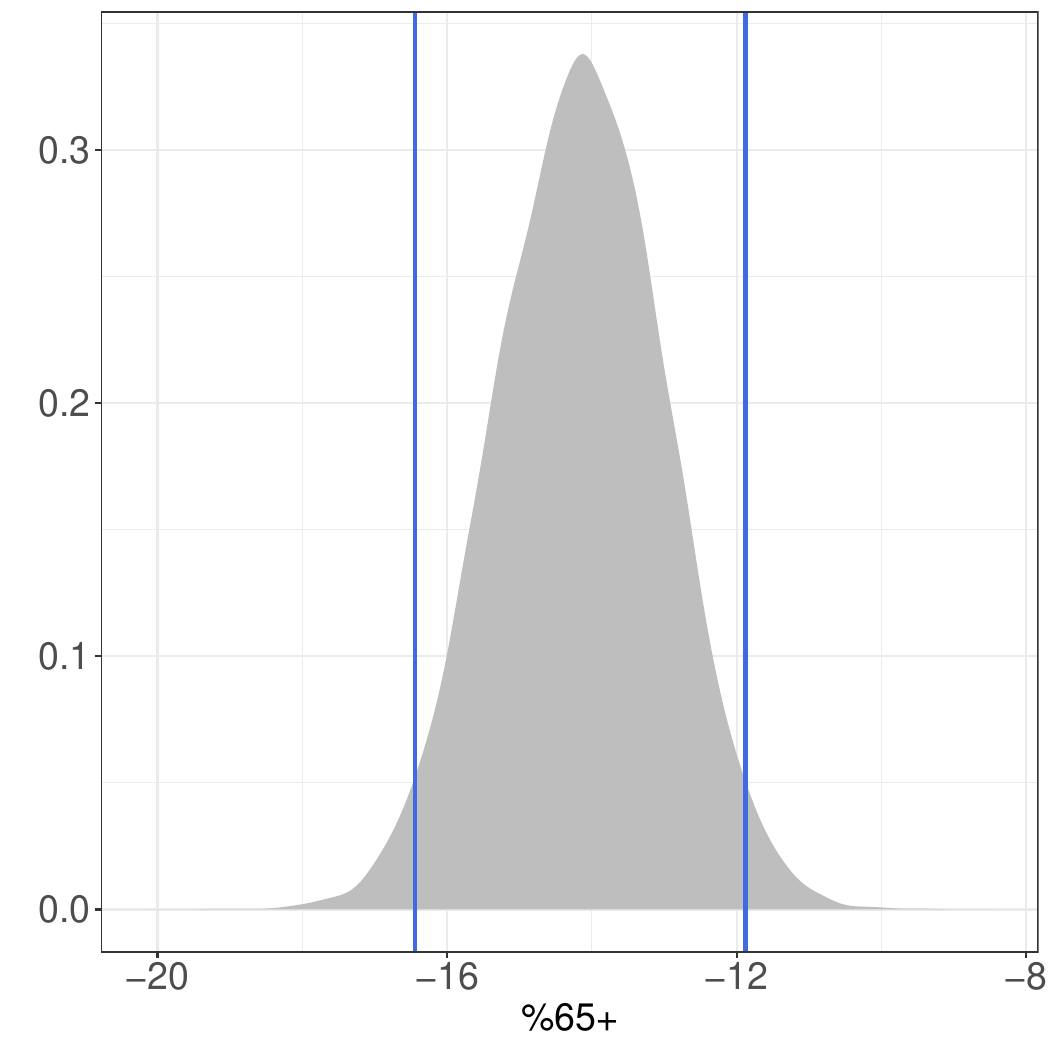}\\
\includegraphics[width=0.4\textwidth]{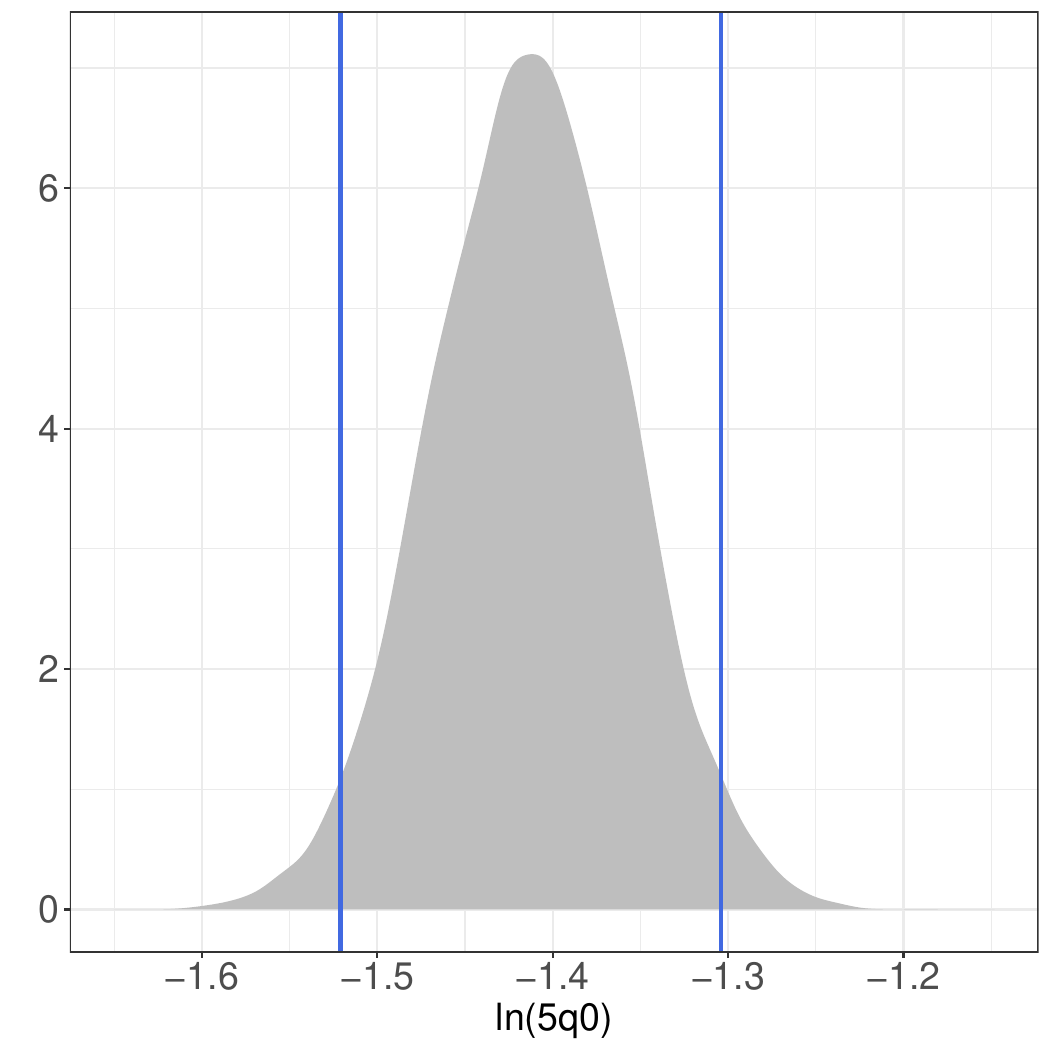} &
\includegraphics[width=0.4\textwidth]{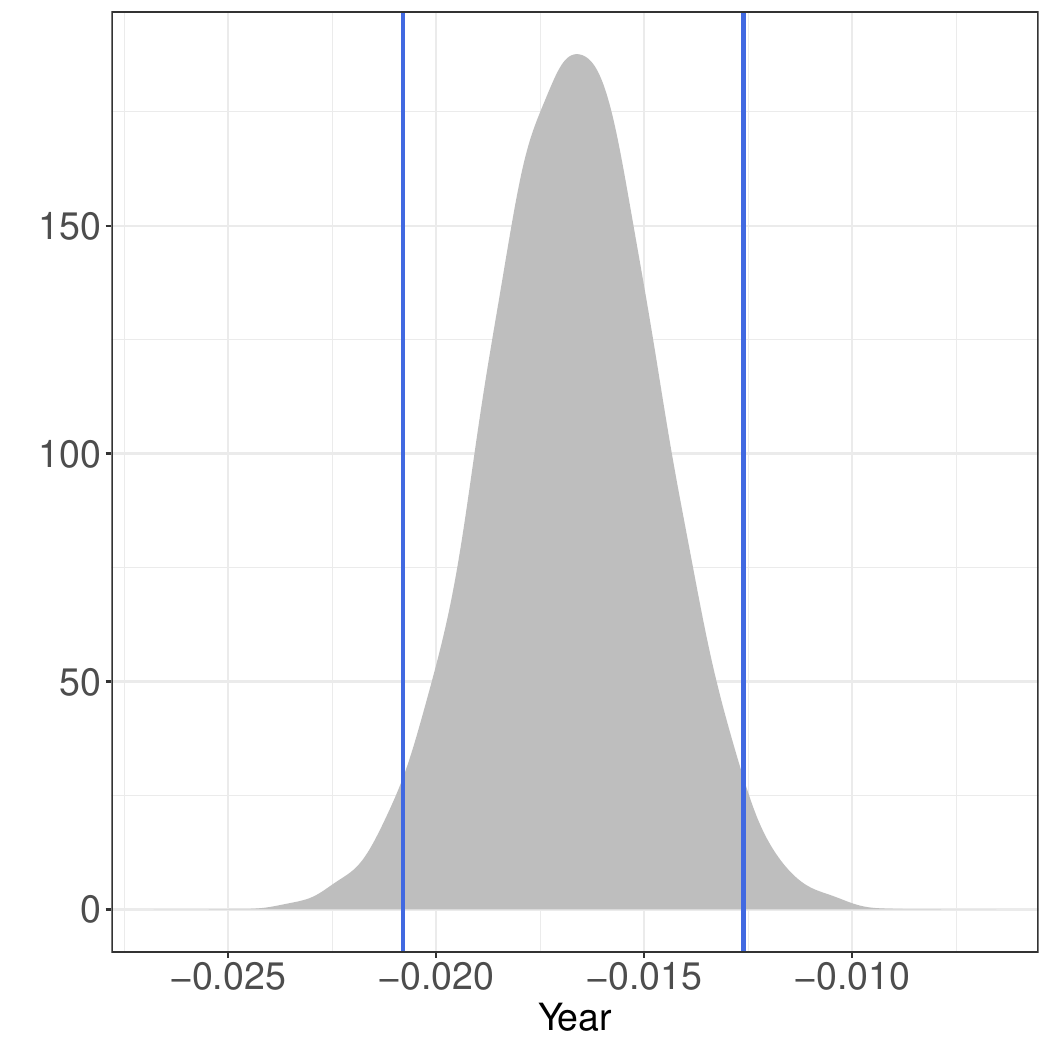}
\end{tabular}
\end{center}
\vspace{-0.5cm}
\caption{
Densities with credible intervals (blue lines). The  models are based on an  dataset updated to 2019, which uses GBD death estimates based on the GBD 2019 and  comprises 120 countries and 2,748 country-years from 1970-2019  \citep{collaborators2020global}, females and model 2 using a Half-Cauchy prior
for the local scale of the errors}
\label{fig:fig_both8}
\end{figure}

\clearpage

\begin{figure}[ht]

\begin{center}
\begin{tabular}{ccc}
\includegraphics[width=0.4\textwidth]{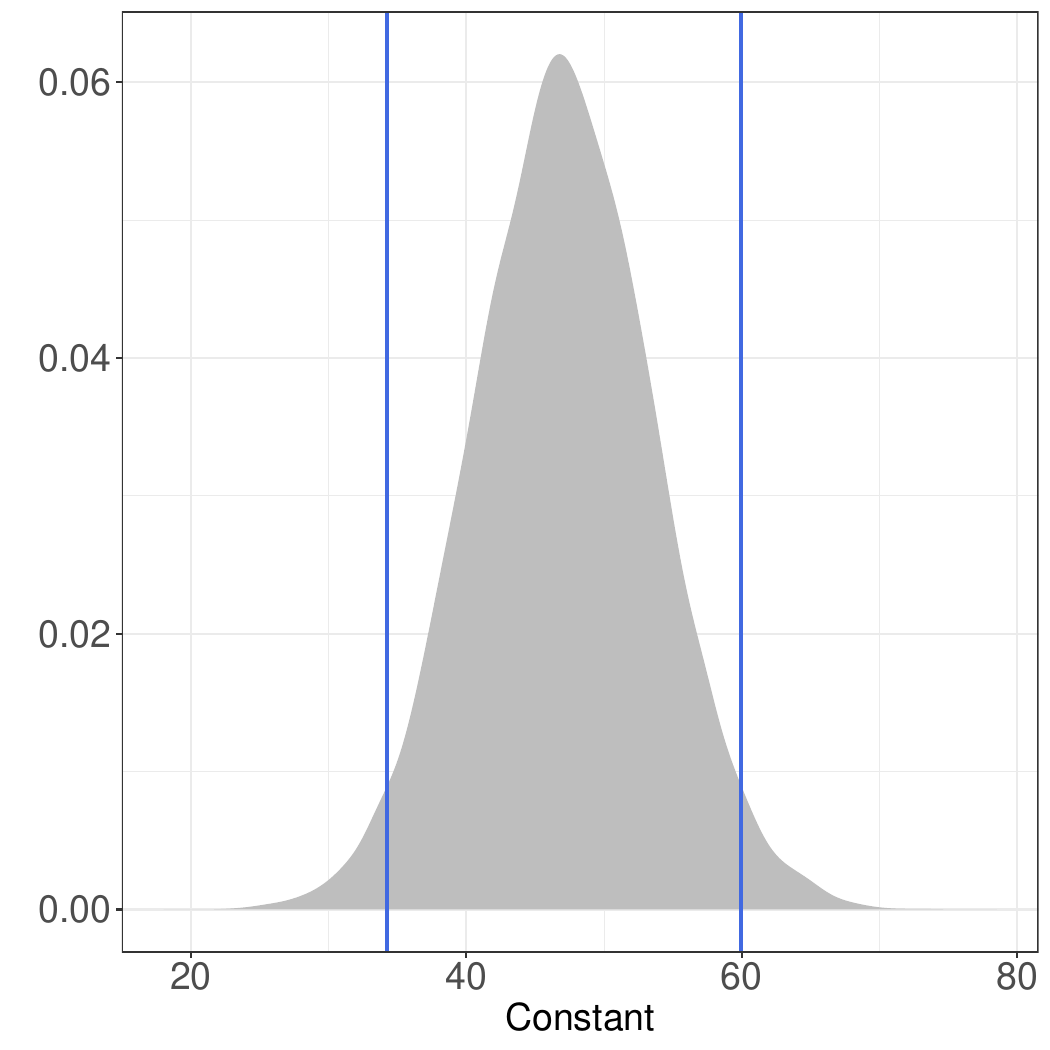} &
\includegraphics[width=0.4\textwidth]{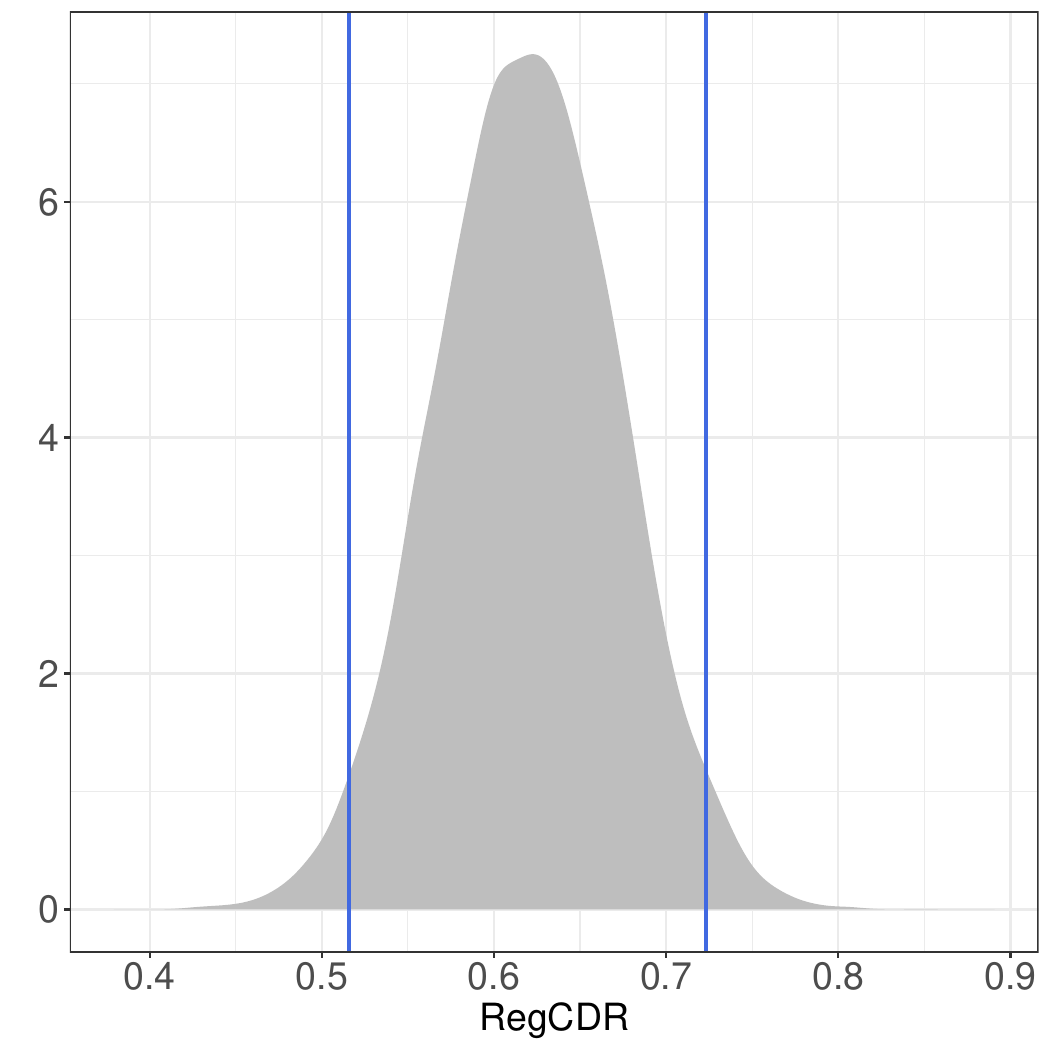}\\
\includegraphics[width=0.4\textwidth]{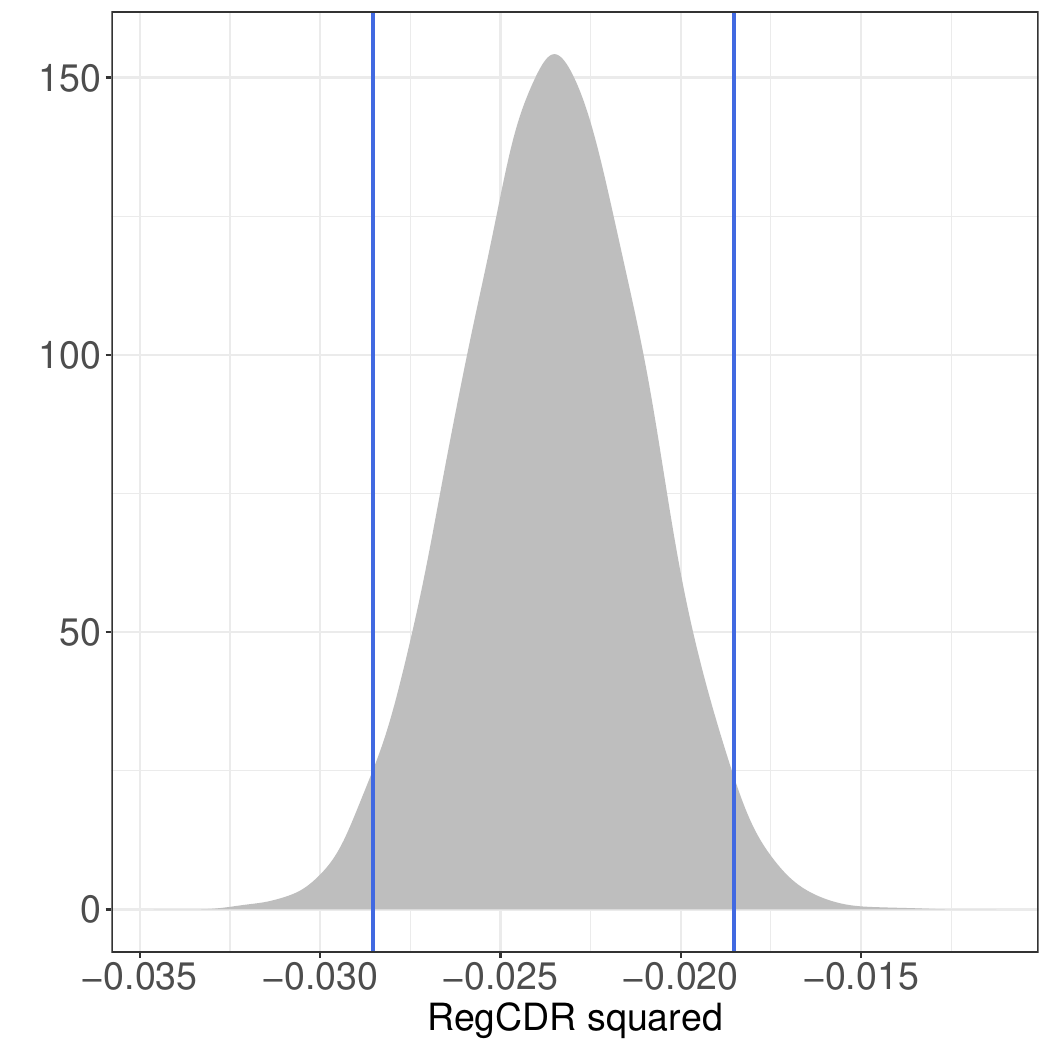} &
\includegraphics[width=0.4\textwidth]{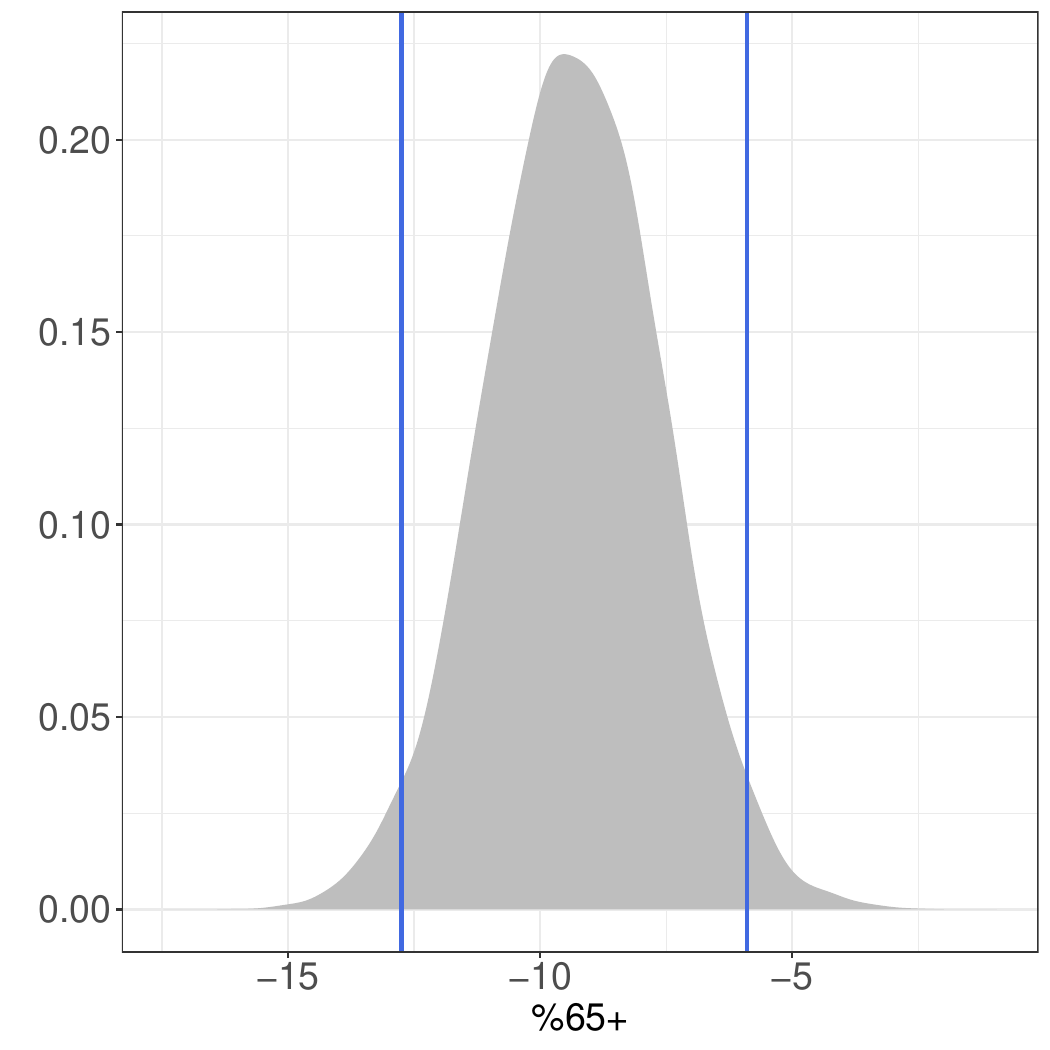}\\
\includegraphics[width=0.4\textwidth]{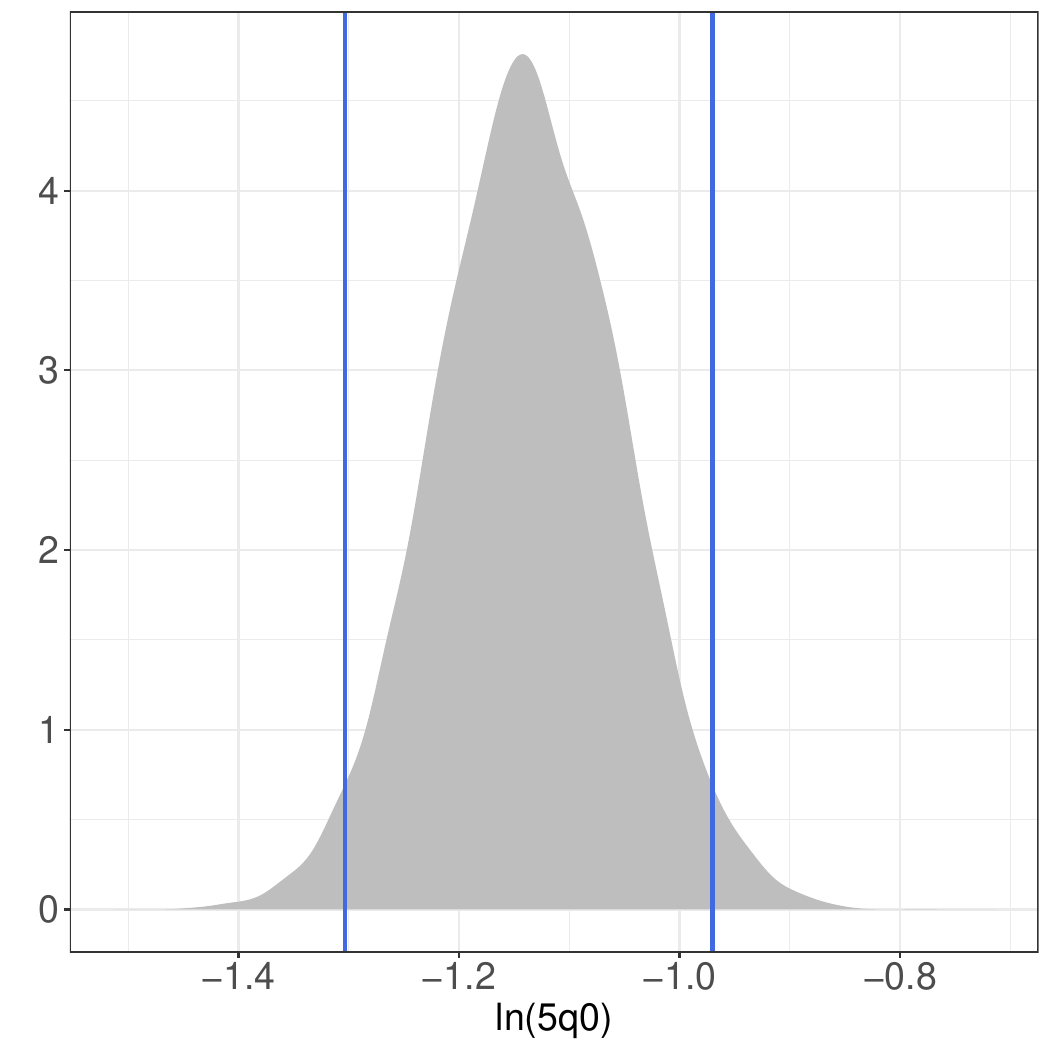} &
\includegraphics[width=0.4\textwidth]{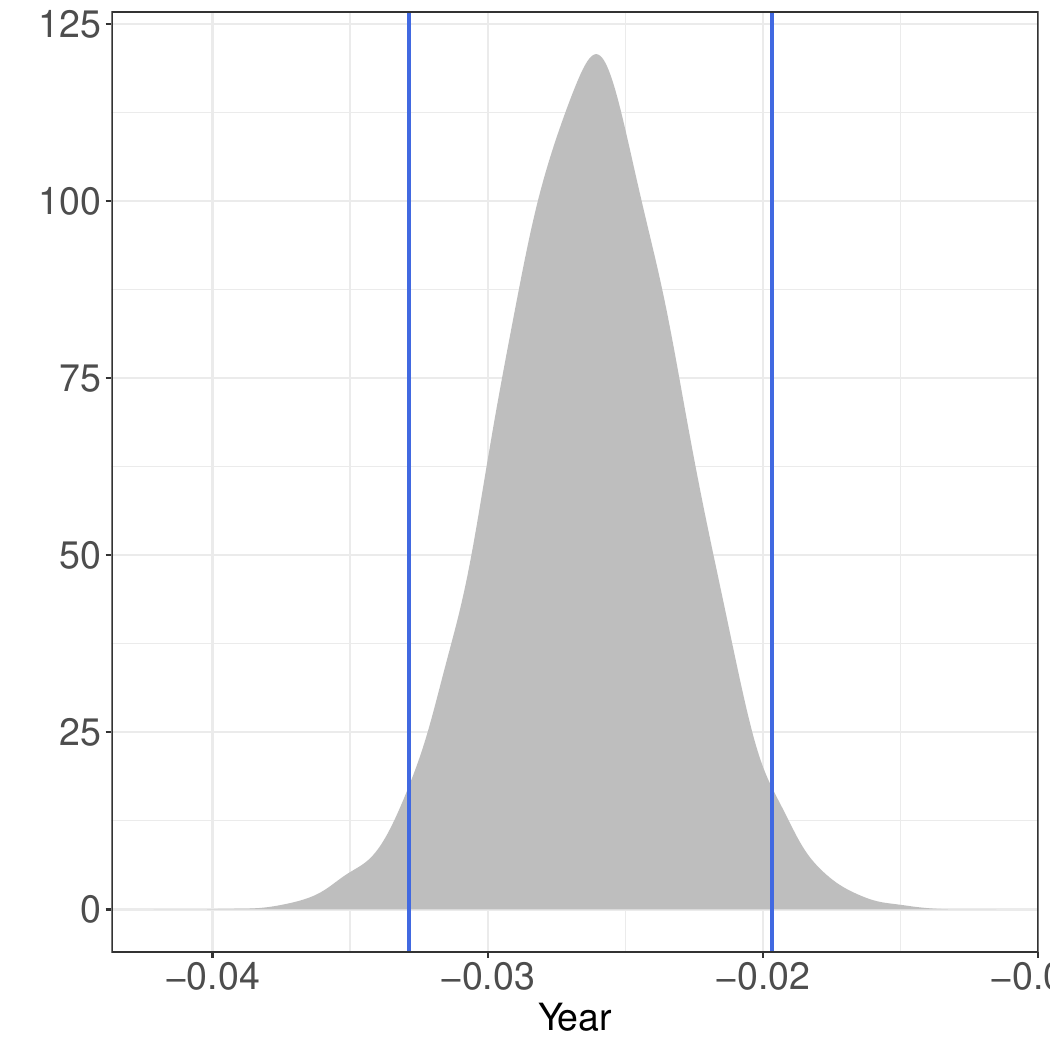}
\end{tabular}
\end{center}
\vspace{-0.5cm}
\caption{
Densities with credible intervals (blue lines). The  models are based on an  dataset updated to 2019, which uses GBD death estimates based on the GBD 2019 and  comprises 120 countries and 2,748 country-years from 1970-2019  \citep{collaborators2020global}, males and model 1 using a Half-Cauchy prior
for the local scale of the errors.}
\label{fig:fig_both9}
\end{figure}

\begin{figure}[ht]
\begin{center}
\begin{tabular}{ccc}
\includegraphics[width=0.45\textwidth]{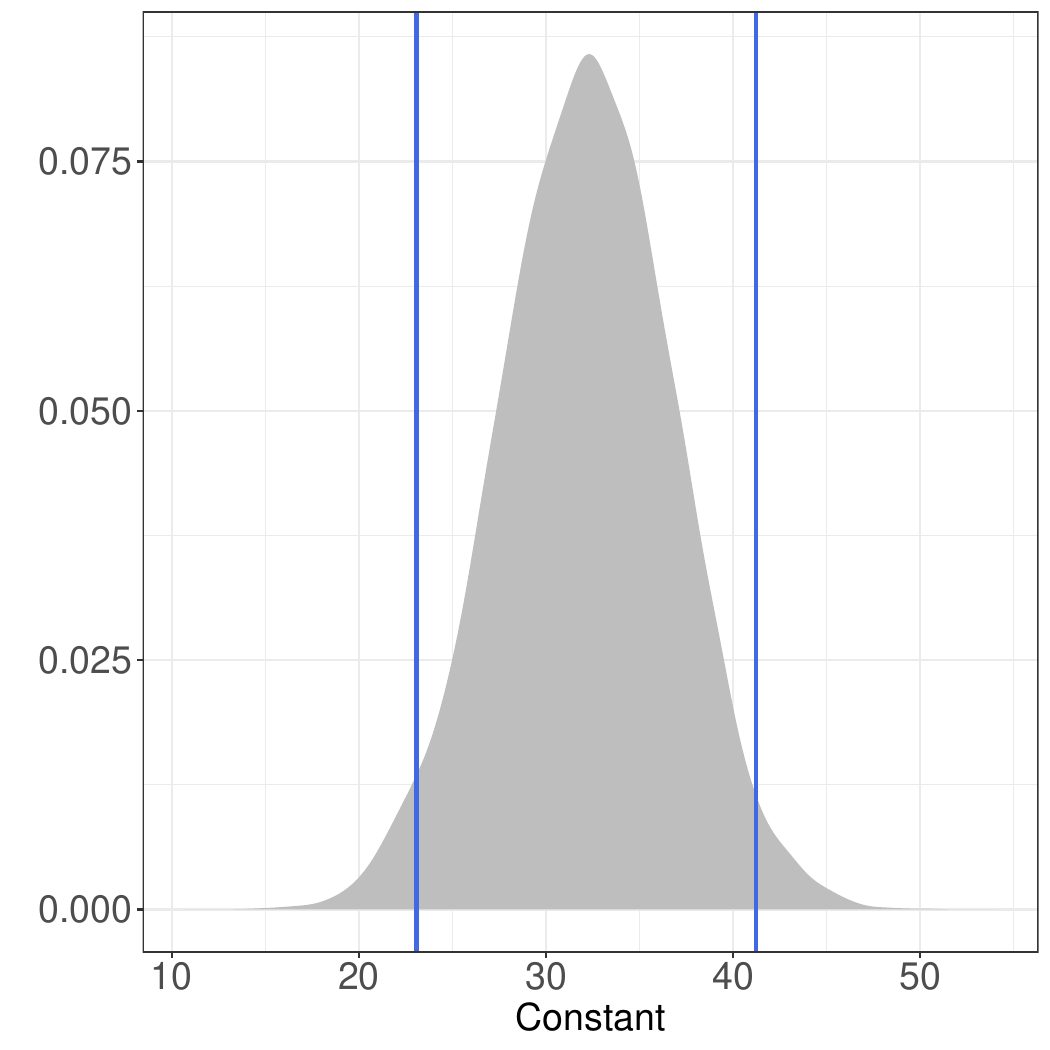} &
\includegraphics[width=0.45\textwidth]{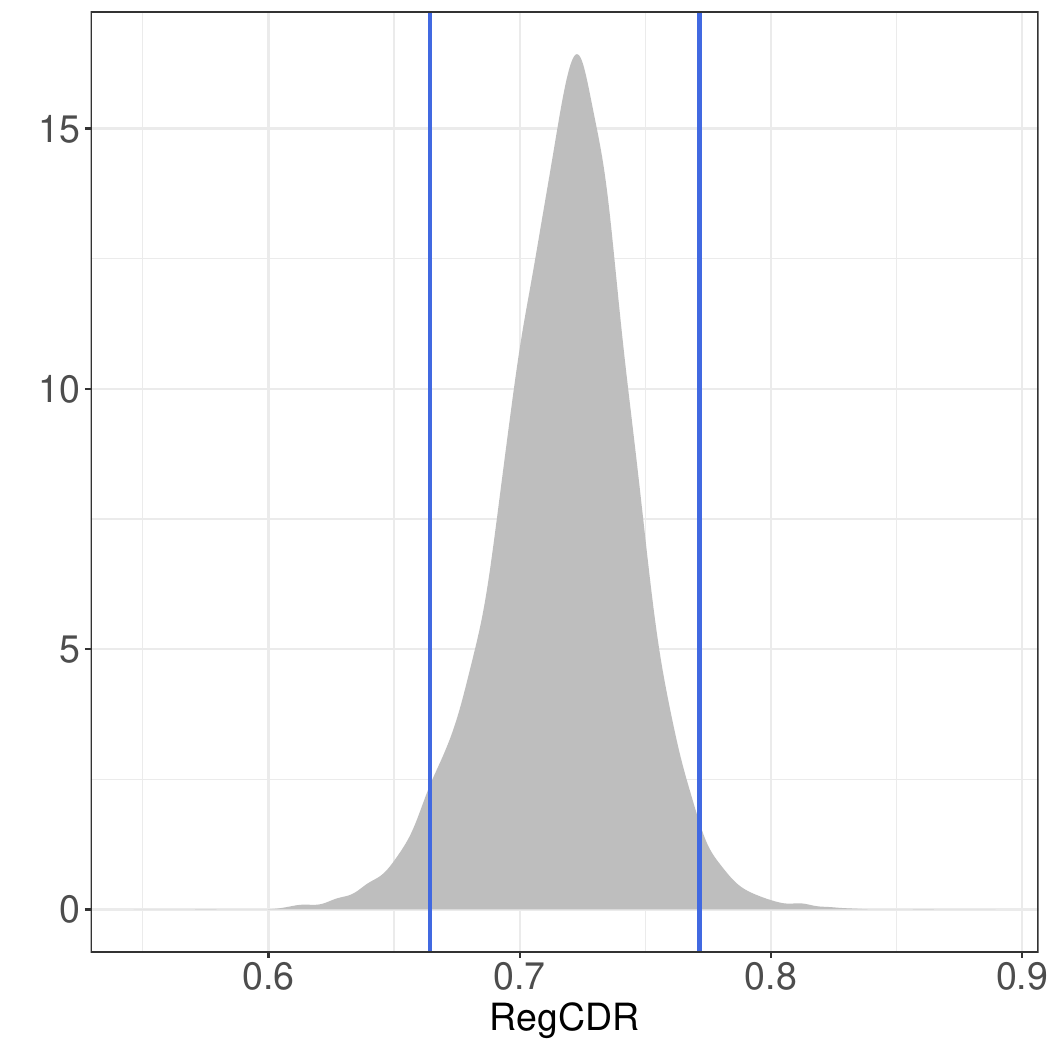}\\
\includegraphics[width=0.45\textwidth]{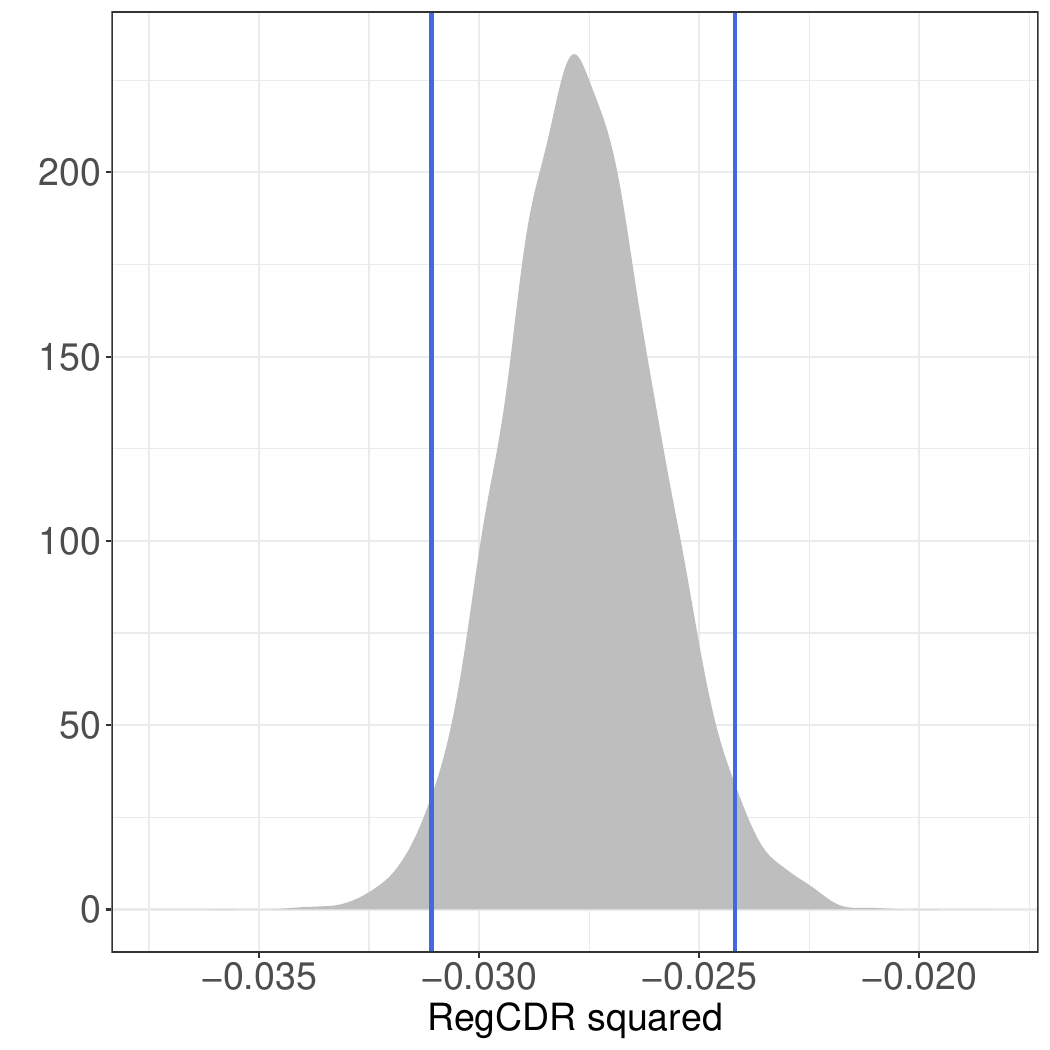} &
\includegraphics[width=0.45\textwidth]{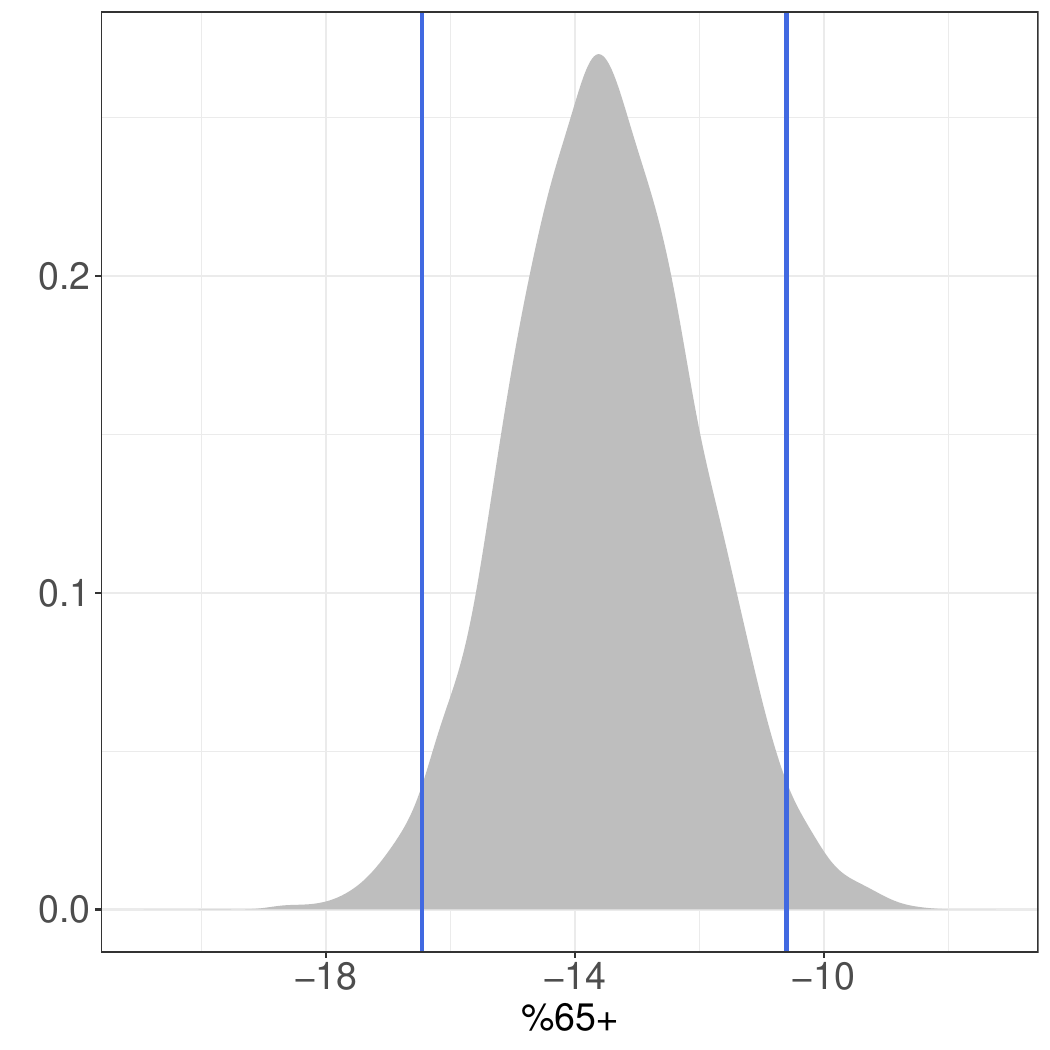}\\
\includegraphics[width=0.45\textwidth]{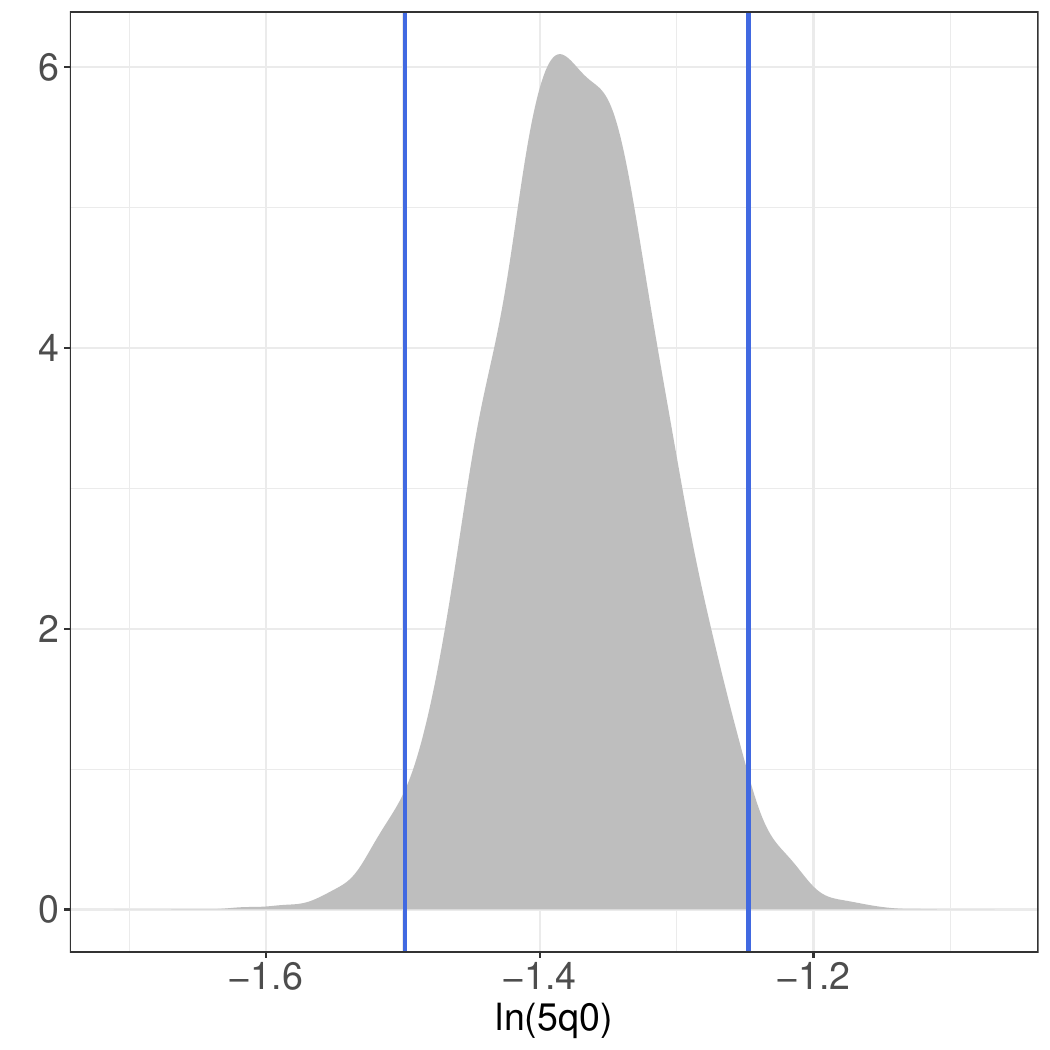} &
\includegraphics[width=0.45\textwidth]{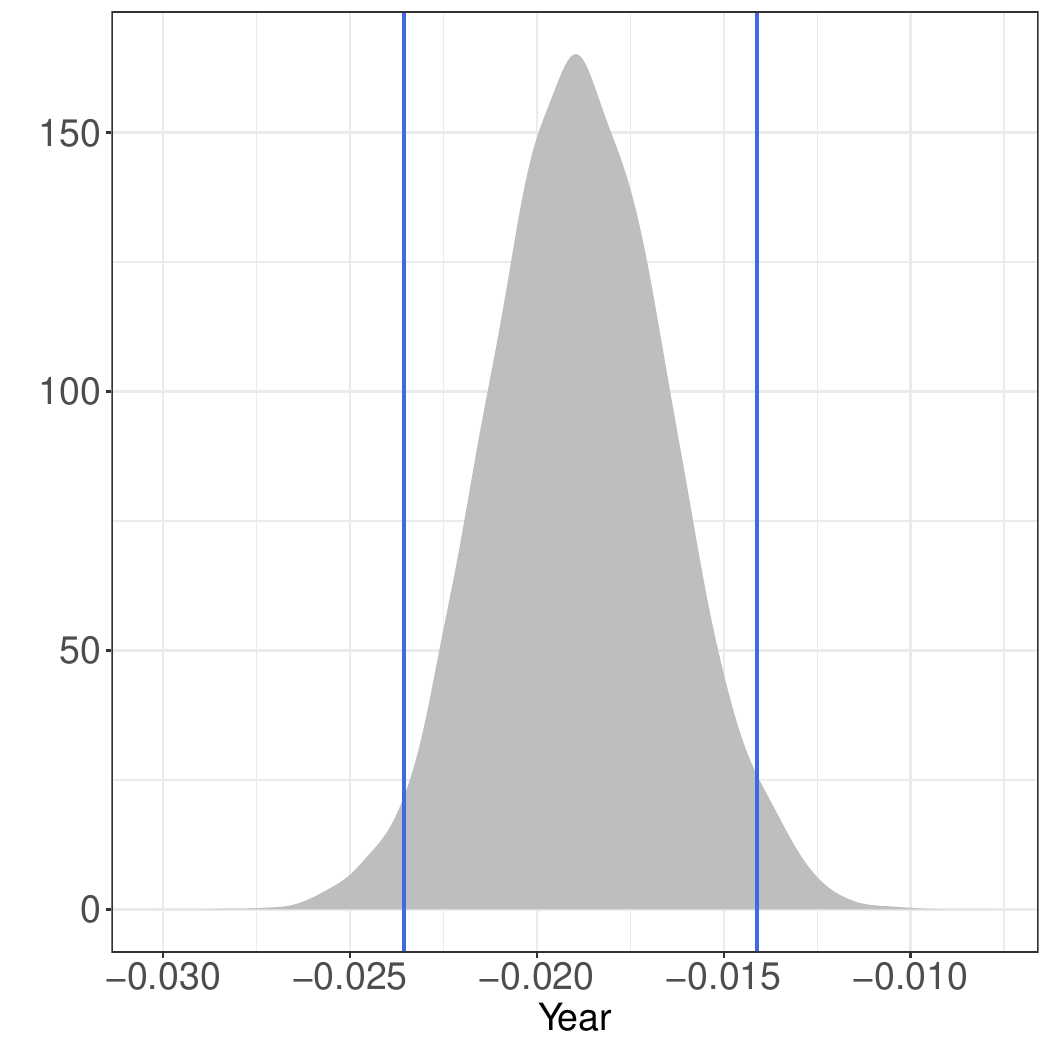}
\end{tabular}
\end{center}
\vspace{-0.5cm}
\caption{
Densities with credible intervals (blue lines). The  models are based on an  dataset updated to 2019, which uses GBD death estimates based on the GBD 2019 and  comprises 120 countries and 2,748 country-years from 1970-2019  \citep{collaborators2020global}, males and model 1 using a Half-Cauchy prior
for the local scale of the errors.}
\label{fig:fig_both10}
\end{figure}

\clearpage

\begin{figure}[ht]

\begin{center}
\begin{tabular}{ccc}
\includegraphics[width=0.4\textwidth]{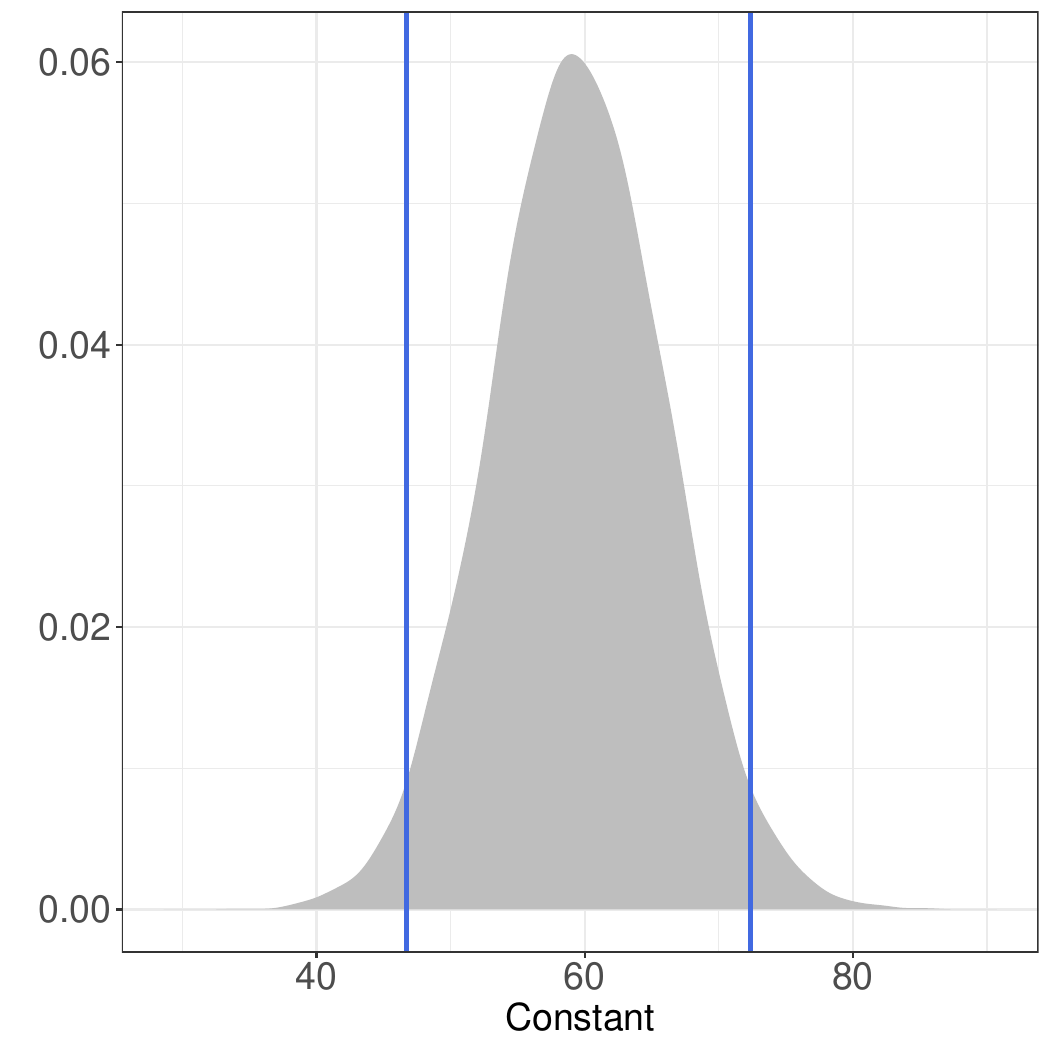} &
\includegraphics[width=0.4\textwidth]{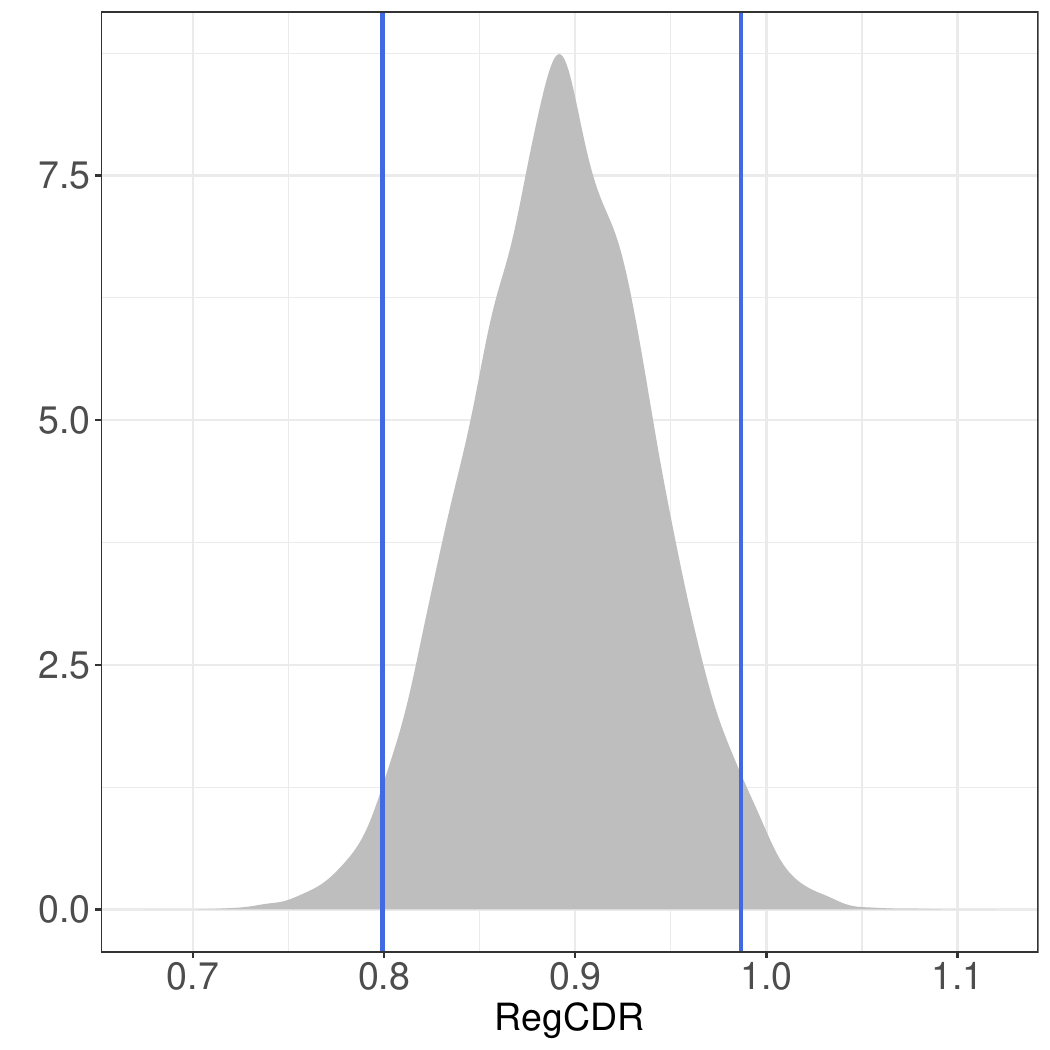}\\
\includegraphics[width=0.4\textwidth]{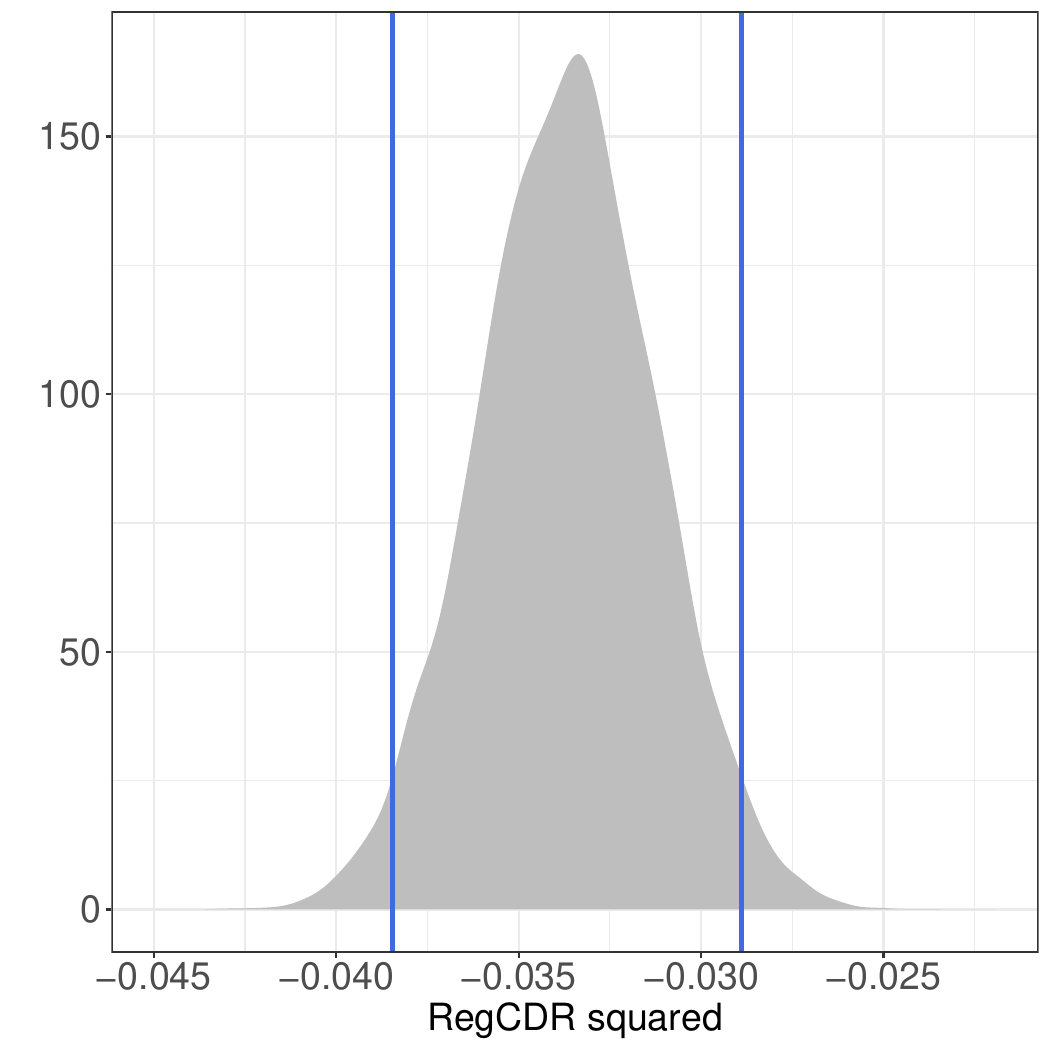} &
\includegraphics[width=0.4\textwidth]{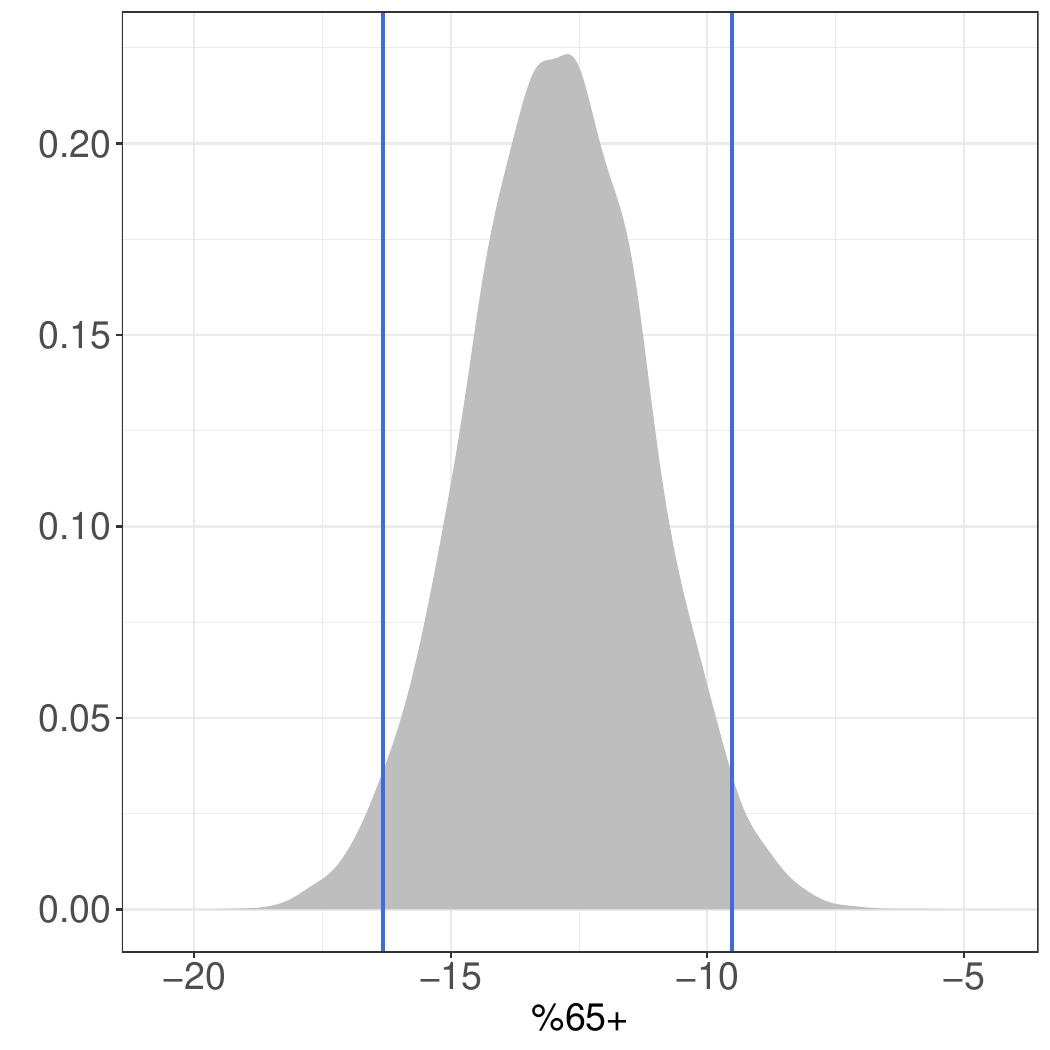}\\
\includegraphics[width=0.4\textwidth]{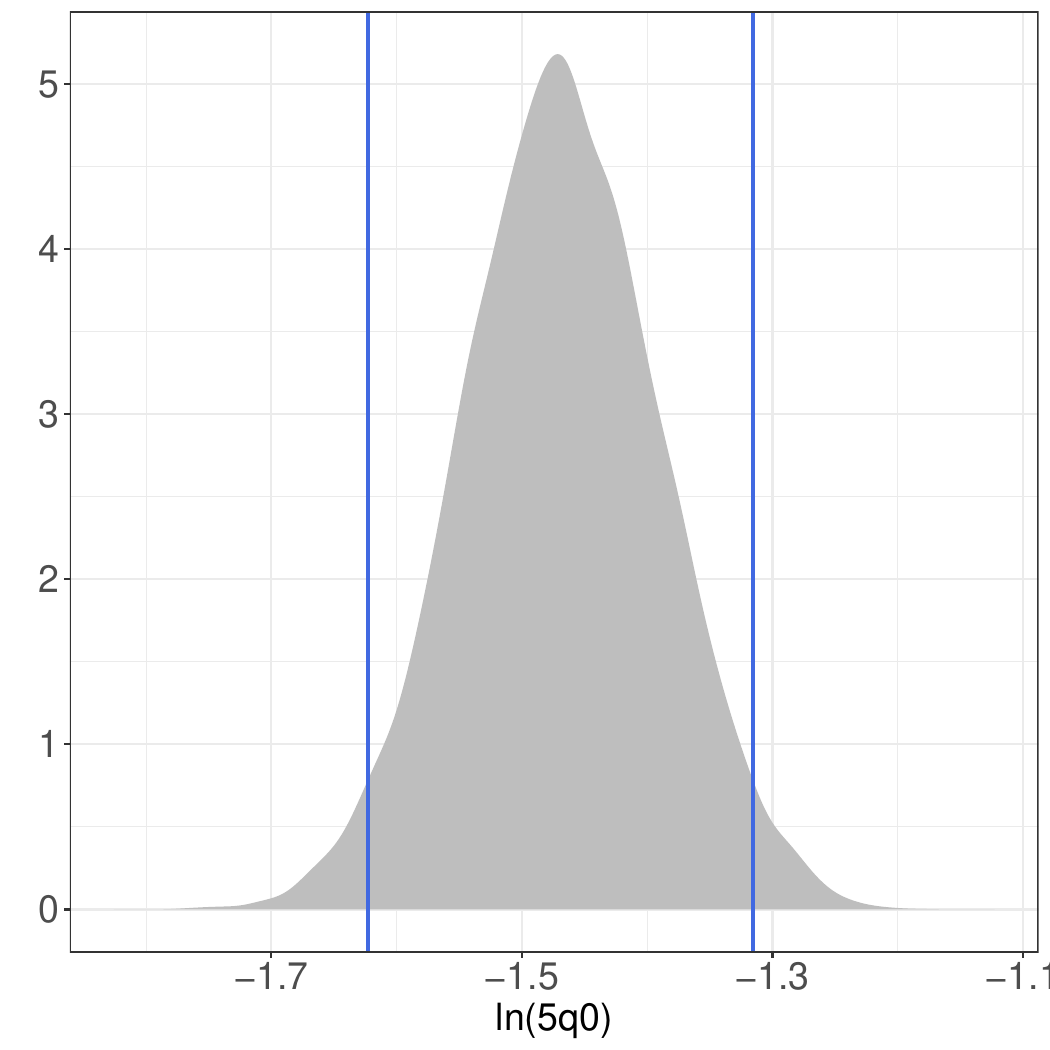} &
\includegraphics[width=0.4\textwidth]{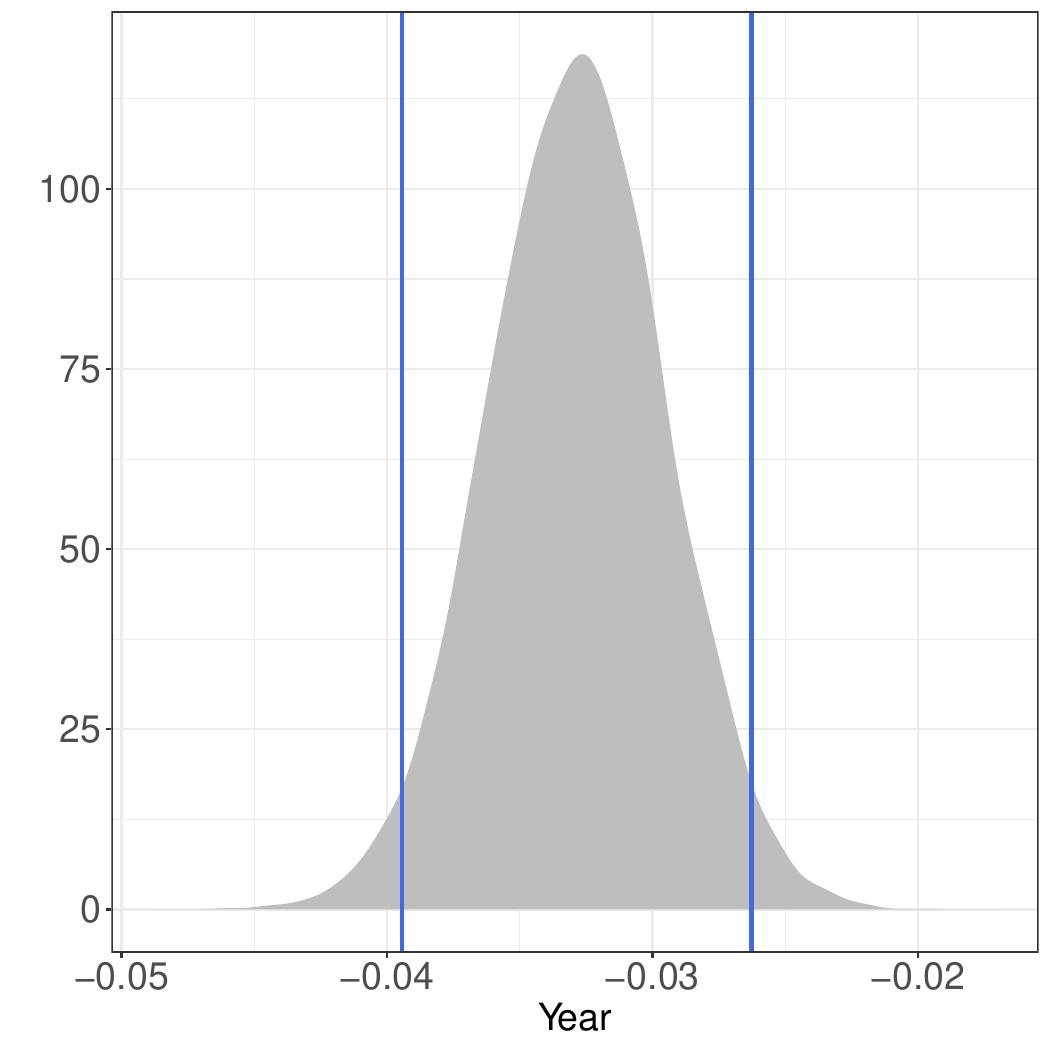}
\end{tabular}
\end{center}
\vspace{-0.5cm}
\caption{
Densities with credible intervals (blue lines). The  models are based on an  dataset updated to 2019, which uses GBD death estimates based on the GBD 2019 and  comprises 120 countries and 2,748 country-years from 1970-2019  \citep{collaborators2020global}, males and model 2 using a Half-Cauchy prior
for the local scale of the errors.}
\label{fig:fig_both11}
\end{figure}

\begin{figure}[ht]
\begin{center}
\begin{tabular}{ccc}
\includegraphics[width=0.45\textwidth]{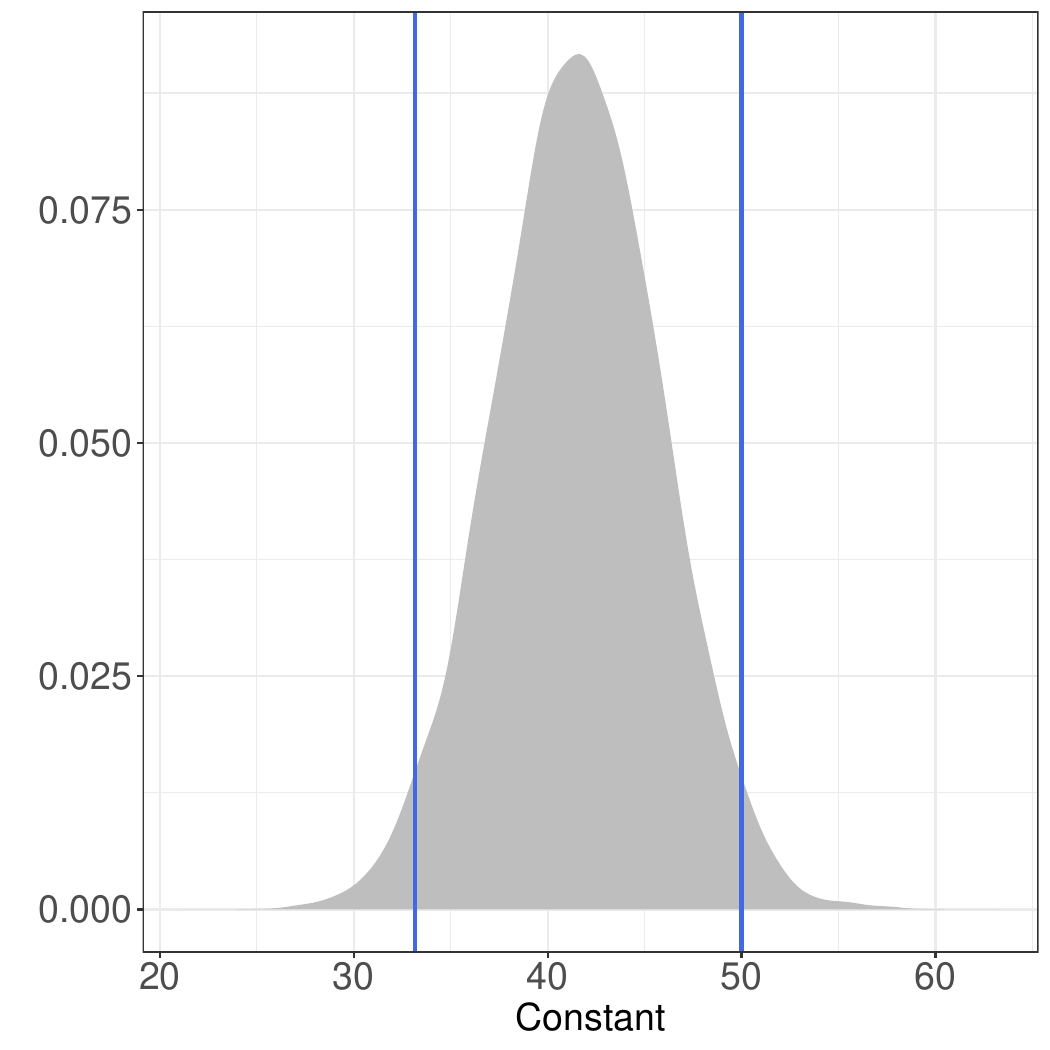} &
\includegraphics[width=0.45\textwidth]{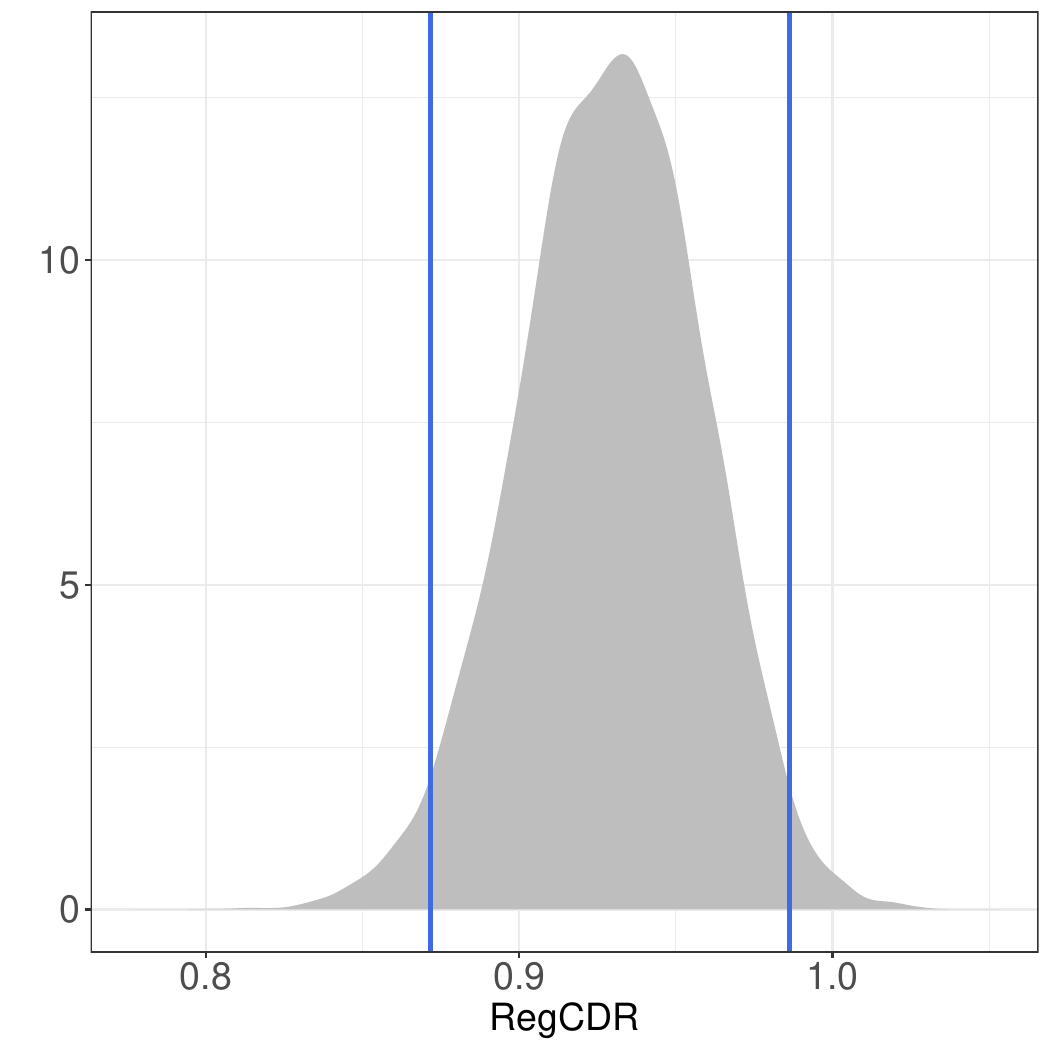}\\
\includegraphics[width=0.45\textwidth]{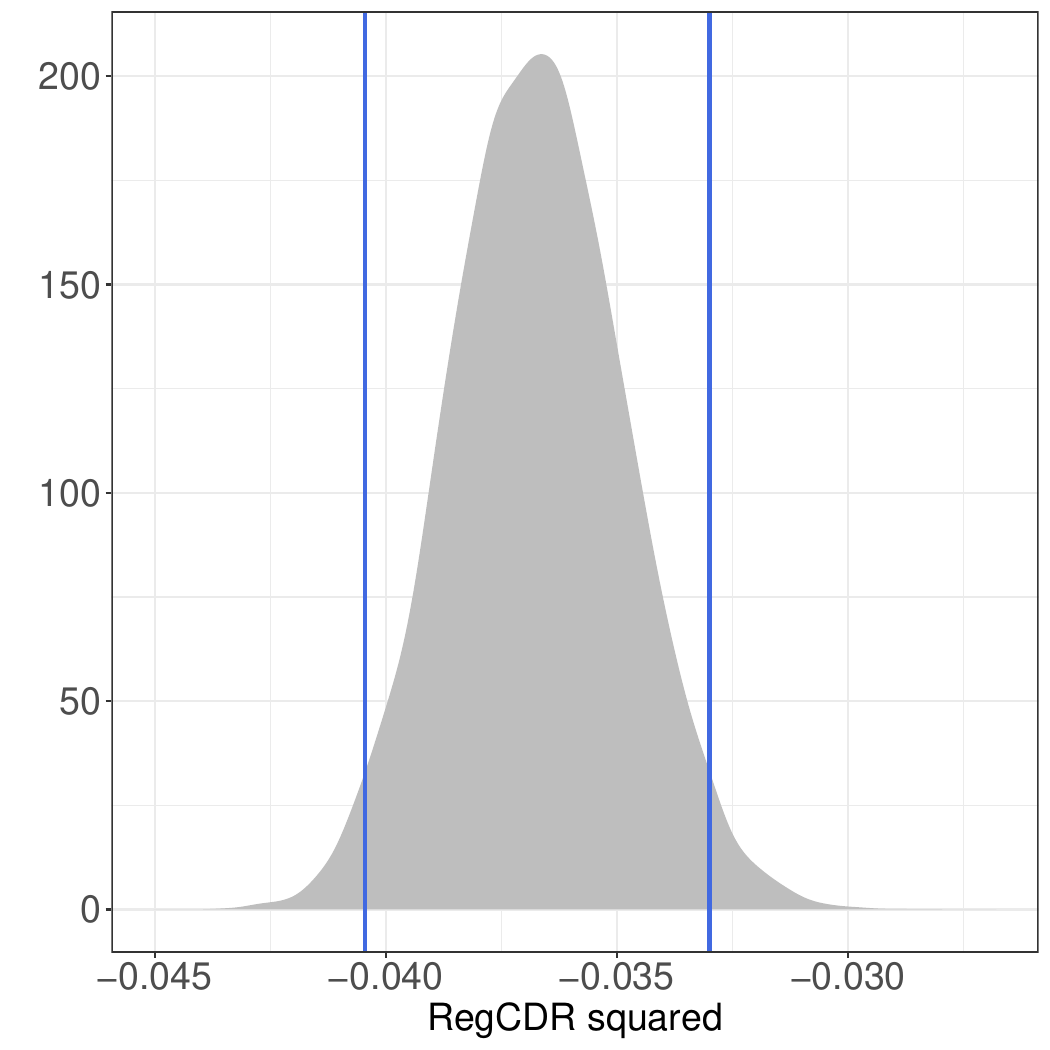} &
\includegraphics[width=0.45\textwidth]{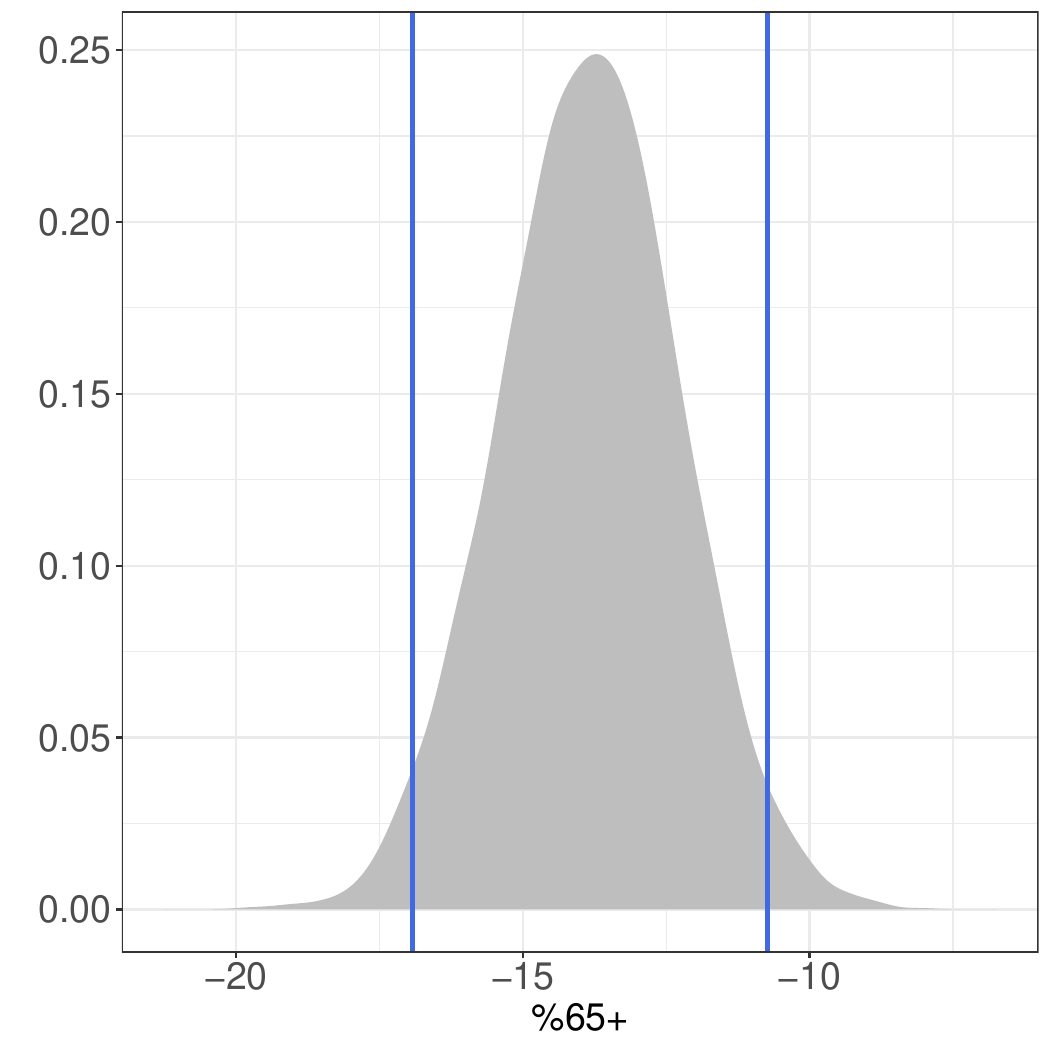}\\
\includegraphics[width=0.45\textwidth]{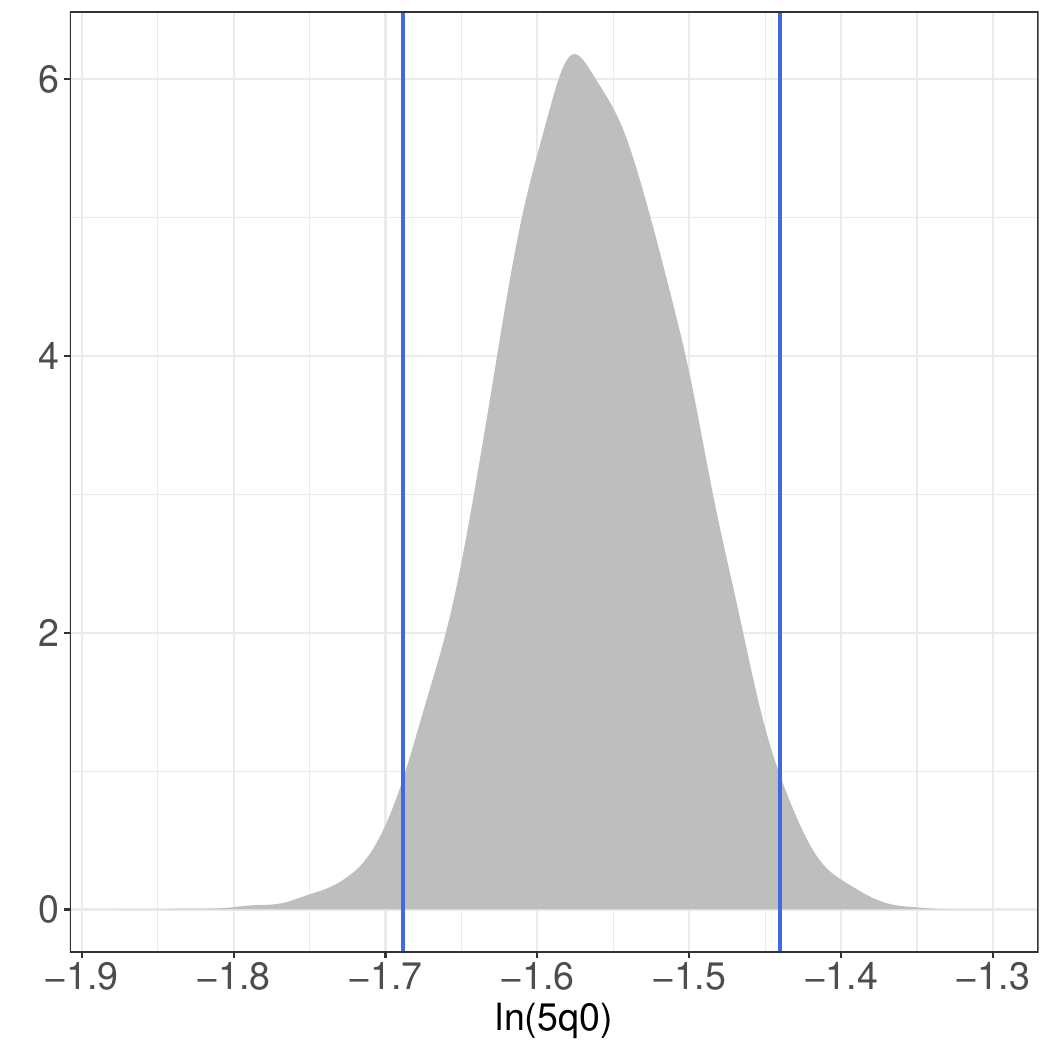} &
\includegraphics[width=0.45\textwidth]{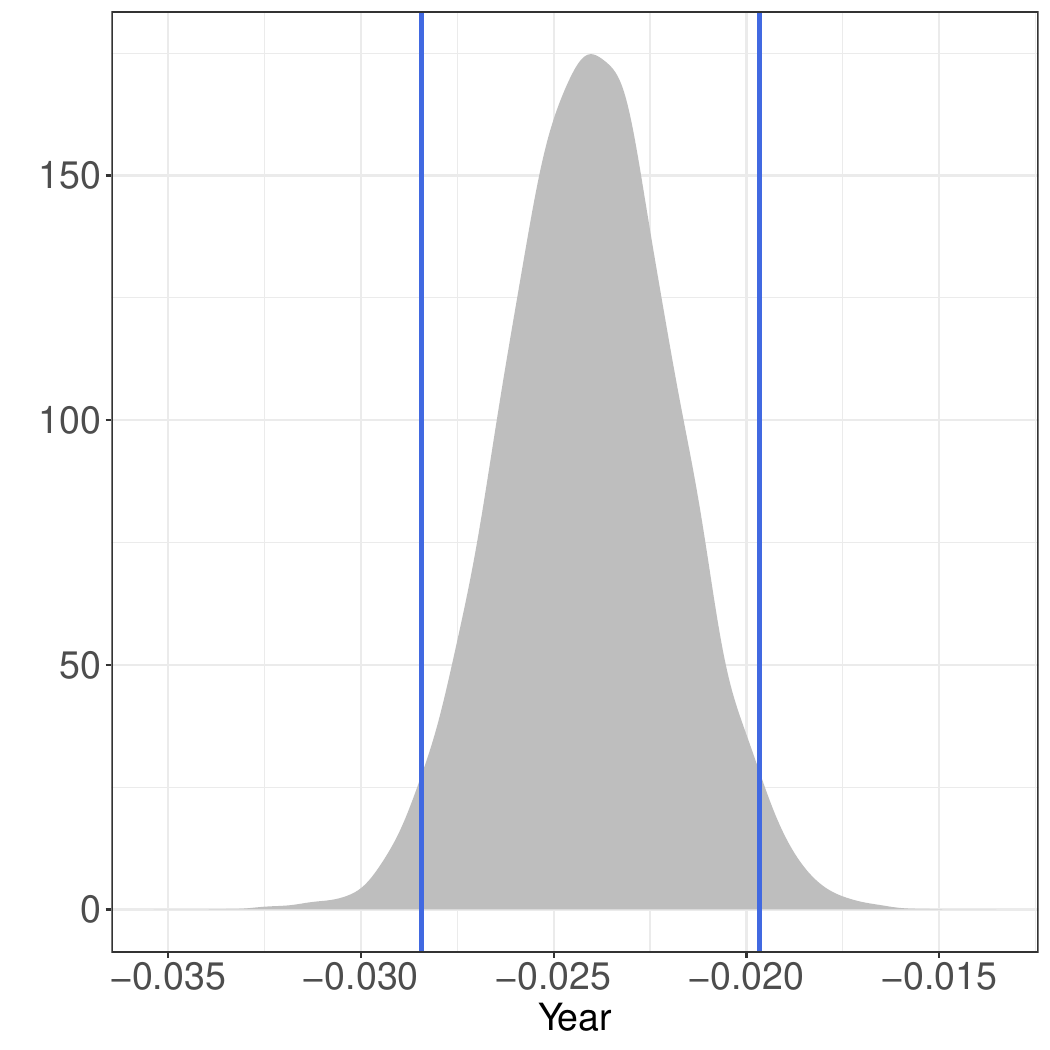}
\end{tabular}
\end{center}
\vspace{-0.5cm}
\caption{
Densities with credible intervals (blue lines). The  models are based on an  dataset updated to 2019, which uses GBD death estimates based on the GBD 2019 and  comprises 120 countries and 2,748 country-years from 1970-2019  \citep{collaborators2020global}, males and model 2 using a Half-Cauchy prior
for the local scale of the errors}
\label{fig:fig_both12}
\end{figure}

\clearpage

\section{Results subnational level}
\enlargethispage{5cm}
\label{sup_post2}

This section contains the posterior predictive  estimates at subnational levels (Departments) and national level and the corresponding 95\% credible intervals under Gamma, HS, LA and Student-t local scale priors fro the random effects and  Half-Cauchy and Gamma global scale priors for the errors. The results for model 1 and 2 for females are presented
in Tables \ref{tab:tab_measures3} and \ref{tab:tab_measures4} and for males in Tables \ref{tab:tab_measures5} and \ref{tab:tab_measures6}.

\begin{table}[ht]
\tiny
\renewcommand{\arraystretch}{1.1}
\begin{tabular}{r>{\columncolor[gray]{0.85}}rrrrrrrrrrrrrr}
\hline
  Department & Census  & Gamma &  $2.5\%$ &  $97.5\%$ & Student-t &  $2.5\%$ &  $97.5\%$ & LA &  $2.5\%$ &  $97.5\%$ & HS &  $2.5\%$ &  $97.5\%$  \\
 \hline
 Antioquia & 94 & 95 & 95 & 96 & 95 & 95 & 96 & 95 & 95 & 96 & 95 & 95 & 96 \\
Atl\'antico & 97 & 94 & 93 & 95 & 94 & 93 & 95 & 94 & 93 & 95 & 94 & 93 & 95 \\
Bogot\'a, D.C. & 97 & 93 & 92 & 94 & 93 & 92 & 94 & 93 & 92 & 94 & 93 & 92 & 94 \\
 Bol\'ivar & 90 & 92 & 91 & 93 & 92 & 91 & 93 & 92 & 91 & 93 & 92 & 90 & 93 \\
Boyac\'a  & 95 & 94 & 92 & 95 & 94 & 93 & 95 & 94 & 93 & 95 & 94 & 93 & 95 \\
 Caldas & 94 & 98 & 97 & 98 & 98 & 97 & 98 & 98 & 97 & 98 & 98 & 97 & 98 \\
 Caquet\'a & 89 & 93 & 92 & 94 & 93 & 92 & 94 & 93 & 92 & 94 & 93 & 92 & 94 \\
 Cauca & 85 & 87 & 85 & 89 & 87 & 85 & 89 & 87 & 85 & 89 & 86 & 84 & 89 \\
 Cesar & 89 & 86 & 84 & 88 & 86 & 83 & 88 & 86 & 83 & 88 & 85 & 83 & 88 \\
C\'ordoba & 89 & 89 & 86 & 90 & 89 & 86 & 90 & 88 & 86 & 90 & 88 & 86 & 90 \\
 Cundinamarca & 94 & 96 & 95 & 97 & 96 & 95 & 97 & 96 & 95 & 97 & 96 & 95 & 97 \\
  \rowcolor[gray]{0.8}
 Choc\'o & 67 & 67 & 62 & 72 & 67 & 62 & 71 & 67 & 62 & 71 & 65 & 61 & 70 \\
 Huila & 93 & 92 & 91 & 94 & 92 & 91 & 94 & 92 & 91 & 94 & 92 & 91 & 94 \\
  \rowcolor[gray]{0.8}
 La Guajira & 38 & 57 & 51 & 62 & 56 & 51 & 62 & 56 & 51 & 61 & 55 & 50 & 60 \\
 Magdalena & 90 & 91 & 89 & 92 & 91 & 89 & 92 & 91 & 89 & 92 & 91 & 89 & 92 \\
 Meta & 92 & 96 & 96 & 97 & 96 & 96 & 97 & 96 & 96 & 97 & 96 & 95 & 97 \\
 Nari\~no & 85 & 89 & 86 & 90 & 89 & 86 & 90 & 89 & 86 & 90 & 88 & 86 & 90 \\
 Norte de Santander & 92 & 93 & 92 & 94 & 93 & 92 & 94 & 93 & 92 & 94 & 93 & 92 & 94 \\
 Quindio & 95 & 96 & 95 & 97 & 96 & 95 & 97 & 96 & 95 & 97 & 96 & 95 & 97 \\
 Risaralda & 95 & 97 & 96 & 97 & 97 & 96 & 97 & 97 & 96 & 97 & 97 & 96 & 97 \\
 Santander & 95 & 93 & 92 & 94 & 93 & 92 & 94 & 93 & 92 & 94 & 93 & 92 & 94 \\
 Sucre & 88 & 95 & 94 & 96 & 95 & 94 & 96 & 95 & 94 & 96 & 94 & 93 & 95 \\
 Tolima & 93 & 93 & 91 & 94 & 93 & 91 & 94 & 93 & 91 & 94 & 93 & 91 & 94 \\
 Valle del Cauca & 96 & 96 & 95 & 96 & 96 & 95 & 97 & 96 & 95 & 97 & 96 & 95 & 97 \\
 Arauca & 86 & 95 & 94 & 96 & 95 & 94 & 96 & 95 & 94 & 96 & 95 & 94 & 96 \\
 Casanare & 86 & 95 & 94 & 96 & 95 & 94 & 96 & 95 & 94 & 96 & 95 & 93 & 96 \\
 Putumayo & 82 & 87 & 84 & 89 & 87 & 84 & 89 & 87 & 84 & 89 & 86 & 84 & 88 \\
San Andr\'es & 95 & 93 & 92 & 94 & 93 & 92 & 94 & 93 & 91 & 94 & 93 & 91 & 94 \\
  \rowcolor[gray]{0.8}
 Amazonas & 76 & 68 & 64 & 72 & 68 & 64 & 72 & 68 & 63 & 72 & 67 & 63 & 71 \\
 Guain\'ia & 56 & 89 & 85 & 91 & 88 & 85 & 91 & 88 & 85 & 91 & 87 & 84 & 90 \\
 Guaviare & 79 & 93 & 91 & 94 & 93 & 91 & 94 & 93 & 91 & 94 & 92 & 90 & 94 \\
  \rowcolor[gray]{0.8}
 Vaup\'es & 62 & 63 & 58 & 68 & 63 & 58 & 68 & 63 & 58 & 67 & 61 & 57 & 67 \\
 Vichada & 62 & 88 & 84 & 90 & 87 & 84 & 90 & 87 & 84 & 90 & 86 & 82 & 89 \\
 \rowcolor[gray]{0.8}
 National & 92 & 93 & 91 & 94 & 93 & 91 & 94 & 93 & 91 & 94 & 93 & 91 & 94 \\
   \hline
\end{tabular}
\caption{\small
Posterior predictive estimates and 95\% credible intervals under Gamma, HS, LA and Student-t local priors for the random effects and  Half-Cauchy local prior for the errors and
census based completeness values. Inf
and Sup represent the 2.5\%  and 97.5\% quantiles from the posterior marginal distribution
of the completeness in each Department and at the national level. Females and model 1. }
\label{tab:tab_measures3}
\end{table}

\clearpage

\begin{table}[ht]
\tiny
\renewcommand{\arraystretch}{1.5}
\begin{tabular}{r>{\columncolor[gray]{0.85}}rrrrrrrrrrrrrr}
\hline
  Department & Census & Gamma &  $2.5\%$ &  $97.5\%$ & Student-t &  $2.5\%$ &  $97.5\%$ & LA &  $2.5\%$ &  $97.5\%$ & HS &  $2.5\%$ &  $97.5\%$  \\
 \hline
Antioquia & 94 & 95 & 94 & 96 & 95 & 94 & 96 & 95 & 94 & 96 & 95 & 94 & 96 \\
Atl\'antico & 97 & 94 & 93 & 95 & 94 & 93 & 95 & 94 & 93 & 95 & 94 & 92 & 95 \\
Bogot\'a, D.C. & 97 & 93 & 92 & 94 & 93 & 92 & 94 & 93 & 92 & 94 & 93 & 91 & 94 \\
 Bol\'ivar & 90 & 91 & 89 & 92 & 91 & 89 & 92 & 91 & 89 & 92 & 91 & 89 & 92 \\
Boyac\'a & 95 & 94 & 93 & 95 & 94 & 93 & 95 & 94 & 93 & 95 & 94 & 93 & 95 \\
 Caldas & 94 & 98 & 97 & 98 & 98 & 97 & 98 & 98 & 97 & 98 & 98 & 97 & 98 \\
 Caquet\'a & 89 & 92 & 91 & 93 & 92 & 91 & 93 & 92 & 91 & 93 & 92 & 90 & 93 \\
Cauca & 85 & 85 & 83 & 88 & 85 & 83 & 88 & 85 & 83 & 88 & 85 & 82 & 88 \\
 Cesar & 89 & 83 & 81 & 86 & 83 & 81 & 86 & 83 & 81 & 86 & 83 & 80 & 86 \\
C\'ordoba & 89 & 87 & 84 & 89 & 87 & 85 & 89 & 87 & 84 & 89 & 87 & 84 & 89 \\
 Cundinamarca & 94 & 95 & 94 & 96 & 95 & 95 & 96 & 95 & 94 & 96 & 95 & 94 & 96 \\
  \rowcolor[gray]{0.8}
 Choc\'o  & 67 & 60 & 56 & 65 & 60 & 56 & 65 & 60 & 55 & 65 & 60 & 55 & 65 \\
Huila & 93 & 93 & 92 & 94 & 93 & 92 & 94 & 93 & 92 & 94 & 93 & 92 & 94 \\
 \rowcolor[gray]{0.8}
 La Guajira & 38 & 53 & 48 & 58 & 53 & 47 & 58 & 52 & 47 & 57 & 52 & 47 & 57 \\
 Magdalena & 90 & 89 & 87 & 91 & 89 & 87 & 91 & 89 & 87 & 91 & 89 & 87 & 91 \\
 Meta & 92 & 95 & 94 & 96 & 95 & 94 & 96 & 95 & 94 & 96 & 95 & 94 & 96 \\
 Nari\~no  & 85 & 89 & 87 & 91 & 89 & 87 & 91 & 89 & 87 & 91 & 89 & 86 & 91 \\
 Norte de Santander & 92 & 93 & 92 & 94 & 94 & 92 & 95 & 93 & 92 & 94 & 93 & 92 & 94 \\
 Quindio & 95 & 96 & 95 & 97 & 96 & 95 & 97 & 96 & 95 & 97 & 96 & 95 & 97 \\
 Risaralda & 95 & 97 & 96 & 97 & 97 & 96 & 97 & 97 & 96 & 97 & 97 & 96 & 97 \\
 Santander & 95 & 94 & 93 & 95 & 94 & 93 & 95 & 94 & 93 & 95 & 94 & 92 & 95 \\
 Sucre & 88 & 93 & 92 & 94 & 93 & 92 & 94 & 93 & 92 & 94 & 93 & 92 & 94 \\
 Tolima & 93 & 94 & 93 & 95 & 94 & 93 & 95 & 94 & 93 & 95 & 94 & 93 & 95 \\
 Valle del Cauca & 96 & 96 & 95 & 97 & 96 & 95 & 97 & 96 & 95 & 97 & 96 & 95 & 97 \\
 Arauca & 86 & 94 & 93 & 95 & 94 & 93 & 95 & 94 & 93 & 95 & 94 & 93 & 95 \\
 Casanare & 86 & 92 & 91 & 94 & 92 & 91 & 94 & 92 & 91 & 94 & 92 & 91 & 93 \\
 Putumayo & 82 & 87 & 85 & 89 & 87 & 85 & 89 & 87 & 85 & 89 & 87 & 85 & 89 \\
San Andr\'es& 95 & 90 & 88 & 91 & 90 & 88 & 91 & 90 & 88 & 91 & 90 & 88 & 91 \\
  \rowcolor[gray]{0.8}
  Amazonas & 76 & 66 & 62 & 70 & 66 & 62 & 70 & 66 & 62 & 70 & 66 & 61 & 70 \\
 Guain\'ia & 56 & 80 & 77 & 83 & 80 & 77 & 83 & 80 & 77 & 83 & 80 & 77 & 83 \\
  Guaviare & 79 & 89 & 87 & 90 & 89 & 87 & 90 & 89 & 87 & 90 & 89 & 86 & 90 \\
   \rowcolor[gray]{0.8}
 Vaup\'es & 62 & 58 & 53 & 62 & 58 & 53 & 62 & 57 & 53 & 62 & 57 & 52 & 62 \\
  Vichada & 62 & 78 & 75 & 81 & 78 & 74 & 81 & 78 & 74 & 81 & 78 & 74 & 81 \\
  \rowcolor[gray]{0.8}
  National & 92 & 92 & 91 & 94 & 93 & 91 & 94 & 93 & 91 & 94 & 93 & 91 & 94 \\
   \hline
\end{tabular}
\caption{\small
Posterior predictive estimates and 95\% credible intervals under Gamma, HS, LA and Student-t local priors for the random effects and  a Half-Cauchy local prior for the errors and
census based completeness values. Inf
and Sup represent the 2.5\%  and 97.5\% quantiles from the posterior marginal distribution
of the completeness in each Department and at the national level. Females and model 2. }
\label{tab:tab_measures4}
\end{table}

\clearpage

\begin{table}[ht]
\tiny
\renewcommand{\arraystretch}{1.5}
\begin{tabular}{r>{\columncolor[gray]{0.85}}rrrrrrrrrrrrrr}
\hline
  Department & Census  & Gamma &  $2.5\%$ &  $97.5\%$ & Student-t &  $2.5\%$ &  $97.5\%$ & LA &  $2.5\%$ &  $97.5\%$ & HS &  $2.5\%$ &  $97.5\%$  \\
 \hline
Antioquia & 92 & 94 & 94 & 95 & 94 & 94 & 95 & 94 & 94 & 95 & 95 & 94 & 95 \\
Atl\'antico & 95 & 93 & 92 & 94 & 93 & 92 & 94 & 93 & 92 & 94 & 93 & 92 & 94 \\
Bogot\'a, D.C. & 95 & 91 & 90 & 93 & 91 & 90 & 93 & 91 & 90 & 93 & 92 & 90 & 93 \\
 Bol\'ivar & 87 & 89 & 87 & 90 & 89 & 87 & 90 & 89 & 87 & 90 & 89 & 87 & 90 \\
Boyac\'a & 92 & 92 & 90 & 93 & 92 & 90 & 93 & 92 & 90 & 93 & 92 & 91 & 93 \\
 Caldas & 92 & 97 & 97 & 98 & 97 & 97 & 98 & 97 & 97 & 98 & 98 & 97 & 98 \\
 Caquet\'a & 87 & 93 & 92 & 94 & 93 & 92 & 94 & 93 & 92 & 94 & 93 & 92 & 94 \\
  Cauca & 81 & 80 & 76 & 82 & 79 & 76 & 82 & 79 & 76 & 82 & 79 & 76 & 82 \\
 Cesar & 87 & 82 & 79 & 85 & 82 & 79 & 84 & 82 & 79 & 84 & 81 & 79 & 83 \\
C\'ordoba & 87 & 85 & 82 & 87 & 85 & 82 & 87 & 85 & 82 & 87 & 85 & 82 & 87 \\
 Cundinamarca & 93 & 95 & 94 & 96 & 95 & 94 & 96 & 95 & 94 & 96 & 96 & 95 & 96 \\
  \rowcolor[gray]{0.8}
 Choc\'o & 64 & 58 & 53 & 63 & 57 & 53 & 62 & 57 & 52 & 62 & 55 & 51 & 59 \\
  Huila & 92 & 89 & 88 & 91 & 89 & 88 & 91 & 89 & 88 & 91 & 89 & 88 & 91 \\
    \rowcolor[gray]{0.8}
  La Guajira & 40 & 50 & 45 & 55 & 49 & 44 & 54 & 48 & 43 & 53 & 46 & 41 & 50 \\
  Magdalena & 86 & 89 & 87 & 90 & 89 & 87 & 90 & 88 & 87 & 90 & 89 & 87 & 90 \\
  Meta & 90 & 95 & 95 & 96 & 95 & 95 & 96 & 95 & 95 & 96 & 96 & 95 & 96 \\
 Nari\~no & 79 & 82 & 79 & 84 & 82 & 79 & 84 & 82 & 79 & 84 & 82 & 79 & 84 \\
  Norte de Santander & 90 & 92 & 91 & 93 & 92 & 91 & 93 & 92 & 91 & 93 & 92 & 91 & 93 \\
  Quindio & 94 & 95 & 94 & 96 & 95 & 94 & 96 & 95 & 94 & 96 & 96 & 95 & 97 \\
  Risaralda & 92 & 96 & 96 & 97 & 96 & 96 & 97 & 96 & 96 & 97 & 97 & 96 & 97 \\
  Santander & 93 & 91 & 90 & 92 & 91 & 90 & 92 & 91 & 90 & 92 & 91 & 90 & 92 \\
  Sucre & 85 & 92 & 91 & 93 & 92 & 91 & 93 & 92 & 91 & 93 & 92 & 91 & 94 \\
  Tolima & 90 & 89 & 87 & 91 & 89 & 87 & 91 & 89 & 87 & 91 & 90 & 88 & 91 \\
  Valle del Cauca & 94 & 96 & 95 & 97 & 96 & 95 & 97 & 96 & 95 & 97 & 96 & 96 & 97 \\
  Arauca & 85 & 96 & 96 & 97 & 96 & 96 & 97 & 96 & 96 & 97 & 96 & 96 & 97 \\
  Casanare & 87 & 93 & 91 & 94 & 93 & 91 & 94 & 93 & 91 & 94 & 93 & 92 & 94 \\
  Putumayo & 76 & 81 & 79 & 84 & 81 & 78 & 84 & 81 & 78 & 83 & 80 & 77 & 82 \\
San Andr\'es & 92 & 96 & 96 & 97 & 96 & 96 & 97 & 96 & 96 & 97 & 96 & 96 & 97 \\
    \rowcolor[gray]{0.8}
  Amazonas & 62 & 72 & 68 & 76 & 72 & 68 & 75 & 71 & 67 & 75 & 70 & 66 & 73 \\
 Guain\'ia  & 49 & 84 & 82 & 87 & 84 & 81 & 86 & 84 & 81 & 86 & 84 & 81 & 86 \\
  Guaviare & 82 & 94 & 93 & 95 & 94 & 93 & 95 & 94 & 93 & 95 & 94 & 93 & 95 \\
    \rowcolor[gray]{0.8}
 Vaup\'es & 54 & 57 & 52 & 62 & 56 & 51 & 61 & 55 & 51 & 60 & 53 & 49 & 57 \\
  Vichada & 59 & 80 & 77 & 83 & 79 & 76 & 82 & 79 & 76 & 82 & 79 & 76 & 82 \\
  \rowcolor[gray]{0.8}
  National & 89 & 91 & 90 & 92 & 91 & 90 & 92 & 91 & 90 & 92 & 91 & 90 & 92 \\ \hline
  \end{tabular}
\caption{\small
Posterior predictive estimates and 95\% credible intervals under Gamma, HS, LA and Student-t local priors for the random effects and  a Half-Cauchy local prior for the scale of the errors and  the observed estimates from
the census. Inf
and Sup represent the 2.5\%  and 97.5\% quantiles from the  marginal posterior distribution
of the completeness in each Department and at the national level. Males and model 1. }
\label{tab:tab_measures5}
\end{table}

\clearpage

\begin{table}[ht]
\tiny
\renewcommand{\arraystretch}{1.5}
\begin{tabular}{r>{\columncolor[gray]{0.85}}rrrrrrrrrrrrrr}
\hline
  Department & Census  & Gamma &  $2.5\%$ &  $97.5\%$ & Student-t &  $2.5\%$ &  $97.5\%$ & LA &  $2.5\%$ &  $97.5\%$ & HS &  $2.5\%$ &  $97.5\%$  \\
 \hline
 Antioquia & 92 & 94 & 93 & 95 & 94 & 93 & 95 & 94 & 93 & 95 & 94 & 93 & 95 \\
Atl\'antico & 95 & 90 & 89 & 92 & 90 & 89 & 92 & 90 & 89 & 92 & 91 & 89 & 92 \\
Bogot\'a, D.C. & 95 & 89 & 87 & 90 & 89 & 87 & 90 & 89 & 87 & 90 & 89 & 88 & 90 \\
 Bol\'ivar & 87 & 85 & 83 & 87 & 85 & 83 & 87 & 85 & 83 & 87 & 85 & 83 & 87 \\
Boyac\'a & 92 & 91 & 90 & 93 & 91 & 90 & 93 & 91 & 90 & 93 & 92 & 90 & 93 \\
 Caldas & 92 & 97 & 96 & 97 & 97 & 96 & 97 & 97 & 96 & 98 & 97 & 97 & 98 \\
 Caquet\'a & 87 & 92 & 91 & 93 & 92 & 91 & 93 & 92 & 91 & 93 & 92 & 91 & 93 \\
 Cauca & 81 & 77 & 73 & 80 & 77 & 73 & 80 & 77 & 73 & 80 & 77 & 73 & 80 \\
 Cesar & 87 & 77 & 74 & 80 & 77 & 74 & 79 & 77 & 74 & 79 & 77 & 74 & 79 \\
C\'ordoba & 87 & 78 & 75 & 81 & 78 & 75 & 81 & 78 & 75 & 81 & 79 & 76 & 82 \\
  Cundinamarca & 93 & 93 & 92 & 94 & 93 & 92 & 94 & 93 & 93 & 94 & 94 & 93 & 95 \\
   \rowcolor[gray]{0.8}
 Choc\'o & 64 & 47 & 42 & 52 & 47 & 42 & 51 & 46 & 42 & 51 & 46 & 41 & 51 \\
 Huila & 92 & 90 & 88 & 91 & 90 & 88 & 91 & 90 & 88 & 91 & 90 & 88 & 91 \\
  \rowcolor[gray]{0.8}
 La Guajira & 40 & 41 & 36 & 45 & 40 & 36 & 45 & 40 & 35 & 44 & 39 & 35 & 44 \\
  Magdalena & 86 & 84 & 82 & 86 & 84 & 82 & 86 & 84 & 82 & 86 & 84 & 82 & 86 \\
 Meta & 90 & 94 & 93 & 95 & 94 & 93 & 95 & 94 & 93 & 95 & 94 & 94 & 95 \\
 Nari\~no  & 79 & 82 & 79 & 84 & 82 & 79 & 84 & 82 & 79 & 84 & 82 & 79 & 85 \\
 Norte de Santander & 90 & 92 & 91 & 93 & 92 & 91 & 93 & 92 & 91 & 93 & 92 & 91 & 93 \\
  Quindio & 94 & 95 & 94 & 96 & 95 & 94 & 96 & 95 & 94 & 96 & 96 & 95 & 97 \\
   Risaralda & 92 & 96 & 95 & 96 & 96 & 95 & 96 & 96 & 95 & 96 & 96 & 95 & 97 \\
   Santander & 93 & 92 & 90 & 93 & 92 & 91 & 93 & 92 & 91 & 93 & 92 & 91 & 93 \\
   Sucre & 85 & 88 & 86 & 90 & 88 & 87 & 90 & 88 & 87 & 90 & 89 & 87 & 90 \\
   Tolima & 90 & 90 & 88 & 92 & 91 & 89 & 92 & 91 & 89 & 92 & 91 & 89 & 93 \\
   Valle del Cauca & 94 & 96 & 95 & 96 & 96 & 95 & 96 & 96 & 95 & 96 & 96 & 95 & 97 \\
   Arauca & 85 & 95 & 94 & 95 & 95 & 94 & 95 & 95 & 94 & 95 & 95 & 94 & 96 \\
   Casanare & 87 & 90 & 89 & 92 & 90 & 89 & 92 & 90 & 89 & 92 & 91 & 89 & 92 \\
   Putumayo & 76 & 77 & 74 & 80 & 77 & 74 & 79 & 77 & 74 & 79 & 77 & 74 & 79 \\
San Andr\'es & 92 & 95 & 95 & 96 & 95 & 95 & 96 & 95 & 95 & 96 & 95 & 95 & 96 \\
  \rowcolor[gray]{0.8}
   Amazonas & 62 & 60 & 56 & 64 & 60 & 56 & 64 & 60 & 56 & 64 & 59 & 56 & 63 \\
 Guain\'ia & 49 & 72 & 68 & 75 & 72 & 68 & 75 & 72 & 68 & 75 & 71 & 68 & 75 \\
   Guaviare & 82 & 93 & 92 & 94 & 93 & 92 & 94 & 93 & 92 & 94 & 93 & 92 & 94 \\
   \rowcolor[gray]{0.8}
 Vaup\'es & 54 & 53 & 49 & 58 & 53 & 48 & 57 & 52 & 48 & 57 & 52 & 47 & 56 \\
   Vichada & 59 & 69 & 66 & 73 & 69 & 65 & 73 & 69 & 65 & 73 & 69 & 65 & 73 \\
   \rowcolor[gray]{0.8}
   National & 89 & 89 & 88 & 91 & 89 & 88 & 91 & 89 & 88 & 91 & 90 & 88 & 91 \\    \hline
\end{tabular}
\caption{\small
Posterior predictive estimates and 95\% credible intervals under Gamma, HS, LA and Student-t local priors for the random effects and  a Half-Cauchy local prior for the scale of the errors and  the observed estimates from
the census. Inf
and Sup represent the 2.5\%  and 97.5\% quantiles from the  marginal posterior distribution
of the completeness in each Department and at the national level. Males and model 2. }
\label{tab:tab_measures6}
\end{table}



\end{document}